\documentclass[12pt]{article}

\usepackage{lmodern}
\usepackage{amssymb,amsmath,amsfonts,amsthm,geometry,graphicx}
\usepackage{color,setspace,sectsty,comment,natbib}
\usepackage{pdflscape,array,epstopdf,subcaption,booktabs,forest}
\usepackage{dirtytalk,accents,float,tikz,multirow,etoolbox,pdfpages}
\usepackage{marginnote,enumitem,rotating,fancyvrb}
\usepackage[normalem]{ulem}
\usepackage[parfill]{parskip}
\usepackage[font=small,labelfont=bf]{caption}
\usepackage[
    colorlinks = true,
    linkcolor  = blue,
    citecolor  = blue,
    urlcolor   = blue
]{hyperref}
\usepackage[ruled]{algorithm2e} 

\SetAlFnt{\small}
\SetAlCapFnt{\small}
\SetAlCapNameFnt{\small}
\SetAlCapHSkip{0pt}
\IncMargin{-\parindent}
\usepackage[hang,flushmargin]{footmisc}

\DeclareMathOperator*{\argmax}{argmax}

\DeclareMathOperator{\E}{\mathbb{E}}

\newtheorem{theorem}{Theorem}
\newtheorem*{theorem*}{Theorem}
\newtheorem{assumption}{Assumption}[section]
\newtheorem{proposition}{Proposition}

\newtheorem{lemma}{Lemma}[section]
\newtheorem{corollary}{Corollary}

\newtheorem{definition}{Definition}
\newtheorem{remark}{Remark}

\newcommand{\vvec}{\mathbf v}

\newcolumntype{L}[1]{>{\raggedright\let\newline\\arraybackslash\hspace{0pt}}m{#1}}
\newcolumntype{C}[1]{>{\centering\let\newline\\arraybackslash\hspace{0pt}}m{#1}}
\newcolumntype{R}[1]{>{\raggedleft\let\newline\\arraybackslash\hspace{0pt}}m{#1}}

\setcitestyle{authoryear}

\usepackage[T1]{fontenc}
\usepackage{mathtools,longtable,microtype,titletoc}
\newcommand{\Pp}{\mathbb P}
\newcommand{\R}{\mathbb R}
\newcommand{\ellb}{\ell_\beta}
\DeclareMathOperator{\logit}{logit}
\theoremstyle{plain}
\newtheorem{apptheorem}{Theorem}[section]
\newtheorem{appproposition}[apptheorem]{Proposition}
\newtheorem{applemma}[apptheorem]{Lemma}
\newtheorem{appcorollary}[apptheorem]{Corollary}
\theoremstyle{definition}
\newtheorem{appdefinition}[apptheorem]{Definition}
\newtheorem{appassumption}[apptheorem]{Assumption}
\theoremstyle{remark}
\newtheorem{appremark}[apptheorem]{Remark}
\theoremstyle{plain}
\hypersetup{pdftitle={Learning to Choose with Categories: Theory and an Application to Algorithmic Prominence},pdfauthor={Philippe Jehiel and Aviman Satpathy}}

\begin{document}

\begin{titlepage}
\makeatletter
\renewcommand{\thanks}[1]{%
  \stepcounter{footnote}%
  \protected@xdef\@thefnmark{\thefootnote}%
  \@makefnmark
  \protected@xdef\@thanks{\@thanks
    \protect\footnotetext[\the\c@footnote]{#1}}}
\makeatother

\title{Learning to Choose with Categories: \\ Theory and an Application to Algorithmic Prominence\thanks{We thank Evan Friedman, Johannes Hörner, Giacomo Lanzani, Margarita Pavlova, Ariel Rubinstein, Jean-Marc Tallon, Olivier Tercieux, and Giacomo Weber for valuable comments. The AI tools, Refine.ink and GPT 6 Astra, were used to review the paper for consistency, accuracy, and clarity. All remaining errors are our own. This work has benefited from the support of the Agence Nationale de la Recherche in France through the program Investissements d’Avenir ANR-17-EURE-0001.}}
\author{Philippe Jehiel\thanks{Paris School of Economics \& University College London; \href{mailto:jehiel@enpc.fr}{jehiel@enpc.fr}} \and Aviman Satpathy\thanks{Paris School of Economics; \href{mailto:aviman.satpathy@psemail.eu}{aviman.satpathy@psemail.eu}}}
\date{November 9, 2026 \\ \small{\href{https://drive.google.com/file/d/1dMmRfR6TUBI4rosUwU-t5WKMOiXYDh5p/view?usp=share_link}{Click here for the latest version}}}
\maketitle

\thispagestyle{empty} 
\setcounter{page}{0} 

\begin{abstract}
\noindent We introduce Coarse Q-learning, in which a decision-maker partitions alternatives into coarse similarity classes, chooses among them by logit, and learns class valuations by averaging experienced payoffs. Even in stationary bandit problems with independently drawn menus, choices change the composition of pooled feedback. This endogenous sampling can generate a unique globally attracting mixed equilibrium with limiting indifference, multiple stable strict rankings, or persistent cycling. Learning can reduce welfare even when the coarse representation admits efficient choice. An application demonstrates how digital platforms may exploit this feedback by promoting high-quality products to sustain purchases of inferior products under a common brand. \medskip

\noindent\textbf{Keywords:} Coarse Inference, Misspecified Learning, Reinforcement Learning \\
\noindent\textbf{JEL Codes:} C62, C73, D83, D91
\medskip
\end{abstract}

\end{titlepage}
\pagebreak
\newpage

\section{Introduction}
\label{sec:Intro}

Consider a consumer, Alice, who has recently moved to Paris and dines out regularly. To discover her taste in wine, she orders a bottle to share with friends on each occasion. Facing a vast and unfamiliar selection across restaurants, Alice simplifies her choice by relying on wine categories based on region of origin, such as Bordeaux or Chianti, and price range. She typically favors the category she currently values the most, occasionally explores others, and uses each bottle's realized experience to update her assessment of its category. We ask how she will compare these categories after accumulating sufficient experience.

One might expect Alice to eventually learn each category's average quality, much as a bandit learner discovers an arm's expected payoff. This intuition, however, overlooks how her choices determine the composition of that average. Restaurants differ both in the categories they offer and in the wines representing each category. When Alice comes to favor Bordeaux, she orders it more often in restaurants where Chianti is also available, but her choice is unchanged where Bordeaux is the only available category. Relative to Chianti, her Bordeaux observations therefore place more weight on wines offered alongside Chianti. Even though every wine has a stationary payoff distribution and restaurant menus are drawn independently from a fixed distribution, the mixture of wines she samples within a category changes with her valuations. A similarity class is generally not a bandit arm with a fixed mixture distribution as the choices based on its assessment change the evidence used to assess it.

We study this endogenous sampling feedback mechanism in a stationary bandit problem with independently drawn menus and independent payoff shocks. Alternatives are partitioned into \textit{similarity classes}, and the decision-maker maintains one valuation per class. She chooses among available classes using the logit rule, selects an alternative within the chosen class by a fixed tie-break, and updates the class valuation toward its realized payoff. Harmonic visit-count learning rates make this valuation a sample average of observed payoffs. We call this dynamic process \textit{Coarse Q-learning} (CQL). It combines familiar reinforcement learning \citep{Skinner,Caplin2008,Glimcher2011} with categorization \citep{rosch1978principles}. Memory and sample constraints motivate tracking a few categories rather than many alternatives; limited salience provides a natural basis for organizing those categories around attributes such as region and price \citep{Tversky,Shleifer2012,Shleifer2013}.

With a separate class for every alternative, learning recovers the alternatives' expected payoffs in the long-run, and increasingly sensitive choice selects the best available options. Coarsening changes this conclusion because valuations and the distributions they estimate are jointly determined. The resulting feedback can produce robust indifference, several stable strict rankings, or persistent fluctuations. These outcomes require neither non-stationary fundamentals nor strategic interaction: they arise from pooling heterogeneous observations whose composition depends on the learner's own choices. Throughout, sensitivity to valuation differences is held fixed during learning; we then study long-run outcomes as sensitivity increases, so that mistakes induced solely by logit exploration vanish in the limit.

Our central result identifies conditions under which experience leads to \textit{indifference}. Return to Bordeaux and Chianti, and suppose, for simplicity, that the wines offered on their own are equally rewarding on average and sufficiently better than those offered when the two categories compete. This may reflect the expertise of restaurants specializing in a particular wine region, whose carefully curated offerings compare favorably with those of establishments catering to a broader range of tastes. A preference for Bordeaux then increases its exposure to less favorable competitive experiences, while Chianti's assessment relies more heavily on favorable experiences where it is the only option. This feedback corrects the initial preference. In the high-sensitivity limit, a strict ranking cannot sustain itself: the category losing every comparison would be assessed through its favorable singleton experiences, undermining its lower valuation. Learning instead sustains indifference and nontrivial mixing between categories, even when their representatives in the competitive menu have different expected payoffs. The result extends beyond this illustration. 

With singleton menus for every class and a connected pattern of competition, a sufficiently large common increase in singleton payoffs yields a unique equilibrium to which learning converges from every initial assessment. As sensitivity increases, its limiting policy continues to mix in at least one competitive menu. Indifference thus arises for open sets of primitives, rather than for equal objective payoffs. Nor is the resulting randomization simply residual exploration - the composition of observations must remain balanced to sustain the equal limiting valuations. More generally, we identify a self-correcting feedback condition under which stronger competition from a rival weakly improves a class's sampled payoff, ensuring global convergence to a unique equilibrium, even in the absence of singleton menus.

Reversing the payoff comparison can reverse this feedback. If singleton experiences are sufficiently unfavorable, a class favored in competitive menus receives better observations, while a neglected class is assessed mainly through its disappointing singleton experiences. Under suitably rich menu support, several strict rankings become self-confirming and locally attracting from nearby valuations. Initial assessments and realized experiences can then have persistent path-dependent effects on choice. With three or more classes, and at least one class yielding heterogeneous expected payoffs across competitive menus, corrective and reinforcing feedback can also interact to prevent convergence of learning altogether. We construct an open set of decision problems with only three classes in which, at every sufficiently high sensitivity, the unique equilibrium is unstable and our stochastic learning almost surely has a non-constant periodic limit set. Persistent cycling can therefore survive stationary fundamentals, sustained exploration, and diminishing learning rates.

To establish these results, we derive the mean-field dynamics of harmonic visit-count coarse Q-learning in continuous time using stochastic approximation. Their steady states, \textit{smooth valuation equilibria} (SVE), exist at every finite sensitivity, and their high-sensitivity limits refine the valuation equilibria (VE) of \citet{Jehiel2007}. Using o-minimal geometry, we show that this refinement is generically finite even when the VE set contains continua. We also characterize precisely which strict VE survive vanishing logit exploration and show that they are approximated by locally stable SVE at all sufficiently high sensitivities.\footnote{A strict VE survives the refinement if and only if each class's valuation equals its average payoff across the menus that minimize its valuation disadvantage to competing chosen classes, weighted by their probabilities.}

The welfare analysis separates the cost of coarse representation from the additional cost of coarse learning. Our benchmark is a fully informed oracle using the same categories, within-class choice rule, and common class scores across menus. Its optimal policy is generically unique and can always be implemented by a strict ranking of classes. When each menu has a unique best class, a common ranking selects it in every menu if and only if the required comparisons are acyclic. Even when this efficient policy is representable, learning need not attain it. Singleton payoffs make the distinction transparent as they contribute the same unavoidable payoff to every ranking, yet change the ranking or mixture that learning sustains. In the wine illustrations considered below, Bordeaux is better when both categories are available, but learning can either favor Chianti or mix between them indefinitely in the high-sensitivity limit. Reweighting feedback flexibly across menus and histories cannot generally close this welfare gap. We establish an impossibility result for payoff-assessment learning by constructing a family of decision problems in which any coarse learner who eventually assesses homogeneous classes correctly and uses the same monotone stochastic choice rule across menus suffers a strictly positive welfare loss relative to the oracle in the long-run. 

Later in the paper, we develop an application that shows how a digital platform that sells its own products along with third-party ones can exploit this feedback through algorithmic prominence \citep{LeeMusolff2026,Fradkin}. Consumers pool experiences across products of unknown, heterogeneous quality sold under the platform's common in-house brand label, while a competitor offers constant quality. Prominence removes the competitor from consideration in selected product states. Under complete information, promoting the high-quality product cannot increase sales, since consumers already prefer it to the competitor. Under coarse learning, however, the platform’s optimal deterministic policy can promote only the high-quality product, with its favorable feedback sustaining purchases of inferior products under the same brand, even when the competitor remains available. The resulting disappointments offset the favorable promoted experiences, leaving consumers indifferent between brands. In our illustrations, this policy diverts sales from the competitor while preserving consumer welfare at its no-prominence level. Yet consumers could achieve a higher payoff using the same categories by buying the promoted high-quality product and otherwise choosing the competitor, a ranking that payoff-based learning cannot sustain. We also show how randomized prominence can further raise platform profits at consumers’ expense, while prominence costs can make deterministic policies optimal.

Section~\ref{sec:Model} presents the learning model and its mean-field dynamics. Section~\ref{sec:results} develops the equilibrium, learning, and welfare results, with illustrations of their mechanisms. Section~\ref{sec:platform-learning} studies algorithmic prominence. Sections~\ref{sec:discuss} and~\ref{sec:litreview} discuss Bayesian foundations, interpretations, and related literature. Section~\ref{sec:conclude} concludes. All proofs are in the appendix.

\section{Model}
\label{sec:Model}

\subsection{Primitives}
\label{sec:primitives}

Let \(\mathcal A\) be a finite set of alternatives and \(\Psi:=\{\psi\subseteq\mathcal A:\psi\neq\varnothing\} \) the set of menus. Time is discrete, \(k\in\mathbb N_0\). In each period \(k\), Nature draws a menu \(\psi_k\) independently according to an exogenous distribution \(f\in\Delta(\Psi)\). Alice chooses an alternative \(a\in\psi_k\) and observes a reward \(r_{a,k}\). Each alternative \(a\) has a payoff distribution \(R_a\in\Delta(\mathbb R)\), which does not depend on the menu in which it appears. For each \(a\), the sequence \(\{r_{a,k}\}_{k\geq0}\) is i.i.d. with law \(R_a\), and these sequences are independent across alternatives and independent of the menu process. We assume uniform square-integrability: \( \sup_{a\in\mathcal A}\mathbb E\!\left[r_{a,k}^2\right] \leq M_2<\infty, \) and write \(\mu_a:=\mathbb E[r_{a,k}]\) for the expected payoff of alternative \(a\). Alice is risk-neutral and does not know the payoff distributions. She learns about expected payoffs from the rewards she observes.\footnote{A standard interpretation is that each alternative $a$ has a deterministic \emph{fundamental} payoff $\mu_a$, and the observed payoff is a noisy signal $r_{a,k}=\mu_a+\varepsilon_{a,k}$, where $\{\varepsilon_{a,k}\}_{k\ge 0}$ are i.i.d.\ mean-zero shocks. The stochasticity of payoffs may reflect either imperfect monitoring of the fundamentals or idiosyncratic taste shocks.}

Alice represents the alternatives through an exogenously fixed partition \(\mathcal S\) into \textit{similarity classes}. Let
\(
\rho:\mathcal A\to\mathcal S,
\) with 
\( \rho(a)=s\ \text{if }a\in s,
\) be the similarity projection assigning each alternative to its class. Alice represents an alternative only through its class and maintains one payoff estimate per class to guide her choices. This coarse representation economizes on both memory and data. When alternatives are numerous, maintaining a separate payoff estimate for each places substantial demands on memory, while learning about rarely sampled alternatives rests on few observations. Organizing learning and choice around a smaller set of categories reduces the number of estimates Alice must retain and allows a broader set of experiences within a class to inform the payoff estimate, in line with the principles of categorization studied in psychology \citep{rosch1978principles,anderson1991adaptive}.

Attribute salience provides a natural basis for organizing these categories. Alternatives typically differ along many payoff-relevant dimensions, only some of which attract Alice’s attention. A consumer, for example, may distinguish wines by region and price while paying little attention to vintage or producer. We assume a fixed attention mechanism: an attribute is either salient or not. Alice groups together alternatives that coincide on the salient dimensions and learns a common expected payoff for each class. She treats the payoffs sampled from a class as i.i.d.\ draws conditional on a common but unknown stationary distribution. In reality, her classes need not reflect objective payoff homogeneity, because the salient attributes need not coincide with the attributes that are most predictive of payoffs.

\subsection{Discrete-time Learning Model}
\label{sec:dynamics}

Let $\mathcal T=(\mathcal A,\Psi,f,r)$ denote the stage decision problem defined by the primitives above. At the beginning of period $k$, Alice has a valuation $v_k(s)\in\mathbb R$ for each similarity class $s\in\mathcal S$, representing her current estimate of its expected payoff. Write $\mathbf v_k=(v_k(s))_{s\in\mathcal S}$ for the valuation vector, with $\mathbf v_0$ specifying her initial assessments. Alice perceives the realized menu $\psi_k$ through the set of represented classes, \( \omega_k:=\rho[\psi_k] =\{\rho(a):a\in\psi_k\}. \) The history of observed class menus, chosen classes, and realized payoffs generate her information filtration $\mathcal F_k=\sigma\!\operatorname{-alg}\bigl((\omega_t,s_t,r_t)_{t<k}\bigr)$ at the beginning of period $k$. Her initial assessments $\mathbf v_0$ and this history determine her current valuations $\mathbf v_k$. Given $\mathbf v_k$ and the realized class menu $\omega_k$, she chooses among the available classes according to the multinomial logit (softmax) rule \citep{Luce1959}, commonly used in reinforcement learning models and logistic fictitious play \citep{Sutton1998,fudenberg1998theory,hofbauer2002global}:
\[
\sigma(s\mid\omega_k;\mathbf v_k)
=
\frac{\exp \bigl(\beta v_k(s)\bigr)}
{\displaystyle\sum_{j\in\omega_k}\exp \bigl(\beta v_k(j)\bigr)},
\qquad s\in\omega_k,
\]
with $\sigma(s\mid\omega_k;\mathbf v_k)=0$ for unavailable classes. The parameter $\beta\in[0,\infty)$ measures Alice's sensitivity to valuation differences and remains fixed throughout the learning process. At $\beta=0$, she chooses uniformly among the available classes. Higher sensitivity places greater weight on classes with higher valuations, although every available class retains positive probability at any finite $\beta$. Holding valuations fixed, choice probabilities concentrate on the highest-valued classes as $\beta\to\infty$. We are particularly interested in analyzing the long-run behavior of the learning process for a decision-maker with sufficiently large sensitivity $\beta<\infty$.

Conditional on the chosen class $s_k$, an alternative $a_k$ is drawn exogenously and uniformly from $s_k\cap\psi_k$.\footnote{Uniform tie-breaking uniquely obeys label-independence among alternatives within the chosen class. More generally, we allow for any fixed exogenous tie-breaking kernel $\Upsilon(\cdot\mid s,\psi)$ supported on $s\cap\psi$.}
Thus, Alice's valuations determine the probability assigned to each available class, which is then divided equally among its available members.\footnote{At fixed valuations, duplicating an alternative within a class leaves the perceived class menu and class-choice probabilities unchanged. It only redistributes probability within that class, avoiding the usual issue of probability inflation associated with IIA in the \say{duplicates problem} of \citet{Debreu1960}.}  Denoting the cardinality of a set by $|\cdot|$, the resulting probability of choosing alternative $a\in\psi_k$ is
\begin{equation}\label{eq:choice}
\nu_k(a\mid\psi_k,\mathbf v_k)
:=
\Pr(a_k=a\mid\psi_k,\mathbf v_k)
=
\frac{\sigma(\rho(a)\mid\omega_k;\mathbf v_k)}
{|\rho(a)\cap\psi_k|}.
\end{equation}
After the choice, Alice observes the realized payoff $r_k:=r_{a_k,k}$, but receives no counterfactual feedback about unchosen alternatives. She updates the valuation of the chosen class according to the temporal-difference (TD) recursion:
\begin{equation}\label{eq:update}
v_{k+1}(s)
=
v_k(s)
+
\mathbf 1\{s_k=s\}\,\alpha_k(s)
\bigl(r_k-v_k(s)\bigr),
\qquad s\in\mathcal S.
\end{equation}
The reward prediction error $r_k-v_k(s_k)$ measures the difference between the observed payoff and Alice's current assessment of the chosen class. A positive prediction error raises its valuation, while a negative one lowers it. The valuations of unchosen classes remain unchanged. This recursion is the bandit variant of asynchronous Q-learning \citep{Watkins1992,Sutton1998}, with valuations maintained at the coarser similarity class level.\footnote{In the appendix, we allow for vanishing UCB-style optimism bonuses that encourage finite-sample exploration and show our asymptotic analysis is robust to such modifications of the learning rule. Also, we omit the bootstrapping term usually considered in MDPs since our i.i.d.\ menu process is action-independent.}
We use harmonic learning rates indexed by class visit counts. Let
\begin{equation}\label{eq:visit-count-stepsize}
N_k(s):=\sum_{t=0}^{k-1}\mathbf 1\{s_t=s\},
\qquad
\alpha_k(s):=\frac{1}{\kappa_s+N_k(s)+1},
\qquad \kappa_s>0,
\end{equation}
where $N_k(s)$ counts the observations of class $s$ collected before period $k$. The parameter $\kappa_s$ determines the weight attached to Alice's initial assessment. The learning rate declines as she accumulates observations of that class: periods in which the class is not chosen leave the weight on its next observation unchanged. Under the maintained assumptions, valuations remain bounded almost surely and every class that appears with positive probability is sampled infinitely often. Consequently, the effective step-sizes satisfy the usual \citet{Robbins-et-al} conditions,
$\sum_k\mathbf 1\{s_k=s\}\alpha_k(s)=\infty$
and
$\sum_k\mathbf 1\{s_k=s\}\alpha_k(s)^2<\infty,$
almost surely. The boundedness and sampling arguments are provided in the appendix.

This choice of harmonic learning rates gives valuations a direct interpretation as averages of past observations (see Section~\ref{sec:conjugate} for a Bayesian foundation). Iterating \eqref{eq:update} yields
\begin{equation}\label{eq:class-sample-mean}
v_k(s)
=
\frac{
\kappa_s v_0(s)
+
\sum_{t=0}^{k-1}\mathbf 1\{s_t=s\}r_t
}{
\kappa_s+N_k(s)
}.
\end{equation}
Every observed payoff from class $s$ receives the same weight in its current valuation, while the initial assessment carries the weight of $\kappa_s$ observations.\footnote{$\kappa_s$ can be interpreted as a prior pseudo-count: a larger value makes Alice's initial assessment more resistant to revision, although its weight vanishes as observations accumulate. Under suitable conjugate priors, \eqref{eq:class-sample-mean} coincides with the Bayesian posterior mean of a class's expected payoff under Alice's misspecified (subjectively) stationary model in the one-parameter natural exponential family; see Sec.~\ref{sec:conjugate} for details.}

The averaging rule pools observations whose composition is shaped by Alice's choices. Different menus may offer different representatives of the same class, and her current valuations determine how frequently she selects that class in each menu. The conditional distribution of payoffs entering its sample average can therefore change as valuations evolve, even though each objective alternative has a fixed payoff distribution. Valuations influence the evidence used to revise them, creating the feedback at the heart of our coarse learning dynamics.

Together, \eqref{eq:choice}--\eqref{eq:visit-count-stepsize}, the initial valuations $\mathbf v_0$, and the pseudo-counts $(\kappa_s)_{s\in\mathcal S}$ specify the learning process in the given environment. The next subsection develops a continuous-time approximation that tracks both valuations and class visit frequencies to study its long-run behavior. With harmonic visit-count step sizes, Alice's response to a new observation depends on how often she has sampled the corresponding class. To describe the evolution of her valuations, we therefore keep track of the class sampling frequencies as well. Let
\(
\chi:=\sum_{s\in\mathcal S}\kappa_s,\) and \[
x_{k,s}:=\dfrac{\kappa_s+N_k(s)}{\chi+k},
\ s\in\mathcal S.
\]
The vector $\mathbf x_k$ records the empirical frequencies of class choices, including the initial observation weights $\{\kappa_s\}_{s\in\mathcal S}$. Its coordinates are strictly positive and sum to one, and the contribution of the initial weights vanishes as observations accumulate.\footnote{For our analysis, we consider and retain only the classes that are available in at least one supported menu.}

\subsection{Continuous-time Asymptotic Approximation}
\label{sec:continuousmodel}

To describe the expected motion of this process, fix a valuation vector $\vvec$. The probability that Alice chooses class $s$ in a period is
\(
\Lambda_s(\vvec)
:=
\sum_{\psi:\,s\cap\psi\neq\varnothing}
f(\psi)\,\sigma(s\mid\rho[\psi];\vvec).
\)
Conditional on choosing $s$ from menu $\omega = \rho[\psi]$, uniform within-class selection yields expected payoff
\(
\mu_s(\psi)
:=
\dfrac{1}{|s\cap\psi|}
\sum_{a\in s\cap\psi}\mu_a.
\)
Averaging over the menus from which Alice chooses $s$ yields
\[
g_s(\vvec)
:=
\frac{
\displaystyle
\sum_{\psi:\,s\cap\psi\neq\varnothing}
f(\psi)\,\sigma(s\mid\rho[\psi];\vvec)\,\mu_s(\psi)
}{
\displaystyle
\sum_{\psi:\,s\cap\psi\neq\varnothing}
f(\psi)\,\sigma(s\mid\rho[\psi];\vvec)
}.
\]
$g_s(\vvec)$ is the expected payoff observed when class $s$ is chosen under valuations $\vvec$. The difference $g_s(\vvec)-v_s$ is Alice's expected prediction error for that class. A positive expected prediction error raises its valuation, while a negative one lowers it. After rescaling time to account for the diminishing learning rates, the mean-field asynchronous dynamics are
\begin{equation}
\begin{aligned}
\dot v_s
&=
\frac{\Lambda_s(\vvec)}{x_s}
\bigl(g_s(\vvec)-v_s\bigr),\\
\dot x_s
&=
\Lambda_s(\vvec)-x_s,
\end{aligned}
\qquad s\in\mathcal S.
\label{eq:differential}
\end{equation}
Here, $\mathbf x$ is a strictly positive probability vector, and the derivatives are taken with respect to learning time. The two equations reflect the distinction between current behavior and accumulated experience. Class $s$ is currently chosen with probability $\Lambda_s(\vvec)$, while its learning rate depends on its historical sampling frequency $x_s$. At a sufficiently late date $k$, a new observation of $s$ receives weight approximately $1/[(\chi+k)x_s]$. Combining this weight with the probability of observing $s$ yields the factor $\Lambda_s(\vvec)/x_s$ in the valuation equation. 

For a given expected prediction error, valuations adjust faster when a class is currently sampled more often than its historical frequency would suggest. The frequency equation describes how that history itself changes: $x_s$ rises when the current probability of choosing $s$ exceeds its accumulated share, and falls when the inequality is reversed. Diminishing learning rates also make the effect of payoff and sampling noise small over fixed intervals of learning time. This connection between the stochastic recursion and the ODE can be stated precisely using standard stochastic-approximation results \citep{kushner2003}.

\begin{proposition}[Continuous-time approximation]
\label{prop:continuous-approximation}
Under the maintained assumptions, fix $\beta\in[0,\infty)$, positive initial weights
$\{\kappa_s\}_{s\in\mathcal S}$, and a finite initial valuation vector $\vvec_0$. Then, \( \sup_{k\geq 0}\|\vvec_k\|<\infty,\) and \( \liminf_{k\to\infty}\min_{s\in\mathcal S}x_{k,s}>0 \) almost surely. Let $Z(t)$ be the piecewise-linear interpolation of
$\{(\vvec_k,\mathbf x_k)\}_{k\geq 0}$ at times $\{\tau_k\}_{k\geq 0}$, and let $\Phi_h$ denote the flow of \eqref{eq:differential}. Then $Z$ is an asymptotic pseudo-trajectory of this flow: almost surely, for every $T>0$,
\[
\lim_{t\to\infty}
\sup_{0\leq h\leq T}
\bigl\|Z(t+h)-\Phi_h(Z(t))\bigr\|
=0.
\]
Consequently, the set of subsequential limits of $\{(\vvec_k,\mathbf x_k)\}_{k\geq 0}$ is almost surely nonempty, compact, connected, and internally chain transitive
for the asynchronous joint ODE \eqref{eq:differential}.\footnote{The proof is provided in the appendix. Internal chain transitivity means that any two points in the invariant set can be connected by a finite chain of orbit segments lying in the set, with arbitrarily small jumps between segments and an arbitrarily large lower bound on their durations; see \citet{Benaim1999}.}
\end{proposition}

The first conclusion establishes boundedness of the valuation process and continued sampling of every class. The approximation statement then gives a precise sense in which the ODE describes learning: starting sufficiently late, the interpolated process follows the ODE from its current state arbitrarily closely over any fixed interval of learning time.\footnote{The corresponding clock is
$\tau_0=0$ and
\(
\tau_k
=
\sum_{j=0}^{k-1}\dfrac{1}{\chi+j+1}.
\)
Thus, $\tau_k\sim\log k$: a fixed interval of learning time contains progressively more observations as $k$ grows.} Its long-run outcomes are therefore constrained by the recurrent behavior of the joint ODE system.

Menu variation affects the evidence pooled within a class in two ways. First, different menus may contain different representatives of the same class. Their expected payoffs may generally differ, so the mean payoff from choosing a class can vary across menus, even though each alternative has a fixed payoff distribution.\footnote{The same mean-field ODE construction accommodates menu-dependent primitive payoff distributions $R_{a,\psi}$, under the corresponding conditional independence and uniform second-moment assumptions. Writing $\mu_a(\psi)$ for their means, one replaces $\mu_s(\psi)$ by $|s\cap\psi|^{-1}\sum_{a\in s\cap\psi}\mu_a(\psi)$. We maintain menu-independent primitive payoffs to isolate the menu dependence generated by coarse aggregation.} Second, Alice’s valuations determine how often she selects that class in each menu through competition with the other available classes. The mixture of menus contributing to its payoff history is therefore selected by her own choices. Current valuations shape the composition of the evidence she observes, and this evidence in turn shapes future valuations. This feedback makes learning \textit{active} and \textit{endogenous}.

\textbf{Reduction of the tree to class menus:}
The mean-field dynamics admit a simpler tree representation that groups together menus containing the same similarity classes. Let
\(
\Omega:=\{\omega\subseteq\mathcal S:\omega\neq\varnothing\}, \) 
and \( \Psi_\omega:=\{\psi\in\Psi:\rho[\psi]=\omega\}. \)
The probability of encountering class menu $\omega$ is
\(
p(\omega):=\sum_{\psi\in\Psi_\omega}f(\psi).
\)
For each $\omega$ with $p(\omega)>0$ and each $s\in\omega$, define
\[
\pi_s(\omega)
\;=\;
\frac{\displaystyle \sum_{\psi\in\Psi_{\omega}} f(\psi)\,\mu_s(\psi)}
     {\displaystyle \sum_{\psi\in\Psi_{\omega}} f(\psi)}\,,
\qquad
\mu_s(\psi)\;:=\;\frac{1}{|s\cap\psi|}\sum_{a\in s\cap\psi}\mu_a.
\]
This is the expected payoff from choosing class $s$ when the perceived menu is $\omega$. To see why the averaging weights depend only on menu probabilities, recall that for a fixed \(\vvec\) Alice uses the same class-choice probabilities in every objective menu inducing $\omega$. Conditional on choosing $s$, these menus therefore retain their relative occurrence probabilities.

Writing
$\sigma_\omega^s(\vvec):=\sigma(s\mid\omega;\vvec)$,
the class-selection probabilities and conditional expected payoffs become
\(\Lambda_s(\vvec)
=
\sum_{\substack{\omega\in\Omega:\,s\in\omega}}
p(\omega)\,\sigma_\omega^s(\vvec),\) and 
\[g_s(\vvec)
=
\frac{
\displaystyle
\sum_{\substack{\omega\in\Omega:\,s\in\omega}}
p(\omega)\,\sigma_\omega^s(\vvec)\,\pi_s(\omega)
}{
\displaystyle
\sum_{\substack{\omega\in\Omega:\,s\in\omega}}
p(\omega)\,\sigma_\omega^s(\vvec)
}.
\]
The reduced decision tree
$\mathcal T'=(\mathcal S,\Omega,p,\pi)$
thus preserves both $\Lambda$ and $g$, and hence the joint mean-field dynamics in \eqref{eq:differential}. Figure~\ref{fig:bob-trees} illustrates the reduction. We use this class-menu representation throughout the analysis.

\begin{figure}[ht]
  \centering
  \begin{subfigure}{0.65\linewidth}
    \centering
    \resizebox{\linewidth}{!}{%
      \begin{tikzpicture}[scale=2.2]
      \tikzstyle{solid node}=[circle,draw,inner sep=1.2,fill=black]
      \tikzstyle{hollow node}=[circle,draw,inner sep=1.2]
      \tikzstyle{level 1}=[level distance=10mm,sibling distance=3.5cm]
      \tikzstyle{level 2}=[level distance=12mm,sibling distance=1.5cm]
      \tikzstyle{level 3}=[level distance=15mm,sibling distance=1cm]
      \node(0)[solid node,label=above:{$r$}]{}
      child{node[solid node,label=above left:{$\psi_1$}]{}
        child{node[hollow node,label=below:{$(3)$}]{} edge from parent node[left]{\textcolor{red}{Bordeaux (2019)}}}
        child{node[hollow node,label=below:{$(0)$}]{} edge from parent node[right]{\textcolor{blue}{Chianti (2014)}}}
        edge from parent node[left,xshift=5, yshift = 15]{$\dfrac{3}{10}$}
      }
      child{node[solid node,label=above right:{$\psi_2$}]{}
        child{node[hollow node,label=below:{$(2)$}]{} edge from parent node[left]{\textcolor{blue}{Chianti (2016)}}}
        child{node[hollow node,label=below:{$(0)$}]{} edge from parent node[right]{\textcolor{blue}{Chianti (2014)}}}
        edge from parent node[right, yshift = 10] {$\dfrac{6}{10}$}
      }
      child{node[solid node,label=above right:{$\psi_3$}]{}
        child{node[hollow node,label=below:{$(2)$}]{} edge from parent node[left]{\textcolor{blue}{Chianti (2016)}}}
        child{node[hollow node,label=below:{$(1)$}]{} edge from parent node[right]{\textcolor{red}{Bordeaux (2013)}}}
        edge from parent node[right,xshift=-5, yshift = 15]{$\dfrac{1}{10}$}
      };
      \end{tikzpicture}
    }
    \caption{Alice's original decision tree $\mathcal{T}$}
    \label{fig:original-tree}
  \end{subfigure}\hfill
  \begin{subfigure}{0.35\linewidth}
    \centering
    \resizebox{\linewidth}{!}{%
      \begin{tikzpicture}[scale=2.1]
      \tikzstyle{solid node}=[circle,draw,inner sep=1.2,fill=black]
      \tikzstyle{hollow node}=[circle,draw,inner sep=1.2]
      \tikzstyle{level 1}=[level distance=10mm,sibling distance=3.5cm]
      \tikzstyle{level 2}=[level distance=12mm,sibling distance=1.5cm]
      \tikzstyle{level 3}=[level distance=15mm,sibling distance=1cm]
      \node(0)[solid node,label=above:{$r'$}]{}
      child{node[solid node,label=above left:{$\omega_1$}]{}
        child{node[hollow node,label=below:{$(2.5)$}]{} edge from parent node[left]{\textcolor{red}{Bordeaux}}}
        child{node[hollow node,label=below:{$(0.5)$}]{} edge from parent node[right]{\textcolor{blue}{Chianti}}}
        edge from parent node[left,xshift=5, yshift = 15]{$\dfrac{4}{10}$}
      }
      child{node[solid node,label=above right:{$\omega_2$}]{}
        child{node[hollow node,label=below:{$(1)$}]{} edge from parent node[right]{\textcolor{blue}{Chianti}}}
        edge from parent node[right, yshift = 15] {$\dfrac{6}{10}$}
      };
      \end{tikzpicture}
    }
    \caption{Alice's reduced decision tree $\mathcal{T'}$}
    \label{fig:reduced-tree}
  \end{subfigure}
  \caption{Example of tree-reduction for the mean-field CQL dynamics}
  \label{fig:bob-trees}
\end{figure}

\textbf{Elementary benchmarks:} The following four special benchmarks eliminate the feedback from valuations to expected class payoffs. In each case, $g(\vvec)$ is a constant vector $\bar{\boldsymbol\pi}$, and the mean-field dynamics~\eqref{eq:differential} trivially converge to these valuations from every initial condition. The stochastic learning process~\eqref{eq:update} likewise satisfies
$\vvec_k\to\bar{\boldsymbol\pi}$ almost surely.\footnote{For the joint asynchronous dynamics,
\(
\frac{d}{dt}\bigl[x_s(t)(v_s(t)-\bar\pi_s)\bigr]
=
-x_s(t)(v_s(t)-\bar\pi_s).
\)
Bounded valuations and positive class availability ensure that sampling frequencies eventually stay away from zero. Hence $\vvec(t)\to\bar{\boldsymbol\pi}$ and
$\mathbf x(t)\to\Lambda(\bar{\boldsymbol\pi})$. The almost-sure convergence of the stochastic process follows from Proposition~\ref{prop:continuous-approximation}.}

\begin{enumerate}
\item[\textbf{(i)}]
\emph{Pure exploration:}
When $\beta=0$, Alice chooses uniformly among the available classes, independently of her valuations. The expected payoff observed from class $s$ is therefore a constant \(
\bar\pi_s
=
\dfrac{
\sum_{\omega\ni s}
p(\omega)|\omega|^{-1}\,\pi_s(\omega)}
{
\sum_{\omega\ni s}p(\omega)|\omega|^{-1}
}.
\)
Menu composition affects this average, but Alice's valuations do not affect which menus contribute to it.

\item[\textbf{(ii)}]
\emph{Finest partition:}
When every similarity class contains a single alternative, Alice maintains a separate valuation for each alternative. Every observation entering the valuation of $\rho(a):=a$ then has conditional mean $\mu_a$, so
$g_{\rho(a)}(\vvec)=\mu_a$.
Learning recovers each alternative's expected payoff asymptotically as in standard Q-learning over alternatives.

\item[\textbf{(iii)}]
\emph{Constant class availability:}
Suppose every menu contains a representative of every class, so that $p(\mathcal S)=1$. Then
\(
g_s(\vvec)
=
(\sigma_{\mathcal S}^s(\vvec)\,\pi_s(\mathcal S))/
(\sigma_{\mathcal S}^s(\vvec))
=
\pi_s(\mathcal S).
\)
Although different alternatives may appear across physical menus, the set of competing classes is always the same. Alice therefore chooses a given class with the same probability in every menu, leaving the composition of its observed payoffs independent of valuations.

\item[\textbf{(iv)}]
\emph{Homogeneous class means:}
The finest-partition benchmark extends to any partition satisfying
$\mu_a=\bar\pi_s$ for every $a\in s$.
In that case, $\mu_s(\psi)=\bar\pi_s$ whenever class $s$ is available, and hence $g_s(\vvec)=\bar\pi_s$. Changes in the composition of sampled alternatives leave expected class payoffs unchanged as they are constant across menus.
\end{enumerate}

These benchmarks identify special circumstances in which the composition of experience has no feedback effect on expected class payoffs. With heterogeneous alternatives and varying sets of competing classes, Alice's valuations determine her choices which can in turn change the distribution of expected rewards entering her class estimates. This feedback raises two questions: which valuations are consistent with the choices they induce, and when does learning converge to such valuations? We call these self-consistent steady states \emph{smooth valuation equilibria (SVE)} and turn next to their structure and stability.

\section{Results}
\label{sec:results}

We begin by analyzing the structure of SVE, their high-sensitivity limits LSVE, and the local stability of strict limiting equilibria. We then examine global convergence, the effects of singleton-payoff shifts on indifference and multiplicity, and persistent cycling. Finally, we distinguish the welfare costs of coarse representation and learning. Examples throughout illustrate the feedback between choice and experience driving these outcomes. 

\subsection{Equilibrium Structure}
\label{sec:SVE}

For equilibrium and local stability analysis, it is convenient to introduce the \emph{normalized mean-field dynamics}
\begin{equation}\label{eq:normalized}
\dot v_s=g_s(\vvec)-v_s,
\qquad s\in\mathcal S.
\end{equation}
Under these simpler dynamics, each valuation adjusts at unit rate toward the expected payoff generated by the current choice policy. The endogenous composition of observed payoffs is captured by the same map $g$, while adjustment rates are synchronized across classes. Its steady states correspond one-to-one to the steady states of the asynchronous joint dynamics~\eqref{eq:differential}: a valuation vector $\vvec^*$ is stationary if and only if $\vvec^*=g(\vvec^*)$, and the corresponding visit frequencies are uniquely determined by $\mathbf x^*=\Lambda(\vvec^*)$. The two systems also agree on the local asymptotic stability of corresponding hyperbolic steady states. At corresponding steady states, the Jacobian of the asynchronous joint dynamics~\eqref{eq:differential} has the eigenvalues of
$Dg(\vvec^*)-I$, together with $|\mathcal S|-1$ additional eigenvalues equal to $-1$. Hence hyperbolicity and, for hyperbolic steady states, local asymptotic stability coincide across the two systems.\footnote{A steady state is hyperbolic (hence, regular) if its Jacobian has no eigenvalue with zero real part. Additionally, it is locally asymptotically stable if every such eigenvalue has a negative real part.} The appendix formally establishes both the stated correspondences. 

Thus, at a steady state, Alice's valuation of each class equals the expected payoff she receives conditional on choosing it. Her choices determine which observations enter this assessment, while the assessment itself guides her choices. A smooth valuation equilibrium is a fixed point that requires these two elements to be mutually consistent.

\begin{definition}[Smooth valuation equilibrium]
\label{def:SVE}
Fix a reduced decision tree $\mathcal T'=(\mathcal S,\Omega,p,\pi)$
and sensitivity $\beta\in[0,\infty)$. A pair $(\vvec,\sigma)$ is a
\textbf{smooth valuation equilibrium} (SVE) if:
\[
\sigma_\omega^s
=
\mathbf{1}\{s\in\omega\}\frac{\exp(\beta v_s)}{\displaystyle\sum_{r\in\omega}\exp(\beta v_r)},
\qquad
v_s
=
\frac{
\displaystyle\sum_{\omega\ni s}
p(\omega)\,\sigma_\omega^s\,\pi_s(\omega)
}{
\displaystyle\sum_{\omega\ni s}
p(\omega)\,\sigma_\omega^s
}.
\]
The first equality applies to every supported menu $\omega$, and the second to every $s\in\mathcal S$.
\end{definition}
The first condition specifies logit-optimal choices given valuations. The second, consistency, requires each class valuation to equal the expected payoff in the sample generated by those choices. Thus, Alice correctly assesses the average payoff of each class under her own sampling behavior. The equilibrium condition allows the alternatives within a class to have different expected payoffs: their contributions to its valuation depend endogenously on how frequently they are encountered and chosen. Equivalently, equilibrium valuations satisfy $\vvec^*=g(\vvec^*;\beta)$, with the associated choice
profile uniquely determined by the logit rule. 
At every finite sensitivity, each available class is
chosen with positive probability. Consequently, its
equilibrium valuation is an average of its expected
payoffs across all supported menus in which it appears,
with strictly positive weight on each such menu.
Writing
\[
\Omega_s^+
:=
\{\omega\in\operatorname{supp}(p):s\in\omega\},
\qquad
K
:=
\prod_{s\in\mathcal S}
\left[
\min_{\omega\in\Omega_s^+}\pi_s(\omega),
\max_{\omega\in\Omega_s^+}\pi_s(\omega)
\right],
\]
every SVE valuation lies in the relative interior of $K$. We denote the set of equilibrium valuations by
\(
\mathcal V(\beta)
:=
\{\vvec\in K:\vvec=g(\vvec;\beta)\}.
\)

\textbf{High-sensitivity limits of SVE:} To understand equilibrium behavior when choices respond strongly to valuation differences, we study SVE as $\beta\to\infty$, holding the decision problem fixed. A vector $\vvec^\infty$ is a \emph{limiting smooth valuation equilibrium} (LSVE) if there exist $\beta_n\to\infty$ and $\vvec_n\in\mathcal V(\beta_n)$ such that $\vvec_n\to\vvec^\infty$. We write $\mathcal V(\infty)$ for the set of all such limits.\footnote{This limit concerns sensitivity across learning environments. Within each environment, $\beta$ remains fixed along each learning path for all forward time.} The natural benchmark is the $\arg\max$ (or greedy) counterpart of SVE, called a \emph{valuation equilibrium} (VE), following the analogous solution concept introduced by \citet{Jehiel2007} for extensive-form games. A VE consists of a valuation vector and a choice profile such that Alice chooses only among valuation-maximizing classes in each menu, allowing menu-specific mixing over ties. Each class chosen with positive aggregate probability is valued consistently at its expected payoff conditional on being chosen. For classes that are never chosen, payoff consistency imposes no additional restriction, although their valuations remain subject to the menu-wise optimality conditions (see the appendix for a formal definition). Let $\mathcal{VE}$ denote the set of valuation vectors that can be supported in this way.

Theorem~\ref{th:correspondence} proves the existence of SVE, describes how their structure varies with sensitivity, and relates their high-sensitivity limits to VE. In particular, while every accumulation point of the SVE correspondence as $\beta\to\infty$ is a VE; generically, the converse need not hold.

\begin{theorem}\label{th:correspondence}
Fix a decision tree $(\mathcal S,\Omega,p,\pi)$. For $\beta\ge 0$, let
\( \mathcal V(\beta):=\{\vvec\in K:\ \vvec=g(\vvec;\beta)\} \) denote the SVE set, and let $\mathcal V(\infty)$ denote the set of accumulation points of $\{\mathcal V(\beta)\}_{\beta\uparrow\infty}$.  Then:

\begin{enumerate}
\item[(i)] For every $\beta< \infty$, $\mathcal V(\beta)$ is nonempty and compact, and the correspondence $\beta\mapsto\mathcal V(\beta)$ is upper hemicontinuous and compact-valued on $[0,\infty)$. The SVE graph $\{(\vvec,\beta):\vvec\in\mathcal V(\beta)\}$ has finitely many path-connected components. For generic parameters \((p,\pi)\) and all but finitely many $\beta\in[0,\infty)$, $\mathcal V(\beta)$ is finite and consists of an odd number
of hyperbolic equilibria, each lying on a locally unique
real-analytic branch in $\beta$.\footnote{Throughout, genericity is understood with
respect to the numerical parameters of a fixed reduced
tree structure: we hold $\mathcal S$, $\Omega$, and
$\operatorname{supp}(p)$ fixed and vary the positive
menu probabilities and the class-menu expected payoffs.
The resulting parameter space carries the relative
Euclidean topology, with probabilities constrained to
sum to one, and Lebesgue measure on its affine hull.
A property holds \emph{generically} if it holds on a
comeager subset of this space that also has full measure. Valuation vectors and sensitivities carry
their usual Euclidean topologies, and the equilibrium
graph carries the corresponding product topology.}

\item[(ii)] Every accumulation point of SVE as $\beta\to\infty$ is a VE: $\mathcal V(\infty)\subseteq \mathcal{VE}$. Equivalently, for every $\varepsilon>0$, $\exists$ $\widehat\beta<\infty$ such that $\forall\ \beta\ge\widehat\beta$, every $\mathbf v\in\mathcal V(\beta)$ lies in an $\varepsilon$-neighborhood of some $\mathbf v^\infty\in\mathcal {VE}$. For generic \((p,\pi)\), $\mathcal V(\infty)$ is a nonempty finite subset of $\mathcal{VE}$.\footnote{The appendix also shows that the SVE connected component containing the unique equilibrium at $\beta=0$ contains an equilibrium at every finite sensitivity and, generically, has a unique high-sensitivity limit.}
\end{enumerate}
\end{theorem}

Theorem~\ref{th:correspondence} establishes that smooth valuation equilibrium is well-defined and places substantial restrictions on its multiplicity. Equilibria exist at every nonnegative sensitivity and, for generic decision trees, are hyperbolic and finite in number at all but finitely many sensitivities. More strikingly, their high-sensitivity limits form a nonempty finite subset of VE.

The finiteness of the limiting set is substantive because the underlying set of VE need not be finite, even for generic primitives. Optimal choice leaves the valuation of a never-chosen class unrestricted by payoff consistency. An underdetermined VE can also arise when every class is chosen: ties permit menu-specific mixing, and these mixing probabilities can adjust together with valuations while preserving both indifference and payoff consistency. Such continua can persist under small perturbations of the primitive payoffs and probabilities.\footnote{For example, consider three equally likely menus $\{A,B\}$, $\{B,C\}$, and $\{C,A\}$, with payoffs $(1,0)$ to the classes listed in each menu. For every $t\in(0,1)$, choosing the first-listed class with probability $t$ in each menu generates conditional mean payoff $t$ for every class. Thus $(t,t,t)$ is a VE valuation vector. The consistency equations are regular in the three mixing probabilities, so a local continuum of such valuation vectors persists under small perturbations of the primitives. The unique LSVE here is $(1/2,1/2,1/2)$.} 

The proof sidesteps this difficulty by studying the entire SVE correspondence as sensitivity varies. We show that the SVE graph in $(\vvec,\beta)$-space is definable in the o-minimal structure $\mathbb R_{\exp}$ (real ordered field with the exponential function) since its equilibrium equations involve only algebraic operations and exponentials. Generic regularity via standard transversality arguments makes this graph one-dimensional; o-minimality then ensures finitely many path-connected components and,
after compactifying sensitivity, finitely many boundary points
corresponding to $\beta=\infty$. This yields a generically finite refinement even when VE form a continuum, retaining precisely those equilibrium valuations that persist as limits of SVE under vanishing logit exploration. This refinement also has an axiomatic foundation.\footnote{On a rich domain of menus, positivity, continuity, label- and context-independence, IIA, translation invariance, and monotonicity
characterize the logit family up to nonnegative sensitivity \citep{breitmoser}.}

We next study the stability and convergence of the CQL dynamics.\footnote{Local asymptotic stability requires trajectories starting sufficiently close to remain close and converge to the equilibrium. Global asymptotic stability additionally requires convergence from every initial condition.} We assess local stability using the eigenvalues of the Jacobian of~\eqref{eq:differential}, and establish global  convergence through strict Lyapunov functions or contraction arguments in the $\sup$ norm.

\subsection{Strict Equilibria: Refinement and Local Stability}
\label{sec:strictLSVE}

Theorem~\ref{th:correspondence} establishes that every LSVE (high-sensitivity limit of SVE) is a VE, but the converse need not hold. We now characterize which strict VE survive this refinement. The additional restriction concerns classes that Alice never chooses under the limiting optimal policy. Although VE leaves their valuations unrestricted, these classes continue to be sampled at every finite sensitivity. Their valuations must therefore remain consistent with the observations generated by this
increasingly rare exploration as $\beta\to\infty$.

\begin{definition}[Strict VE]
\label{def:StVE}
A VE valuation $\vvec^*$ is strict if: (i) every supported menu $\omega$ has a unique valuation-maximizing class chosen in equilibrium, denoted by \( s^*(\omega):=\arg\max_{j\in\omega}v_j^*; \) and (ii) the payoff-consistent valuations induce a strict total order on the set of classes chosen in at least one supported
menu, \(
\mathcal W(\vvec^*)
:=
\{s^*(\omega):\omega\in\operatorname{supp}(p)\}.
\)
\end{definition}
Thus, distinct classes in $\mathcal W(\vvec^*)$ have distinct valuations, while no ordering among never-chosen classes is required. We refer to condition (i) alone as \emph{menu-wise strictness}. For generic decision trees, menu-wise strictness already implies condition (ii).\footnote{For any fixed pure choice profile, VE consistency determines each chosen class's valuation as a conditional payoff average. Equality between two such valuations is a nongeneric restriction on the payoffs. Since there are finitely many pure choice profiles, the conclusion holds simultaneously across all menu-wise strict VE.} 

To understand how the remaining classes are assessed, consider the valuation loss associated with choosing class $s$ in menu $\omega$:
\(
\Delta_s(\omega;\vvec^*)
:=
v_{s^*(\omega)}^*-v_s^*,
\ \omega\in\Omega_s^+.
\)
Under logit choice, a larger loss makes selection exponentially less likely as sensitivity increases. Consequently, conditional on choosing a class that is sub-optimal in every menu, Alice's observations exponentially concentrate on the menus in which its valuation disadvantage is the smallest. For the class's valuation to survive vanishing logit exploration, it must equal the $p$-weighted average payoff across these \emph{closest-loss menus} denoted by
\(
\Omega_s^\star(\vvec^*)
:=
\arg\min_{\omega\in\Omega_s^+}
\Delta_s(\omega;\vvec^*).
\)
For a class chosen in some menu, the minimum loss is zero,
and its closest-loss menus are precisely those in
which it is optimal. The next result shows that closest-loss consistency exactly
characterizes which strict VE are LSVE, and that each such LSVE is approached
by a locally unique, locally stable family of SVE
defined at every sufficiently large sensitivity.

\begin{theorem}[Characterization and local stability]
\label{th:strictpurestable}
Let $\vvec^*$ be a strict VE valuation. Then:
\begin{enumerate}
\item[(i)] The valuation vector $\vvec^*$ is an LSVE if and only if
it satisfies closest-loss consistency:
\begin{equation}
\label{eq:closest-loss-consistency}
v_s^*
=
\frac{
\displaystyle
\sum_{\omega\in\Omega_s^\star(\vvec^*)}
p(\omega)\pi_s(\omega)
}{
\displaystyle
\sum_{\omega\in\Omega_s^\star(\vvec^*)}
p(\omega)
},
\qquad s\in\mathcal S.
\end{equation}
\item[(ii)] Whenever \eqref{eq:closest-loss-consistency} holds, there exist a neighborhood $U$ of $\vvec^*$ and $\widehat\beta<\infty$ such that, $\forall\ \beta\geq\widehat\beta$, $U$ contains a unique SVE,
denoted by $\vvec^\beta$, and
\(
\vvec^\beta\rightarrow\vvec^*
\ \text{as }\beta\to\infty.
\)
This SVE is locally exponentially stable in the normalized dynamics~\eqref{eq:normalized} and the corresponding equilibrium
\(
\bigl(\vvec^\beta,\Lambda(\vvec^\beta;\beta)\bigr)
\)
is locally exponentially stable under the asynchronous dynamics~\eqref{eq:differential}.
\end{enumerate}
\end{theorem}

Thus, closest-loss consistency identifies which strict patterns of choice can be supported in equilibrium by valuations learned from the observations that those choices generate. This sampling structure also explains local stability of nearby SVE for sufficiently large sensitivity. Near a strict equilibrium, small valuation perturbations preserve the optimal class in each menu and thus the set of closest-loss menus for each class. Hence, at high sensitivity, they have little effect on the composition of an optimal class's observations, with nearly fixed relative weights. Even for a rarely chosen class, whose overall sampling frequency can respond sharply to valuation shocks, its conditional expected payoff remains almost unchanged as the relative weights of its closest-loss menus remain nearly invariant. This  allows the mean-reverting update rule~\eqref{eq:normalized} to correct valuation errors by pulling assessments back toward these nearly fixed payoff targets. This yields local attraction under the normalized dynamics and,
by the local-stability correspondence, under the harmonic
visit-count dynamics~\eqref{eq:differential}.

Therefore, if a decision tree admits a strict VE that satisfies closest-loss consistency \eqref{eq:closest-loss-consistency}, the set of asymptotically stable SVE for sufficiently high sensitivity in the CQL dynamics is nonempty. In the appendix, we show that every locally asymptotically stable SVE of the joint asynchronous dynamics~\eqref{eq:differential} has strictly positive probability of being the long-run outcome of the stochastic CQL process \eqref{eq:update}, from any finite initial valuations and positive prior counts. Thus, asymptotic stability serves as a dynamic equilibrium selection criterion. 

As a special case that shall be useful later, suppose $\vvec^*$ is a menu-wise strict VE at which every class wins in at least one supported menu. Each class's closest-loss menus are precisely those it wins, so ordinary VE consistency implies closest-loss consistency. The appendix shows that even if classes that win in different menus have equal valuations, $\vvec^*$ is the limit of a locally unique branch of locally exponentially stable SVE defined for all sufficiently large $\beta$, as in Theorem~\ref{th:strictpurestable}(ii). This applies in particular when every singleton menu is supported. 

\subsection{Convergence with Constant Competitive-Menu Payoffs}
\label{sec:competitiveconstant}

We now establish the global convergence of the asynchronous CQL dynamics~\eqref{eq:differential} to SVE in all decision trees with at most two similarity classes. This insight follows from the more general result stated below, which allows arbitrarily many classes provided each class’s expected payoff is constant across the set of competitive menus in which it appears. Formally, suppose that, for each
class $s$, there exists a constant $b_s\in\mathbb R$ such that
\begin{equation}
\label{eq:constant-competitive-payoffs}
\pi_s(\omega)=b_s
\qquad
\text{for every }\omega\in\Omega_s^+
\text{ with }|\omega|\geq2.
\end{equation}
Thus, a class $s$ yields the same expected payoff whenever it competes with another class. Its expected payoff in the singleton menu $\{s\}$, when supported, is unrestricted - allowing its valuation to eventually lie between $\pi_s\{s\}$ and $b_s$.\footnote{We call a menu \emph{competitive} if it contains at least two similarity classes. A singleton menu may nevertheless contain several alternatives: it is the available class, rather than the available alternative, that is unique.} 
While condition~\eqref{eq:constant-competitive-payoffs} holds trivially
with two classes, it also holds for any number of classes whenever
every competitive menu features at least one alternative from each class, so that $\mathcal S$ is the only non-singleton menu in the reduced tree. This condition preserves endogenous sampling bias. As a class becomes more attractive, Alice chooses it more frequently in competitive menus, increasing the weight of those observations relative to its singleton observations. We show this restriction  ensures global convergence of learning while allowing for both multiplicity of equilibria and indifference with endogenous mixing in the high-sensitivity limit.
\begin{theorem}[]
\label{th:competitiveconstant}
Suppose \eqref{eq:constant-competitive-payoffs} holds. For any sensitivity $\beta\in[0,\infty)$, every solution of the joint asynchronous dynamics~\eqref{eq:differential}, starting from $\vvec(0)\in\mathbb R^{\mathcal S}$ and $\mathbf x(0)\in\operatorname{ri}\Delta(\mathcal S)$, approaches the equilibrium set
\(
\mathcal E(\beta)
:=
\{(\vvec,\Lambda(\vvec;\beta)):
\vvec\in\mathcal V(\beta)\}.
\) Generically, since $\mathcal V(\beta)$ is finite for a.e.\ $\beta$, each trajectory converges to an individual hyperbolic equilibrium.\footnote{With two classes, the SVE set is finite at every sensitivity, yielding convergence to an isolated equilibrium.} 
\end{theorem}
Also, for fixed sensitivity and menu probabilities, convergence generically survives all sufficiently small perturbations of the constant-competitive-payoff benchmark~\eqref{eq:constant-competitive-payoffs}. 
This restriction makes competitive-menu observations of a class interchangeable in expectation. Alice's choices therefore affect its assessed payoff $g_s(\vvec)$ only through the balance between its singleton and competitive observations. Since singleton observations arrive at a fixed rate, the class's total sampling frequency determines their weight in its assessment. This dependence on aggregate sampling frequencies aligns choice adjustments across menus with a common potential whose payoff component has conditional mean payoffs as marginal returns and whose entropy component captures logit randomization. At a given logit policy, a marginal shift toward a class with a larger expected prediction error raises the potential. The proof combines this alignment between choice and payoff feedback with a nontrivial smoothing argument based on harmonic averaging to establish convergence toward the SVE set.

The potential need not be concave in general, allowing for multiple stable SVE and therefore path dependence. However, if each class’s singleton payoff weakly exceeds its competitive payoff on average, competitive observations dilute favorable singleton feedback, creating diminishing returns to sampling. The potential is then strictly concave in choice probabilities, yielding a unique, globally attracting SVE, as illustrated below. 

The following two-class illustrations based on Fig.~\ref{fig:multiple-tree} demonstrate how convergence can coexist with path dependence or produce indifference in the high-sensitivity limit. Consider two classes, Bordeaux ($B$) and Chianti ($C$).
The menus $\{B,C\}$, $\{B\}$, and $\{C\}$ occur with equal probability.
The wines offered in the competitive menu yield expected payoffs
$2$ and $1$, respectively, while those offered in either singleton
menu yield a common $z \in \mathbb R$. Changing $z$ reverses the feedback between choices and the evidence
used to evaluate each class.\footnote{At an SVE with competitive-menu
choice probability $q$, the normalized Jacobian has eigenvalues
\(-1\) and \(
-1+\beta q(1-q)
\left[
\dfrac{2-z}{(1+q)^2}
+\dfrac{1-z}{(2-q)^2}
\right].
\)
The joint asynchronous Jacobian has one additional eigenvalue $-1$.} Writing
$q=(1+e^{-\beta(v_B-v_C)})^{-1}$ for the probability of choosing
Bordeaux in the binary menu, the conditional payoff means are
\(
g_B(q)=\dfrac{2q+z}{1+q},
\) and\ \(
g_C(q)=\dfrac{1-q+z}{2-q}.
\)

\begin{figure}[ht]
    \centering
    \begin{tikzpicture}[scale=0.735, font=\footnotesize]
    \tikzstyle{solid node}=[circle,draw,inner sep=1.2,fill=black]
    \tikzstyle{hollow node}=[circle,draw,inner sep=1.2]
    \tikzstyle{level 1}=[level distance=15mm,sibling distance=3.5cm]
    \tikzstyle{level 2}=[level distance=15mm,sibling distance=1.5cm]
    \tikzstyle{level 3}=[level distance=15mm,sibling distance=1cm]
    \node(0)[solid node,label=above:{$r$}]{}
    child{node[solid node,label=above left:{$\omega_1$}]{}
    child{node[hollow node,label=below:{$2$}]{} edge from parent node[left]{\textcolor{blue}{$B_1$}}}
    child{node[hollow node,label=below:{$1$}]{} edge from parent node[right]{\textcolor{blue}{$C_1$}}}
    edge from parent node[left,xshift=5, yshift = 10]{$1/3$}
    }
    child{node[solid node,label=above left:{$\omega_2$}]{}
    child{node[hollow node,label=below:{$z$}]{} edge from parent node[left]{\textcolor{blue}{$B_2$}}}
    edge from parent node[right, yshift = 0] {$$1/3$$}
    }
    child{node[solid node,label=above right:{$\omega_3$}]{}
    child{node[hollow node,label=below:{$z$}]{} edge from parent node[left]{\textcolor{blue}{$C_3$}}}
    edge from parent node[right,xshift=-5, yshift = 10] {$1/3$}
    };
    \end{tikzpicture}
    \caption{Example of a Decision Tree with Two Similarity Classes}
    \label{fig:multiple-tree}
\end{figure}

\begin{figure}
    \centering
    \includegraphics[width=0.9\linewidth]{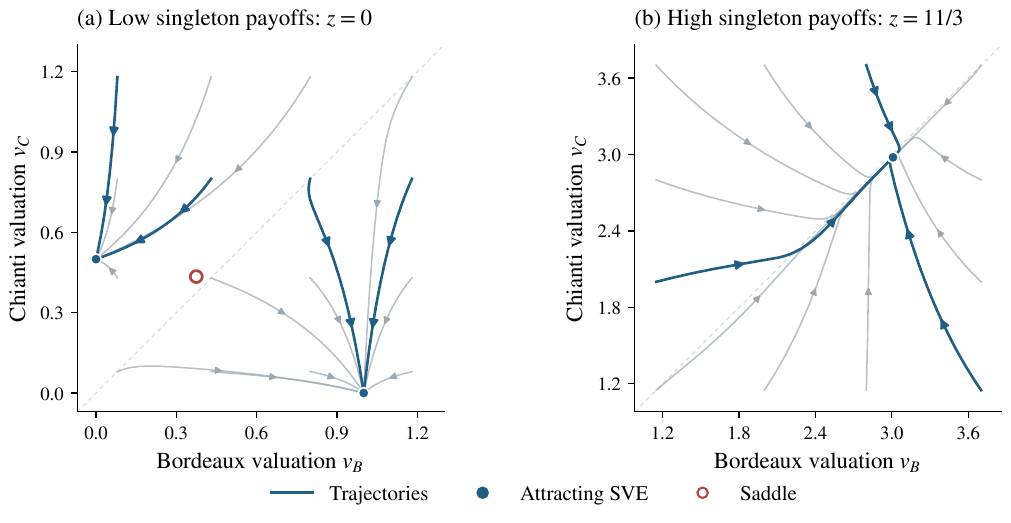}
    \caption{Valuation projections of the asynchronous joint
    dynamics at $\beta=20$, initialized with $x_B=x_C=1/2$.
    (a) Low singleton payoffs, $z=0$: two attracting SVE and a saddle.
    (b) High singleton payoffs, $z=11/3$: a unique globally attracting SVE.}
\label{fig:phase}
\end{figure}

\textbf{Path dependence:} Let $z=0$. Choosing a class more frequently in the competitive menu increases the share of favorable observations in its payoff history, reinforcing its valuation advantage. For sufficiently large $\beta$, there are two locally attracting SVEs approaching the strict VEs $(v_B,v_C)=(1,0)$ with Bordeaux chosen in the binary menu and $(0,\tfrac12)$ with Chianti chosen in the binary menu; both sets of eigenvalues converge to $-1$. A third mixed SVE is a saddle, with one positive eigenvalue. Thus initial assessments can lead to different long-run rankings, even though Bordeaux yields the higher payoff whenever Alice chooses between the two classes. Figure~\ref{fig:phase}(a) illustrates these distinct basins of attraction.

\textbf{Limiting indifference:} Now let $z=11/3$. More frequent selection in the competitive menu exposes a class disproportionately to unfavorable payoffs below its singleton benchmark, correcting and shrinking its valuation advantage. The SVE is unique for every $\beta\ge0$, and all eigenvalues are negative. Theorem~\ref{th:competitiveconstant} therefore gives global convergence under both dynamics, as illustrated in Figure~\ref{fig:phase}(b). As $\beta\to\infty$,
\(
(v_B,v_C)\rightarrow(3,3),
\) and \(
q\rightarrow2/3.
\)
Alice consequently becomes indifferent between Bordeaux and Chianti despite
their unequal competitive-menu payoffs. This limiting indifference is sustained by endogenous mixing with Bordeaux chosen twice as often
as Chianti when both are available.

\subsection{Self-Correcting Feedback and Global Convergence}
\label{sec:negative-covariance}

We now identify a condition under which endogenous sampling leads to a unique long-run outcome, even with heterogeneous competitive-menu payoffs. This mechanism concerns which observations used in the assessment of a class get displaced when a competing class becomes more attractive. If these observations have relatively low payoffs on average, the residual observations improve the assessment of the class losing choice share. This response counteracts the initial change in relative attractiveness and stabilizes valuations and choices.

To express this condition, fix a sensitivity $\beta\ge0$ and let
\(
w_s^{\vvec}(\omega)
:=
p(\omega)\sigma_\omega^s(\vvec)/ \Lambda_s(\vvec),
\ \omega\in\Omega_s^+,
\)
denote the distribution of menus conditional on choosing class $s$. These are the weights with which observations from different menus enter its assessed expected payoff $g_s(\vvec)$.  For another class $r\neq s$, the probability $\sigma_\omega^r(\vvec)$ measures the exposure of an $s$-observation from menu $\omega$ to crowding by $r$, with $\sigma_\omega^r(\vvec)=0$ when $r\notin\omega$. An increase in $v_r$ relative to $s$ proportionally reduces the probability of choosing $s$ most acutely in menus where $r$ already has a high choice probability. Simply differentiating the conditional expected payoff map w.r.t.\ $v_r$ gives
\begin{equation}
\begin{aligned}
\frac{\partial g_s(\vvec)}{\partial v_r}
=
-\beta\sum_{\omega\in\Omega_s^+}
w_s^{\vvec}(\omega)
\bigl[\pi_s(\omega)-g_s(\vvec)\bigr]
\sigma_\omega^r(\vvec)
=
-\beta\,
\operatorname{Cov}_{w_s^{\vvec}}
\left(\pi_s(\omega),\sigma_\omega^r(\vvec)\right),
\ r\neq s.
\end{aligned}
\label{eq:payoff-crowding-identity}
\end{equation}
Nonpositive payoff-crowding covariance implies a rise in $v_r$ marginally displaces observations of class $s$ that are no better, on average, than those currently used to assess it. Their removal weakly raises $g_s(\vvec)$, providing a corrective response to the initial drop in the valuation of $s$ relative to $r$. If this condition holds for each pair of classes over all $\vvec\in \operatorname{ri} K$, we obtain the following result on global convergence to a unique SVE with arbitrarily many classes.

\begin{theorem}[Global convergence under non-positive payoff-crowding covariance]
\label{th:negative-covariance}
Fix a tree $\mathcal T'$ and $\beta\in[0,\infty)$. Suppose that, for every $\vvec\in K$ and every pair of distinct classes $s,r\in\mathcal S$,
\begin{equation}
\operatorname{Cov}_{w_s^{\vvec}}
\left(
\pi_s(\omega),\sigma_\omega^r(\vvec)
\right)
\leq 0.
\label{eq:negative-payoff-crowding}
\end{equation}
Then there is a unique SVE $\vvec^\beta$, and the corresponding equilibrium
\(
\bigl(\vvec^\beta,\Lambda(\vvec^\beta)\bigr)
\)
is globally asymptotically stable under the joint asynchronous dynamics~\eqref{eq:differential}, for arbitrary finite initial valuations and positive initial sampling frequencies.
\end{theorem}

The covariance condition makes selection self-correcting: raising a rival's valuation weakly improves a class's assessed payoff, while raising its own valuation weakly lowers it.\footnote{This mechanism resembles the \emph{self-defeating externalities} of stable games \citep{HOFBAUER2009}. Their condition imposes monotonicity on the population payoff map; ours restricts how valuations change the composition of payoff observations. Neither condition implies the other.} Since a common shift in valuations leaves choice probabilities unchanged, $g(\vvec+c\mathbf1)=g(\vvec)$ and hence $Dg(\vvec)\mathbf1=0$. The normalized drift $F_s(\vvec):=g_s(\vvec)-v_s$ therefore satisfies
\[
\frac{\partial F_s}{\partial v_s}
=-1-\sum_{r\neq s}
\left|\frac{\partial g_s}{\partial v_r}\right|
=-1+\beta\sum_{r\neq s}
\operatorname{Cov}_{w_s^{\vvec}}
\bigl(\pi_s(\omega),\sigma_\omega^r(\vvec)\bigr)
\le -1,
\qquad \vvec\in K.
\]
The restoring own-valuation effect thus strictly dominates the combined cross effects, making the largest absolute valuation error shrink uniformly on $K$.\footnote{This argument immediately yields local stability and uniqueness of SVE (by Poincaré-Hopf index theorem).} This gives contraction for the normalized dynamics. The same error bound yields convergence for the asynchronous joint dynamics because their adjustment factors $\Lambda_s(\vvec)/x_s$ remain positive and bounded away from zero along each trajectory. Valuation convergence then implies $x\to\Lambda(\vvec^*)$ through $\dot x_s=\Lambda_s(\vvec)-x_s$. The argument also covers initial valuations outside the attracting, positively invariant box $K$. Proposition~\ref{prop:continuous-approximation} therefore gives
\(
(\vvec_k,\mathbf x_k)
\rightarrow
\bigl(\vvec^\beta,\Lambda(\vvec^\beta)\bigr)
\ \text{almost surely}
\)
under stochastic CQL~\eqref{eq:update}. We next give a simple condition on mean payoffs that ensures~\eqref{eq:negative-payoff-crowding}.

\subsection{Singleton Payoffs: Indifference and Multiplicity}
\label{sec:singleton-payoffs}

Singleton menus, though not necessary for the preceding result, provide a natural source of this corrective feedback. We now examine how their payoffs shape this feedback and, separately, whether limiting choices are strict or involve indifference. Singleton menus supply observations whose frequency is unaffected by competition. Whenever $\{s\}$ occurs, Alice samples class $s$ regardless of her valuations. Yet the weight of these observations in her assessment of $s$ depends on how frequently she samples it elsewhere. Observations obtained without competition between classes can thus shape subsequent choices in competitive menus. 

This role motivates interpreting singleton menus as \emph{anchor menus}: situations in which one class is effectively the only viable option. Indeed, richer menus containing a \textit{focal class} alongside strictly obviously dominated classes provide a similar source of observations at high sensitivity, when the focal class is selected with probability close to one.\footnote{A sufficient separation condition is that the highest expected payoff of each accompanying class, across all its supported menus, lies strictly below the focal class's lowest expected payoff across its supported menus. Equilibrium valuations inherit this gap, so the probability of selecting an accompanying class vanishes exponentially with sensitivity. The formal results retain singleton class menus.} For the remainder of this subsection, we maintain the following assumption.
\begin{assumption}
\label{as:singleton-support}
There are at least two similarity classes, every singleton menu has positive probability, and the co-occurrence graph is connected. This graph has vertex set $\mathcal S$ and an edge between distinct classes $s$ and $r$ whenever they appear together in a supported menu.
\end{assumption}
Starting from a fixed decision tree $\mathcal T'$, introduce a common singleton-payoff shift $z\in\mathbb R$:
\begin{equation}
\pi_s^z(\omega)
:=
\pi_s(\omega)+z\,\mathbf 1\{\omega=\{s\}\}.
\label{eq:singleton-shift}
\end{equation}
Positive shifts make singleton experiences more favorable, while negative shifts make them less favorable. Menu probabilities and competitive-menu expected payoffs remain fixed, as do differences between singleton expected payoffs across classes. We examine how this change in the evidence available to Alice affects learning and its long-run outcomes.

\subsubsection{High Singleton Payoffs and Limiting Indifference}
\label{sec:limiting-indifference}

Sufficiently favorable singleton experiences have two consequences. They generate the corrective feedback identified in the preceding subsection, ensuring convergence to a unique SVE. A sufficiently large common shift also makes choices that are strict in every menu inconsistent with the assessments they generate, forcing indifference as $\beta\to\infty$.

\begin{theorem}[]
\label{th:uniquemixedstable}
Under Assumption~\ref{as:singleton-support}, there is a finite threshold $\bar z\geq0$, depending on the underlying decision tree, such that, for every fixed $z>\bar z$:
\begin{enumerate}
\item[(i)]
For every $\beta<\infty$, there is a unique SVE $\vvec^\beta(z)$. The corresponding equilibrium $\bigl(\vvec^\beta(z),\Lambda(\vvec^\beta(z))\bigr)$
is globally asymptotically stable under the joint dynamics~\eqref{eq:differential}. 

\item[(ii)]
As $\beta\to\infty$, the valuations $\vvec^\beta(z)$ and their associated logit choice profiles converge to a VE
$\bigl(\sigma^\infty(z),\vvec^\infty(z)\bigr)$.
This limit exhibits indifference in at least one competitive menu:
\(
\left|
\argmax_{s\in\omega}v_s^\infty(z)
\right|\geq2
\
\text{for some }\omega\in\operatorname{supp}(p)
\text{ with }|\omega|\geq2.
\)
Moreover, $\sigma^\infty(z)$ assigns positive probability to at least two classes in some supported menu.
\end{enumerate}
\end{theorem}

The convergence conclusion follows from non-positive payoff-crowding covariance. Each class is selected whenever its singleton menu occurs, so singleton observations retain a share of its conditional payoff distribution bounded away from zero, uniformly over valuations and sensitivity. For sufficiently large $z$, these favorable observations raise its assessed payoff above every payoff it receives in competitive menus. A competing class therefore displaces relatively poor observations, improving the assessment of the class it crowds out. Theorem~\ref{th:negative-covariance} then gives uniqueness and convergence. The appendix provides the bounds establishing a threshold $\bar z$ that works for every finite sensitivity.

Limiting indifference requires a further high-sensitivity limit argument. Suppose every supported menu had a unique valuation maximizer. A lowest-valued class would then be selected only in its singleton menu, so its valuation would incorporate the entire shift $z$. A class beating it in a competitive menu would also receive observations from that menu, incorporating only a fraction of the singleton shift into its valuation. For sufficiently large $z$, the supposedly lower-valued class would consequently have the higher assessment. This contradicts the proposed strict choices, ruling out valuation equilibria with a unique maximizer in every menu. Thus, the sufficiently large common shift overturns every pattern of menu-wise strict choices. Heterogeneous primitive expected payoffs can therefore generate persistent indifference through the endogenous composition of experience (see Fig.~\ref{fig:phase}(b) for example).
Part (ii) shows that the unique SVE branch selects a unique limiting VE with nontrivial mixing in at least one competitive menu, even though infinitely many VE may coexist. A stronger full-support restriction on menu probabilities yields indifference across all classes.

\begin{assumption}[]
\label{assm:monotonicity}
All singleton menus are equally likely, and
\(
\omega\subseteq\omega'
\ \Longrightarrow\
p(\omega)\leq p(\omega').
\)
\end{assumption}
For instance, uniform probabilities over all nonempty menus satisfy this condition. Under Assumption~\ref{assm:monotonicity}, sufficiently high singleton payoffs force indifference across all classes.
\begin{proposition}[Total indifference]
\label{th:fullymixedstable}
Under Assumption~\ref{assm:monotonicity}, there is a finite threshold $\bar z_1$ such that, for every $z>\bar z_1$, every VE satisfies
\(
v_s=v_r
\ \text{for all }s,r\in\mathcal S.
\)
\end{proposition}

Any strict valuation gap would give the higher-valued class a larger sampling frequency and a smaller singleton share. A sufficiently large positive shift would then raise the lower-valued class's assessment more strongly, contradicting the gap. The LSVE selected by the unique branch of globally stable SVE is unique, fully-mixed and exhibits total indifference. Retaining inclusion monotonicity, this result  persists with nearly equal singleton probabilities.

However, high singleton payoffs are not necessary for
\eqref{eq:negative-payoff-crowding}. The appendix provides
additional primitive sufficient conditions, including examples
without singleton menus and one with intermediate singleton
payoffs that also sustains robust limiting indifference and mixing.

\subsubsection{Low Singleton Payoffs and Multiplicity of Strict VE}
\label{sec:multipleVE}

When singleton experiences are sufficiently unfavorable, the composition of observations can reinforce the valuation differences that guide choice. A class selected infrequently in competitive menus is assessed largely through its poor singleton experiences, sustaining its low valuation. More frequent competitive choices dilute these unfavorable observations. Different patterns of selection can consequently support different self-confirming rankings. Under Assumption~\ref{as:singleton-support}, every class is selected in its singleton menu. Theorem~\ref{th:strictpurestable} implies that every strict VE is an LSVE, approached by a branch of locally stable SVE as $\beta\to\infty$.

\begin{theorem}[]
\label{th:generallow}
Fix a reduced decision tree $\mathcal T'$ satisfying Assumption~\ref{as:singleton-support}. There is a finite threshold $\underline z$ such that, for every $z<\underline z$, the following statements hold:
\begin{enumerate}
\item[(i)]
If the decision tree is generic, there exists at least one strict VE.
\item[(ii)]
If every binary menu is supported, there exist at least $|\mathcal S|$ distinct strict VE, one with each class uniquely lowest-ranked. This conclusion does not require genericity.
\end{enumerate}
For each fixed such $z$, there is a finite sensitivity threshold $\widehat\beta(z)$ such that, for every $\beta\geq\widehat\beta(z)$, a unique SVE lies in a sufficiently small neighborhood of each strict VE. These SVE converge to their respective VE as $\beta\to\infty$ and are locally exponentially stable under~\eqref{eq:differential}.
\end{theorem}

The mechanism is transparent when choices are held fixed. For a pure (deterministic) choice rule $\varphi$ selecting one class in each menu, let
\(
D_s^\varphi
:=
\sum_{\omega:\,\varphi(\omega)=s}p(\omega)
\)
be the aggregate probability of selecting class $s$ across all menus featuring it. Its consistent valuation is
\begin{equation}
v_s^\varphi(z)
=
a_s^\varphi+\lambda_s^\varphi z,
\qquad
\lambda_s^\varphi
:=
\frac{p(\{s\})}{D_s^\varphi},
\label{eq:singleton-exposure}
\end{equation}
where $a_s^\varphi$ is its conditional expected payoff before the shift. The coefficient $\lambda_s^\varphi$ is the share of singleton observations in its payoff history. As $z$ becomes increasingly negative, differences in these shares dominate the fixed differences in baseline payoffs. Classes whose observations are concentrated in singleton menus receive the largest effective penalty.

The existence argument constructs a priority order whose choices generate the valuation comparisons needed to sustain it. With every binary menu supported, this construction can place any designated class $s$ last. Every other class $r$ then receives competitive observations from $\{s,r\}$, so its singleton share is strictly below one, whereas $s$ bears the entire singleton penalty. The construction also orders the remaining classes strictly: assigning a binary menu to an earlier-priority class removes a positive source of competitive observations from the later-priority class, increasing the latter's eventual singleton share. For sufficiently negative shifts, the resulting valuations reproduce the entire priority order. The appendix provides the construction and the associated thresholds. Genericity in part (i) provides a strict total order on valuations. Theorem~\ref{th:strictpurestable} then supplies the nearby locally exponentially stable SVE for sufficiently large sensitivity. 

Under the stronger Assumption~\ref{assm:monotonicity}, sufficiently low singleton payoffs make every strict ranking self-confirming. For any proposed strict ranking, the menu probabilities ensure that higher-ranked classes receive more competitive observations and place lesser weight on the singleton penalty, so lower-ranked assessments are depressed more strongly. The appendix proves, for each sufficiently negative $z$, all sufficiently large sensitivities admit
at least $|\mathcal S|!$ locally exponentially stable SVE, one associated with each strict ranking (see Fig.~\ref{fig:phase}(a)). Initial valuations then affect which valuation ranking is sustained in the long-run, as the choices they induce shape the evidence subsequently used to evaluate each class.

\subsection{Instability and Persistent Cycling}
\label{sec:limitcycle}

The preceding results establish stability and convergence in several important classes of decision problems. However, they leave open the possibility that learning can fail to converge in a decision tree with three or more similarity classes. We now answer this question in the affirmative by showing that stationary decision problems can indeed generate persistent oscillations in the learning dynamics. In fact, with just three similarity classes, we show there are open sets of decision problems for which sufficiently sensitive choice admits no asymptotically stable SVE and learning persistently oscillates.
\begin{figure}[ht]
\centering
\begin{tikzpicture}[scale=0.85,
    font=\footnotesize,
    solid node/.style={
        circle,draw,fill=black,inner sep=1.2pt
    },
    hollow node/.style={
        circle,draw,inner sep=1.2pt
    },
    level 1/.style={
        level distance=13mm,sibling distance=34mm
    },
    level 2/.style={
        level distance=16mm,sibling distance=15mm
    }
]
\node[solid node,label=above:{$r$}]{}
child{
    node[solid node,label=left:{$\omega_1$}]{}
    child{
        node[hollow node,label=below:{$\varepsilon$}]{}
        edge from parent node[left]{
            \textcolor{red}{$A_1$}
        }
    }
    child{
        node[hollow node,label=below:{$1$}]{}
        edge from parent node[right]{
            \textcolor{orange}{$B_1$}
        }
    }
    edge from parent node[above left]{$1/3$}
}
child{
    node[solid node,label=right:{$\omega_2$}]{}
    child{
        node[hollow node,label=below:{$-1$}]{}
        edge from parent node[left]{
            \textcolor{orange}{$B_2$}
        }
    }
    child{
        node[hollow node,label=below:{$3/2$}]{}
        edge from parent node[right]{
            \textcolor{blue}{$C_2$}
        }
    }
    edge from parent node[right]{$1/3$}
}
child{
    node[solid node,label=right:{$\omega_3$}]{}
    child{
        node[hollow node,label=below:{$-3/4$}]{}
        edge from parent node[right]{
            \textcolor{blue}{$C_3$}
        }
    }
    edge from parent node[above right]{$1/3$}
};
\end{tikzpicture}
\caption{A stationary decision tree generating persistent cycling. Colors identify similarity classes; terminal labels give the expected payoffs of individual alternatives with $\varepsilon>0$.}
\label{fig:smallcycle}
\end{figure}

Consider the decision tree in Figure~\ref{fig:smallcycle}. Five alternatives form three similarity classes: $A=\{A_1\}$, $B=\{B_1,B_2\}$, and $C=\{C_2,C_3\}$. The induced class menus are $\{A,B\}$, $\{B,C\}$, and $\{C\}$, each occurring with equal probability $1/3$. The feedback operates through the observations entering class $B$ which links the two competitive menus. 

$B$ pays well when competing with $A$ and poorly when competing with $C$. Suppose Alice begins by choosing $A$ more often. She then observes fewer of $B$'s favorable payoffs from $\{A,B\}$, so her assessment of $B$ deteriorates. This reinforces $A$'s position and also encourages her to choose $C$ over $B$. Choosing $C$ has the opposite effect on $B$'s selected sample, however, as it removes unfavorable observations from $\{B,C\}$. As $C$ gains ground, its favorable competitive payoffs also raise its own valuation, further displacing $B$ from the menu $\{B,C\}$ where $B$ performs poorly. The observations that remain for the rarely chosen $B$ are consequently more favorable, allowing its valuation to recover. That recovery sets up the next reversal. As Alice returns to $B$, its unfavorable payoffs in $\{B,C\}$ re-enter her sample. At the same time, choosing $C$ less often in that menu leaves its valuation increasingly shaped by its poor singleton payoff, reinforcing Alice's return to $B$. The renewed exposure to $B$'s unfavorable payoffs pulls its valuation down again, allowing $A$ to regain ground. The cycle is sustained by these changes in the composition of experience: every alternative's mean payoff is fixed, but Alice's choices continually change the evidence used to evaluate its class.\footnote{The interaction between $B$ and $C$ creates a negative feedback loop: a higher valuation of $C$ improves the composition of $B$'s observations, whereas a higher valuation of $B$ worsens the composition of $C$'s observations. This interacts with the self-reinforcing feedback in $\{A,B\}$, where choosing $A$ more often removes favorable observations of $B$. These opposing feedbacks sustain oscillation at high sensitivity.}

At $\varepsilon=0$, the construction is particularly transparent: the unique SVE is $\vvec^*=(0,0,0)$ at every finite sensitivity, and it is linearly unstable whenever $\beta>8$. A small positive value of $\varepsilon$ yields the following robust conclusion (see Figure \ref{fig:RPSphaseportrait} for an illustration).

\begin{theorem}[Persistent cycling]
\label{th:cycle}
Fix the menu support $\Omega_+=\{\{A,B\},\{B,C\},\{C\}\}$. There exist $\varepsilon>0$, an open neighborhood $\mathcal U$ of the decision tree in Figure~\ref{fig:smallcycle}, and a finite $\bar\beta$ such that, for every tree in $\mathcal U$ and every finite $\beta\geq\bar\beta$:
\begin{enumerate}[label=(\roman*)]
    \item There is a unique SVE, which is linearly unstable under both the normalized~\eqref{eq:normalized} and joint asynchronous dynamics~\eqref{eq:differential}.

    \item Every trajectory of both the normalized and the joint asynchronous dynamics outside the equilibrium's stable manifold has a non-trivial periodic orbit as its omega-limit set.

    \item From any finite initial valuations and positive prior counts, the limit set of $(\vvec_k,\mathbf x_k)$ under the stochastic CQL process~\eqref{eq:update} is almost surely a non-constant periodic orbit of the asynchronous dynamics~\eqref{eq:differential}. Neither valuations nor choice probabilities converge.
\end{enumerate}
\end{theorem}

Here openness is relative to the space of positive menu probabilities and expected payoffs of alternatives on the specified support. The proof uses a fixed $\varepsilon>0$, allowing all three menu probabilities and all five payoffs to vary in a sufficiently small neighborhood. Thus the absence of a stable equilibrium and the persistence of oscillation survive perturbations of the decision problem. Harmonic averaging makes the argument tractable. In the
sufficient-statistic coordinates $(\mathbf x,\mathbf y)$, where $y_s=x_sv_s$, three affine discrepancies decay
exponentially, leaving an invariant attracting plane. The unique equilibrium is a source on this plane.
Boundedness and the Poincar\'e--Bendixson theorem then imply that every nonstationary planar trajectory has a periodic omega-limit set.\footnote{The full joint state space is five-dimensional. Its equilibrium has a three-dimensional stable manifold, whose intersection with the invariant plane consists only of the equilibrium.} The appendix extends the conclusion to the full joint dynamics and combines the planar limit-set results of \citet{Benaim1999} with a verification that the stochastic process almost surely avoids the unstable equilibrium.
\begin{figure}[ht]
  \centering
  \includegraphics[width=0.81\linewidth]{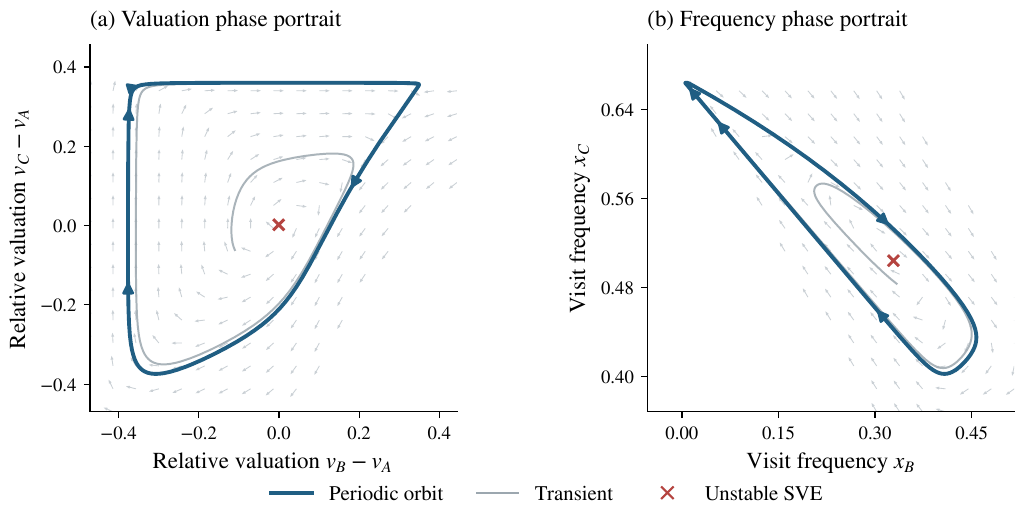}
  \caption{Phase portraits of the joint CQL dynamics for tree in Fig.~\ref{fig:smallcycle} with $\varepsilon=0.01$ and $\beta=20$}
  \label{fig:RPSphaseportrait}
\end{figure}

The construction is minimal as cycling requires at least $3$ classes (by Theorem~\ref{th:competitiveconstant}) and $5$ alternatives.\footnote{Instability does not require this particular menu arrangement. The appendix gives further examples with three binary menus whose payoff advantages form a cycle, supplemented either by a ternary menu, without singletons, or by all three singleton menus at intermediate payoff levels.} Its three-menu support is also irreducible. The self-correcting feedback condition~\eqref{eq:negative-payoff-crowding} fails here as competition removes favorable observations: a higher valuation of $A$ crowds $B$ out of its favorable menu, lowering $g_B$; a higher valuation of $B$ similarly crowds $C$ out of its favorable menu, lowering $g_C$. Building on this observation, our next result characterizes the \textit{absence of asymptotically stable SVE} in decision problems with \textit{three classes} under the joint dynamics~\eqref{eq:differential}. At an SVE $\vvec^*$, define the trace and the determinant of the linearized normalized dynamics~\eqref{eq:normalized} in the two independent valuation differences by
\[
T(\vvec^*)
:=
-2+\beta\sum_{s\ne r}
\operatorname{Cov}_{w_s^{\vvec^*}}
\bigl(\pi_s(\omega),\sigma_\omega^r(\vvec^*)\bigr),
\quad
d(\vvec^*):=\det\!\bigl(I-Dg(\vvec^*)\bigr).
\]
\begin{lemma}[Absence of stable SVE with three classes]
\label{lem:three-class-stability}
Fix $\beta\in[0,\infty)$ and a decision tree with three classes and suppose all SVE are hyperbolic.\footnote{By Theorem~\ref{th:correspondence}, this assumption holds for generic $(p,\pi)$ at all but finitely many sensitivities.} At least one SVE has $d(\vvec^*)>0$.
No asymptotically stable SVE exists if and only if
\(
T(\vvec^*)>0
\ \text{at every SVE with }d(\vvec^*)>0.
\)
\end{lemma}
The two eigenvalues governing relative valuations have sum $T$ and product $d$, with $-2$ capturing the unit restoring force in each independent valuation difference. Under the payoff-crowding covariance condition~\eqref{eq:negative-payoff-crowding}, every covariance term is nonpositive, so $T\leq-2$. The same quantities also characterize instability with three classes: at equilibria with $d>0$, positive covariances must outweigh negative ones sufficiently to overcome these restoring forces, giving $T>0$. SVE with $d<0$ are already unstable. Thus, no stable SVE exist precisely when this aggregate reinforcing feedback dominates at every SVE with $d>0$. 

\subsection{Coarse representation and the welfare cost of learning}
\label{sec:coarse-welfare}

The preceding result shows that with three or more classes, Alice's valuations need not converge. Even when her valuations converge, her choices need not attain the best expected payoff her coarse representation permits. To separate the costs of coarse representation and coarse learning, consider a fully informed oracle who knows the menu distribution and expected payoffs. He uses Alice's similarity partition, within-class implementation, and logit choice rule, with one
common valuation per class across menus.\footnote{At any $\beta>0$, logit depends only on $\beta\vvec$, so unrestricted valuations allow us to normalize $\beta=1$.} We allow high-sensitivity limits of logit policies, making deterministic choices admissible.\footnote{Over finite valuation vectors, the benchmark above is generally
only a supremum.} Unlike Alice, however, he chooses these valuations freely, without requiring consistency with experienced payoffs.

The proposition below characterizes the oracle's optimal choice policy and gives necessary and sufficient conditions for its efficiency. For a fixed strict ranking $\succ$, write $\max_\succ\omega$ for the highest-ranked class in menu $\omega$. When each supported menu has a unique payoff-maximizing class, define the
\emph{efficient-comparison graph} $G^{\mathrm{eff}}$ by drawing an arrow $s\to r$ whenever $s$ is the (objective) payoff-maximizer in a supported menu containing $r\ne s$. 
 
\begin{proposition}[The informed coarse benchmark]
\label{prop:coarse-oracle}
An optimal oracle policy is represented by a strict class
ranking, and its expected payoff is
\[
W^{\mathrm{coarse}}
=
\max_{\succ}
\sum_{\omega\in\Omega}
p(\omega)\,\pi_{\max_\succ\omega}(\omega).
\]
Generically, every optimal class-choice policy is pure.
If each supported menu has a unique payoff-maximizing class,
then
\(
W^{\mathrm{coarse}}
=
\sum_{\omega\in\Omega}
p(\omega)\max_{s\in\omega}\pi_s(\omega)
\ \Longleftrightarrow\
G^{\mathrm{eff}}\text{ is acyclic}.
\)
\end{proposition}

Logit randomization cannot improve this benchmark. Its random-utility representation amounts to drawing a strict ranking of classes independently of the menu and choosing the highest-ranked available class. Expected welfare is therefore a weighted average of the welfare generated by strict rankings, and cannot exceed the best of them. Conversely, any strict ranking can be implemented in the limit by increasing the gaps between appropriately ordered valuations. The graph records the comparisons required for efficiency: an acyclic graph can be completed into a common ranking that selects the best class in every menu, whereas a directed cycle
makes these requirements inconsistent.\footnote{Efficiency here holds the within-class implementation fixed. Unrestricted informed choice among objective alternatives
may yield a higher payoff.}

The welfare gap between unconstrained informed choice and this benchmark measures the cost of coarse representation; any further shortfall under coarse learning reflects the welfare cost of learning. We examine high-sensitivity limits of long-run outcomes to isolate losses that persist as logit noise vanishes. Singleton payoffs illustrate the distinction: they contribute the same payoff to every ranking, yet can change the outcome selected by learning. In the two-class illustrations in Figure~\ref{fig:phase}, the oracle is efficient and always chooses Bordeaux when both classes are available. The limiting Chianti-favoring strict outcome with low singleton payoffs loses $1/3$ per period relative to this benchmark, while the limiting mixed outcome with high singleton payoffs loses $1/9$. Thus even globally convergent learning can sustain a welfare loss while the optimal policy is representable by the same class scores.

Our harmonic benchmark gives equal weight to observations from the same class, regardless of the menus in which they arise. Alice observes those menus, so in principle she could instead weight their feedback differently while retaining one valuation per class. Such reweighting might correct the choice distortions in the stylized two-class illustrations above. What follows shows why it cannot provide a general remedy for closing the welfare gap.

Indeed, a persistent welfare loss can arise even when learning converges globally to a unique SVE whose high-sensitivity limit is strict. Consider two equally likely binary menus, $\{A,C\}$ and $\{B,C\}$, with deterministic payoffs $(0.6,1)$ and $(0.4,0)$, respectively. The oracle ranking $B\succ C\succ A$ selects
the best class in each menu and yields $W^{\mathrm{coarse}}=0.7$. Thus, coarse representation permits efficient choice. Coarse Q-learning, however, cannot sustain these comparisons. Every observation of $A$ yields $0.6$, while every observation of $B$ yields $0.4$, so their learned ordering contradicts the efficient ranking.
For every finite sensitivity, this example has a unique, globally stable SVE, to which harmonic visit-count learning converges almost surely. As sensitivity
increases, its valuations approach $(0.6,0.4,0)$, a strict VE that selects $A$ and $B$ in their respective menus and yields welfare $0.5$. The limiting loss is thus $0.2$ per period, despite convergence to a strict outcome and the absence of
singleton menus.

The obstruction extends beyond CQL. Call a learning rule \emph{constant-payoff calibrated} if a class whose observations always yield $c$ is assessed at $c$ asymptotically, almost surely. Constant-payoff calibration is a minimal consistency requirement for any reasonable payoff-assessment learning rule. Suppose Alice uses common class valuations and the same monotone choice rule across menus.\footnote{For binary menus, $\Pr_k(s\mid\{s,r\};\vvec)=Q_k(v_s,v_r)$, with $Q_k$ nondecreasing in its first argument and nonincreasing in its second, and $Q_k(x,y)+Q_k(y,x)=1$. Scores are determined before the menu is drawn.} Calibration eventually ranks $A$ above $B$, making $C$ weakly less likely to be chosen against $A$ than against $B$, although efficiency requires the opposite. Changing $C$'s assessment cannot resolve this conflict. The following result formalizes this observation. 
\begin{proposition}[A limit to payoff-assessment learning]
\label{prop:welfare-reweighting}
Under the given choice restrictions, no learning rule satisfying constant-payoff calibration can guarantee long-run average welfare at least $W^{\mathrm{coarse}}$ in every decision tree. In the preceding example, every such rule yields
\[
\limsup_{K\to\infty}\frac1K\sum_{k=0}^{K-1}r_k
\leq \frac12
<\frac{7}{10}=W^{\mathrm{coarse}}
\qquad\text{almost surely}.
\]
\end{proposition}
The appendix establishes a positive welfare gap on an open set of payoff and menu-probability specifications with this support. The bound requires only calibration of $A$ and $B$; neither $C$'s valuation nor the choice rule need converge. Constant-payoff calibrated learning covers arbitrary nonnegative weighted averages of own-payoff observations and initial scores whenever the averages are eventually well-defined and initial-score shares vanish. The weights may depend arbitrarily
on past menus and rewards, even with full knowledge of the primitives. Efficient common scores encode priorities across menus; learning them as estimates of a class's expected payoff need not recover those priorities. The impossibility identifies this limitation, even when the efficient policy is representable by a common ranking of categories.

\section{Algorithmic Prominence \& Consumer Learning}
\label{sec:platform-learning}

Online marketplaces influence which products consumers consider through search, recommendations, and algorithmic prominence \citep{LeeMusolff2026}, with evidence that Amazon gives preferential search rankings to its own brands \citep{Fradkin}. Moreover, by emphasizing a common in-house brand label and presenting its underlying products homogeneously with a uniform price, a hybrid platform may encourage consumers to attend to brand identity rather than less salient product-specific attributes. Consumers who organize learning around these shared cues pool experiences across the brand, as in our limited-salience interpretation of similarity classes. When consumers learn through such categories, prominence also shapes the experiences that inform subsequent choices. In this section, we demonstrate as a natural application of our coarse Q-learning framework, how a digital platform can use favorable purchases of its own brand to sustain demand for inferior products sold under the same umbrella label, including when competitors remain available. 

Consider a consumer who uses three similarity classes: the platform's in-house brand $H$, a competitor $C$, and an outside option $O$ (shopping elsewhere). In each period, a product state $i\in\{1,\ldots,J\}$ is drawn independently and uniformly. The platform's product has expected payoff $\mu_i$ and margin $m_i\geq0$, while the competitor and the outside option have constant expected payoffs $\mu_C>\mu_O$. Prices are fixed, and we normalize the platform's profit to zero when its own product is not purchased. The consumer maintains one valuation per class and follows our harmonic
visit-count coarse Q-learning rule~\eqref{eq:update}. Throughout, suppose
\begin{equation}
\bar\mu:=\frac1J\sum_{i=1}^J\mu_i<\mu_O<\mu_C,
\qquad \max_i\mu_i>\mu_O.
\label{eq:platform-randomized-assumptions}
\end{equation}
The in-house brand's expected payoff is below the outside option on average, although some of its products are better. The platform commits to a stationary prominence rule $\tau=(\tau_i)_{i=1}^J\in[0,1]^J$. After observing state $i$, it promotes its product with probability
$\tau_i$, using an independent draw each period. Prominence removes the competitor from consumers' consideration set, reducing the class menu from $\{H,C,O\}$ to $\{H,O\}$. We initially assume that prominence is costless. For each fixed prominence rule and finite sensitivity, the appendix establishes (by extending Theorem \ref{th:competitiveconstant}) that the induced stationary decision problem has a unique SVE. The corresponding joint equilibrium is globally asymptotically stable under \eqref{eq:differential}, and harmonic visit-count learning \eqref{eq:update} converges to it almost surely. We evaluate platform profit and consumer welfare by taking the high-sensitivity limit of these long-run outcomes. Let $d_i(\tau)$ denote the resulting probability of an in-house product sale in state $i$, whether or not the product is promoted. The platform seeks to maximize expected profit per period, $\Pi^\infty(\tau):=J^{-1}\sum_{i=1}^J m_i d_i(\tau)$, breaking ties in favor of less frequent prominence.

\subsection{Optimal Deterministic Prominence}
\label{sec:platform-spillovers}

To begin, restrict the platform's strategies to deterministic prominence $\tau=(\tau_i)_{i=1}^J\in\{0,1\}^J$. We start with two states, $h$ and $\ell$, satisfying
\( \mu_h > \mu_C > \mu_O > \mu_\ell\) and \(m_h,m_\ell>0\). Under complete information (the consumer observes the state and knows the expected payoffs), prominence cannot increase sales: the high-quality product is preferred to the competitor, while the low-quality product
is worse than the outside option.\footnote{The same benchmark arises under learning with the finest partition: valuations converge to alternatives' expected payoffs, and the subsequent high-sensitivity limit yields informed choice in every menu.} Our tie-breaking convention
therefore selects no prominence. Under coarse learning, however, no prominence leaves the overall brand valued at $\bar\mu<\mu_O$, so the consumer chooses the competitor in the limit. 

\textit{Promoting only the high-quality product} changes this outcome. Its favorable payoffs raise the common brand valuation and encourage purchases of the low-quality product when the competitor is available. Those disappointing payoffs pull the valuation back down. In the limit, the two sources of
evidence endogenously balance at indifference $v_H=\mu_C$, with purchase probabilities $1$ and
\(
q_2=(\mu_h-\mu_C)/(\mu_C-\mu_\ell)\in(0,1)
\)
in states $h$ and $\ell$, respectively.
Limiting platform profit is
$(m_h+q_2m_\ell)/2$. The other three deterministic rules yield zero limiting
brand sales, so \textit{high-only prominence is uniquely optimal} among them. The
same corrective feedback behind our limiting-indifference result in Theorem~\ref{th:uniquemixedstable} thus sustains inferior purchases despite continued learning from their payoffs.\footnote{While there are no literal singleton menus in this illustration, under high-state prominence only, $\{H,O\}$ serves as a focal menu for $H$ with Alice choosing $H$ almost surely in the long-run, high-sensitivity limit.}

Adding a medium-quality product makes the platform's trade-off more
substantive. Suppose now there are three states, $h,m,\ell$, with
\begin{equation}
\mu_h>\mu_C>\mu_m>\mu_O>\mu_\ell,
\qquad
\frac{\mu_h+\mu_m}{2}\leq\mu_C,
\qquad m_h,m_m,m_\ell>0,
\label{eq:platform-three-states}
\end{equation}
and $\bar\mu<\mu_O$ as before. Under complete information, prominence increases sales only in state $m$; our tie-breaking convention therefore selects the medium-only prominence policy. Under coarse learning, that policy also sells only the medium-quality product: its payoff calibrates the brand valuation to $\mu_m$, leaving even the high-quality product unpurchased. Promoting both
the high- and medium-quality products improves on this policy. Promoting only the high-quality product can be more profitable still. Under high-only prominence, the in-house brand's limiting purchase probabilities in
states $(h,m,\ell)$ are $(1,q_3,q_3)$, where
\[
q_3:=\frac{\mu_h-\mu_C}{2\mu_C-\mu_m-\mu_\ell}\in(0,1).
\]
High-and-medium prominence instead induces purchases $(1,1,0)$. The platform compares the additional low-quality sales under the first rule with the additional medium-quality sales
under the second. If $m_\ell \gg m_m$, i.e.\ if the margins from the former are sufficiently higher (as under a uniform umbrella-brand retail price $P$ where unit production costs increase with quality, so that net consumer payoffs $\mu_i = u_i - P$ rise with quality while margins $m_i = P - c_i$ fall), the platform optimally uses deterministic prominence in the high state only.

\begin{proposition}[Optimal deterministic prominence]
\label{prop:platform-deterministic}
Under \eqref{eq:platform-randomized-assumptions} and
\eqref{eq:platform-three-states}, the optimal deterministic rule,
with ties resolved in favor of less frequent prominence, is
\begin{equation}
\tau^*=
\begin{cases}
(1,0,0),
& \text{if }\ q_3m_\ell\geq(1-q_3)m_m,\\[3pt]
(1,1,0),
& \text{if }\ q_3m_\ell<(1-q_3)m_m,
\end{cases}
\label{eq:platform-deterministic-comparison}
\end{equation}
with coordinates ordered as $(h,m,\ell)$. The profit-maximizing rule is unique whenever
$q_3m_\ell\neq(1-q_3)m_m$. The maximal limiting
profit is
\(
\Pi_{\mathrm{det}}^*
=(m_h+\max\{q_3(m_m+m_\ell),\,m_m\})/3.
\)
\end{proposition}

Promoting the medium-quality product secures its sales but dilutes the favorable evidence that sustains $v_H=\mu_C$ under high-only prominence. Relative to promoting both products, high-only prominence sacrifices a fraction $1-q_3$ of medium-product sales while generating a fraction $q_3$ of low-product sales. Condition~\eqref{eq:platform-deterministic-comparison} compares these gains
and losses. Note that with positive margins, while the complete-information optimal prominence rule is independent of their magnitudes, the coarse-learning optimum depends on them.

Under high-only prominence, consumer welfare equals $\mu_C$, its no-prominence level, in both illustrations. Thus, the platform gains from prominence, not at the cost of the consumer, but of its rival. Nonetheless, from the consumer's standpoint, the same coarse representation permits a better allocation. Holding prominence fixed, the ranking $C\succ H\succ O$ selects the high-quality product under prominence and the competitor elsewhere, yielding an additional payoff of $(\mu_h-\mu_C)/J$ for $J=2,3$. CQL cannot sustain this ranking: purchasing the brand only in state $h$ would raise its valuation to $\mu_h>\mu_C$, inducing purchases elsewhere. The gap isolates the welfare cost of learning beyond
the restriction to a common class ranking.

\subsection{Optimal randomized prominence}
\label{sec:platform-randomized}

Beyond deterministic prominence, randomization lets the platform adjust the composition of
purchases more precisely, expanding inferior product sales while accepting a lower learned brand valuation. Return to the two-state illustration and consider the prominence strategy $\tau=(1,t)$: always
promote the high-quality product and promote the low-quality
product with probability $t$. Initially, low-state prominence
only replaces purchases that already occurred in competitive
menus. Once $t>q_2$, however, it increases low-quality sales.
For
\[
q_2<t<t_2^*,
\qquad \text{where }
t_2^*:=\frac{\mu_h-\mu_O}{\mu_O-\mu_\ell}\in(q_2,1),
\]
the promoted mean $(\mu_h+t\mu_\ell)/(1+t)$ lies strictly
between $\mu_O$ and $\mu_C$. Consumers therefore buy the
brand under prominence and the competitor otherwise.
Low-quality sales increase with $t$, raising profit even
as the learned brand valuation falls. The platform approaches
its profit supremum by taking $t\uparrow t_2^*$, pushing
the brand valuation toward the outside option from above.
The three-state illustration follows the same logic.
Consider $\tau=(1,1,t)$: always promote the high- and
medium-quality products, and randomize prominence in the
low state. For
\[
0<t<t_3^*,
\qquad \text{where }
t_3^*:=\frac{\mu_h+\mu_m-2\mu_O}{\mu_O-\mu_\ell}\in(q_3,1),
\]
consumers again buy the brand precisely when it is promoted.
Randomization thus secures medium-quality sales while
retaining profitable low-quality sales, resolving the
tradeoff between the two deterministic candidates.
The profit suprema in the two illustrations are
\[
\Pi_2^{\sup}=\frac{m_h+t_2^*m_\ell}{2},
\qquad
\Pi_3^{\sup}=\frac{m_h+m_m+t_3^*m_\ell}{3},
\]
each strictly exceeding every deterministic rule's profit.
These gains, however, have a welfare cost. In both illustrations,
every sequence approaching the respective profit supremum
yields consumer welfare and competitor sales share converging to
\[
W_j^*
=
\mu_C-\frac{\mu_C-\mu_O}{j}(j-1+t_j^*)<\mu_C,
\qquad
c_j^*=\frac{1-t_j^*}{j}>0,
\qquad j=2,3.
\]
The competitor retains a positive share, but consumer welfare
falls below its no-prominence level. Along the stated
policies, the platform's randomization generates these sales
shares: consumers choose the brand under prominence and the
competitor otherwise.

The same logic characterizes optimal randomized performance
with any number of states. Let $d_i\in[0,1]$ denote the limiting
probability of a brand purchase in state $i$. Valuation
consistency equates the learned brand payoff with the
purchase-weighted average of product expected payoffs. Positive
limiting sales therefore require
\begin{equation}
\sum_i(\mu_i-\mu_O)d_i\geq0:
\label{eq:platform-profit-bound}
\end{equation}
favorable purchases must, on average, offset the unfavorable feedback
from products below the outside option. The supremum of
limiting profit is exactly the maximum of
$J^{-1}\sum_i m_i d_i$ over sales allocations satisfying
this constraint.
With nonnegative margins, a maximizing allocation can
fully include every product with $\mu_i\geq\mu_O$.
Products below the outside option are then included in
decreasing order of $m_i/(\mu_O-\mu_i)$, until the constraint
binds, with at most one product partially included.
The platform thus prioritizes margin per unit of unfavorable
feedback. Slightly reducing prominence in the partially included below-outside-option state, or in a fully included such state if none is partially included, keeps the promoted mean just above
$\mu_O$ and approaches the bound. Hence approximately optimal
stationary rules need randomize in at most one state.\footnote{
With strictly positive margins, no fixed stationary rule
attains the supremum. At a rule whose promoted mean equals
$\mu_O$, the limiting purchase probability under prominence
is $1/2$; approaching that mean from above instead sustains
purchase with probability $1$. Each rule is evaluated after
learning and the high-sensitivity limit, before taking the
policy limit. At any fixed finite sensitivity, an exact
profit-maximizing rule exists and can also be chosen to
randomize in at most one state.}
The appendix gives the characterization and proof.

\textbf{Heterogeneous consumers:} The randomized-policy characterization~\eqref{eq:platform-profit-bound} extends to heterogeneous partitions. Suppose a positive fraction of consumers use the coarse brand category, the remainder distinguish individual products, and the platform applies a common prominence rule. The same family of policies approaches the profit supremum for every positive coarse share. Promoting every product at least as good as the outside option already maximizes profit from finest-partition consumers, while adjusting prominence below that threshold affects brand sales only to coarse consumers. In our illustrations, finest-partition consumers also suffer from low-quality prominence: they reject the promoted product but lose access to the competitor and take the outside option. Coarse learning can therefore affect welfare even for consumers who distinguish every product.\footnote{Consumer composition can change the optimal deterministic policy. In the three-state illustration~\eqref{eq:platform-deterministic-comparison}, a larger finest-partition share favors high-and-medium prominence over high-only prominence. The appendix characterizes this comparison and proves the heterogeneous-consumer extension.}

\textbf{Regulation with prominence caps:} In the all-coarse case, a cap on the frequency of prominence at the optimal deterministic frequency may seem promising but need not preserve consumer welfare. The platform can reallocate prominence from favorable products toward inferior ones with higher margins, changing the composition of purchases without increasing its prominence frequency. The appendix characterizes the profit supremum under any cap and shows it can be approached by stationary prominence rules that randomize in at most two states, and two-state randomization can be necessary even with a positive share of finest-partition consumers.\footnote{More generally, for any fixed consumer mix, every stationary rule admits a replacement with the same prominence frequency and weakly higher limiting population profit that randomizes in at most two states.} It identifies when policies approaching the profit supremum preserve no-prominence welfare, yet gives an example where approximately optimal policies lower welfare under every positive cap, though the welfare loss vanishes as the cap falls to $0$.

\subsection{Optimal costly stationary prominence}
\label{sec:platform-costly}

In practice, prominence may require advertising and promotional expenditure, or displace other valuable uses of consumers' attention. We show such costs can make deterministic prominence optimal even when randomization is available, while preserving consumer welfare at its no-prominence level. Suppose each use incurs a cost $c\geq0$, regardless of whether a purchase follows. The platform maximizes limiting net profit
$\Pi^\infty(\tau)-cT(\tau)$, where $T(\tau):=J^{-1}\sum_i\tau_i$ is its prominence frequency.
Consumer payoffs and learning remain unchanged. Costs make purchases of the in-house brand in competitive menus particularly valuable to the platform since they generate revenue without requiring prominence.

To illustrate, let's specialize the two-state environment to
\[
(\mu_h,\mu_\ell)=(4,0),\qquad
\mu_O=\frac52,\qquad
\mu_C=3,\qquad
m_h=m_\ell=1.
\]
Without costs, the platform approaches its profit supremum
by always promoting the high-quality product and promoting
the low-quality product with probability tending to $3/5$.
At $c=2/3$, however, high-only prominence uniquely maximizes
net profit among all stationary rules.\footnote{High-only
prominence is uniquely optimal precisely for
$c\in[4/9,8/9]$, as established in the appendix.}
It induces purchases $(1,1/3)$: favorable high-quality
experiences sustain low-quality purchases in the competitive
menu, keeping the brand valuation at $\mu_C=3$.
Gross profit is $2/3$ and prominence frequency is $1/2$,
giving net profit $1/3$.

The distinction between sales and prominence is essential.
Promoting the low-quality product with probability $1/3$
generates exactly the same sales, but requires prominence
in $2/3$ of periods and reduces net profit to $2/9$.
The sequence approaching the costless profit supremum
generates more sales, yet its net profit approaches only
$4/15$ at this cost. Competitive mixing economizes on
prominence, while the additional sales generated by more
frequent promotion no longer justify their cost.
Consumer welfare consequently equals its no-prominence
level $3$, compared with the lower costless-optimal limit $13/5$.

Costs can also change which products receive prominence.
Return to the three-state illustration and suppose all
margins equal $m>0$. Since $q_3<1/2$, high-and-medium
prominence is optimal among deterministic rules when
prominence is costless. At $c=m$, high-only prominence
instead becomes uniquely optimal among all stationary
rules, including randomized ones. The high-state sale
just covers its prominence cost; its favorable feedback
supports profitable purchases in both unpromoted states.
Net profit and consumer welfare are
\[
\Pi_h^\infty-cT_h
=
\frac{m(1+2q_3)-m}{3}
=
\frac{2mq_3}{3}>0,
\qquad
W_h^\infty=\mu_C.
\]
Thus costs can favor concentrating prominence on the best product, with competitive mixing sustaining sales of the others. For arbitrary $J$, we optimize over the closure of feasible sales allocations. Positive sales require a purchase-weighted brand mean payoff between $\mu_O$ and $\mu_C$, by valuation consistency. For each allocation, we minimize the prominence needed to implement it, allowing approximation at the outside-option boundary.
When the brand valuation lies strictly
between these thresholds, every brand purchase requires
prominence. At valuation $\mu_C$, competitive purchases
instead provide a baseline:
\(
d_i=q+(1-q)\tau_i,
\)
where $d_i$ is the purchase probability in state $i$ and $q$
is the purchase probability in a competitive menu.
Holding sales fixed, increasing $q$ reduces the prominence
required. The least costly implementation therefore uses
$q=\min_i d_i$. If an optimal allocation takes only the
values $q>0$ and $1$, deterministic prominence implements
it: promote the states with purchase probability one and
let competitive mixing supply purchases elsewhere.

The net-profit supremum equals the maximum gross profit
net of this exact or limiting prominence cost. An optimal policy,
when one exists, or a sequence approaching the supremum
can be chosen to randomize in at most one state.\footnote{
A positive supremum is attained if and only if some maximizing
sales allocation has a purchase-weighted mean strictly above
$\mu_O$. Otherwise it is approached at the outside-option
boundary, as in the costless case.}
The appendix gives the formal optimization problem and proof.
Whenever the optimum sustains brand valuation $\mu_C$,
consumer welfare also equals $\mu_C$: brand purchases yield
that payoff on average, and remaining purchases go to the
competitor. Costs can therefore preserve welfare through
competitive mixing, although they need not select this
outcome in every decision problem.

\section{Discussion}
\label{sec:discuss}

\subsection{A misspecified Bayesian learning interpretation}
\label{sec:conjugate}

Limited salience provides a natural source of misspecification. Alice treats payoffs from alternatives within the same similarity class as independent draws from a common stationary distribution, so experience with any one alternative informs her assessment of the entire class. Yet alternatives that coincide on salient dimensions may differ systematically in expected payoffs because of attributes she overlooks. Alice therefore interprets variation arising from these differences as noise around a common class mean.

We show that while coarse Q-learning requires neither a subjective likelihood for payoffs nor a prior, its harmonic updating rule admits an exact Bayesian interpretation under a misspecified model of class-level payoffs. Suppose Alice treats the payoffs observed from each class $s$, conditional on an unknown mean $\mu_s$, as independent draws from $N(\mu_s,\sigma_s^2)$, where $\sigma_s^2>0$ is known in her subjective model. She regards this payoff distribution as independent of the menu and the particular alternative implementing the class. Suppose further that her class means have
independent priors
\(
\mu_s\sim N\!\left(v_0(s), \sigma_s^2/\kappa_s\right),
\ \kappa_s>0.
\)
These are subjective assumptions: the objective payoff distributions need not be Gaussian. Since the normal likelihood is strictly positive at every real-valued payoff, Bayesian updating remains well-defined after every finite realized history. Let $\mathcal H_k$ denote Alice's observed history strictly before date $k$, and let $N_k(s):=\sum_{t=0}^{k-1}\mathbf 1\{s_t=s\}$. Normal conjugacy yields a Gaussian posterior for $\mu_s$ with mean and variance
\[
\begin{aligned}
m_k(s)
&:=\mathbb E_A[\mu_s\mid\mathcal H_k]
=\frac{\kappa_s v_0(s)+
\sum_{t=0}^{k-1}\mathbf 1\{s_t=s\}r_t}
{\kappa_s+N_k(s)},\\
V_k(s)
&:=\operatorname{Var}_A(\mu_s\mid\mathcal H_k)
=\frac{\sigma_s^2}{\kappa_s+N_k(s)},
\end{aligned}
\]
where the subscript $A$ indicates that posterior moments are computed under Alice's subjective model. The posterior mean satisfies
\[
m_{k+1}(s)
=
m_k(s)
+\mathbf 1\{s_k=s\}
\frac{r_k-m_k(s)}{\kappa_s+N_k(s)+1}.
\]
This is exactly the CQL update in \eqref{eq:update} with the harmonic learning rates in \eqref{eq:visit-count-stepsize}. Since $m_0(s)=v_0(s)$, it follows that $m_k(s)=v_k(s)$ at every date along every realized path. The parameter $\kappa_s$  has the usual Bayesian interpretation as a prior pseudo-count.
The representation also describes how Alice's subjective uncertainty evolves. Her posterior variance of the class mean falls whenever that class is sampled and remains unchanged otherwise. For every fixed $\beta\ge0$, Prop.~\ref{prop:continuous-approximation} implies every class is sampled with positive lower limiting frequency. Consequently, \(V_k(s)\to 0\) for every \(s\in\mathcal S,\) almost surely under the objective law. Thus, Alice becomes arbitrarily confident about each class's expected payoff in her subjective model.

Vanishing posterior uncertainty does not, however, require convergence of posterior means. Under the conditions of Theorem~\ref{th:cycle}, the stochastic joint process almost surely has a non-trivial periodic orbit as its limit set, and valuations and choice probabilities do not converge. In these environments, Alice's posteriors become increasingly concentrated around means whose vector continues to move. She interprets her observations as evidence about fixed class-level payoff distributions, while her changing valuations alter the composition of the menus and alternatives generating those observations. Endogenous selection can therefore sustain non-stationary behavior even as subjective uncertainty vanishes.\footnote{The posterior-mean correspondence extends to other regular one-parameter natural exponential families with scalar sufficient statistic $T(r)$ and Diaconis-Ylvisaker conjugate priors \citep{DiaconisYlvisaker1979}. With a positive prior pseudo-count and a prior mean in the interior of the convex support of $T(r)$, the posterior expectation of the mean statistic is the corresponding pseudo-count-weighted average of the observed $T(r)$ values. Objectively realized observations must be compatible with the subjective likelihood, and the posterior must remain proper after every finite history. When $T(r)=r$, this reproduces the original CQL update; otherwise, the correspondence concerns the recursion using $T(r)$ as its update target. Examples include Bernoulli-Beta, Binomial-Beta, and Poisson-Gamma specifications.}

The Bayesian interpretation of valuation updating has a complementary
choice-theoretic representation. Given current posterior means
$\mathbf v_k$ and menu $\omega$, suppose Alice chooses
$\sigma\in\Delta(\omega)$ to maximize current subjective expected
payoff net of an entropy cost:
\[
\max_{\sigma\in\Delta(\omega)}
\left\{
\sum_{s\in\omega}\sigma_s v_k(s)
-\frac{1}{\beta}\sum_{s\in\omega}
\sigma_s\log\!\left(\frac{\sigma_s}{1/|\omega|}\right)
\right\},
\qquad \beta>0.
\]
The cost is proportional to the relative entropy from
uniform choice, penalizing concentration on particular classes and thereby encouraging exploration.
The unique optimizer is
\[
\sigma_s
=\frac{\exp\{\beta v_k(s)\}}
{\sum_{r\in\omega}\exp\{\beta v_k(r)\}},
\qquad s\in\omega,
\]
precisely our logit rule. Larger $\beta$ lowers the cost
relative to payoff differences; uniform choice is recovered as
$\beta\downarrow0$. Thus the benchmark combines Bayesian updating
under Alice's subjective payoff model with entropy-regularized
choice based on current posterior means.

\subsection{Cognitive Complexity \& Choice Overload}
\label{sec:complexity}

Beyond limited salience, CQL also admits a tractability interpretation, taking the similarity partition as given. Let $n:=|\mathcal S|$. Separately evaluating each class conditional on the menu would remove the choice distortions seen above as Alice would asymptotically learn the  expected payoff of each class in every menu correctly, and her welfare would approach the efficient benchmark $\sum_{\omega\in\Omega}
p(\omega)\max_{s\in\omega}\pi_s(\omega)$ as $\beta\to\infty$. However, it would require one valuation $v(s,\omega)$ for every supported pair with $s\in\omega$, giving
\(
\sum_{\omega\in\operatorname{supp}(p)}|\omega|
\leq n\cdot2^{n-1}
\)
coordinates, with equality when all nonempty class menus are supported. Alternatively, directly evaluating performance under every strict class ranking guarantees the oracle welfare $W^{\mathrm{coarse}}$ asymptotically, yet entails $n!$ candidates, exceeding $n\cdot2^{n-1}$ for $n\geq5$. On the contrary, CQL maintains only $n$ valuations and $n$ visit counts, independently of the number of menus.

The burden of conditioning on menus is also statistical: each
$v(s,\omega)$ is updated only when that particular menu occurs and
class $s$ is chosen. Rarely encountered pairs therefore provide little
feedback. CQL pools observations across menus, reducing the number of
quantities to estimate. This economy comes at a cost when class payoffs
vary across menus. Current valuations determine which menus contribute
observations to each class, so the pooled assessment reflects a
composition of experiences shaped by Alice's own choices. The feedback
between valuations and this endogenous selection drives our results.

The finite-class representation also accommodates infinitely many
objective alternatives. Within-class choice can be specified by a
fixed measurable probability kernel supported on the available
alternatives. Whenever the induced reduced decision tree is finite
and satisfies our maintained sampling and moment assumptions, our
results apply to that tree.

We also note that menu variation need not reflect changes in physical availability. As the prominence application illustrates, it can instead arise from changes in which alternatives enter the consumer's consideration set. More generally, our analysis applies to such effective menus whenever their evolution satisfies the maintained independence assumptions. The similarity partition describes how Alice groups alternatives, while the considered menu determines which alternatives, and hence which classes, enter her current choice.

Finally, singleton-payoff shifts admit an interpretation as
reduced-form choice aversion or affinity across classes
\citep{iyengar2000choice}. If
$u_a(\psi)=\mu_a(\psi)+\xi(\psi)$, the component $\xi(\psi)$ is common
to all alternatives in a menu and leaves within-menu payoff
comparisons unchanged, but affects assessments pooled across menus.
Setting
\(
\xi(\psi)=z\,\mathbf 1\{|\rho[\psi]|=1\}
\)
assigns a bonus to the absence of between-class choice when $z>0$,
and a penalty when $z<0$. Under the conditions of
Theorems~\ref{th:uniquemixedstable} and~\ref{th:generallow},
sufficiently large positive shifts produce limiting indifference,
whereas sufficiently negative shifts yield multiple locally stable
SVE at high sensitivity when all binary menus are supported.

\section{Related Literature}
\label{sec:litreview}
Our closest point of contact is \citet{Jehiel2007} who introduce valuation equilibrium (VE) for extensive-form games with similarity classes. They motivate VE through learning
but leave the convergence of learning over similarity classes open. We analyze such learning explicitly in stationary single-agent bandit problems. Its steady states are smooth valuation equilibria (SVE), whose high-sensitivity limits refine VE. Using o-minimal geometry, we show that the set of limiting SVE is generically finite even when VE form continua. For strict VE, closest-loss consistency characterizes this refinement: valuations of classes left unchosen by the limiting policy must reflect
the mean payoff in menus in which their valuation disadvantage is smallest. Every strict VE satisfying this condition is approached by locally stable SVE at all sufficiently high sensitivities. Additionally, we show how the endogenous composition of observations can produce either limiting indifference or multiple self-confirming strict rankings or oscillating rankings when no equilibrium is stable.

The endogenous composition of experience also connects CQL to Bayesian learning under misspecification. Berk-Nash equilibrium \citep{Esponda} requires beliefs to best fit the observations generated by behavior, while \citet{YamamotoJET} characterize learning through the dynamics of empirical action frequencies. Our harmonic update has an exact posterior-mean interpretation under a subjective Gaussian model, although objective payoff laws need not be Gaussian. With finite reward outcomes, product-Dirichlet priors also give an exact Bayesian interpretation. Under \textit{singleton-menu support}, the argmax counterpart of CQL (optimal choices as opposed to logit-optimal) embeds in the framework of \citet{Lanzani}: VE correspond to strongly uniform Berk--Nash equilibria, and strict VE inherit local stochastic stability from sufficiently concentrated priors. Applying their strict-stability theorem to finite menu-contingent plans also gives every mixed VE a stable purification through arbitrarily small payoff shocks. The appendix develops these connections formally. They do not by themselves determine the refinement induced by logit exploration, where we fix sensitivity during learning before taking its high-sensitivity limit. We also identify general conditions for global convergence of the CQL dynamics: constant competitive-menu payoffs ensure
convergence toward the SVE set, while nonpositive payoff-crowding covariance ensures convergence to a unique SVE, even when singleton menus are not in support. 

Our learning rule builds on payoff-assessment models that update values directly from realized rewards \citep{sarin1999,Samet,COMINETTI}. Our departure is to pool rewards across alternatives in a class, so the distribution being assessed depends on the learner's choices across menus. This connects the paper to analogy-based expectation equilibrium \citep{jehiel2005analogy} in games, where players aggregate observations when forming expectations, and case-based decision theory \citep{Gilboa}, where evaluations depend on outcomes in similar past cases.\footnote{VE is also conceptually related to Personal Equilibrium \citep{Spiegler}: Alice's choices generate the data used to fit a misspecified subjective model, and the fitted beliefs/valuations in turn justify those choices.} We take the similarity partition as given and study the feedback it creates between choice and experience. The resulting distortions can persist even when every class is sampled indefinitely. Coarsening therefore changes long-run behavior, beyond its familiar benefits for memory and sample complexity.

Our cycling result contributes to the study of persistent oscillation under misspecification \citep{Nyarko,Romanyuk,YamamotoJET}. It isolates endogenous selection of observations under a fixed similarity partition as a source of non-convergence. We show the existence of an open set of decision trees with three classes and three independently drawn menus, in which sufficiently sensitive choice yields a unique unstable SVE, and the stochastic learning process almost surely has a nontrivial periodic limit set. Coarse inference can thus sustain oscillation in a stationary bandit problem even as learning rates vanish and, under the Bayesian interpretation, subjective uncertainty disappears.

The application adds a learning channel to the economic analysis of the effects of algorithmic prominence and self-preferencing by online marketplaces on consumer choice \citep{LeeMusolff2026, Hagiu} and derives novel conclusions: prominence changes the experiences pooled under a brand and thereby demand for the platform's own brand when the competitor remains available. We characterize optimal stationary prominence and show how its costs can favor deterministic policies that sustain competitive mixing, while improving consumer welfare relative to costless promotion by preserving its no-prominence level. 
Our use of similarity classes connects CQL to the theory of coarse thinking. \citet{Mullainathan} show how categorical inference transfers the meaning of information across situations, creating scope for persuasion. Our application develops a related dynamic channel of category transference: prominence changes actual purchases, and favorable experiences with promoted products sustain demand for other inferior products under the same brand. Brand assessments are thus jointly determined with the purchases they induce. 

The endogenous feedback mechanism also connects CQL to selection neglect. \citet{Jehiel2018} studies investors who assess projects using the returns of previously implemented projects without correcting for the selection governing the observed sample.\footnote{A related mechanism appears in \citet{Esponda2008}: naive buyers learn average quality from purchased goods but neglect how changing their bids alters the composition of purchases. In CQL, choices similarly determine the composition of payoff observations pooled within each class, although Alice interprets each class as a stationary mixture bandit arm.} In our setting, a single agent generates her own sample through repeated choices across menus and pools the experienced payoffs within classes. We characterize how the neglected dependence of this sample on her valuations shapes learning. We show selection can produce corrective feedback and globally convergent learning whose high-sensitivity limit sustains indifference and mixing, or reinforcing feedback supporting distinct stable strict rankings. These effects arise with stationary payoff distributions and independently drawn menus, and persist despite continued sampling and vanishing learning rates. Our welfare analysis separates the restriction imposed by common class valuations from the additional losses induced by learning those valuations as payoff assessments. An informed oracle can optimize within the same coarse representation, yet we show mere constant-payoff calibration and a monotone choice rule preclude attaining its optimal performance in some decision problems, even with arbitrary history-dependent reweighting of observations.

\section{Conclusion}
\label{sec:conclude}
Coarse Q-learning makes the composition of experience endogenous even in stationary bandit problems with independently drawn menus. Although each alternative has a fixed payoff distribution with independent draws, the mixture sampled within a class depends on the valuations guiding choice. This feedback can produce globally attracting equilibria with limiting indifference and endogenous mixing, multiple stable strict rankings, or persistent oscillations. These outcomes arise from the interaction between coarse representation and the selection of observations, without non-stationary fundamentals or strategic interaction. The welfare cost of coarse learning extends beyond the restrictions imposed by coarse representation. Even when a better policy is available under the same categories and common class scores, learning those scores as averages of experienced payoffs need not recover it. Reweighting observations across menus does not provide a general remedy: efficient scores encode priorities across menus that need not agree with the ordering of payoff estimates.

The application shows why this distinction matters for digital marketplaces. Prominence shapes the experiences pooled under a brand, so promoting good products can sustain purchases of inferior ones even when competitors remain available. With costless prominence, randomized policies can yield higher profit than every deterministic rule while reducing consumer welfare below its no-prominence level. Prominence costs can instead make deterministic policies sustaining competitive mixing optimal and preserve that welfare benchmark. These findings highlight the composition of promoted experiences, alongside the frequency of prominence, as a determinant of market outcomes.
A natural next step for future research is to study when decision-makers revise their categories, and whether they can learn better choice priorities while retaining a small set of valuations.

\bibliographystyle{apalike}
\bibliography{ref}

@String{Springer = "Springer-Verlag" }

@article{Jehiel2007,
	author = {Jehiel, Philippe and Samet, Dov},
	title = {Valuation equilibrium},
	journal = {Theoretical Economics},
	volume = {2},
	number = {2},
	year = {2007},
	pages = {163-185},
	url = {https://econtheory.org/ojs/index.php/te/article/view/20070163/0}
}

@article{Benaim1999,
     author = {Bena{\"\i}m, Michel},
     title = {Dynamics of stochastic approximation algorithms},
     journal = {S\'eminaire de probabilit\'es de Strasbourg},
     pages = {1--68},
     publisher = {Springer - Lecture Notes in Mathematics},
     volume = {33},
     year = {1999},
     zbl = {0955.62085},
     mrnumber = {1767993},
     language = {en},
     url = {http://www.numdam.org/item/SPS_1999__33__1_0/}
}

@article{COMINETTI,
title = {A payoff-based learning procedure and its application to traffic games},
journal = {Games and Economic Behavior},
volume = {70},
number = {1},
pages = {71-83},
year = {2010},
issn = {0899-8256},
doi = {https://doi.org/10.1016/j.geb.2008.11.012},
url = {https://www.sciencedirect.com/science/article/pii/S0899825608002200},
author = {Roberto Cominetti and Emerson Melo and Sylvain Sorin}
}

@article{Watkins1992,
  author = {Watkins, Christopher J. C. H. and Dayan, Peter},
  day = 01,
  doi = {10.1007/BF00992698},
  issn = {1573-0565},
  journal = {Machine Learning},
  month = may,
  number = 3,
  pages = {279--292},
  title = {Q-learning},
  url = {https://doi.org/10.1007/BF00992698},
  volume = 8,
  year = 1992
}

@book{Sutton1998,
  author = {Sutton, Richard S. and Barto, Andrew G.},
  edition = {Second},
  publisher = {The MIT Press},
  title = {Reinforcement Learning: An Introduction},
  url = {http://incompleteideas.net/book/the-book-2nd.html},
  year = {2018 }
}

@article{sarin1999,
  title={Payoff assessments without probabilities: A simple dynamic model of choice},
  author={Sarin, Rajiv and Vahid, Farshid},
  journal={Games and Economic Behavior},
  volume={28},
  number={2},
  pages={294--309},
  year={1999},
  publisher={Elsevier}
}

@book{rosch1978principles,
  title={Principles of categorization},
  author={Rosch, Eleanor and Lloyd, Barbara B},
  year={1978},
  publisher={MIT press}
}

@book{fudenberg1998theory,
  title={The Theory of Learning in Games},
  author={Fudenberg, D. and Levine, D.K.},
  isbn={9780262061940},
  lccn={97039957},
  url={https://books.google.fr/books?id=G6vTQFluxuEC},
  year={1998},
  publisher={MIT Press}
}

@article{hofbauer2002global,
  title={On the global convergence of stochastic fictitious play},
  author={Hofbauer, Josef and Sandholm, William H},
  journal={Econometrica},
  volume={70},
  number={6},
  pages={2265--2294},
  year={2002},
  publisher={Wiley Online Library}
}

@article{jehiel2005analogy,
  title={Analogy-based expectation equilibrium},
  author={Jehiel, Philippe},
  journal={Journal of Economic theory},
  volume={123},
  number={2},
  pages={81--104},
  year={2005},
  publisher={Elsevier}
}

@article{iyengar2000choice,
  title={When choice is demotivating: Can one desire too much of a good thing?},
  author={Iyengar, Sheena S and Lepper, Mark R},
  journal={Journal of personality and social psychology},
  volume={79},
  number={6},
  pages={995},
  year={2000},
  publisher={American Psychological Association}
}

@book{kushner2003,
  title={Stochastic Approximation and Recursive Algorithms and Applications},
  author={Kushner, H. and Yin, G.G.},
  isbn={9780387008943},
  lccn={03045459},
  series={Stochastic Modelling and Applied Probability},
  url={https://books.google.fr/books?id=_0bIieuUJGkC},
  year={2003},
  publisher={Springer New York}
}

@article{Robbins-et-al,
author = {Herbert Robbins and Sutton Monro},
title = {{A Stochastic Approximation Method}},
volume = {22},
journal = {The Annals of Mathematical Statistics},
number = {3},
publisher = {Institute of Mathematical Statistics},
pages = {400 -- 407},
year = {1951},
doi = {10.1214/aoms/1177729586},
URL = {https://doi.org/10.1214/aoms/1177729586}
}

@article{Tversky,
    author = {Tversky, A.},
    title = {Elimination by aspects: a theory of choice.},
    journal = {Psychological Review},
    year = {1972},
    volume = {79},
    issue = {4},
    doi = {10.1037/h0032955}
}

@article{Shleifer2012,
    author = {Bordalo, Pedro and Gennaioli, Nicola and Shleifer, Andrei},
    title = "{Salience Theory of Choice Under Risk}",
    journal = {The Quarterly Journal of Economics},
    volume = {127},
    number = {3},
    pages = {1243-1285},
    year = {2012},
    month = {06},
    issn = {0033-5533},
    doi = {10.1093/qje/qjs018},
    url = {https://doi.org/10.1093/qje/qjs018},
    eprint = {https://academic.oup.com/qje/article-pdf/127/3/1243/30456683/qjs018.pdf},
}

@article{Shleifer2013,
 ISSN = {00223808, 1537534X},
 URL = {http://www.jstor.org/stable/10.1086/673885},
 author = {Pedro Bordalo and Nicola Gennaioli and Andrei Shleifer},
 journal = {Journal of Political Economy},
 number = {5},
 pages = {803--843},
 publisher = {The University of Chicago Press},
 title = {Salience and Consumer Choice},
 urldate = {2024-09-21},
 volume = {121},
 year = {2013}
}

@book{Luce1959,
  title={Individual choice behavior},
  author={Luce, R Duncan},
  volume={4},
  year={1959},
  publisher={Wiley New York}
}

@article{Debreu1960,
 ISSN = {00028282},
 URL = {http://www.jstor.org/stable/1813477},
 title = {Review of Individual Choice Behavior},
 author = {Gerard Debreu},
 journal = {The American Economic Review},
 number = {1},
 pages = {186--188},
 publisher = {American Economic Association},
 reviewed-author = {R. Duncan Luce},
 urldate = {2025-09-24},
 volume = {50},
 year = {1960}
}

@article{anderson1991adaptive,
  title={The adaptive nature of human categorization.},
  author={Anderson, John R},
  journal={Psychological review},
  volume = {98(3)},
  pages = {409},
  year={1991},
  publisher={American Psychological Association}
}

@book{Dries1998, 
    place={Cambridge}, 
    series={London Mathematical Society Lecture Note Series}, 
    title={Tame Topology and O-minimal Structures}, 
    publisher={Cambridge University Press}, 
    author={van den Dries, L.}, year={1998}, 
    collection={London Mathematical Society Lecture Note Series}}

@article{LVD,
 ISSN = {0003486X, 19398980},
 URL = {http://www.jstor.org/stable/2118545},
 author = {Lou van den Dries and Angus Macintyre and David Marker},
 journal = {Annals of Mathematics},
 number = {1},
 pages = {183--205},
 publisher = {[Annals of Mathematics, Trustees of Princeton University on Behalf of the Annals of Mathematics, Mathematics Department, Princeton University]},
 title = {The Elementary Theory of Restricted Analytic Fields with Exponentiation},
 urldate = {2025-11-20},
 volume = {140},
 year = {1994}
}

@article{DiaconisYlvisaker1979,
  author  = {Diaconis, Persi and Ylvisaker, Donald},
  title   = {Conjugate Priors for Exponential Families},
  journal = {The Annals of Statistics},
  year    = {1979},
  volume  = {7},
  number  = {2},
  pages   = {269-281}
}

@article{Romanyuk,
author = {Fudenberg, Drew and Romanyuk, Gleb and Strack, Philipp},
title = {Active learning with a misspecified prior},
journal = {Theoretical Economics},
volume = {12},
number = {3},
pages = {1155-1189},
doi = {https://doi.org/10.3982/TE2480},
url = {https://onlinelibrary.wiley.com/doi/abs/10.3982/TE2480},
eprint = {https://onlinelibrary.wiley.com/doi/pdf/10.3982/TE2480},
year = {2017}
}

@article{Nyarko,
title = {Learning in mis-specified models and the possibility of cycles},
journal = {Journal of Economic Theory},
volume = {55},
number = {2},
pages = {416-427},
year = {1991},
issn = {0022-0531},
doi = {https://doi.org/10.1016/0022-0531(91)90047-8},
url = {https://www.sciencedirect.com/science/article/pii/0022053191900478},
author = {Yaw Nyarko}
}

@article{Lanzani,
author = {Fudenberg, Drew and Lanzani, Giacomo and Strack, Philipp},
title = {Limit Points of Endogenous Misspecified Learning},
journal = {Econometrica},
volume = {89},
number = {3},
pages = {1065-1098},
doi = {https://doi.org/10.3982/ECTA18508},
url = {https://onlinelibrary.wiley.com/doi/abs/10.3982/ECTA18508},
eprint = {https://onlinelibrary.wiley.com/doi/pdf/10.3982/ECTA18508},
year = {2021}
}

@article{Esponda,
author = {Esponda, Ignacio and Pouzo, Demian},
title = {Berk–Nash Equilibrium: A Framework for Modeling Agents With Misspecified Models},
journal = {Econometrica},
volume = {84},
number = {3},
pages = {1093-1130},
doi = {https://doi.org/10.3982/ECTA12609},
url = {https://onlinelibrary.wiley.com/doi/abs/10.3982/ECTA12609},
eprint = {https://onlinelibrary.wiley.com/doi/pdf/10.3982/ECTA12609},
year = {2016}
}

@article{YamamotoJET,
title = {Asymptotic behavior of Bayesian learners with misspecified models},
journal = {Journal of Economic Theory},
volume = {195},
pages = {105260},
year = {2021},
issn = {0022-0531},
doi = {https://doi.org/10.1016/j.jet.2021.105260},
url = {https://www.sciencedirect.com/science/article/pii/S0022053121000776},
author = {Ignacio Esponda and Demian Pouzo and Yuichi Yamamoto}
}

@article{Gilboa,
 title = {Case-Based Decision Theory},
 ISSN = {00335533, 15314650},
 URL = {http://www.jstor.org/stable/2946694},
 author = {Itzhak Gilboa and David Schmeidler},
 journal = {The Quarterly Journal of Economics},
 number = {3},
 pages = {605--639},
 publisher = {Oxford University Press},
 urldate = {2025-12-18},
 volume = {110},
 year = {1995}
}

@article{Samet,
title = {Learning to play games in extensive form by valuation},
journal = {Journal of Economic Theory},
volume = {124},
number = {2},
pages = {129-148},
year = {2005},
issn = {0022-0531},
doi = {https://doi.org/10.1016/j.jet.2004.09.004},
url = {https://www.sciencedirect.com/science/article/pii/S0022053104002169},
author = {Philippe Jehiel and Dov Samet}
}

@article{Glimcher2011,
    author = {Paul W. Glimcher },
    title = {Understanding dopamine and reinforcement learning: The dopamine reward prediction error hypothesis},
    journal = {PNAS},
    volume = {108},
    number = {3},
    pages = {15647-15654},
    year = {2011},
    doi = {10.1073/pnas.1014269108},
    URL = {https://www.pnas.org/doi/abs/10.1073/pnas.1014269108},
    eprint = {https://www.pnas.org/doi/pdf/10.1073/pnas.1014269108}
}

@article{Caplin2008,
    author = {Caplin, Andrew and Dean, Mark},
    title = {Dopamine, Reward Prediction Error, and Economics},
    journal = {The Quarterly Journal of Economics},
    volume = {123},
    number = {2},
    pages = {663-701},
    year = {2008},
    month = {05},
    issn = {0033-5533},
    doi = {10.1162/qjec.2008.123.2.663},
    url = {https://doi.org/10.1162/qjec.2008.123.2.663},
    eprint = {https://academic.oup.com/qje/article-pdf/123/2/663/5441135/123-2-663.pdf},
}

@article{breitmoser,
  title={An axiomatic foundation of conditional logit},
  author={Breitmoser, Yves},
  journal={Economic Theory},
  volume={72},
  number={1},
  pages={245--261},
  year={2021},
  publisher={Springer}
}

@book{Skinner,
  title={The Behavior of Organisms: An Experimental Analysis},
  author={Skinner, B.F.},
  isbn={9780996453905},
  url={https://books.google.fr/books?id=S9WNCwAAQBAJ},
  year={1938},
  publisher={}
}

@article{Spiegler,
 ISSN = {00335533, 15314650},
 URL = {https://www.jstor.org/stable/26372664},
 author = {Ran Spiegler},
 journal = {The Quarterly Journal of Economics},
 number = {3},
 pages = {1243--1290},
 publisher = {Oxford University Press},
 title = {BAYESIAN NETWORKS AND BOUNDEDLY RATIONAL EXPECTATIONS},
 urldate = {2026-04-24},
 volume = {131},
 year = {2016}
}

@article{LeeMusolff2026,
  title={Two-sided markets shaped by platform-guided search},
  author={Lee, Kwok Hao and Musolff, Leon},
  journal={},
  year={2025}
}

@article{HOFBAUER2009,
title = {Stable games and their dynamics},
journal = {Journal of Economic Theory},
volume = {144},
number = {4},
pages = {1665-1693.e4},
year = {2009},
issn = {0022-0531},
doi = {https://doi.org/10.1016/j.jet.2009.01.007},
url = {https://www.sciencedirect.com/science/article/pii/S0022053109000222},
author = {Josef Hofbauer and William H. Sandholm}
}

@article{Fradkin,
Author = {Farronato, Chiara and Fradkin, Andrey and MacKay, Alexander},
Title = {Self-Preferencing at Amazon: Evidence from Search Rankings},
Journal = {AEA Papers and Proceedings},
Volume = {113},
Year = {2023},
Month = {May},
Pages = {239–43},
DOI = {10.1257/pandp.20231068},
URL = {https://www.aeaweb.org/articles?id=10.1257/pandp.20231068}}

@article{Mullainathan,
    author = {Mullainathan, Sendhil and Schwartzstein, Joshua and Shleifer, Andrei},
    title = {Coarse Thinking and Persuasion},
    journal = {The Quarterly Journal of Economics},
    volume = {123},
    number = {2},
    pages = {577-619},
    year = {2008},
    month = {05},
    issn = {0033-5533},
    doi = {10.1162/qjec.2008.123.2.577},
    url = {https://doi.org/10.1162/qjec.2008.123.2.577},
    eprint = {https://academic.oup.com/qje/article-pdf/123/2/577/5441119/123-2-577.pdf},
}

@article{Hagiu,
author = {Hagiu, Andrei and Teh, Tat-How and Wright, Julian},
title = {Should platforms be allowed to sell on their own marketplaces?},
journal = {The RAND Journal of Economics},
volume = {53},
number = {2},
pages = {297-327},
doi = {https://doi.org/10.1111/1756-2171.12408},
url = {https://onlinelibrary.wiley.com/doi/abs/10.1111/1756-2171.12408},
eprint = {https://onlinelibrary.wiley.com/doi/pdf/10.1111/1756-2171.12408},
year = {2022}
}

@article{Jehiel2018,
ISSN = {00028282, 19447981},
URL = {https://www.jstor.org/stable/26528018},
author = {Philippe Jehiel},
journal = {The American Economic Review},
number = {6},
pages = {1582--1597},
publisher = {American Economic Association},
title = {Investment Strategy and Selection Bias: An Equilibrium Perspective on Overoptimism},
urldate = {2026-10-21},
volume = {108},
year = {2018}
}

@article{Esponda2008,
Author = {Esponda, Ignacio},
Title = {Behavioral Equilibrium in Economies with Adverse Selection},
Journal = {American Economic Review},
Volume = {98},
Number = {4},
Year = {2008},
Month = {September},
Pages = {1269–91},
DOI = {10.1257/aer.98.4.1269},
URL = {https://www.aeaweb.org/articles?id=10.1257/aer.98.4.1269}}

@book{Kuznetsov2013,
  title={Elements of Applied Bifurcation Theory},
  author={Kuznetsov, Y.A.},
  isbn={9781475724219},
  lccn={94041847},
  series={Applied Mathematical Sciences},
  url={https://books.google.fr/books?id=ZFntBwAAQBAJ},
  year={2013},
  publisher={Springer New York}
}

@article{Borkar1997,
  author={Borkar, Vivek S.},
  title={Stochastic approximation with two time scales},
  journal={Systems \& Control Letters},
  volume={29}, number={5}, pages={291--294}, year={1997},
  doi={10.1016/S0167-6911(97)90015-3}
}

@article{BD1996,
  author={Brandi\`ere, Odile and Duflo, Marie},
  title={Les algorithmes stochastiques contournent-ils les pi\`eges?},
  journal={Annales de l'Institut Henri Poincar\'e, Probabilit\'es et Statistiques},
  volume={32}, number={3}, pages={395--427}, year={1996},
  url={https://www.numdam.org/item/AIHPB_1996__32_3_395_0/}
}

\clearpage
\appendix
\singlespacing
\fontsize{11}{13.2}\selectfont
\setlength{\parskip}{0pt plus 1pt}
\setlength{\parindent}{17pt}
\numberwithin{equation}{section}
\renewcommand{\theHequation}{appendix.\thesection.\arabic{equation}}
\let\theorem\apptheorem
\let\endtheorem\endapptheorem
\let\proposition\appproposition
\let\endproposition\endappproposition
\let\lemma\applemma
\let\endlemma\endapplemma
\let\corollary\appcorollary
\let\endcorollary\endappcorollary
\let\definition\appdefinition
\let\enddefinition\endappdefinition
\let\assumption\appassumption
\let\endassumption\endappassumption
\let\remark\appremark
\let\endremark\endappremark
\startcontents[appendix]
\begin{center}
{\LARGE\bfseries Appendix}
\end{center}

\noindent This appendix supplies proofs for the theorems, propositions, lemma, and corollaries in the manuscript.

\paragraph{Maintained primitives and convention:}
There are finitely many alternatives $a \in \mathcal A$ with stationary independent reward laws $R_a$ (with finite second moments uniformly bounded by $M_2$), and a fixed partition into similarity classes $\mathcal S$. Menus are drawn independently over time, independently of rewards. Within a chosen class, the tie-breaking kernel is fixed and exogenous. Its reduction gives a nonempty finite supported class-menu set $\Omega^+$, probabilities $p(\omega)>0$ summing to one, and conditional class-menu means $\pi_s(\omega)$. We discard classes with zero availability: $P_s=\sum_{\omega\ni s}p(\omega)>0$ for every $s\in\mathcal S$. At each fixed finite sensitivity $\beta\geq0$, write
\[
 \sigma_\omega^s(v;\beta)=
 \frac{\exp(\beta v_s)}{\sum_{r\in\omega}\exp(\beta v_r)},\qquad
 \Lambda_s(v;\beta)=\sum_{\omega\ni s}p(\omega)\sigma_\omega^s(v;\beta),
\]
with $\sigma_\omega^s=0$ for $s\notin\omega$. Let $N_k(s)$ count class visits before period $k$, let $\kappa_s>0$, and let $\chi=\sum_s\kappa_s$. The maintained recursion is
\[
 v_{k+1}(s)=v_k(s)+
 \frac{\mathbf1\{s_k=s\}}{\kappa_s+N_k(s)+1}
          \bigl(r_k-v_k(s)\bigr),\qquad
 x_{k,s}=\frac{\kappa_s+N_k(s)}{\chi+k}.
\]
Sensitivity remains fixed along a learning path; limits as $\beta\to\infty$ compare the long-run behavior of different fixed environments (except for the impossibility result in the section on welfare analysis where we explicitly permit time-varying sensitivity). Statements about
stochastic periodic limit sets concern the harmonic learning clock.

\paragraph{Proof guide:}
The following table identifies every formal claim in the paper. Supporting claims are proved in the adjacent subsections, including attainability, the VE continuum, the logit axioms, the two-class illustrations, and cycling minimality.

\begin{longtable}{@{}>{\raggedright\arraybackslash}p{0.65\linewidth}
 >{\raggedright\arraybackslash}p{0.29\linewidth}@{}}
\toprule
Main-text statement & Appendix proof\\\midrule
Continuous-time approximation & Propositions~\ref{app:model:boundedness}, \ref{app:model:apt}\\
Equilibrium and local-stability correspondence & Proposition~\ref{app:model:spectrum}\\
Structure of SVE and high-sensitivity limits & Theorem~\ref{app:eq:structure}\\
Strict VE characterization and local stability & Theorem~\ref{app:eq:strict}\\
Strict-equilibrium coverage corollary & Corollary~\ref{app:eq:coverage}\\
Constant competitive-menu payoffs & Theorem~\ref{app:potential:theorem}\\
Robustness to payoff heterogeneity & Corollary~\ref{app:potential:robust}\\
Convergence with two classes & Corollary~\ref{app:potential:two}\\
Convergence under nonpositive payoff-crowding covariance & Theorem~\ref{app:single:convergence}\\
High singleton shifts and Limiting Indifference & Theorem~\ref{app:single:high}\\
Total indifference under monotone menu probabilities & Proposition~\ref{app:single:total}\\
Additional primitive conditions for corrective feedback & Proposition~\ref{app:single:primitive-crowding}\\
Low singleton shifts and strict multiplicity & Theorem~\ref{app:single:low}\\
Every strict ranking can be sustained & Proposition~\ref{app:single:all-rankings}\\
Nearly equal singleton probabilities & Proposition~\ref{app:single:near-equal}\\
Persistent cycling & Theorem~\ref{app:cycle:th:robust}, Proposition~\ref{app:cycle-normalized:theorem}\\
Common-score oracle benchmark and efficiency & Proposition~\ref{app:welfare:oracle}\\
Efficient scores in own-payoff intervals & Proposition~\ref{app:welfare:score-feasibility}\\
Path-based welfare bounds & Lemma~\ref{app:welfare:path-errors}, Proposition~\ref{app:welfare:path-bound}, Corollary~\ref{app:welfare:path-online}\\
Singleton welfare and the wine illustrations & Lemma~\ref{app:welfare:singletons}, equation~\eqref{app:welfare:wines}\\
Limits of payoff aggregation & Proposition~\ref{app:welfare:static}, Theorem~\ref{app:welfare:online}\\
Learning under stationary prominence & Proposition~\ref{app:platform:convergence}, Lemma~\ref{app:platform:limits}\\
Two-state illustration and its welfare gap & Lemma~\ref{app:platform:two}\\
Optimal deterministic prominence (three states) & Proposition~\ref{app:platform:three}\\
Optimal randomized prominence & Proposition~\ref{app:platform:randomized}\\
Approximately optimal prominence and welfare & Corollary~\ref{app:platform:random-examples}\\
Heterogeneous consumer partitions & Proposition~\ref{app:platform:heterogeneous}, Corollary~\ref{app:platform:heterogeneous-examples}\\
Prominence caps and welfare & Propositions~\ref{app:platform:caps}, \ref{app:platform:caps-welfare}, \ref{app:platform:caps-example}\\
General policy reduction and randomization bounds & Lemma~\ref{app:platform:policy-reduction}, Corollary~\ref{app:platform:sparse-optima}\\
Optimal costly prominence and the example & Propositions~\ref{app:platform:costly}, \ref{app:platform:costly-example}\\
Bayesian learning: equilibrium, strict stability, and purification & Proposition~\ref{prop:fls-connection}\\
Convergence in balanced binary constant-sum decision trees & Proposition~\ref{app:balanced:global} \\
\bottomrule
\end{longtable}

\medskip
\begingroup\small
\subsection*{Appendix contents}
\printcontents[appendix]{}{1}{\setcounter{tocdepth}{1}}
\endgroup
\clearpage
\section{The stochastic learning process and its mean field ODE limit}
\label{app:model}

This section proves the two propositions in the model section and the
elementary benchmark claims. Throughout, the sensitivity $\beta$ is fixed
and finite. Write $m=|\mathcal S|$, omit classes that have zero availability,
and let
\[
 P_s=\sum_{\omega\ni s}p(\omega)>0,\qquad
 R_s(v)=\sum_{\omega\ni s}p(\omega)\sigma^s_\omega(v)\pi_s(\omega),
 \qquad g_s(v)=\frac{R_s(v)}{\Lambda_s(v)}.
\]
Here $R_s$ denotes an unconditional expected reward contribution of a class, not a
primitive reward distribution. Let $\chi=\sum_s\kappa_s$ and
$\gamma_{k+1}=(\chi+k+1)^{-1}$. All initial valuations are finite, and
$\kappa_s>0$. An interior of a probability simplex always means its
\emph{relative} interior.

\subsection{Reduction and deterministic preliminary facts}

\begin{lemma}[Reduction to class menus]
\label{app:model:reduction}
Fix a within-class tie-breaking kernel $\Upsilon(a\mid s,\psi)$,
independent of time, past observations, and valuations, with support in
$s\cap\psi$. Define
\[
 \mu_s(\psi)=\sum_{a\in s\cap\psi}\Upsilon(a\mid s,\psi)\mu_a,
 \quad
 p(\omega)=\sum_{\psi:\rho[\psi]=\omega}f(\psi),
 \quad
 \pi_s(\omega)=\frac{\sum_{\psi:\rho[\psi]=\omega}f(\psi)\mu_s(\psi)}
 {p(\omega)}
\]
for $p(\omega)>0$. These reduced primitives preserve $\Lambda$, $R$, and
$g$, and therefore preserve both mean-field systems. Uniform within-class
selection is the special case $\Upsilon(a\mid s,\psi)=|s\cap\psi|^{-1}$.
\end{lemma}
\begin{proof}
The class choice probability is the same in all objective menus inducing
$\omega$. For such a menu and $s\in\omega$, Bayes' law gives
\[
 \Pr(\psi\mid\omega,s\text{ chosen},v)
 =\frac{f(\psi)\sigma^s_\omega(v)}
 {p(\omega)\sigma^s_\omega(v)}=\frac{f(\psi)}{p(\omega)}.
\]
Consequently the conditional reward mean is $\pi_s(\omega)$. Summing over
class menus gives the stated expressions for $\Lambda$ and $R$.
In fact the conditional reward law is the fixed mixture
\[
 \widetilde R_{s,\omega}
 =\sum_{\psi:\rho[\psi]=\omega}\frac{f(\psi)}{p(\omega)}
   \sum_{a\in s\cap\psi}\Upsilon(a\mid s,\psi)R_a.
\]
Using this mixture preserves the law of the observed class-menu, class,
and reward process, conditional on its history.
\end{proof}

Let \(\underline\pi_s=\min_{\omega\ni s:\,p(\omega)>0}\pi_s(\omega),\) and \( \overline\pi_s=\max_{\omega\ni s:\,p(\omega)>0}\pi_s(\omega).\)
\begin{lemma}[Well-posedness, boundedness, and positive frequencies]
\label{app:model:ode-compactness}
The normalized system $\dot v=g(v)-v$ and the joint asynchronous system
\begin{equation}
 \dot v_s=\frac{\Lambda_s(v)}{x_s}(g_s(v)-v_s),
 \qquad \dot x_s=\Lambda_s(v)-x_s
 \label{app:model:joint}
\end{equation}
have unique global forward solutions from every finite valuation vector
and, for the joint asynchronous system, every $x(0)\in\operatorname{ri}\Delta(\mathcal S)$.
Each valuation remains in
\[
 [\min\{v_s(0),\underline\pi_s\},
   \max\{v_s(0),\overline\pi_s\}].
\]
Along each joint trajectory, $\inf_{t\geq0}x_s(t)>0$. Every equilibrium
valuation belongs to $\prod_s[\underline\pi_s,\overline\pi_s]$, and at
least one equilibrium exists.
\end{lemma}
\begin{proof}
For finite $\beta$, all available-class logits and all $\Lambda_s$ are
strictly positive, and $g,\Lambda$ are smooth. Since $g_s(v)$ is a convex
combination of $\pi_s(\omega)$, the valuation vector field points inward
on the indicated coordinate intervals. Thus the valuation box is forward
invariant, for both systems. On this compact box there is $\ell_s>0$
with $\Lambda_s(v)\geq\ell_s$. The frequency equation gives
\begin{equation}
 x_s(t)=e^{-t}x_s(0)+\int_0^t e^{-(t-u)}\Lambda_s(v(u))\,du
 \geq e^{-t}x_s(0)+(1-e^{-t})\ell_s.
 \label{app:model:frequency-filter}
\end{equation}
Moreover $\sum_s\dot x_s=1-\sum_sx_s$, so the simplex is invariant.
The solution stays in a compact subset of the domain of its smooth vector
field, which proves global existence. Equilibrium valuations satisfy
$v=g(v)$ and hence belong to the payoff box. The continuous map $g$ maps
that nonempty compact convex box into itself, so Brouwer's theorem yields
a fixed point, including when some coordinate intervals are singletons.
Together with $x=\Lambda(v)$ it is a joint equilibrium.
\end{proof}
The auxiliary coordinates
\begin{equation}
 y_s=x_sv_s
 \label{app:model:sufficient-statistics}
\end{equation}
will be useful repeatedly. On $x_s>0$, they are a smooth change of
coordinates, and the joint system becomes
\begin{equation}
 \dot x_s=\Lambda_s(y/x)-x_s,\qquad
 \dot y_s=R_s(y/x)-y_s,
 \label{app:model:xy-ode}
\end{equation}
where division is coordinatewise. This follows by differentiating
$y_s=x_sv_s$; the terms involving $x_s v_s$ cancel exactly.

\subsection{Adaptive sampling and the second-moment assumption}

Let $\mathcal F_k$ contain the entire history before period $k$'s realized menu and choice, including initial conditions and past objective menus, chosen alternatives, and observed rewards. The current valuation and visit counts are $\mathcal F_k$-measurable. Let $\mathcal G_k$ augment $\mathcal F_k$ by the current menu and the random choices, before the reward is observed.
The independence assumptions give
\begin{equation}
 \mathbb E[r_k\mid\mathcal G_k]=\mu_{a_k},\qquad
 \mathbb E[r_k^2\mid\mathcal G_k]\leq M_2.
 \label{app:model:conditional-rewards}
\end{equation}
The larger filtration is used only for the proof; it grants Alice no
additional information.
\begin{lemma}[A visit-count martingale law]
\label{app:model:visit-martingale}
Let $I_{k,s}=\mathbf1\{s_k=s\}$ and
$\xi_{k,s}=I_{k,s}(r_k-\mu_{a_k})$. Almost surely, for every class $s$,
the series
\begin{equation}
 \sum_{k\geq0}\frac{\xi_{k,s}}{\kappa_s+N_k(s)+1}
 \label{app:model:visit-mart-series}
\end{equation}
converges. On $\{N_k(s)\to\infty\}$,
\begin{equation}
 \frac{\sum_{j<k}\xi_{j,s}}{\kappa_s+N_k(s)}\longrightarrow0.
 \label{app:model:visit-slln}
\end{equation}
\end{lemma}
\begin{proof}
By iterated conditioning in \eqref{app:model:conditional-rewards},
$\mathbb E[\xi_{k,s}\mid\mathcal F_k]=0$ and
\[
 \mathbb E[\xi_{k,s}^2\mid\mathcal F_k]
 \leq M_2\,\mathbb E[I_{k,s}\mid\mathcal F_k].
\]
The denominator in \eqref{app:model:visit-mart-series} is predictable.
Furthermore, by monotone convergence and by counting the successive
actual visits to $s$,
\[
 \sum_{k\geq0}\mathbb E\left[
 \frac{\xi_{k,s}^2}{(\kappa_s+N_k(s)+1)^2}\right]
 \leq M_2\mathbb E\sum_{k\geq0}
 \frac{I_{k,s}}{(\kappa_s+N_k(s)+1)^2}
 \leq M_2\sum_{j\geq1}\frac1{(\kappa_s+j)^2}<\infty.
\]
The partial sums are therefore an $L^2$-bounded martingale and converge
almost surely. On a sample path with infinitely many visits, discard the
zero terms and index the remaining terms by visit number. Kronecker's
lemma with denominators $\kappa_s+j$ proves \eqref{app:model:visit-slln}. No positive lower bound on sampling frequencies is used in this lemma.
\end{proof}

\begin{proposition}[Boundedness and infinite sampling]
\label{app:model:boundedness}
Under the maintained reward assumptions,
\[
 \sup_k\|v_k\|<\infty,\qquad
 \liminf_{k\to\infty}\min_sx_{k,s}>0
 \quad\hbox{almost surely},
 \qquad x_{k,s}=\frac{\kappa_s+N_k(s)}{\chi+k}.
\]
In particular every class is visited infinitely often, and its effective
learning rates have divergent sum and square-summable squares.
\end{proposition}
\begin{proof}
Induction in the update rule gives the exact identity
\begin{equation}
 v_k(s)=\frac{\kappa_sv_0(s)+\sum_{j<k}I_{j,s}r_j}
 {\kappa_s+N_k(s)}.
 \label{app:model:sample-average}
\end{equation}
Replace each reward by $\mu_{a_j}+(r_j-\mu_{a_j})$. On a path with
infinitely many visits, the contribution of the second term tends to zero
by Lemma~\ref{app:model:visit-martingale}; the first term is a convex
combination of finitely many finite means and the initial value. Hence
the valuation sequence for that class is bounded. On a path with finitely
many visits it is eventually constant and has only finitely many finite
values. Taking the intersection of these probability-one events over the
finite class set proves boundedness before invoking continued sampling.
The bounded martingale differences $I_{k,s}-\Lambda_s(v_k)$ satisfy
\begin{equation}
 \frac{N_k(s)}k-\frac1k\sum_{j<k}\Lambda_s(v_j)\longrightarrow0
 \quad\hbox{almost surely}.
 \label{app:model:frequency-slln}
\end{equation}
For completeness, convergence of
$\sum_k[I_{k,s}-\Lambda_s(v_k)]/(k+1)$ follows from square summability of
its variances, and Kronecker's lemma gives the assertion. On any path
with $\sup_k\|v_k\|_\infty\leq B<\infty$,
\[
 \Lambda_s(v_k)\geq \frac{P_s}{m}e^{-2\beta B}>0.
\]
Combining this pathwise bound with \eqref{app:model:frequency-slln} proves
the positive lower limiting frequency and thus infinitely many visits.
Finally, the nonzero effective gains are precisely
$(\kappa_s+j)^{-1}$, $j=1,2,\dots$. Their first powers sum to infinity
and their squares have finite sum.
\end{proof}

Because the primitive family is finite and stationary, the maintained
second-moment bound also implies uniform integrability of the squared
primitive rewards.

\subsection{Proof of the continuous-time approximation}

\begin{proposition}[Asymptotic pseudo-trajectory]
\label{app:model:apt}
Let $\tau_0=0$, $\tau_k=\sum_{j<k}\gamma_{j+1}$, and let $Z(t)$ be the
affine interpolation of $(v_k,x_k)$ at these times. Almost surely, for
every $T>0$,
\[
 \lim_{t\to\infty}\sup_{0\leq h\leq T}
 \|Z(t+h)-\Phi_h(Z(t))\|=0,
\]
where $\Phi$ is the flow of \eqref{app:model:joint}. The set of
subsequential limits is nonempty, compact, connected, invariant, and
internally chain transitive for this flow.
\end{proposition}
\begin{proof}
Define $y_{k,s}=x_{k,s}v_k(s)$. By
\eqref{app:model:sample-average}, the following identities are exact:
\begin{align}
 x_{k+1,s}-x_{k,s}&=\gamma_{k+1}(I_{k,s}-x_{k,s}),\notag\\
 y_{k+1,s}-y_{k,s}&=\gamma_{k+1}(I_{k,s}r_k-y_{k,s}).
 \label{app:model:xy-sa}
\end{align}
Write $W_k=(x_k,y_k)$, and let $F$ be the vector field in
\eqref{app:model:xy-ode}. Then
\[
 W_{k+1}-W_k=\gamma_{k+1}\{F(W_k)+U_{k+1}\},
 \qquad
 U_{k+1}=
 \begin{pmatrix}I_k-\Lambda(v_k)\\ I_kr_k-R(v_k)\end{pmatrix},
\]
where $I_k$ is the chosen-class unit vector. The conditional expectation of
$U_{k+1}$ is zero. Since $\|I_k\|^2=1$ and
$\|I_kr_k\|^2=r_k^2$, its conditional second moment is bounded by
$1+M_2$ (or any larger fixed constant). Thus
\begin{equation}
 \sum_k\gamma_{k+1}U_{k+1}\quad\hbox{converges almost surely in every
 coordinate}.
 \label{app:model:calendar-martingale}
\end{equation}
This uses $\sum_k\gamma_{k+1}^2<\infty$.
By Proposition~\ref{app:model:boundedness}, $W_k$ eventually lies in a
compact subset of the open set $\{(x,y):x_s>0\}$. On a slightly larger
compact set $F$ is bounded and Lipschitz. The affine interpolation
$\overline W$ therefore satisfies, on every fixed future time window,
the integral equation for $F$ up to an error tending uniformly to zero.
Here are the two ingredients. First, every tail partial sum of the noise
series in \eqref{app:model:calendar-martingale} is uniformly small; the
partial first and last interpolation intervals are also small because
$\gamma_{k+1}U_{k+1}\to0$. Second,
$W_{k+1}-W_k\to0$, and the error from replacing $F(\overline W(u))$ by
the piecewise constant $F(W_k)$ is bounded by the Lipschitz constant
times the largest mesh increment, multiplied by the window length.
Gronwall's inequality now gives the asymptotic-pseudo-trajectory property
for \eqref{app:model:xy-ode}. This is also the criterion in
\citet[Proposition~4.1]{Benaim1999}.
The smooth coordinate map $Q(x,y)=(y/x,x)$ is uniformly Lipschitz on
the eventual compact set and maps the flow of
\eqref{app:model:xy-ode} to that of \eqref{app:model:joint}.
Consequently $Q(\overline W)$ is an asymptotic pseudo-trajectory for the
joint flow. Its distance from the affine interpolation $Z$ tends to zero:
on each mesh interval both are uniformly close to the common left
endpoint, since the mesh increments tend to zero. Hence $Z$ has the same
property and the same limit set. Precompactness and
\citet[Theorem~5.7]{Benaim1999} give internal chain transitivity and the
stated compactness, connectedness, and invariance properties.
\end{proof}

\begin{corollary}[Consequences and their scope]
\label{app:model:apt-consequences}
If the joint flow has a globally asymptotically stable equilibrium, the
stochastic process converges to it almost surely. If $v_k$ converges
almost surely on an event, its limit is an SVE on that event and
$x_k\to\Lambda(\lim_kv_k)$.
\end{corollary}
\begin{proof}
A compact internally chain transitive set intersecting the basin of an
attractor lies in that attractor; apply this standard attractor property
to the singleton global attractor and Proposition~\ref{app:model:apt}
(equivalently, use \citet[Theorem~6.9]{Benaim1999}).
Alternatively, stability and uniform attraction on a compact neighborhood
of the sample path imply trapping of every sufficiently late small
perturbation. If $v_k\to v^*$, equation
\eqref{app:model:frequency-slln} and continuity imply
$x_k\to\Lambda(v^*)$. The resulting singleton limit set of an asymptotic
pseudo-trajectory is invariant, so it must be an equilibrium, which
implies $g(v^*)=v^*$.
\end{proof}

\subsection{Equilibria, linearization, and elementary benchmarks}

\begin{proposition}[Exact local correspondence]
\label{app:model:spectrum}
The equilibrium valuation vectors of the joint asynchronous and normalized systems coincide. At an equilibrium $(v^*,\Lambda(v^*))$, the joint Jacobian on $\mathbb R^m\times T\Delta$ has the eigenvalues of $Dg(v^*)-I_m$, together with $m-1$ additional eigenvalues equal to $-1$, counting algebraic multiplicities. Thus, hyperbolicity and the number of unstable directions agree. For hyperbolic equilibria, local asymptotic stability under one system is equivalent to local exponential stability under both systems.
\end{proposition}
\begin{proof}
The frequency equation at an equilibrium forces $x=\Lambda(v)>0$; the
valuation equation then forces $v=g(v)$. Conversely these conditions
give an equilibrium. Let
$a_s(v,x)=\Lambda_s(v)/x_s$. At an equilibrium $a_s=1$, and the factor
$g_s(v)-v_s$ multiplying every derivative of $a_s$ vanishes. Therefore
the full linearization is
\[
 \begin{pmatrix}
 Dg(v^*)-I_m&0\\
 D\Lambda(v^*)&-I_{T\Delta}
 \end{pmatrix},
 \qquad T\Delta=\{u\in\mathbb R^m:\sum_su_s=0\}.
\]
The lower left block has image in $T\Delta$ because
$\sum_s\Lambda_s\equiv1$. The characteristic polynomial is the product
of the two diagonal-block characteristic polynomials. The claims follow
from the linearization theorem for a hyperbolic equilibrium. In
particular there are $m-1$, rather than $m$, additional frequency
eigenvalues on the simplex state space.
\end{proof}

\begin{proposition}[Constant-target benchmarks]
\label{app:model:constant-target}
Suppose $g(v)\equiv c$. The normalized and joint asynchronous systems each have a
unique globally asymptotically stable equilibrium, with valuation $c$
and, in the joint system, frequency $\Lambda(c)$. The stochastic
harmonic-visit process converges to these values almost surely.
This applies to pure exploration, the finest partition, full class
availability, and homogeneous within-class expected payoffs.
\end{proposition}
\begin{proof}
The normalized solution is $v(t)=c+(v(0)-c)e^{-t}$. In the joint system,
\eqref{app:model:xy-ode} gives
\[
 \frac{d}{dt}(y_s-c_sx_s)=-(y_s-c_sx_s),
 \qquad
 v_s(t)-c_s=\frac{x_s(0)}{x_s(t)}(v_s(0)-c_s)e^{-t}.
\]
Lemma~\ref{app:model:ode-compactness} bounds $x_s(t)$ away from zero, so
$v(t)\to c$. Formula~\eqref{app:model:frequency-filter} then implies
$x(t)\to\Lambda(c)$. The unique equilibrium is locally exponentially
stable by Proposition~\ref{app:model:spectrum}, because $Dg=0$, and hence
is globally asymptotically stable. The stochastic conclusion follows
from Corollary~\ref{app:model:apt-consequences}.

When $\beta=0$, the constant target is
\[
 c_s=\frac{\sum_{\omega\ni s}p(\omega)|\omega|^{-1}\pi_s(\omega)}
 {\sum_{\omega\ni s}p(\omega)|\omega|^{-1}}.
\]
For the finest partition it is the corresponding alternative mean $\mu_a$.
If $p(\mathcal S)=1$, the common logit factor cancels in the ratio
defining $g_s$, leaving $c_s=\pi_s(\mathcal S)$. Finally, homogeneous
means $\mu_a=c_s$ within class $s$ imply $\pi_s(\omega)=c_s$ whenever
$s\in\omega$. These arguments establish all four benchmarks.
\end{proof}

\subsection{Inverse-propensity normalization with estimated propensities}

The normalized dynamics can be asymptotically implemented without knowing the menu distribution or the current aggregate propensities. Alice estimates selection probabilities on a faster timescale than she updates valuations,
following the two-timescale stochastic-approximation principle of
\citet{Borkar1997}. Initialize
$\widehat\Lambda_0\in\operatorname{ri}\Delta(\mathcal S)$ and update
\emph{every} class's estimate each period:
\begin{align}
 \widehat\Lambda_{k+1,s}
 &=\widehat\Lambda_{k,s}
   +b_k(I_{k,s}-\widehat\Lambda_{k,s}),\notag\\
 v_{k+1}(s)
 &=v_k(s)+\alpha_k\frac{I_{k,s}}{d_{k,s}}(r_k-v_k(s)),
 \qquad d_{k,s}=\max\{\widehat\Lambda_{k,s},\varepsilon_k\}.
 \label{app:model:two-timescale-ipw}
\end{align}
The valuation update uses the estimate available \emph{before} the current
choice. The deterministic floor $\varepsilon_k\downarrow0$ prevents small
initial estimates from generating large inverse weights. One concrete
choice of gains is
\begin{equation}
 \alpha_k=(k+2)^{-1},\qquad b_k=(k+2)^{-2/3},\qquad
 \varepsilon_k=(k+2)^{-1/6}.
 \label{app:model:ipw-gains}
\end{equation}
In particular, $\alpha_k/b_k\to0$: the estimator adjusts faster than
valuations. For a frozen valuation vector $v$, its mean equation is
$\dot\ell=\Lambda(v)-\ell$, whose unique globally attracting equilibrium
is $\Lambda(v)$. The following result establishes tracking, including when
valuations themselves do not converge.

\begin{proposition}[Learning the normalizing propensities]
\label{app:model:ipw-tracking}
Under the maintained assumptions and the gains
\eqref{app:model:ipw-gains}, the recursion
\eqref{app:model:two-timescale-ipw} satisfies, almost surely,
\[
 \sup_k\|v_k\|<\infty,\qquad
 \widehat\Lambda_{k,s}-\Lambda_s(v_k)\longrightarrow0
 \quad\text{for every }s.
\]
The denominator floor is inactive from some finite random time onward.
The affine interpolation of $v_k$ on the clock
$\tau_k=\sum_{j<k}\alpha_j$ is an asymptotic pseudo-trajectory of
$\dot v=g(v)-v$.
\end{proposition}
\begin{proof}
The gains satisfy
\[
 \sum_k\alpha_k=\sum_k b_k=\infty,\qquad
 \sum_k(\alpha_k/\varepsilon_k)^2+\sum_kb_k^2<\infty,
 \qquad \frac{\alpha_k}{b_k\varepsilon_k}\longrightarrow0.
\]
Also $0<b_k<1$, so the estimates remain in the simplex interior, and
$c_{k,s}:=\alpha_k I_{k,s}/d_{k,s}\in[0,1]$.
By \eqref{app:model:conditional-rewards}, the martingale
\[
 M_{n,s}=\sum_{k<n}c_{k,s}(r_k-\mu_{a_k})
\]
has summable increment variances, bounded by
$M_2\sum_k(\alpha_k/\varepsilon_k)^2$. It therefore converges almost
surely. Put $u_k=v_k-M_k$. Coordinatewise,
\[
 u_{k+1,s}=(1-c_{k,s})u_{k,s}
            +c_{k,s}(\mu_{a_k}-M_{k,s}).
\]
The targets on the right are bounded on each such sample path, since the
primitive means form a finite set and $M_k$ converges. Convex averaging
therefore bounds $u_k$, and hence $v_k$.

We next prove tracking directly. Write $M_\infty=\lim_kM_k$ and define
$L_k=\Lambda(u_k+M_\infty)$. Finite-sensitivity logit propensities are
globally Lipschitz. Since
$u_{k+1,s}-u_{k,s}=c_{k,s}(\mu_{a_k}-v_k(s))$, boundedness gives
\[
 L_k-\Lambda(v_k)\longrightarrow0,\qquad
 \|L_{k+1}-L_k\|\leq C\alpha_k/\varepsilon_k=o(b_k)
\]
for a finite path-dependent constant $C$. Thus
$e_k=\widehat\Lambda_k-L_k$ obeys
\[
 e_{k+1}=(1-b_k)e_k+b_k(\zeta_{k+1}+\rho_k),\qquad
 \zeta_{k+1}=I_k-\Lambda(v_k),\qquad \rho_k\longrightarrow0,
\]
where explicitly
$\rho_k=\Lambda(v_k)-L_k-(L_{k+1}-L_k)/b_k$.
The series $\sum_kb_k\zeta_{k+1}$ converges almost surely, since
$\zeta_{k+1}$ are bounded martingale differences. Variation of constants now
gives $e_k\to0$: products of the factors $1-b_k$ tend to zero, the
weighted averages of $\rho_k$ vanish, and summation by parts makes the
contribution of the convergent martingale series vanish. Consequently
$\widehat\Lambda_k-\Lambda(v_k)\to0$.

On every bounded valuation path, positive availability and finite
$\beta$ imply $\inf_k\Lambda_s(v_k)>0$. Tracking therefore bounds the
estimated propensities away from zero eventually, and the vanishing
floor becomes inactive. Conditioning before the current menu gives
\begin{equation}
 \mathbb E[v_{k+1}(s)-v_k(s)\mid\mathcal F_k]
 =\alpha_k\frac{\Lambda_s(v_k)}{d_{k,s}}
       (g_s(v_k)-v_k(s)).
 \label{app:model:estimated-propensity}
\end{equation}
The ratio tends to one. After localization to bounded valuation sets,
the centered valuation increments have summable variances, bounded by
a constant times $\sum_k(\alpha_k/\varepsilon_k)^2$. Their martingale
series converges. The drift error relative to $g(v_k)-v_k$ vanishes,
so the integral-equation and Gronwall argument in
Prop.~\ref{app:model:apt} proves the asymptotic
pseudo-trajectory property on the slow clock.
\end{proof}

This implementation uses only the chosen class, its realized reward, and the learner's current estimates. Faster adjustment lets the propensity estimator track the current policy. Note that the harmonic visit-count benchmark, whose historical frequencies adjust on the same timescale as valuations, has the joint asynchronous mean field \eqref{app:model:joint}.

\subsection{Attainability and positive-probability convergence}

\begin{proposition}[Capture of a stable SVE from every finite prior]
\label{app:model:attainability}
Every SVE is attainable by the harmonic-visit process: from any
finite initial valuations and positive prior counts, every neighborhood
of its sufficient-statistic representation has positive probability of
being visited at a deterministic date later than any prescribed date.
Every locally asymptotically stable SVE of the joint
mean field has positive probability of being the stochastic process's
limit, from every such prior.
\end{proposition}
\begin{proof}
Let $v^*$ be an SVE, $x^*=\Lambda(v^*)$, and $y^*=x^*v^*=R(v^*)$.
Freeze the logit policy at $v^*$ and let $H=(e_s,r e_s)$ be the
one-period outcome vector in the $(x,y)$ coordinates. The primitive
reward assumptions give a finite expectation, and
\(
 \mathbb E[H]=(x^*,y^*)=:w^*.
\)
The expectation of an integrable finite-dimensional random vector lies
in the closed convex hull of its support: otherwise a separating linear
functional would separate its expectation from all its realizations,
which is impossible after taking expectations. Therefore every
neighborhood of $w^*$ contains a convex combination of finitely many
supported outcome vectors. Approximate its weights by rational numbers,
repeat the corresponding finite outcome word, and require each reward
to lie in a sufficiently small bounded interval about the selected
support point. Choose specific objective menus and implemented
alternatives realizing those support points; their probabilities and
the probabilities of the reward intervals are positive.

Every finite prescribed word of these events has positive probability
under the actual adaptive process. Menus have their fixed positive
probabilities, the chosen classes have positive logit probabilities at
all finite valuations, and the selected implementation probabilities
are positive. On the prescribed bounded reward events all valuations
remain finite. Independence of the next reward conditional on the
implemented alternative then gives positive probability to the entire
word. In view of the exact formula
\[
 W_n=\frac{\chi W_0+\sum_{k<n}H_{k+1}}{\chi+n},
\]
an arbitrarily long repetition makes the prior negligible and places
$W_n$ in the desired neighborhood. This proves attainability, without
requiring any individual reward to equal its expectation.

We verify that attainability implies capture under the present
second-moment assumption. Let $w^*$ be locally asymptotically stable.
Choose a compact neighborhood $K$ in its basin and a larger compact
neighborhood $K'$ in the domain $x_s>0$. Stability and attraction give
a time $T$ and a tracking tolerance $\delta>0$ such that every path
starting sufficiently near $w^*$ and remaining within $\delta$ of each
successive length-$T$ solution segment stays in $K'$, and its values at
the ends of these segments stay in $K$. This is the deterministic
trapping argument for an asymptotically stable equilibrium; it follows
by choosing $K$ inside a stable neighborhood and using uniform
attraction on $K$, or from \citet[Lemma~6.8]{Benaim1999}.

On $K'$ the field is bounded and Lipschitz. For the martingale in
\eqref{app:model:calendar-martingale}, the conditional maximal
inequality gives, at any deterministic date $n$,
\[
 \Pr\!\left(\sup_{l\geq n}\left\|
 \sum_{k=n}^{l}\gamma_{k+1}U_{k+1}\right\|>\eta
 \,\middle|\,\mathcal F_n\right)
 \leq \frac{C}{\eta^2}\sum_{k=n}^\infty\gamma_{k+1}^2,
\]
where $C$ is deterministic. On the complementary event all future
window noise sums and individual noise increments are at most $2\eta$.
The integral-equation estimate in the proof of
Proposition~\ref{app:model:apt}, stopped upon leaving $K'$, consequently
bounds every length-$T$ tracking error by
$C_T(\eta+\sup_{k\geq n}\gamma_{k+1})$. Choose $\eta$ small enough and
then $n$ large enough to make this less than $\delta$ and the displayed
exception probability less than $1/2$. Successive trapping prevents
exit from $K'$. On this event the almost-sure APT property and
\citet[Theorem~6.9]{Benaim1999} imply $W_k\to w^*$.
Thus, conditional on reaching a sufficiently small neighborhood at a
large deterministic date, convergence has probability at least $1/2$.
Attainability supplies positive probability of that event. The change
of coordinates $v=y/x$ yields convergence of the original state.
\end{proof}

\subsection{A UCB-style vanishing exploration-bonus extension}

\begin{proposition}[Vanishing bonuses in the update]
\label{app:model:vanishing-bonus}
Keep the same logit choice rule and harmonic gains, but replace the
observed update target by $r_k+b_{k,s}$ when class $s$ is selected.
Suppose $b_{k,s}$ is $\mathcal F_k$-measurable and
$\max_s|b_{k,s}|\to0$ almost surely. Then boundedness, positive lower
limiting frequencies, and the joint asymptotic-pseudo-trajectory
conclusion continue to hold, with the same limiting ODE.
\end{proposition}
\begin{proof}
The sample-average identity acquires the additional numerator
$\sum_{j<k}I_{j,s}b_{j,s}$. On paths with infinitely many visits, its
ratio to $\kappa_s+N_k(s)$ tends to zero by the Ces\`aro argument on the
visit subsequence; on paths with finitely many visits it has a finite
limit. The boundedness and continued-sampling proof is otherwise
unchanged. In sufficient statistics the conditional drift gains
$\Lambda_s(v_k)b_{k,s}$ in its $y_s$ component, a term tending to zero.
The additional centered noise is
$(I_{k,s}-\Lambda_s(v_k))b_{k,s}$; after localization to bounded bonuses,
its squared harmonic increments are summable. Its martingale series
therefore converges. A drift remainder tending to zero contributes
vanishing integral error on fixed learning-time windows, proving the
APT assertion as before.
\end{proof}

\section{Equilibrium structure and strict-equilibrium refinement}
\label{app:eq:section}

This section proves Theorem~1 of the main text, the characterization and
local-stability Theorem~2 for strict VE, and its full-coverage corollary.
Throughout, the support $\Omega^+:=\operatorname{supp}(p)$ is finite,
$p(\omega)>0$ for $\omega\in\Omega^+$, and every class occurs in at least
one supported menu. Set $n:=|\mathcal S|$ and
$\Omega_s^+:=\{\omega\in\Omega^+:s\in\omega\}$.
We use the following notation, displaying sensitivity when needed:
\begin{align}
 \sigma_\omega^s(v;\beta)
 &=\mathbf 1_{\{s\in\omega\}}
   \frac{\exp(\beta v_s)}{\sum_{j\in\omega}\exp(\beta v_j)},
 &\Lambda_s(v;\beta)&=\sum_{\omega\in\Omega_s^+}p(\omega)
                    \sigma_\omega^s(v;\beta),\label{app:eq:logit}\\
 w_s(\omega;v,\beta)
 &=\frac{p(\omega)\sigma_\omega^s(v;\beta)}{\Lambda_s(v;\beta)},
 &g_s(v;\beta)&=\sum_{\omega\in\Omega_s^+}
             w_s(\omega;v,\beta)\pi_s(\omega).
 \label{app:eq:assessment}
\end{align}
Let $m_s:=\min_{\omega\in\Omega_s^+}\pi_s(\omega)$,
$M_s:=\max_{\omega\in\Omega_s^+}\pi_s(\omega)$, and
$K:=\prod_s[m_s,M_s]$. For clarity, a valuation equilibrium is a pair $(v,\sigma)$ such that
$\sigma_\omega$ is a probability vector supported on
$\arg\max_{s\in\omega}v_s$ and, on writing
$D_s(\sigma):=\sum_{\omega\in\Omega_s^+}p(\omega)\sigma_\omega^s$, for a class $s$ whose $D_s(\sigma)>0$,
\begin{equation}
 v_s
 =\frac{1}{D_s(\sigma)}\sum_{\omega\in\Omega_s^+}p(\omega)\sigma_\omega^s\pi_s(\omega)
 \qquad(s\in\mathcal S).
 \label{app:eq:VE}
\end{equation}
If $D_s(\sigma)=0$, this equation imposes no restriction on $v_s$.
Write $\mathcal{VE}$ for the set of valuation vectors admitting such a
profile, $\mathcal V(\beta):=\{v\in K:v=g(v;\beta)\}$ for the SVE correspondence for a fixed $\beta\in[0,\infty)$, and $\mathcal V(\infty)$ for its accumulation set as $\beta\to\infty$.

\subsection{The SVE correspondence}

\begin{theorem}[Equilibrium structure]
\label{app:eq:structure}
For every fixed tree, the following statements hold.
\begin{enumerate}
\item For each finite $\beta\geq0$, $\mathcal V(\beta)$ is nonempty,
compact, and contained in $\operatorname{ri}K$. The correspondence is
upper hemicontinuous. Its graph
\[
 \mathcal G:=\{(v,\beta)\in K\times[0,\infty):v=g(v;\beta)\}
\]
has finitely many path-connected components. The unique
equilibrium at $\beta=0$ belongs to the principal component whose projection covers all
of $[0,\infty)$.
\item The set $\mathcal V(\infty)$ is nonempty and compact,
$\mathcal V(\infty)\subseteq\mathcal{VE}$, and
\[
 \lim_{\beta\to\infty}\sup_{v\in\mathcal V(\beta)}
       \operatorname{dist}(v,\mathcal V(\infty))=0.
\]
\item Fix the menu support and vary $(p,\pi)$ in the relative Euclidean
parameter space specified below. On a comeager set of full Lebesgue
measure, $\mathcal G$ is a one-dimensional real-analytic embedded
manifold with its sole boundary point at $\beta=0$. Except at finitely
many sensitivities, all equilibria are hyperbolic and their number is
finite and odd with each having a locally unique real-analytic continuation
in $\beta$. The set $\mathcal V(\infty)$ is finite. The component
containing the equilibrium at zero sensitivity is eventually a single
real-analytic graph and has exactly one high-sensitivity limit.
\end{enumerate}
\end{theorem}

In particular, part (iii) implies hyperbolicity of every SVE at all
sufficiently large finite sensitivities.

\begin{proof}
\emph{Existence, compactness, and upper hemicontinuity:}
At finite sensitivity, $\Lambda_s(v;\beta)>0$ for every $v$ and class
$s$. The weights in \eqref{app:eq:assessment} are positive and sum to
one. Thus $g$ is jointly real analytic in $(v,\beta)$ and maps
$\mathbb R^n$ into $\operatorname{ri}K$. Brouwer's theorem, applied on
the possibly lower-dimensional compact convex set $K$, gives a fixed
point. The fixed-point set is closed in $K$, hence compact. If
$\beta_j\to\beta<\infty$ and $v_j\in\mathcal V(\beta_j)$, compactness
of $K$ and continuity of $g$ imply that every limit point belongs to
$\mathcal V(\beta)$. This proves upper hemicontinuity because $K$ is
a common compact codomain. At $\beta=0$, $g$ is independent of $v$,
so there is exactly one equilibrium $v^0$ and
$D_v(v-g(v;0))=I$.

\emph{Degree $1$:}
Set $h(v,\beta):=v-g(v;\beta)$ and introduce the enlarged open box
\[
 Q:=\prod_s(m_s-1,M_s+1).
\]
For any $c\in K$, the homotopy
$v-[(1-t)c+t g(v;\beta)]$, $0\leq t\leq1$, has no zero on
$\partial Q$: its bracketed term lies in $K\subset Q$. Homotopy
invariance of Brouwer degree consequently gives
\begin{equation}
 \deg(h(\cdot,\beta),Q,0)=1
 \qquad(0\leq\beta<\infty).
 \label{app:eq:degree}
\end{equation}
Whenever every equilibrium is regular, each is isolated. There are then finitely many: an infinite set in the compact zero set would have an accumulation point, which would itself be a non-isolated zero. Their local degrees are $\operatorname{sgn}\det D_vh\in\{-1,1\}$ and sum to one. Their number is therefore odd. More generally, an index sum is asserted only when the equilibria are isolated.

\emph{Definability and continuation from zero sensitivity:}
The real field with exponentiation is o-minimal, and the formulas in \eqref{app:eq:logit}--\eqref{app:eq:assessment} use only its definable operations. Thus $\mathcal G$ is definable in the o-minimal structure $\mathbb R_{\exp}$. We use the standard consequences of o-minimality that definable sets have finitely many connected components, each definably path connected; that the frontier has smaller dimension than a nonempty definable set; and that bounded one-variable definable functions are eventually monotone and have limits at an endpoint. These facts follow from cell decomposition and monotonicity; see \citet{LVD,Dries1998}.

Fix $b>0$ and let $\mathcal G_b:=\mathcal G\cap(K\times[0,b])$. It is compact and
has finitely many connected components. Suppose the component $C_0$ containing $(v^0,0)$ did not meet $K\times\{b\}$. Because the components are finitely many disjoint compact sets, there exists an open neighborhood $U$ of $C_0$ in $Q\times[0,b]$ such that $\overline U\cap\mathcal G_b=C_0$, $\overline U$ misses the top face, and the relative boundary of $U$ contains no zero of $h$. For $U_t:=\{v:(v,t)\in U\}$, the integer $\deg(h(\cdot,t),U_t,0)$ is locally constant in $t$. Here is the usual excision argument for the varying slices: at a given $t_0$, surround the compact set of zeros in $U_{t_0}$ by finitely many open sets with closures inside $U_{t_0}$; for sufficiently nearby $t$,
these closures remain inside $U_t$ and contain all its zeros. Otherwise a sequence of omitted zeros would converge to a zero on the relative boundary or to an omitted zero in $U_{t_0}$. Excision
then reduces the degree to a fixed open set, where homotopy invariance applies. The same argument covers an empty zero set and the endpoints by one-sided neighborhoods. The degree at zero is
$+1$, whereas at $b$ it is zero, a contradiction. Thus $C_0$ meets the top face. As $b$ is arbitrary, the component of the full graph containing $(v^0,0)$ projects onto $[0,\infty)$. It is the only
component with this property because the fiber at zero is a singleton.

\emph{High-sensitivity limits and VE consistency:}
Compactness of $K$ makes $\mathcal V(\infty)$ nonempty. More
explicitly it is the nested intersection
\(
 \mathcal V(\infty)=
 \bigcap_{b>0}\overline{\bigcup_{\beta\geq b}\mathcal V(\beta)},
\)
so it is compact. Let $\beta_j\to\infty$ and $v^j\in\mathcal V(\beta_j)$ converge to $v^*$. Passing to a subsequence, every menu's logit probability vector converges to a profile $\sigma^*$. If $v_s^*<\max_{r\in\omega}v_r^*$, the corresponding gap is bounded below by a positive constant for large
$j$, whence $\sigma_\omega^s(v^j;\beta_j)\to0$. Thus the limit profile chooses only maximizers. Multiplying the fixed-point equation by $\Lambda_s(v^j;\beta_j)$ and taking limits gives the balance identity $v_s^*D_s(\sigma^*)=\sum_{\omega\in\Omega_s^+}p(\omega)\sigma_\omega^{*,s}\pi_s(\omega)$. For $D_s(\sigma^*)>0$, division gives \eqref{app:eq:VE}; when $D_s(\sigma^*)=0$, the identity imposes no restriction. Hence $v^*\in\mathcal{VE}$. If the stated uniform-distance conclusion failed, a sequence of equilibria at sensitivities tending to infinity would stay at a positive distance from $\mathcal V(\infty)$; its convergent subsequence in $K$ contradicts the definition of that set. Moreover, $\mathcal{VE}$ is closed: the menu-wise optimality conditions
and the balance identities above define a closed set of pairs $(v,\sigma)$,
and the space of choice profiles is compact.

\emph{Generic regularity of the entire equilibrium graph:}
The parameter space is
\[
 \Theta=\left\{p\in\mathbb R^{\Omega^+}:
 p(\omega)>0,\ \sum_{\omega\in\Omega^+}p(\omega)=1\right\}
 \times\mathbb R^{\{(s,\omega):\omega\in\Omega^+,\ s\in\omega\}}.
\]
It carries its relative Euclidean topology and Lebesgue measure in its affine coordinates. Write $d:=\dim\Theta$. On the fixed open domain $\mathbb R^n\times(0,\infty)\times\Theta$, consider
$\Phi(v,\beta,\theta):=v-g(v;\beta,\theta)$. For each $s$, choose one supported menu $\omega_s$ containing $s$. The derivative of $\Phi$ with respect to the $n$ distinct payoff coordinates
$\pi_s(\omega_s)$ is diagonal, with entries $-w_s(\omega_s;v,\beta)<0$. Consequently $\Phi$ is a submersion, and $\mathcal M:=\Phi^{-1}(0)$ is a smooth manifold of dimension $d+1$.

Let $P:\mathcal M\to\Theta$ be parameter projection. A direct linear-algebra argument shows that $\theta$ is a regular value of $P$ if and only if
\begin{equation}
 \operatorname{rank}D_{(v,\beta)}\Phi(v,\beta,\theta)=n
 \quad\hbox{at every zero with }\beta>0.
 \label{app:eq:fullregular}
\end{equation}
Indeed, write the derivative as $A\xi+B\zeta$. Surjectivity of $A$ lets one solve $A\xi=-B\zeta$ for every parameter direction $\zeta$, which is precisely surjectivity of $DP$ on the tangent
space of $\mathcal M$. Conversely, that solvability implies $\operatorname{im}B\subseteq\operatorname{im}A$; surjectivity of $[A\ B]$ then forces $A$ to be surjective.

Sard's theorem says that the set of critical values of $P$ has measure zero. It is also meager: exhaust $\mathcal M$ by countably many compact sets; the image of the closed critical set intersected with each such compact set is compact and null, hence nowhere dense. Their union is the critical-value set. Therefore the \emph{same} set of regular parameters is both comeager and of full
measure. For each such parameter, \eqref{app:eq:fullregular} and the analytic regular-level-set theorem make the positive-sensitivity graph a one-dimensional analytic embedded manifold. At $\beta=0$, $D_vh=I$, so the analytic implicit-function theorem supplies a boundary chart and no additional boundary points occur. Consequently, the graph has dimension one.

\emph{Exceptional sensitivities and finite limit sets:}
Fix a regular parameter and let $\xi:\mathcal G\to[0,\infty)$ be sensitivity projection. On the
interior manifold, a value of $\xi$ is regular exactly when $D_vh$ is invertible at every equilibrium in that fiber. To see this, a tangent vector $(u,t)$ obeys $D_vh\,u+\partial_\beta h\,t=0$. If $D_vh$ is invertible, a tangent vector with $t=1$ exists. If it is singular, its kernel
already gives a nonzero tangent vector with $t=0$; the tangent space is one-dimensional, so $D\xi=0$ there. The critical-value set is definable and null by Sard's theorem. A definable null
subset of the real line is finite. At every other sensitivity, compactness, \eqref{app:eq:degree}, and the analytic implicit-function theorem give the claimed finiteness, parity, and local branches.
Compactify sensitivity by $t=(1+\beta)^{-1}$ and put
\[
 \widehat{\mathcal G}:=
 \{(v,t)\in K\times(0,1]:(v,t^{-1}-1)\in\mathcal G\}.
\]
This is a one-dimensional definable set. Its frontier $\overline{\widehat{\mathcal G}}\setminus\widehat{\mathcal G}$ has dimension strictly below one and is therefore finite.
The points added at $t=0$ are exactly $\mathcal V(\infty)\times\{0\}$. This proves finiteness of
$\mathcal V(\infty)$.

\emph{The unique limit of the principal component:}
Sensitivity projection is proper: its inverse image of a compact interval is closed in the compact product of that interval with $K$. Its restriction to a component is proper as well. The
principal component $C$ is a connected one-dimensional manifold with exactly one boundary point. Since its projection is unbounded, it is non-compact; the classification of connected one-dimensional manifolds with boundary identifies $C$ with $[0,\infty)$. Properness implies that $\beta$ tends to infinity along its sole non-compact end. On this component $\xi$ is analytic and
non-constant. Its critical points are isolated by the analytic identity theorem and form a definable set, so there are only finitely many. Choose a compact initial arc containing them all,
and choose $b$ greater than the largest sensitivity on that arc. On the remaining end $\xi$ has nonzero derivative, is strictly monotone, and tends to infinity. Thus $C\cap(K\times(b,\infty))$ is a single graph $\{(\varphi(\beta),\beta):\beta>b\}$, where $\varphi$ is definable and real analytic. Each bounded coordinate of $\varphi$ has a limit by definable monotonicity. The resulting vector is the unique high-sensitivity limit of $C$ and belongs to $\mathcal V(\infty)\subseteq\mathcal{VE}$.

\emph{Generic hyperbolicity:}
We now strengthen regularity to hyperbolicity for every finite $n$.
On the universal equilibrium manifold
$\mathcal M=\Phi^{-1}(0)$ above, which has dimension $d+1$, set
\begin{equation}
 J(v,\beta,\theta):=D_vg(v;\beta,\theta)-I_n,
 \qquad
 H(v,\beta,\theta):=
 \det\!\left(J\otimes I_n+I_n\otimes J\right),
 \label{app:eq:hyperbolicity-test}
\end{equation}
where $\otimes$ denotes the Kronecker product. The eigenvalues of
the matrix inside the determinant are the sums $\lambda_i+\lambda_j$
of eigenvalues of the real matrix $J$, counted with multiplicity.
If $J$ has a purely imaginary eigenvalue $i\omega$, including zero,
then $-i\omega$ is also an eigenvalue, so $H=0$.
Only this implication is needed: $H$ can also vanish at a hyperbolic
matrix with opposite nonzero eigenvalues.

Fix any $(v,\beta,p,\pi)\in\mathcal M$ and, keeping $v$ fixed as a
center, vary the reduced payoffs according to
\begin{equation}
 \pi_s^{\,t}(\omega)=v_s+t\bigl(\pi_s(\omega)-v_s\bigr).
 \label{app:eq:equilibrium-preserving-payoffs}
\end{equation}
This perturbation is admissible in the stated parameter space and
leaves $p$ unchanged. Since the conditional weights depend on
valuations, sensitivity, and $p$, but not on payoffs, for every test
valuation $u$ we have
\[
 g(u;\beta,p,\pi^{\,t})
 =v+t\bigl(g(u;\beta,p,\pi)-v\bigr).
\]
Thus $v$ remains an equilibrium for every $t$, and along this curve
$J_t=tD_vg(v;\beta,p,\pi)-I_n$.
The corresponding $H(t)$ is a polynomial in $t$ with
$H(0)=(-2)^{n^2}\neq0$. Consequently, arbitrarily close to $t=1$
there are points of $\mathcal M$ at which $H\neq0$.
The definable set $Z:=\{(v,\beta,\theta)\in\mathcal M:H=0\}$
therefore has empty relative interior in the $(d+1)$-dimensional
manifold $\mathcal M$, and hence $\dim Z\leq d$.

For $\theta\in\Theta$, write
$Z_\theta:=\{(v,\beta):(v,\beta,\theta)\in Z\}$.
The definable fiber-dimension theorem implies that
\[
 E:=\{\theta\in\Theta:\dim Z_\theta\geq1\}
 \quad\hbox{is definable and satisfies}\quad \dim E\leq d-1.
\]
Indeed, on the definable set of parameters with fiber dimension $k$,
the total dimension is the parameter-set dimension plus $k$; for
$k\geq1$ this bounds the parameter-set dimension by $d-1$.
Thus $E$ has measure zero and is nowhere dense, and every fiber
$Z_\theta$ outside $E$ is finite (or empty), because a
zero-dimensional definable set is finite. These are standard
consequences of definable cell decomposition \citep{Dries1998}.
Intersecting $\Theta\setminus E$ with the comeager full-measure
parameter set already obtained above preserves all preceding
conclusions. For each parameter in this intersection, the projection
of $Z_\theta$ onto the sensitivity axis is finite and contains every
positive sensitivity admitting a nonhyperbolic SVE. At $\beta=0$,
$D_vg=0$ and $J=-I_n$, so hyperbolicity is automatic.
Finally, Proposition~\ref{app:model:spectrum} shows that the joint
linearization adds only $n-1$ eigenvalues equal to $-1$.
The same conclusion therefore holds for both dynamics, with the
finite exceptional set depending on $(p,\pi)$.
\end{proof}

\subsection{Two supporting claims in the equilibrium discussion}

\begin{lemma}[A robust continuum of VE valuations]
\label{app:eq:VEcontinuum}
Consider three equally likely menus $\{A,B\}$, $\{B,C\}$, and
$\{C,A\}$, with payoffs $(1,0)$ to the classes listed in that
order. Every $(t,t,t)$ with $0<t<1$ is a VE valuation. A local
continuum of fully mixed VE valuations persists for every tree
in an open neighborhood of these primitives, holding the three-menu
support fixed. At the symmetric primitives, however, the unique
SVE is $(1/2,1/2,1/2)$ at every finite sensitivity, so this is the
unique LSVE.
\end{lemma}
\begin{proof}
Let $q_1,q_2,q_3$ be the probabilities of choosing the first-listed
class in the three menus. With all valuations equal to $t$, VE
consistency at the symmetric primitives is
\[
 (1-t)q_1+tq_3=t,\qquad
 tq_1+(1-t)q_2=t,\qquad
 tq_2+(1-t)q_3=t.
\]
The choice $q_1=q_2=q_3=t$ solves these equations. Their derivative
with respect to $q$ has determinant
$(1-t)^3+t^3>0$. For each fixed $t_0\in(0,1)$, the
implicit-function theorem therefore gives a smooth interior
solution $q=q(t,p,\pi)$ for every $t$ in a neighborhood of $t_0$
and every $(p,\pi)$ sufficiently near the symmetric primitives.
The unnormalized consistency equations are smooth in these
parameters; their derivative differs at the symmetric point from
the displayed matrix only by a nonzero common scalar. Since all
classes are tied, every such interior profile is greedy-optimal.
Thus these are continua of distinct VE valuation vectors, not
only continua of different mixing profiles at one valuation.

For finite $\beta>0$, the assessment map at the symmetric
primitives satisfies
\begin{align*}
 \operatorname{sgn}(g_A(v)-1/2)&=\operatorname{sgn}(v_C-v_B),\\
 \operatorname{sgn}(g_B(v)-1/2)&=\operatorname{sgn}(v_A-v_C),\\
 \operatorname{sgn}(g_C(v)-1/2)&=\operatorname{sgn}(v_B-v_A).
\end{align*}
These identities follow by comparing, for each class, its logit
selection probabilities in its payoff-one and payoff-zero menus.
At a fixed point choose a highest-valued class and, by cyclic
relabeling, call it $A$. The second and third identities give
$v_B\geq1/2$ and $v_C\leq1/2$. Hence $v_C\leq v_B$, and the
first identity gives $v_A\leq1/2$. Since $v_A\geq v_B$, both
equal $1/2$. The first identity then forces $v_C=v_B$. Uniform
choice at $\beta=0$ yields the same unique fixed point directly.
\end{proof}

\begin{lemma}[The stated odds restrictions characterize logit]
\label{app:eq:logitaxioms}
On a rich domain admitting every real valuation profile and menus
with at least three distinct labels, suppose choice probabilities
are positive and continuous. Suppose the odds of any two available
classes with valuations $a,b$ equal a common function $r(a,b)$,
independent of the other available classes and of their labels.
If choice is invariant to common valuation shifts and is weakly
increasing in a class's own valuation, then, for some $\beta\geq0$,
\[
 \sigma_\omega^s(v)=
 \frac{\exp(\beta v_s)}{\sum_{j\in\omega}\exp(\beta v_j)}.
\]
Conversely every such logit rule has these properties.
\end{lemma}
\begin{proof}
Consistency of odds in a three-class menu gives
$r(a,c)=r(a,b)r(b,c)$ for every $a,b,c\in\mathbb R$.
With $w(a):=r(a,0)>0$, this yields
$r(a,b)=w(a)/w(b)$ and $w(0)=1$. Common-shift invariance implies
$r(a+c,c)=r(a,0)$, or $w(a+c)=w(a)w(c)$. The continuous
function $\log w$ is additive on $\mathbb R$, so
$w(a)=\exp(\beta a)$ for a real constant $\beta$. Own-valuation
monotonicity implies $\beta\geq0$, and strict monotonicity rules
out zero. Normalizing the pairwise odds gives the displayed
formula. Its converse properties follow immediately. This
characterization concerns the specified rich domain.
\end{proof}

\begin{lemma}[Closest-loss concentration]
\label{app:eq:generalclosest}
Suppose $\beta_j\to\infty$ and $v^j\to v^*$, without requiring
that $v^j$ be equilibria. Define
\[
 \Delta_s(\omega;v):=\max_{r\in\omega}v_r-v_s,
 \qquad
 \Omega_s^\star(v^*):=
 \arg\min_{\omega\in\Omega_s^+}\Delta_s(\omega;v^*).
\]
Then
$\sum_{\omega\notin\Omega_s^\star(v^*)}
 w_s(\omega;v^j,\beta_j)\to0$ for every class $s$.
Consequently every limiting SVE satisfies
\[
 v_s^*\in\operatorname{co}
 \{\pi_s(\omega):\omega\in\Omega_s^\star(v^*)\}.
\]
\end{lemma}
\begin{proof}
The assertion is immediate if every supported menu containing $s$
is a closest-loss menu. Otherwise let $\delta>0$ be the smallest
excess loss of a non-minimizing menu at $v^*$. By continuity and
finiteness, for all large $j$ each such menu has loss at least
$\delta/2$ larger than the current minimum
$d_j:=\min_{\omega\in\Omega_s^+}\Delta_s(\omega;v^j)$.
For every menu,
\[
 \frac1n e^{-\beta_j\Delta_s(\omega;v^j)}
 \leq\sigma_\omega^s(v^j;\beta_j)
 \leq e^{-\beta_j\Delta_s(\omega;v^j)}.
\]
Choose a current minimizing menu and set
$p_{\min,s}:=\min_{\omega\in\Omega_s^+}p(\omega)>0$.
The denominator $\Lambda_s$ is at least
$(p_{\min,s}/n)e^{-\beta_jd_j}$, whereas the total numerator
over nonminimizing menus is at most
$e^{-\beta_j(d_j+\delta/2)}$. Their ratio is at most
$(n/p_{\min,s})e^{-\beta_j\delta/2}$ and tends to zero. Passing
to a subsequential limit of the finitely many weights and using
$v^j=g(v^j;\beta_j)$ gives the convex-hull conclusion for SVE.
\end{proof}

\subsection{Strict equilibria and the closest-loss characterization}

A \emph{strict VE} has a unique winner $a(\omega)$ in each supported menu and pairwise distinct valuations among its winner classes $\mathcal W:=\{a(\omega):\omega\in\Omega^+\}$. For such a vector, define
\begin{equation}
 \bar g_s(v^*):=
 \frac{\sum_{\omega\in\Omega_s^\star(v^*)}
                p(\omega)\pi_s(\omega)}
      {\sum_{\omega\in\Omega_s^\star(v^*)}p(\omega)}.
 \label{app:eq:closestaverage}
\end{equation}

\begin{lemma}[Uniform conditional-weight and derivative estimates]
\label{app:eq:uniform}
Let $v^*$ be menu-wise strict. Suppose, for every $s$, there is a
unique class
\[
 t_s=\arg\min\{v_t^*:t=a(\omega),\ \omega\in\Omega_s^+\}.
\]
This condition holds for every strict VE. There exist $r,\eta,C>0$
such that, on $B:=\{v:\|v-v^*\|_\infty\leq r\}$ and for
$\beta\geq1$,
\begin{align}
 \sup_{v\in B}\sum_{\omega\in\Omega_s^+}
 \left|w_s(\omega;v,\beta)-
 \frac{p(\omega)\mathbf1_{\{\omega\in\Omega_s^\star(v^*)\}}}
      {\sum_{\omega'\in\Omega_s^\star(v^*)}p(\omega')}\right|
 &\leq Ce^{-\beta\eta},\label{app:eq:weightsbound}\\
 \sup_{v\in B}\|g(v;\beta)-\bar g(v^*)\|_\infty
 &\leq Ce^{-\beta\eta},\label{app:eq:gbound}\\
 \sup_{v\in B}\|D_vg(v;\beta)\|_\infty
 &\leq C\beta e^{-\beta\eta}.
 \label{app:eq:Dgbound}
\end{align}
Here the matrix norm is the induced maximum row-sum norm.
\end{lemma}
\begin{proof}
Write $\Omega_s^\star:=\Omega_s^\star(v^*)$. The unique closest
winner condition implies
$\Omega_s^\star=\{\omega\in\Omega_s^+:a(\omega)=t_s\}$.
Choose $r$ sufficiently small that, throughout $B$, all menu winners
remain $a(\omega)$ and all the positive differences
\[
 v_{a(\omega)}-v_j\quad(j\in\omega\setminus\{a(\omega)\}),
 \qquad
 v_{a(\omega)}-v_{t_s}
 \quad(\omega\in\Omega_s^+\setminus\Omega_s^\star)
\]
are bounded below by one common $\eta>0$. Empty collections of
comparisons impose no restriction. Define
\[
 Z_\omega:=1+\sum_{j\in\omega\setminus\{a(\omega)\}}
                     e^{-\beta(v_{a(\omega)}-v_j)}.
\]
Then $1\leq Z_\omega\leq n$ and
$|Z_\omega-1|\leq(n-1)e^{-\beta\eta}$. The exact identity
\[
 \sigma_\omega^s=
 e^{-\beta(v_{a(\omega)}-v_s)}/Z_\omega
\]
allows the common factor $e^{-\beta(v_{t_s}-v_s)}$ to be canceled
from the conditional weights. Put
\[
 a_{s\omega}:=
 \frac{e^{-\beta(v_{a(\omega)}-v_{t_s})}}{Z_\omega},
 \quad A_s:=\sum_{\omega\in\Omega_s^+}p(\omega)a_{s\omega},
 \quad P_s^\star:=\sum_{\omega\in\Omega_s^\star}p(\omega)>0.
\]
Thus $w_s(\omega)=p(\omega)a_{s\omega}/A_s$. On
$\Omega_s^\star$, $a_{s\omega}=1/Z_\omega$, whereas outside
it $a_{s\omega}\leq e^{-\beta\eta}$. Consequently
\[
 A_s\geq P_s^\star/n,\qquad
 E_s:=\sum_{\omega\in\Omega_s^+}p(\omega)
       |a_{s\omega}-\mathbf1_{\{\omega\in\Omega_s^\star\}}|
       \leq n e^{-\beta\eta},
 \qquad |A_s-P_s^\star|\leq E_s.
\]
Normalizing the weights changes their total variation by at most
$2E_s/A_s\leq(2n^2/P_s^\star)e^{-\beta\eta}$. This proves
\eqref{app:eq:weightsbound}, and bounded payoffs give
\eqref{app:eq:gbound}. The lower bound is for the \emph{factored}
denominator $A_s$; $\Lambda_s$ itself need not be uniformly
positive as sensitivity increases.

For the derivative estimate, differentiating the logit quotient and
using exact payoff centering,
\begin{align}
 \partial_{v_k}g_s
 &=-\beta\sum_{\omega\in\Omega_s^+}w_s(\omega)
           \sigma_\omega^k(\pi_s(\omega)-g_s)\notag\\
 &=-\beta\sum_{\omega\in\Omega_s^+}w_s(\omega)
       (\sigma_\omega^k-\mathbf1_{\{k=t_s\}})
       (\pi_s(\omega)-g_s).
 \label{app:eq:centeredderivative}
\end{align}
Indeed, $\sum_\omega w_s(\omega)(\pi_s(\omega)-g_s)=0$, so
the terms constant across menus cancel. With
$R:=\max_s(M_s-m_s)$, the absolute row sum is bounded by
\[
 \beta R\sum_{\omega\in\Omega_s^+}w_s(\omega)
                       \|\sigma_\omega-e_{t_s}\|_1.
\]
On $\Omega_s^\star$, the last norm is at most
$2(n-1)e^{-\beta\eta}$. On its complement it is at most two,
and the total conditional mass there is $O(e^{-\beta\eta})$
by \eqref{app:eq:weightsbound}. This proves
\eqref{app:eq:Dgbound}, uniformly on $B$.
\end{proof}

\begin{theorem}[Strict VE: necessity, sufficiency, and local stability]
\label{app:eq:strict}
Let $v^*$ be a strict VE. Then $v^*\in\mathcal V(\infty)$ if and
only if $v^*=\bar g(v^*)$. Whenever this condition holds, a fixed
neighborhood of $v^*$ contains exactly one SVE $v^\beta$ for every
sufficiently large $\beta$, and
\[
 \|v^\beta-v^*\|_\infty\leq Ce^{-\beta\eta}
\]
for some constants $C,\eta>0$. This equilibrium is locally
exponentially stable under the normalized system
$\dot v=g(v;\beta)-v$, and $(v^\beta,\Lambda(v^\beta;\beta))$
is locally exponentially stable under the joint system
\[
 \dot v_s=\frac{\Lambda_s(v;\beta)}{x_s}(g_s(v;\beta)-v_s),
 \qquad \dot x_s=\Lambda_s(v;\beta)-x_s,
 \qquad x\in\operatorname{ri}\Delta(\mathcal S).
\]
\end{theorem}
\begin{proof}
The strict ordering of winner classes makes each $t_s$ in
Lemma~\ref{app:eq:uniform} unique. Suppose first that equilibria
$v^j\to v^*$ occur along $\beta_j\to\infty$. Eventually they
belong to the fixed ball $B$ of that lemma. The fixed-point identity
and \eqref{app:eq:gbound} imply
\[
 \|v^j-\bar g(v^*)\|_\infty
 \leq Ce^{-\beta_j\eta}\longrightarrow0.
\]
Thus $v^*=\bar g(v^*)$, proving necessity even for classes that
are never chosen at the limiting VE.

Conversely, suppose $v^*=\bar g(v^*)$. For all sufficiently large
$\beta$, the lemma gives
$g(B;\beta)\subseteq B_{r/2}(v^*)$ and
$\sup_B\|D_vg\|_\infty\leq q_\beta<1/2$. Since $B$ is convex,
$g$ is a contraction on $B$, so Banach's theorem gives exactly one
fixed point there. It belongs to $K$ because every value of $g$
belongs to $K$. Evaluating \eqref{app:eq:gbound} at the fixed
point proves the exponential approximation bound, and hence the
LSVE property. In particular the interior of $B$ is the fixed
neighborhood required in the paper.
On each boundary face of $B$ the normalized field points strictly
inward. Thus $B$ is positively invariant. If $v(t)$ starts in $B$,
variation of constants for $z(t):=v(t)-v^\beta$ and the contraction
bound give
\[
 e^t\|z(t)\|_\infty
 \leq\|z(0)\|_\infty+
 q_\beta\int_0^t e^u\|z(u)\|_\infty\,du.
\]
Gronwall's inequality yields
$\|z(t)\|_\infty\leq e^{-(1-q_\beta)t}\|z(0)\|_\infty$.
This proves local exponential stability for the normalized system.
Equivalently, every eigenvalue of $Dg(v^\beta)-I$ has real part
at most $-1+q_\beta<-1/2$.

At a joint equilibrium $x^\beta=\Lambda(v^\beta;\beta)$,
the factors $\Lambda_s/x_s$ equal one and all prediction errors
vanish. In local coordinates on the simplex tangent space, the
joint derivative is the block-triangular matrix
\begin{equation}
 \begin{pmatrix}
 Dg(v^\beta)-I&0\\
 D\Lambda(v^\beta)&-I
 \end{pmatrix}.
 \label{app:eq:jointlinearization}
\end{equation}
Its extra $n-1$ eigenvalues are $-1$. It is therefore Hurwitz,
and the standard linearization theorem for a smooth ODE gives
local exponential stability of the joint equilibrium. 
\end{proof}

\begin{corollary}[Coverage, including full singleton support]
\label{app:eq:coverage}
If a menu-wise strict VE selects every class in at least one
supported menu, it is an LSVE and all the nearby-equilibrium and
local-stability conclusions of Theorem~\ref{app:eq:strict} hold.
In particular this applies when every class has a supported
singleton menu. Distinct winners need not have distinct valuations
in this corollary.
\end{corollary}
\begin{proof}
For each $s$, some menu is won by $s$. In every other menu
containing $s$, the winner either is $s$ or has valuation strictly
above $v_s^*$. Thus the unique closest winner is $t_s=s$, and
the closest-loss menus are exactly those won by $s$. VE
consistency gives $v_s^*=\bar g_s(v^*)$. The uniform estimates
and the sufficiency proof above apply. A class always wins its
own singleton menu, proving the assertion.
\end{proof}

\begin{lemma}[Generic ordering of winning classes]
\label{app:eq:winnerordering}
For generic reduced-tree parameters, every menu-wise strict VE
orders its winning classes strictly.
\end{lemma}
\begin{proof}
There are finitely many pure menu-choice profiles. Fix one and
write $A_s$ for the nonempty collection of menus it assigns to
winner $s$. VE consistency determines
\[
 v_s=\frac{\sum_{\omega\in A_s}p(\omega)\pi_s(\omega)}
            {\sum_{\omega\in A_s}p(\omega)}.
\]
For two distinct winners $s,t$, their equality is a smooth
hypersurface in parameter space: its derivative with respect to
any $\pi_s(\omega)$ with $\omega\in A_s$ is nonzero. It is
relatively closed, nowhere dense, and Lebesgue null. The union
over the finitely many profiles and winner pairs has the same
properties. Outside that union, no menu-wise strict VE has tied
winner valuations.
\end{proof}

\section{Constant competitive-menu payoffs and two-class convergence}
\label{app:potential}

Suppose $\pi_s(\omega)=b_s$ whenever $s\in\omega$ and $|\omega|\geq2$.
Put $a_s=p(\{s\})$, interpreting this probability as zero if the singleton
is unsupported, and put $d_s=a_s(\pi_s(\{s\})-b_s)$, with $d_s=0$ when
$a_s=0$. If a class appears only alone, choose any $b_s$ and define $d_s$
by the same identity. Then
\begin{equation}
 R_s(v)=b_s\Lambda_s(v)+d_s,\qquad
 g_s(v)=b_s+\frac{d_s}{\Lambda_s(v)}.
 \label{app:potential:R}
\end{equation}

\begin{theorem}[Constant competitive-menu payoffs]
\label{app:potential:theorem}
For each fixed finite $\beta$, every normalized trajectory approaches the
SVE set, and every joint trajectory approaches the corresponding joint
equilibrium set. If the SVE set is finite, each trajectory converges to one
equilibrium. For each fixed $\beta$, finiteness holds generically within
the family satisfying the stated payoff restriction.
If $d_s\geq0$ for every class, the SVE is unique.
\end{theorem}
\begin{proof}
The case $\beta=0$ is the constant-target benchmark. Assume $\beta>0$.
Let $\rho=(\rho_\omega^s)$ be an interior menu-wise choice profile and set
$X_s(\rho)=\sum_{\omega\ni s}p(\omega)\rho_\omega^s$. Consider
\begin{equation}
 \Phi(\rho)=\sum_s\{b_sX_s(\rho)+d_s\log X_s(\rho)\}
 -\frac1\beta\sum_\omega p(\omega)
       \sum_{s\in\omega}\rho_\omega^s\log\rho_\omega^s.
 \label{app:potential:Phi}
\end{equation}
All expressions below are evaluated on compact interior sets, so their
derivatives are bounded. In a simplex tangent direction its gradient is
represented by
\[
 \frac{1}{p(\omega)}\frac{\partial\Phi}{\partial\rho_\omega^s}
 =b_s+\frac{d_s}{X_s(\rho)}
      -\frac1\beta(1+\log\rho_\omega^s).
\]

\emph{Normalized dynamics.} Set $\rho=\sigma(v)$ and $F=g(v)-v$.
Since $X(\rho)=\Lambda(v)$ and
$\log\sigma_\omega^s=\beta v_s-\log\sum_{j\in\omega}e^{\beta v_j}$,
the simplex gradient is represented by $F_s$ plus a menu constant.
Differentiating the logits gives
\[
 \dot\rho_\omega^s=\beta\rho_\omega^s
 \left(F_s-\sum_{j\in\omega}\rho_\omega^jF_j\right),\qquad
 \frac{d}{dt}\Phi(\sigma(v))
 =\beta\sum_\omega p(\omega)\operatorname{Var}_{\sigma_\omega(v)}(F)
 \geq0.
\]
Valuations stay in a compact box. LaSalle's invariance argument places
their omega-limit set in the zero-derivative set. There $F_s$ is constant
on each connected component of the class co-occurrence graph. Along a
trajectory remaining in this set, all within-component valuation gaps
are constant, so $g$ is constant and the remaining component-wise common
shift solves $\dot c=-c$. The only such complete trajectories in a compact
set are stationary. Thus the largest compact invariant subset of the
zero-derivative set consists of SVE, proving approach to the SVE set.

\emph{Joint dynamics.} In sufficient statistics,
$e_s=y_s-b_sx_s-d_s$ satisfies $\dot e_s=-e_s$ by
\eqref{app:potential:R}. Introduce an auxiliary profile, with any interior
initial condition, through
\[
 \dot\rho_\omega^s=\sigma_\omega^s(v)-\rho_\omega^s.
\]
It remains a menu-wise probability profile. Moreover
$X(\rho)-x$ solves the same equation $\dot h=-h$. The compact valuation
box bounds all logits away from zero, so eventually $\rho$, $X(\rho)$,
and $x$ stay in compact interior sets. Consequently
\[
 b_s+d_s/X_s(\rho)-v_s=O(e^{-t}).
\]
Writing $\sigma=\sigma(v)$ and using the logit odds identity gives
\begin{align}
 \frac{d}{dt}\Phi(\rho(t))
 &=\frac1\beta\sum_\omega p(\omega)\sum_{s\in\omega}
       (\sigma_\omega^s-\rho_\omega^s)
          \log\frac{\sigma_\omega^s}{\rho_\omega^s}+O(e^{-t})\nonumber\\
 &=:D(t)+O(e^{-t}),\qquad D(t)\geq0.
 \label{app:potential:dissipation}
\end{align}
The error is absolutely integrable and $\Phi$ is bounded along the
trajectory. Thus $\int_0^\infty D(t)\,dt<\infty$. Smoothness on the
eventual compact set makes $D$ uniformly continuous; the elementary
integrability argument for a nonnegative uniformly continuous function
gives $D(t)\to0$. Each summand is nonnegative and vanishes precisely when
its two arguments agree. Hence $\rho-\sigma(v)\to0$ and
$x-\Lambda(v)\to0$. The decaying $e_s$ and
\eqref{app:potential:R} now imply $v-g(v)\to0$. Compactness proves
approach to the joint equilibrium set.

If that set is finite, the connectedness of an omega-limit set gives
convergence to one point. For generic finiteness at a fixed sensitivity,
vary every payoff of class $s$ by an independent constant $t_s$.
This operation preserves constant competitive-menu payoffs and changes
$g_s$ by exactly $t_s$. The equilibrium equation is therefore a
submersion in these class shifts. Parametric regular-value theory gives
regular, hence isolated, equilibria for a comeager full-measure set of
parameters; compactness makes their number finite. The parameter topology
is the relative Euclidean topology of this restricted payoff family.

Finally, when all $d_s\geq0$, the first sum in \eqref{app:potential:Phi}
is concave and the entropy term is strictly concave on the product of
menu simplices. At any SVE its profile is an interior stationary point
of $\Phi$; strict concavity gives at most one such profile. That profile
determines $g$ and therefore the valuations uniquely. Existence was
proved above. This also covers classes with $d_s=0$.
\end{proof}

\subsection{Robustness to payoff heterogeneity}
\label{app:potential:robustness}

The robustness argument concerns the entire recurrent set, not merely
the continuation of individual equilibria. We use the following
perturbation fact on a fixed compact trapping region. All derivatives
and hyperbolicity statements are understood in the affine state space;
in particular, frequency and menu-probability vectors are restricted to
their respective simplex hyperplanes.

\begin{lemma}[Persistence of a finite hyperbolic recurrent set]
\label{app:potential:persistence}
Let $F_\theta$ be $C^1$ vector fields on a neighborhood of a compact set
$K$, with $F_\theta\to F_0$ in $C^1$ there as $\theta\to0$. Suppose $K$
is positively invariant for every sufficiently small $\theta$. If the
chain-recurrent set of the benchmark semiflow on $K$ consists of finitely
many hyperbolic equilibria in the interior of $K$, then, for sufficiently
small $\theta$, its counterpart consists precisely of their hyperbolic
continuations. Their stable and unstable dimensions are unchanged, and
every trajectory in $K$ converges to one of them.
\end{lemma}
\begin{proof}
Write $\varphi_t^\theta$ for the semiflow. A point is chain recurrent if,
for every $\epsilon,T>0$, it can return to itself by a finite sequence
of orbit segments of lengths at least $T$, with jumps smaller than
$\epsilon$ between successive segments. The segment lengths may always
be taken in $[T,2T]$ by subdivision.

The chain-recurrent set is upper semicontinuous in $\theta$. To see this,
suppose $\theta_j\to0$ and chain-recurrent points $z_j\in K$ converge
to $z$. Fix $\epsilon,T>0$. Continuous dependence of solutions gives
uniform convergence of $\varphi_t^{\theta_j}$ to $\varphi_t^0$ on
$[0,2T]\times K$. A sufficiently accurate returning chain for $z_j$
therefore becomes an $(\epsilon,T)$-chain for the benchmark after
replacing its initial and final points by $z$. Uniform continuity of
$\varphi_t^0$ on the same compact set controls the change in the first
segment. Thus $z$ is benchmark chain recurrent. Compactness proves that
every open neighborhood of the benchmark recurrent set contains the
perturbed recurrent set for all sufficiently small $\theta$.

Let the benchmark equilibria be $z_1,\ldots,z_N$. Hyperbolicity and the
implicit function theorem give continuations $z_i(\theta)$ with the
same stable and unstable dimensions. Choose pairwise disjoint compact
neighborhoods $U_i$ in the interior of $K$ sufficiently small that the
only complete orbit remaining in $U_i$ is $z_i(\theta)$, uniformly for
small $\theta$. This familiar local consequence of hyperbolicity can
also be seen from an indefinite quadratic form: in stable and unstable
coordinates, take the difference of positive quadratic forms on the
unstable and stable subspaces. Its derivative is positive definite for
the linearization, and remains so near the continued equilibrium for
$C^1$-close fields. Along a complete orbit confined to this neighborhood,
boundedness of the quadratic form makes the squared distance to the
equilibrium integrable over $\mathbb R$. Uniform continuity implies
that this distance tends to zero in both time directions. The quadratic
form consequently has the same limit zero at both ends, forcing the
orbit to be stationary.

By upper semicontinuity, the perturbed recurrent set lies in
$\bigcup_i\operatorname{int}U_i$. It is invariant, and continuity of
orbits prevents a complete orbit in this set from moving between the
disjoint neighborhoods. The local property just proved therefore makes
it exactly $\{z_i(\theta):1\leq i\leq N\}$. Every trajectory in $K$
has a nonempty connected omega-limit set contained in the chain-recurrent
set. That limit set is a singleton, proving convergence.
\end{proof}

\begin{corollary}[Robustness to payoff heterogeneity]
\label{app:potential:robust}
Fix a finite sensitivity $\beta\geq0$ and the menu distribution $p$.
Suppose the benchmark satisfies constant competitive-menu payoffs and
every SVE is regular:
\begin{equation}
 \det\bigl(I-Dg(v^*)\bigr)\ne0
 \quad\text{for every }v^*\in\mathcal V(\beta).
 \label{app:potential:regular}
\end{equation}
There exists $\varepsilon>0$ such that every conditional-mean payoff
specification $\widetilde\pi$ with
\[
 \|\widetilde\pi-\pi\|_\infty
 :=\max_{\substack{\omega\in\Omega^+\\s\in\omega}}
       |\widetilde\pi_s(\omega)-\pi_s(\omega)|<\varepsilon
\]
has the same finite number of SVE as the benchmark, with unchanged
stable and unstable dimensions under each mean field. Every normalized
trajectory from finite valuations, and every joint trajectory from
finite valuations and strictly positive frequencies, converges to an
equilibrium. Condition~\eqref{app:potential:regular} holds on a comeager,
full-measure subset of the restricted benchmark payoff family at each
fixed $\beta$ and $p$.
\end{corollary}
\begin{proof}
If $\beta=0$, all choice probabilities and sampled-payoff targets are
constant, for every payoff specification. The normalized equation is a
stable linear filter. In joint sufficient statistics both $x$ and $y$
are stable linear filters with constant targets. There is one equilibrium,
all its eigenvalues are $-1$, and every trajectory converges. Henceforth
assume $\beta>0$.

\emph{Step 1: common compact trapping regions and continuous dependence.}
Restrict provisionally to $\|\widetilde\pi-\pi\|_\infty<1$, and put
\[
 m_s=\min_{\omega\ni s}\pi_s(\omega),\qquad
 M_s=\max_{\omega\ni s}\pi_s(\omega),\qquad
 B=\prod_s[m_s-2,M_s+2].
\]
Every perturbed target $\widetilde g_s(v)$ lies in $[m_s-1,M_s+1]$.
Thus $B$ is strictly inward pointing for the normalized valuation field
and for the valuation coordinates of the joint field. It is also
eventually entered from every finite initial valuation vector. For the
joint system, first use the invariant box spanned by the initial values
and $[m_s-1,M_s+1]$. On this box $\Lambda_s$ has a positive lower bound,
so the rates $\Lambda_s/x_s$ are bounded below, since $x_s\leq1$.
Coordinate comparison then gives finite-time entry into the enlarged
box $B$. This argument uses the enlargement, not finite-time entry into
the payoff box itself.

Let
\[
 \lambda=\min_{s,\,v\in B}\Lambda_s(v)>0,\qquad
 0<\delta_x<\lambda.
\]
The set $K_J=B\times\{x\in\Delta(\mathcal S):x_s\geq\delta_x\}$
is compact and positively invariant, and every joint trajectory
eventually enters it. Indeed, after entry into $B$, the equation
$\dot x_s=\Lambda_s(v)-x_s$ forces $x_s$ above $\delta_x$ in finite
time. The normalized system uses $K_N=B$. Neither region depends on
$\widetilde\pi$. On these regions the corresponding vector fields
converge in $C^1$ to their benchmark fields as
$\widetilde\pi\to\pi$: the weights and denominators in $g$ are smooth,
the denominators are bounded away from zero, and the reward numerators
depend linearly on the conditional means.

\emph{Step 2: regularity implies hyperbolicity in the benchmark.}
Define
\[
 L(v)=\frac1\beta\sum_\omega p(\omega)
                    \log\sum_{s\in\omega}e^{\beta v_s}.
\]
Then $\Lambda=\nabla L$ and $H(v):=D\Lambda(v)=D^2L(v)$ is symmetric
positive semidefinite. With $D(v)=\operatorname{diag}(d_s/\Lambda_s(v)^2)$,
identity~\eqref{app:potential:R} gives
\begin{equation}
 Dg(v)=-D(v)H(v),\qquad J_N(v)=-I-D(v)H(v).
 \label{app:potential:real-spectrum}
\end{equation}
The matrices $DH$ and $H^{1/2}DH^{1/2}$ have the same characteristic
polynomial, by $\det(tI-AB)=\det(tI-BA)$. The latter is symmetric.
Consequently every eigenvalue of $J_N$ is real, even if some $d_s$ are
negative or $H$ is singular. Regularity excludes zero eigenvalues and
therefore makes every normalized equilibrium hyperbolic.
Proposition~\ref{app:model:spectrum} adds only $|\mathcal S|-1$ stable
eigenvalues equal to $-1$ for the joint equilibrium. Compactness of the
equilibrium set and regularity make the number of equilibria finite.

\emph{Step 3: the normalized benchmark has no other chain recurrence.}
For a compact positively invariant region $K$, its maximal invariant set
$\mathcal A=\bigcap_{t\geq0}\varphi_t(K)$ consists of the complete
orbits that remain in $K$. It attracts $K$, contains its chain-recurrent
set, and the latter is internally chain recurrent. We use the
finite-critical-values Lyapunov criterion: a strict Lyapunov function on
$\mathcal A$, with finitely many equilibrium values, excludes chain
recurrence outside equilibria \citep[Theorem~5.5 and
Proposition~6.4]{Benaim1999}. One can see the exclusion directly by
choosing, between the function's value at a nonequilibrium point and its
value later on that orbit, a level distinct from all equilibrium values.
Within the compact recurrent set, the corresponding upper level set
traps all sufficiently accurate chains once they enter it, preventing a
return to the original point. The finiteness qualification matters;
monotonicity of a scalar function alone would not suffice.

On $K_N$, the preceding potential calculation gives
\[
 \frac{d}{dt}\Phi(\sigma(v))
 =\beta\sum_\omega p(\omega)
            \operatorname{Var}_{\sigma_\omega(v)}(g(v)-v)\geq0.
\]
Suppose its derivative vanishes at $v$. Write $c=g(v)-v$.
The entries of $c$ coincide within each connected component of the
class co-occurrence graph. Hence every menu's logits are unchanged by
adding $c$, and $u:=v+c=g(v)$ satisfies $g(u)=u$. The entire orbit through
$v$ is
\[
 v(t)=u+e^{-t}(v-u),\qquad t\in\mathbb R,
\]
because the logits and $g$ stay constant along these component-wise
translations. If $v\ne u$, this orbit is unbounded backwards in time.
It cannot lie in the maximal invariant subset $\mathcal A_N$ of $K_N$.
Therefore, on $\mathcal A_N$, the potential increases strictly outside
the finite equilibrium set. The Lyapunov criterion makes that set the
entire chain-recurrent set on $K_N$.

\emph{Step 4: the joint benchmark has no other chain recurrence.}
Adjoin the auxiliary menu-probability profile
\begin{equation}
 \dot\rho_\omega^s=\sigma_\omega^s(v)-\rho_\omega^s.
 \label{app:potential:profile-lift}
\end{equation}
This equation does not affect $(v,x)$ and is adjoined to both the
benchmark and perturbed systems. Each equilibrium has the unique lift
$\rho=\sigma(v)$, and the augmented Jacobian adds only
$\sum_\omega(|\omega|-1)$ eigenvalues equal to $-1$.
Choose $\delta_\rho>0$ smaller than
$\min_{v\in B,\,\omega,\,s\in\omega}\sigma_\omega^s(v)$.
Taking the product of $K_J$ with the menu simplices restricted to
$\rho_\omega^s\geq\delta_\rho$ gives a common compact trapping region
$K_L$ for the augmented systems. The linear filters in
\eqref{app:potential:profile-lift} give eventual entry, and the
augmented fields converge in $C^1$ on a neighborhood of $K_L$.

In the benchmark, let $y_s=x_sv_s$ and
\[
 e_s=y_s-b_sx_s-d_s,\qquad h_s=X_s(\rho)-x_s.
\]
As in the potential proof, $\dot e=-e$ and $\dot h=-h$.
Every complete orbit in $K_L$ therefore has $e=h=0$. On its maximal
invariant subset $\mathcal A_L$ we consequently have
$x=X(\rho)$ and $v_s=b_s+d_s/X_s(\rho)$. The potential identity is
now exact, with no error term:
\begin{equation}
 \frac{d}{dt}\Phi(\rho)
 =\frac1\beta\sum_\omega p(\omega)\sum_{s\in\omega}
 (\sigma_\omega^s(v)-\rho_\omega^s)
          \log\frac{\sigma_\omega^s(v)}{\rho_\omega^s}\geq0.
 \label{app:potential:lift-strict}
\end{equation}
Equality holds exactly when $\rho=\sigma(v)$. Together with $e=h=0$,
this implies $x=\Lambda(v)$ and $v=g(v)$, so the augmented state is an
equilibrium. The potential is thus strict off the finite augmented
equilibrium set on $\mathcal A_L$. The same Lyapunov criterion shows
that this is the entire chain-recurrent set on $K_L$.

\emph{Step 5: perturbation and genericity.}
Steps 1--4 verify all hypotheses of
Lemma~\ref{app:potential:persistence}, separately for the normalized
field on $K_N$ and the augmented joint field on $K_L$. All benchmark
equilibria are in the interiors of these regions; shrink the payoff
neighborhood to one that works for both applications. Every perturbed
trajectory in either region converges to a continued equilibrium.
Eventual entry extends the conclusion to arbitrary finite valuations
and positive initial frequencies. Projecting an augmented trajectory
onto $(v,x)$ proves the claimed joint convergence. The implicit-function
continuations of $g(v)-v=0$ give the same valuation equilibria in both
systems. Their numbers and stable and unstable dimensions are unchanged;
the auxiliary profile only adds stable directions.

Finally, the class-wise payoff shifts used in the proof of
Theorem~\ref{app:potential:theorem} preserve the benchmark restriction
and shift $g$ by arbitrary constant vectors. Parametric regular-value
theory therefore gives \eqref{app:potential:regular} on a comeager,
full-measure subset of the benchmark family, with $p$ and $\beta$ fixed.
This proves the generic form of the main-text corollary.
\end{proof}

The neighborhood in Corollary~\ref{app:potential:robust} is unrestricted
in the conditional means: competitive-menu means need not remain equal,
and singleton means may also change.

\subsection{Convergence with at most two classes}

\begin{corollary}[At most two classes]
\label{app:potential:two}
For every decision tree with at most two classes and every finite
$\beta\geq0$, every trajectory of either mean field converges to an
equilibrium. If the equilibrium is unique, it is globally asymptotically
stable under both systems.
\end{corollary}
\begin{proof}
Only one competitive class menu is possible, so
Theorem~\ref{app:potential:theorem} applies. If there are two classes,
write $u=v_1-v_2$. The normalized gap field is
$h(u)=g_1(u)-g_2(u)-u$. It is real analytic on $\mathbb R$ and is not
identically zero because $g$ is bounded. All its zeros lie in a compact
interval, so it has finitely many zeros. Each zero determines exactly
one valuation equilibrium. The one-class case is immediate.

For completeness, suppose the equilibrium is unique. If the binary
menu is absent, both targets are constant and Lyapunov stability is
immediate. Otherwise write $p>0$ for its probability and define
\[
 Q:=\frac{x_1-a_1}{p},
 \qquad
 e_s:=x_s(v_s-b_s)-d_s.
\]
Thus $Q$ is a historical-frequency coordinate. Since
$a_1+a_2+p=1$, we have
\[
 x_1=a_1+pQ,\qquad x_2=a_2+p(1-Q).
\]
The joint equations give the exact identities
\[
 \dot e_s=-e_s,\qquad
 \dot Q=
 \ell_\beta\!\left(
 b_1+\frac{d_1+e_1}{a_1+pQ}
 -b_2-\frac{d_2+e_2}{a_2+p(1-Q)}
 \right)-Q,
\]
where $\ell_\beta(w)=(1+e^{-\beta w})^{-1}$.

Let $H(Q)$ denote the right-hand side with $e_1=e_2=0$.
Its zeros in $(0,1)$ correspond exactly to joint equilibria,
so it has a unique zero $Q^*\in(0,1)$. Moreover,
$H(0)>0$ and $H(1)<0$, using the continuous extension of
$d_s/x_s$ when $a_s=d_s=0$. Hence $H$ is positive below
$Q^*$ and negative above it.
Choose any sufficiently small closed interval about $Q^*$.
Its endpoint fields point strictly inward, and both denominators
are bounded away from zero on that interval. By continuity,
these inequalities persist for sufficiently small discrepancies
$(e_1,e_2)$. Since $e_s(t)=e_s(0)e^{-t}$, sufficiently small initial
discrepancies remain small for all time, and the interval is
forward invariant. Finally,
\[
 v_s=b_s+\frac{d_s+e_s}{x_s}
\]
shows that $(Q,e_1,e_2)$ are smooth local coordinates near the
equilibrium. Taking the interval and discrepancies sufficiently
small therefore proves joint Lyapunov stability.

For the normalized system the scalar gap field $h$ has the same inward sign property about
its unique zero; the remaining valuation follows a stable linear filter.
This proves normalized Lyapunov stability as well. The already proved
global attraction completes the claim.
\end{proof}

\subsection{The Bordeaux--Chianti illustrations}
For the three equally likely menus $\{B,C\}$, $\{B\}$, and $\{C\}$,
competitive payoffs $(2,1)$, and common singleton payoff $z$, put
$q=\ell_\beta(v_B-v_C)$. Then
\[
 g_B=\frac{2q+z}{1+q},\qquad
 g_C=\frac{1-q+z}{2-q}.
\]
At any SVE the normalized eigenvalues are
\begin{equation}
 -1,\qquad -1+\beta q(1-q)
 \left\{\frac{2-z}{(1+q)^2}+\frac{1-z}{(2-q)^2}\right\}.
 \label{app:potential:wine-eigenvalues}
\end{equation}
Indeed $Dg=g'(q)\beta q(1-q)(1,-1)$ has rank at most one.
The joint system has one further eigenvalue $-1$.

When $z=0$, the strict equilibria associated with the two rankings are
$(1,0)$ and $(0,1/2)$. They produce two locally exponentially stable
branches at high sensitivity by the strict-equilibrium theorem.
Every interior limiting probability must solve $g_B(q)=g_C(q)$, or
$q^2-4q+1=0$. The only root in $(0,1)$ is $q_0=2-\sqrt3$, with common
valuation $1-1/\sqrt3$. The derivative of $g_B-g_C$ is strictly positive
on $[0,1]$. The implicit function theorem applied to
\[
 g_B(q)-g_C(q)-\beta^{-1}\log\frac{q}{1-q}=0
\]
therefore gives exactly one branch near $q_0$; its second eigenvalue in
\eqref{app:potential:wine-eigenvalues} is positive for sufficiently large
$\beta$. Compactness of $[0,1]$ and the three possible limiting
probabilities show that these three branches exhaust the SVE for all
sufficiently large sensitivities. The stable branches have both normalized
eigenvalues tending to $-1$, since $\beta q(1-q)\to0$ there.

When $z=11/3$, both $d_s$ are positive. The SVE is unique, globally
attracting, and locally exponentially stable at every finite sensitivity;
for $\beta>0$ its second normalized eigenvalue is strictly below $-1$.
The limiting probability cannot be zero or one, because the corresponding
limiting valuation gap would have the opposite sign. Its interior limit
solves $g_B(q)=g_C(q)$, with unique solution $q=2/3$ and valuations $(3,3)$.
This proves the indifference and mixing values used.

\subsection{A useful lower bound on the size of a cycling example}
\begin{lemma}[At most one class with heterogeneous menu means]
\label{app:potential:oneheterogeneous}
If $g_s$ is constant for all but at most one class, every trajectory of
either mean field converges to an equilibrium at every finite sensitivity.
Consequently cycling requires at least three classes and at least five
objective alternatives. Deleting any supported menu from the three-menu
construction below eliminates cycling.
\end{lemma}
\begin{proof}
For each constant-target class, normalized valuations converge
exponentially to their target; under the joint system the identity
$x_s(v_s-c_s)=e^{-t}x_s(0)(v_s(0)-c_s)$ gives the same conclusion because
$x_s$ stays away from zero. For the remaining class $j$, after rescaling
time by its bounded, strictly positive joint coefficient $\Lambda_j/x_j$
(or without rescaling in the normalized system), its equation is an
asymptotically autonomous scalar equation with limiting field
$g_j(v_j,c_{-j})-v_j$. This field is analytic, bounded minus $v_j$, and
has only finitely many zeros in the trapping interval. A scalar
asymptotic pseudotrajectory with compact range has its limit set among
these zeros: away from a zero the field has a fixed nonzero sign, which
excludes chain recurrence there. Connectedness gives convergence to
one zero. The frequency filters then converge as well.

With at most two classes the preceding corollary applies. With three
or more classes but at most four alternatives, at most one class can
contain more than one alternative. Every other class has a constant
conditional mean, so the lemma applies. The five-alternative construction
attains this lower bound. Removing either of its binary menus leaves
constant competitive-menu payoffs; removing its singleton leaves only
$B$ with heterogeneous menu means. These observations prove the last
assertion. The alternative-count claim concerns fixed objective reward
laws and a fixed within-class implementation rule, as maintained here.
\end{proof}

\section{Corrective feedback and singleton-payoff shifts}
\label{app:single:section}

Throughout this section, $\Omega^+=\{\omega:p(\omega)>0\}$,
$\Omega_s^+=\{\omega\in\Omega^+:s\in\omega\}$, and
\[
 \Lambda_s(v)=\sum_{\omega\in\Omega_s^+}p(\omega)\sigma_\omega^s(v),
 \qquad
 g_s(v)=\frac{\sum_{\omega\in\Omega_s^+}p(\omega)
                  \sigma_\omega^s(v)\pi_s(\omega)}{\Lambda_s(v)},
 \qquad F(v)=g(v)-v.
\]
The sensitivity $\beta<\infty$ is fixed unless otherwise indicated. Every
class occurs in a supported menu, so $\Lambda_s(v)>0$ at every finite $v$.
The normalized dynamics are $\dot v=F(v)$. The coefficient-one harmonic
visit-count dynamics are
\begin{equation}
 \dot v_s=\frac{\Lambda_s(v)}{x_s}F_s(v),
 \qquad \dot x_s=\Lambda_s(v)-x_s,
 \qquad x\in\operatorname{int}\Delta(\mathcal S).
 \label{app:single:joint}
\end{equation}
Write $m_s=\min_{\Omega_s^+}\pi_s$, $M_s=\max_{\Omega_s^+}\pi_s$, and
$K=\prod_s[m_s,M_s]$; some coordinate intervals may be singletons.

\subsection{Negative payoff--crowding covariance}

\begin{lemma}[Derivative and row-sum identities]
\label{app:single:derivative}
For distinct $s,r$,
\begin{align}
 \partial_{v_r}g_s(v)
 &=\frac{\beta}{\Lambda_s(v)}\sum_{\omega\in\Omega_s^+}
 p(\omega)\sigma_\omega^s(v)\sigma_\omega^r(v)
 [g_s(v)-\pi_s(\omega)]                                      \label{app:single:derivative-form}\\
 &=-\beta\operatorname{Cov}_{w_s^v}
                 (\pi_s(\omega),\sigma_\omega^r(v)),
 \qquad
 w_s^v(\omega)=\frac{p(\omega)\sigma_\omega^s(v)}{\Lambda_s(v)}.
                                                               \label{app:single:covariance}
\end{align}
Here $\sigma_\omega^r=0$ if $r\notin\omega$. Moreover,
$Dg(v)\mathbf 1=0$ and $DF(v)\mathbf 1=-\mathbf 1$.
\end{lemma}
\begin{proof}
For $s\ne r$, $\partial_{v_r}\sigma_\omega^s
=-\beta\sigma_\omega^s\sigma_\omega^r$. Differentiating the quotient
defining $g_s$ gives \eqref{app:single:derivative-form}. Since
$\sum_\omega w_s^v(\omega)\pi_s(\omega)=g_s(v)$, the same expression
is the covariance in \eqref{app:single:covariance}. Logit probabilities
are unchanged by $v\mapsto v+c\mathbf 1$, hence so is $g$.
Differentiation in $c$ proves both row-sum identities. These formulas
remain valid at $\beta=0$, when $g$ is constant.
\end{proof}

\begin{theorem}[Corrective feedback: both dynamics and arbitrary initial values]
\label{app:single:convergence}
Suppose the covariances in \eqref{app:single:covariance} are nonpositive
for every $v\in K$ and every $s\ne r$. There is a unique SVE $v^*$.
The normalized equilibrium $v^*$ and the joint equilibrium
$(v^*,\Lambda(v^*))$ are globally asymptotically stable, for arbitrary
finite initial valuations and strictly positive initial frequencies.
For two normalized solutions initialized in $K$,
\begin{equation}
 \|v(t)-u(t)\|_\infty\le e^{-t}\|v(0)-u(0)\|_\infty.
 \label{app:single:normalized-contraction}
\end{equation}
Under the maintained stochastic-approximation assumptions, the
harmonic visit-count recursion converges almost surely to
$(v^*,\Lambda(v^*))$.
\end{theorem}
\begin{proof}
\emph{Existence and contraction on $K$.}
The continuous map $g$ takes all of $\mathbb R^{\mathcal S}$ into $K$,
so Brouwer's theorem on the possibly lower-dimensional box $K$ gives
a fixed point. For $v,u\in K$, put
\[
 A(v,u)=\int_0^1 DF(u+\theta(v-u))\,d\theta.
\]
Lemma~\ref{app:single:derivative} makes $A$ Metzler with row sums
$-1$. If $z=v-u$ and $|z_i|=\|z\|_\infty$, then
\[
 \operatorname{sgn}(z_i)(Az)_i
 \le \left(A_{ii}+\sum_{j\ne i}A_{ij}\right)\|z\|_\infty
 =-\|z\|_\infty.
\]
The upper right Dini derivative therefore satisfies
$D^+\|z(t)\|_\infty\le-\|z(t)\|_\infty$ for normalized solutions in
$K$. The box is positively invariant because $g_s\in[m_s,M_s]$.
Gronwall's inequality proves \eqref{app:single:normalized-contraction},
and applying this estimate to two fixed points proves uniqueness.

\emph{Bounds along trajectories.}
Let
\[
 H=\prod_s[\min\{m_s,v_s(0)\},\max\{M_s,v_s(0)\}].
\]
For either dynamics, $H$ is positively invariant. On this compact box,
set $d=\min_{s,v\in H}\Lambda_s(v)>0$. The frequency equation gives
\[
 x_s(t)=e^{-t}x_s(0)+\int_0^t e^{-(t-u)}\Lambda_s(v(u))\,du
 \ge\min\{x_s(0),d\}>0.
\]
Consequently the joint adjustment rates
$a_s(t)=\Lambda_s(v(t))/x_s(t)$ are bounded above and below by
strictly positive constants $\overline a$ and $\underline a$.
For normalized dynamics take both constants equal to one.
These bounds also imply global existence. Coordinate comparison at
the endpoints of $K$ yields
\begin{equation}
 \delta(t):=\operatorname{dist}_\infty(v(t),K)
 \le e^{-\underline a t}\delta(0).
 \label{app:single:attraction-box}
\end{equation}
Indeed each positive part $(v_s-M_s)_+$ and $(m_s-v_s)_+$ has Dini
derivative at most $-a_s(t)$ times itself.

\emph{Convergence even when the covariance condition is imposed only on $K$.}
Let $\widehat v$ be the coordinatewise projection of $v$ onto $K$ and
$e=v-v^*$. On the compact set $H$, $g$ is Lipschitz. Put
\[
 B(\widehat v)=\int_0^1Dg(v^*+\theta(\widehat v-v^*))\,d\theta.
\]
Its off-diagonal entries are nonnegative and its row sums vanish.
There is a finite constant $C_H$ such that
\[
 F(v)=(B(\widehat v)-I)e+r(v),
 \qquad \|r(v)\|_\infty\le C_H\|v-\widehat v\|_\infty.
\]
This identity follows by expanding $g(\widehat v)-g(v^*)$ and
absorbing $g(v)-g(\widehat v)-B(\widehat v)(v-\widehat v)$ into
$r(v)$. The previous active-coordinate calculation, now multiplied
by the positive rate $a_i(t)$, gives
\[
 D^+\|e(t)\|_\infty
 \le-\underline a\|e(t)\|_\infty
       +\overline a C_H\delta(t).
\]
Together with \eqref{app:single:attraction-box}, this implies
\begin{equation}
 \|v(t)-v^*\|_\infty
 \le e^{-\underline a t}
   \bigl(\|v(0)-v^*\|_\infty+\overline a C_H\delta(0)t\bigr)
 \longrightarrow0.
 \label{app:single:outside-estimate}
\end{equation}
Continuity of $\Lambda$ and variation of constants in the frequency
equation then give $x(t)\to\Lambda(v^*)$.

The estimates also prove stability. For initial states in a fixed
small neighborhood of the equilibrium, $H$ can be chosen as one
fixed compact box and the rate bounds chosen uniformly. Since
$\delta(0)\le\|v(0)-v^*\|_\infty$, the right side of
\eqref{app:single:outside-estimate} is bounded by a fixed multiple
of the initial valuation error, uniformly in time. The frequency
variation-of-constants estimate adds a fixed multiple of that error
to the initial frequency error. This proves Lyapunov stability and,
together with global attraction, global asymptotic stability.
Finally the stochastic-approximation proposition established for
the harmonic recursion identifies its almost-sure limit set with
an internally chain transitive set of the joint flow. A globally
asymptotically stable equilibrium is its only such limit set in the
compact region visited by a bounded trajectory, proving the
almost-sure assertion.
\end{proof}

\begin{remark}
The box $K$ is globally attracting and positively invariant. It need
not be entered in finite time. For example, when $g_s$ is constant,
a coordinate initialized above that constant approaches it from
above at every finite time. No finite-time absorption claim is
needed in the preceding proof. 
\end{remark}

\subsection{High singleton payoffs}

Assume every singleton is supported, there are at least two classes,
and the co-occurrence graph is connected. Write $q_s=p(\{s\})$,
$a_s=\pi_s(\{s\})$, and $\pi_s^z(\omega)=\pi_s(\omega)+
z\mathbf 1\{\omega=\{s\}\}$. To give explicit sufficient cutoffs, set
\[
 P_s=\sum_{\omega\in\Omega_s^+}p(\omega),\quad
 \alpha_s=q_s/P_s,\quad
 \ell_s=\min_{\substack{\omega\in\Omega_s^+\\|\omega|\ge2}}
                \pi_s(\omega),\quad
 u_s=\max_{\substack{\omega\in\Omega_s^+\\|\omega|\ge2}}
                \pi_s(\omega),
\]
and
\begin{align}
 z_{\mathrm{cov}}&=\max_s
       \left\{\frac{u_s-(1-\alpha_s)\ell_s}{\alpha_s}-a_s\right\},
                                                        \label{app:single:zcov}\\
 M&=\max_{s,\omega\in\Omega_s^+}\pi_s(\omega),\quad
 a_{\min}=\min_s a_s,\quad
 p_C=\min_{\substack{\omega\in\Omega^+\\|\omega|\ge2}}p(\omega),
 \quad q_{\max}=\max_s q_s,\nonumber\\
 \lambda_*&=\frac{q_{\max}}{q_{\max}+p_C}<1,
 \qquad z_{\mathrm{tie}}=\frac{M-a_{\min}}{1-\lambda_*}.
                                                        \label{app:single:ztie}
\end{align}
Connectedness ensures that all competitive-payoff extrema and $p_C$
are well defined.

\begin{theorem}[High singleton shifts]
\label{app:single:high}
Let $\bar z_0=\max\{0,z_{\mathrm{cov}},z_{\mathrm{tie}}\}$. For each
fixed $z>\bar z_0$, every finite sensitivity has a unique SVE and
both dynamics converge globally to the corresponding equilibrium.
The harmonic stochastic recursion converges almost surely under
the maintained assumptions. As $\beta\to\infty$, the unique SVE
and its entire menu-choice profile converge to a VE
$(\sigma^\infty,v^\infty)$ with at least one supported competitive
menu containing multiple valuation maximizers.

Moreover, for every fixed $z>\bar z_0$, the limiting choice profile
$\sigma^\infty$ assigns positive probability to at least two classes
in some supported competitive menu.
\end{theorem}
\begin{proof}
\emph{Uniform corrective feedback.}
The conditional singleton weight is $q_s/\Lambda_s(v)\ge\alpha_s$.
For $z>z_{\mathrm{cov}}$, $a_s+z>\ell_s$ and therefore, uniformly
over all finite $v$ and all finite sensitivities,
\[
 g_s^z(v)\ge\alpha_s(a_s+z)+(1-\alpha_s)\ell_s>u_s.
\]
Singletons contribute zero to every cross derivative in
\eqref{app:single:derivative-form}; all competitive contributions
are nonnegative. Theorem~\ref{app:single:convergence} proves all
finite-sensitivity dynamic conclusions.

\emph{No menu-wise strict VE at high shifts.}
Suppose a VE has a unique maximizer in every supported menu and
let $i$ have minimum valuation. It cannot be selected in a
competitive menu: if it were that menu's unique maximizer, every
other available valuation would be strictly smaller, contradicting
minimality. Hence $v_i=a_i+z\ge a_{\min}+z$. Choose a supported
competitive menu containing $i$, and denote its unique winner by
$j$. If $D_j$ is the total probability of menus won by $j$, then
\[
 v_j=b_j+\lambda_jz,\qquad
 b_j\le M,\qquad
 \lambda_j=q_j/D_j\le\frac{q_j}{q_j+p_C}\le\lambda_*.
\]
For $z>z_{\mathrm{tie}}$, this gives $v_i-v_j\ge
(1-\lambda_*)z+a_{\min}-M>0$, contradicting the victory of $j$.

\emph{A limit of the unique SVE path.}
For fixed $z$, the graph of $v^\beta=g^z(v^\beta;\beta)$ is
definable in $\mathbb R_{\exp}$. Uniqueness makes it the graph of
a single bounded function of $\beta$, since all equilibria lie
in $K(z)$. The o-minimal monotonicity theorem implies convergence
of each coordinate at infinity. Each associated logit probability
is also bounded and definable, so the complete choice profile
converges. If $s$ is not a maximizer in $\omega$ at $v^\infty$,
its logit probability is eventually at most $e^{-\beta\eta}$ for
some $\eta>0$. The limiting choices are thus supported on menu
maximizers. Since $\Lambda_s(v^\beta)\ge q_s>0$, passage to the
limit in the consistency equations is legitimate and proves VE
consistency. The preceding step excludes menu-wise strictness of
this limiting VE.

\emph{Genuine mixing of the selected limit.}
Suppose, to the contrary, that $\sigma^\infty$ is pure in every
supported menu, and write $\varphi(\omega)$ for its chosen class.
Since the menu support is finite, for all sufficiently large finite
$\beta>0$,
\[
 \sigma_\omega^{\varphi(\omega)}(v^\beta)>\frac12
 \qquad\text{for every }\omega\in\Omega^+.
\]
The logit formula therefore implies
$v_{\varphi(\omega)}^\beta>v_r^\beta$ for every
$r\in\omega\setminus\{\varphi(\omega)\}$.
Fix one such $\beta$ and break any remaining valuation ties to obtain
a strict total ranking implementing $\varphi$ in every supported menu.

Let $i$ be the lowest-ranked class. It is selected only in its
singleton menu, so limiting VE consistency gives
\[
 v_i^\infty=a_i+z\ge a_{\min}+z.
\]
Connectedness supplies a supported competitive menu containing $i$;
let $j\ne i$ be its selected class. Define
\[
 D_j=\sum_{\omega:\varphi(\omega)=j}p(\omega),
 \qquad
 b_j=\frac{\sum_{\omega:\varphi(\omega)=j}
                   p(\omega)\pi_j(\omega)}{D_j}.
\]
Since $j$ is selected in its singleton and in this competitive menu,
\[
 D_j\ge q_j+p_C,\qquad b_j\le M,\qquad
 v_j^\infty=b_j+\frac{q_j}{D_j}z
            \le M+\lambda_*z.
\]
Consequently, $z>z_{\mathrm{tie}}$ implies
\[
 v_i^\infty-v_j^\infty
 \ge a_{\min}-M+(1-\lambda_*)z>0.
\]
This contradicts VE maximization in the competitive menu containing
$i$ and selected class $j$. Thus $\sigma^\infty$ mixes in at least
one supported competitive menu.
\end{proof}

\subsubsection{Monotone menu probabilities and total indifference}

In this subsection suppose every nonempty menu is supported,
$q_s=q>0$ for every class, and
$p(\omega)\le p(\omega')$ whenever $\omega\subseteq\omega'$.
Let $M_0=\max_{s,\omega\in\Omega_s^+}|\pi_s(\omega)|$.

\begin{proposition}[Total indifference at every VE]
\label{app:single:total}
For $z>2M_0/q^2$, every VE of the shifted tree has equal valuations
across all classes. For $z>\max\{\bar z_0,2M_0/q^2\}$, the
unique SVE path selects one such VE whose limiting choice profile
satisfies $\sigma^\infty_\omega(s)>0$ for every supported menu
$\omega$ and every $s\in\omega$. No genericity restriction is needed.
\end{proposition}
\begin{proof}
Fix any VE $(\sigma,v)$ and write
\[
 D_s=\sum_{\omega\ni s}p(\omega)\sigma_\omega(s),\qquad
 v_s=b_s+(q/D_s)z,\qquad |b_s|\le M_0.
\]
Suppose $v_i<v_j$, and let
$\mathcal M_i=\{\omega:\sigma_\omega(i)>0\}$. For every
$\omega\in\mathcal M_i$, all classes in $\omega$ have value at
most $v_i$, so $j\notin\omega$ and $j$ is the unique maximizer
in $\omega\cup\{j\}$. The injection
$\omega\mapsto\omega\cup\{j\}$ has image disjoint from $\{j\}$.
Consequently
\[
 D_j\ge q+\sum_{\omega\in\mathcal M_i}p(\omega\cup\{j\})
       \ge q+\sum_{\omega\in\mathcal M_i}p(\omega)
       \ge q+D_i.
\]
Since $D_i,D_j\le1$,
$q/D_i-q/D_j\ge q^2$. Consistency now gives
$v_i-v_j\ge q^2z-2M_0>0$, a contradiction. Thus all valuations tie.

For $z>\max\{\bar z_0,2M_0/q^2\}$, write $\sigma^\beta$ for the
logit profile at the unique SVE $v^\beta$. By
Theorem~\ref{app:single:high}, these profiles and valuations converge
to a VE $(\sigma^\infty,v^\infty)$. Suppose some supported menu $\nu$
contains a class $i$ with $\sigma^\infty_\nu(i)=0$. Choose $j\in\nu$
with $\sigma^\infty_\nu(j)>0$. Their logit odds satisfy
\[
 R_\beta:=\exp\{\beta(v_j^\beta-v_i^\beta)\}
 =\frac{\sigma^\beta_\nu(j)}{\sigma^\beta_\nu(i)}
 \longrightarrow\infty .
\]
Consider any menu $\omega$ with $\sigma^\infty_\omega(i)>0$.
It excludes $j$, since otherwise
$\sigma^\beta_\omega(j)=R_\beta\sigma^\beta_\omega(i)$ would exceed one.
Adding $j$ to this menu gives
\[
 \sigma^\beta_{\omega\cup\{j\}}(j)
 =\frac{R_\beta\sigma^\beta_\omega(i)}
        {1+R_\beta\sigma^\beta_\omega(i)}
 \longrightarrow 1.
\]
The same injection thus gives $D_j\ge q+D_i$ for the selection
masses of $\sigma^\infty$, and hence $q/D_i-q/D_j\ge q^2$.
Since all limiting valuations are equal, payoff consistency gives
$0=v_i^\infty-v_j^\infty\ge -2M_0+zq^2>0$, a contradiction.
Every available class thus has positive probability in the
selected limiting profile.
\end{proof}

\subsection{Additional primitive sources of corrective feedback}
\label{app:single:alternative-primitives}

High singleton payoffs are one way to ensure corrective feedback in the sense of Theorem~\ref{app:single:convergence}. More generally, the relevant comparison is between a class's assessment and the payoffs of the observations displaced by each rival. In this subsection, all menu sums and extrema here range over the
supported set $\Omega^+$. For distinct co-occurring classes $s,r$, define
\[
 Z_\omega(v):=\sum_{j\in\omega}\exp(\beta v_j),
 \qquad
 h_{sr}(v):=
 \frac{\displaystyle\sum_{\omega\ni s,r}
       p(\omega)\pi_s(\omega)/Z_\omega(v)^2}
      {\displaystyle\sum_{\omega\ni s,r}
       p(\omega)/Z_\omega(v)^2}.
\]
Its weights are proportional to
$p(\omega)\sigma_\omega^s(v)\sigma_\omega^r(v)$.
Consequently,
\begin{equation}
 \operatorname{Cov}_{w_s^v}
       \bigl(\pi_s(\omega),\sigma_\omega^r(v)\bigr)
 =
 \mathbb E_{w_s^v}[\sigma_\omega^r(v)]
 \bigl[h_{sr}(v)-g_s(v)\bigr].
 \label{app:single:displaced-mean}
\end{equation}
Nonpositive covariance is equivalent to $h_{sr}(v)\leq g_s(v)$.
Pairs that never co-occur have zero covariance.

\begin{proposition}[A primitive sufficient condition]
\label{app:single:primitive-crowding}
Suppose every singleton menu is supported and every co-occurring pair
$s,r$ has a supported binary menu. Write
\[
 q_s=p(\{s\}),\qquad a_s=\pi_s(\{s\}),\qquad
 P_s=\sum_{\omega\ni s}p(\omega),
\]
and, for each co-occurring pair, let
\[
 p_{sr}=p(\{s,r\}),\qquad
 H_{sr}=\sum_{\omega\supsetneq\{s,r\}}p(\omega),\qquad
 U_{sr}=\max_{\substack{\omega\in\Omega^+\\s,r\in\omega}}\pi_s(\omega).
\]
If
\begin{equation}
 m_s+\frac{q_s}{P_s}(a_s-m_s)
 \geq
 \frac{p_{sr}\pi_s(\{s,r\})+H_{sr}U_{sr}}
      {p_{sr}+H_{sr}}
 \quad\text{for every co-occurring ordered pair }s,r,
 \label{app:single:primitive-crowding-condition}
\end{equation}
then all payoff--crowding covariances are nonpositive for every
$v\in\mathbb R^{\mathcal S}$ and every finite $\beta\geq0$.
\end{proposition}

\begin{proof}
Since $\Lambda_s(v)\leq P_s$ and every payoff is at least $m_s$,
\[
 g_s(v)\geq
 m_s+\frac{q_s}{\Lambda_s(v)}(a_s-m_s)
 \geq m_s+\frac{q_s}{P_s}(a_s-m_s).
\]
For every $\omega\supsetneq\{s,r\}$,
$Z_\omega(v)\geq Z_{\{s,r\}}(v)$. The binary menu therefore
receives at least the fraction $p_{sr}/(p_{sr}+H_{sr})$ of the
weights defining $h_{sr}$. Since
$\pi_s(\{s,r\})\leq U_{sr}$,
\[
 h_{sr}(v)\leq
 \frac{p_{sr}\pi_s(\{s,r\})+H_{sr}U_{sr}}
      {p_{sr}+H_{sr}}.
\]
The conclusion follows from
\eqref{app:single:displaced-mean}.
\end{proof}

This condition permits favorable experiences in larger menus,
provided each rival predominantly displaces poorer observations.
The following example also shows that intermediate singleton
payoffs can sustain robust limiting indifference and mixing.

\begin{corollary}[Intermediate singleton payoffs]
\label{app:single:intermediate-example}
Consider three classes and all seven nonempty menus. Each singleton menu
and each binary menu has probability $10/61$, and the ternary menu
has probability $1/61$. Every represented class receives expected
payoff $1$ in singleton menus, $0$ in binary menus, and $2$ in the
ternary menu.
At every finite sensitivity, all ordered-pair payoff--crowding
covariances are strictly negative, and the unique SVE is
\[
 v^\beta=\frac{32}{61}\mathbf 1.
\]
Choices are uniform within every menu. In the relative parameter
space of reduced primitives $(p,\pi)$ with this seven-menu support,
there is an open neighborhood of the stated specification in
which the covariance condition holds at every finite sensitivity,
every VE has equal class valuations, and the unique SVE converges
as $\beta\to\infty$ to total indifference with strictly positive
choice probabilities in every menu.
\end{corollary}

\begin{proof}
For every class and co-occurring rival,
\[
 g_s(v)\geq\frac{10}{31},
 \qquad
 h_{sr}(v)\leq\frac{2}{11},
 \qquad
 g_s(v)-h_{sr}(v)\geq\frac{48}{341}>0.
\]
Proposition~\ref{app:single:primitive-crowding} and
Theorem~\ref{app:single:convergence} give the covariance signs
and uniqueness. Symmetry forces equal equilibrium valuations.
At equal valuations the conditional payoff is
\[
 \frac{10+2/3}{10+10+1/3}=\frac{32}{61},
\]
independently of sensitivity.

We establish robustness explicitly. For each nonempty proper
subset $H\subset\{A,B,C\}$, consider the compact set
$\mathcal P_H$ of choice profiles that choose only from $H$
whenever the menu intersects $H$. Write $G_s(\sigma)$ for the
conditional payoff induced by a profile $\sigma$. At the
displayed primitives, every $\sigma\in\mathcal P_H$ satisfies,
for $s\in H$ and $r\notin H$,
\[
 \begin{array}{ll}
 |H|=1:
 &G_s(\sigma)=12/31<1/2\leq G_r(\sigma),\\[2pt]
 |H|=2:
 &G_s(\sigma)\leq12/21<1=G_r(\sigma).
 \end{array}
\]
For the first comparison, the sole member of $H$ wins both its
binary menus and the ternary menu, while each other class receives
at most one binary-menu type alongside its singleton observations.
For the second, each member of $H$ must win its binary menu against
the remaining class; that class receives only singleton observations.

Singleton probabilities bound all sampling denominators away from
zero. Continuity, compactness of the sets $\mathcal P_H$, and the
finitely many possible $H$ therefore make these strict comparisons
persist throughout one open neighborhood of primitives.
Shrink this neighborhood so that the strict inequalities in
\eqref{app:single:primitive-crowding-condition} also persist.
A VE with unequal valuations would have a proper highest-valued
group $H$, whose choices belong to $\mathcal P_H$. Its payoff
consistency would contradict the displayed comparisons.
Every nearby VE is therefore totally indifferent.

Fix any specification in this neighborhood. The unique SVE and
its logit choice profile are bounded functions of $\beta$ definable
in $\mathbb R_{\exp}$. The monotonicity argument used in the proof
of Theorem~\ref{app:eq:structure} therefore gives limits as
$\beta\to\infty$, without a genericity restriction. The same
argument applies to
\[
 \lambda_s^\beta:=
 \frac{\exp(\beta v_s^\beta)}
      {\sum_j \exp(\beta v_j^\beta)}.
\]
For every choice profile, singleton support gives
$D_s(\sigma):=\sum_{\omega\ni s}p(\omega)\sigma_\omega^s
\geq q_s>0$. Continuity of $G_s$ consequently implies
\[
 v_s^\infty
 =\lim_{\beta\to\infty}G_s(\sigma^\beta)
 =G_s(\sigma^\infty).
\]
Let $H=\{s:\lambda_s^\infty>0\}$. If $H$ were proper, the limiting
choice profile would belong to $\mathcal P_H$, implying
$v_s^\infty<v_r^\infty$ for $s\in H$, $r\notin H$. But
$\lambda_s^\beta/\lambda_r^\beta\to\infty$ implies
$v_s^\beta>v_r^\beta$ eventually, and hence
$v_s^\infty\geq v_r^\infty$, a contradiction.
Thus every $\lambda_s^\infty>0$, and
\[
 \sigma_\omega^{\infty,s}
 =\frac{\lambda_s^\infty}
        {\sum_{j\in\omega}\lambda_j^\infty}>0
 \qquad(s\in\omega).
\]
This proves full limiting mixing.
\end{proof}

A different construction dispenses with singleton menus altogether.
Here, broader assortments generate poorer focal-class experiences,
while balanced availability ensures that each rival
disproportionately displaces those experiences.

\begin{proposition}[Corrective feedback without singleton menus]
\label{app:single:no-singletons}
Suppose there is one focal class $s$ and $m\geq2$ benchmark classes,
each yielding the same constant expected payoff $c$.
There are no singleton menus. Every menu containing $s$ has the
form $\{s\}\cup B$, where $B$ is a nonempty subset of benchmark
classes. Conditional on $k=|B|$, these subsets are equally likely,
and
\[
 \pi_s(\{s\}\cup B)=h(k),
\]
where $h$ is nonincreasing on the supported sizes.
Then all payoff--crowding covariances are nonpositive on $K$
at every finite sensitivity. No ordering between $c$ and the
values of $h$ is required.
\end{proposition}

\begin{proof}
On $K$, every benchmark valuation equals $c$. Put
$a=\exp(\beta(v_s-c))$. In a menu with $k$ benchmark classes,
$\sigma^s=a/(a+k)$, so conditioning on selection of $s$
preserves uniformity among subsets of a given size.
For any benchmark class $r$,
\[
 \mathbb E_{w_s^v}[\sigma_\omega^r(v)\mid k]
 =\frac{k}{m(a+k)}.
\]
Writing $\nu_s^v$ for the distribution of $k$ conditional on
choosing $s$, we obtain
\[
 \operatorname{Cov}_{w_s^v}(\pi_s,\sigma^r)
 =
 \operatorname{Cov}_{\nu_s^v}
 \left(h(k),\frac{k}{m(a+k)}\right)\leq0,
\]
because the first argument is nonincreasing and the second
strictly increasing. Benchmark rows have zero covariance
because their payoffs are constant.
\end{proof}

For example, support $AB,AC,ABC$, with $p(AB)=p(AC)>0$ and
$p(ABC)>0$. Let $A$ yield payoff $2$ in the binary menus and $0$
in the ternary menu, while $B,C$ always yield $1$.
The condition holds on $K=[0,2]\times\{1\}\times\{1\}$ at every
finite sensitivity. Theorem~\ref{app:single:convergence}
therefore applies even from initial valuations outside $K$.

\subsection{Low singleton payoffs and strict equilibria}

For a pure selector $\varphi$ on the supported menus, define
\[
 D_s^\varphi=\sum_{\omega:\varphi(\omega)=s}p(\omega),\qquad
 b_s^\varphi=\frac{\sum_{\omega:\varphi(\omega)=s}
                    p(\omega)\pi_s(\omega)}{D_s^\varphi},\qquad
 \lambda_s^\varphi=q_s/D_s^\varphi.
\]
All denominators are positive because a selector must choose $s$
in $\{s\}$. Set
\begin{equation}
 v_s^\varphi(z)=b_s^\varphi+\lambda_s^\varphi z.
 \label{app:single:frozen}
\end{equation}
These are the unique valuations consistent with that selector.
If they strictly order the classes in the order used by $\varphi$,
the resulting pair is a strict VE.

\begin{lemma}[Priority construction]
\label{app:single:priority}
There is a generic set of fixed-support decision trees with
singleton support for which some strict VE exists at every
sufficiently negative shift. If every binary menu is supported,
then, without any genericity restriction, for every designated
class $b$ there is a strict VE at every sufficiently negative
shift with $b$ uniquely lowest.
\end{lemma}
\begin{proof}
\emph{The construction.}
Initially every supported menu is unassigned. At a priority stage,
let $R$ be the remaining candidate classes, and for $i\in R$ let
\[
 T_i(R)=\sum_{\substack{\omega\text{ unassigned}\\i\in\omega}}
           p(\omega),\qquad \ell_i(R)=q_i/T_i(R).
\]
Choose a class minimizing $\ell_i(R)$, assign to it every still
unassigned menu containing it, and remove that class from $R$.
The assigned menus are exactly those in which it is highest in
the resulting priority order. Its final singleton share is its
current $\ell_i(R)$. For a later class, removing menus can only
increase its eventual share.

\emph{Generic existence.}
Run the algorithm with all classes as candidates. Impose the
following generic conditions. For each possible remaining set
$R$, whenever distinct $i,j\in R$ both have positive residual
competitive mass $T_i(R)-q_i$ and $T_j(R)-q_j$, their residual
shares differ; also all singleton baseline payoffs $a_i$ differ.
There are finitely many remaining sets. Each excluded equality
of shares is the nonzero polynomial
\[
 q_i[T_j(R)-q_j]-q_j[T_i(R)-q_i]=0.
\]
It is nonzero because the competitive masses do not involve any
singleton probabilities and are both nonzero linear forms in
competitive-menu probabilities. Varying one singleton weight
on the positive cone makes the equality fail, and homogeneity
shows that its restriction to the probability simplex is not
identically zero. The excluded payoff equalities are likewise
proper hyperplanes. Thus these conditions hold on an open dense
set of full measure in the fixed-support parameter space.

While any competitive menu remains, the minimizing share is less
than one and is strictly smaller than every other remaining
share. Hence it is strictly smaller than the final share of
every later class. Once no competitive menu remains, all remaining
classes are selected only in their singleton menus. Order these
last classes by decreasing $a_i$. Their consistent valuations
are $a_i+z$ and remain strictly ordered. For an earlier class and
any later class the strict singleton-share gap makes the earlier
valuation larger for all sufficiently negative $z$. There are
only finitely many comparisons. Taking the minimum of their
finite cutoffs gives one strict ordering for every sufficiently
negative shift, and \eqref{app:single:frozen} supplies its VE.

\emph{A designated bottom class, without genericity.}
Fix $b$, run the same algorithm with candidates
$R=\mathcal S\setminus\{b\}$, and leave $b$ in every original
menu. In particular $\{b,i\}$ remains an unassigned competitive
menu until $i$ is chosen, so every candidate's final share is
strictly below one. At a stage when $i$ is chosen and $j$ is a
later candidate, the unassigned binary menu $\{i,j\}$ is assigned
to $i$. This removes strictly positive mass from the denominator
for $j$, while its singleton numerator remains $q_j>0$. Therefore
\[
 \lambda_j^\varphi>\ell_j(R)\ge\ell_i(R)=\lambda_i^\varphi.
\]
The first inequality holds even when the current minimizers tie.
It proves strict increase of final singleton shares along the
entire selected order. Finally assign the only remaining menu,
$\{b\}$, to $b$, whose final share is one. We have constructed
\[
 \lambda_{i_1}^\varphi<\cdots<
 \lambda_{i_{n-1}}^\varphi<\lambda_b^\varphi=1.
\]
For every sufficiently negative $z$, the affine valuations in
\eqref{app:single:frozen} strictly decrease in this priority order.
Thus they implement $\varphi$, and $b$ is uniquely lowest.
For $n=2$ the argument consists simply of the binary menu
$\{b,i_1\}$ and the two singleton menus.
\end{proof}

\begin{lemma}[Refinement and stability of covered strict equilibria]
\label{app:single:covered}
Suppose every singleton is supported. Every menu-wise strict VE
$v^*$ is approached, for every sufficiently large $\beta$, by a
locally unique SVE $v^\beta$. This SVE is locally exponentially
stable under the normalized dynamics, as is the corresponding
joint equilibrium under \eqref{app:single:joint}.
\end{lemma}
\begin{proof}
Choose a closed ball $B$ about $v^*$ in the ambient valuation space,
small enough that each menu retains its unique winner and all
winning margins are at least $\eta>0$. It is unnecessary, and can
be impossible, to choose this ball inside $K$. Uniformly on $B$,
losing logit probabilities are $O(e^{-\beta\eta})$ and every
logit derivative is $O(\beta e^{-\beta\eta})$. The denominators
of $g_s$ are at least $q_s$. Thus
\[
 \sup_{v\in B}\|g(v;\beta)-v^*\|_\infty=O(e^{-\beta\eta}),
 \qquad
 \sup_{v\in B}\|Dg(v;\beta)\|_\infty=O(\beta e^{-\beta\eta}).
\]
Here the first limit is exactly VE consistency for the fixed pure
selector, with its menu weights proportional to $p$. For large
$\beta$, $g$ maps $B$ into itself and is a strict contraction.
Its unique fixed point $v^\beta$ converges to $v^*$. The normalized
Jacobian is $Dg(v^\beta;\beta)-I$, whose eigenvalues approach
$-1$. At a joint equilibrium the derivative of the valuation
equation with respect to $x$ vanishes, and its derivative with
respect to $v$ is the same normalized Jacobian because
$x_s=\Lambda_s(v^\beta)$. The joint Jacobian is block triangular:
\[
 \begin{pmatrix}
 Dg(v^\beta;\beta)-I&0\\ D\Lambda(v^\beta;\beta)&-I
 \end{pmatrix}.
\]
On the frequency-simplex tangent space it adds $n-1$ eigenvalues
$-1$. All eigenvalues have strictly negative real parts for large
$\beta$, proving local exponential stability of both systems.
\end{proof}

\begin{theorem}[Low shifts and multiplicity of strict VE]
\label{app:single:low}
Fix a tree with singleton support and connected co-occurrence
graph. For a generic such tree, a strict VE exists for every
sufficiently negative shift. If every binary menu is supported,
then for every sufficiently negative shift there are at least
$n=|\mathcal S|$ distinct strict VE, one with each class uniquely
lowest, with no genericity restriction. For each fixed such
shift, every strict VE has a unique nearby SVE for every
sufficiently large sensitivity, and all these nearby equilibria
are locally exponentially stable under both dynamics.
\end{theorem}
\begin{proof}
The first two conclusions are Lemma~\ref{app:single:priority}, taking
the minimum of its finitely many cutoffs over designated bottom
classes. The constructed valuations are distinct because their
uniquely lowest classes differ. Lemma~\ref{app:single:covered}
gives refinement and stability. There are only finitely many
strict VE at a fixed shift: each is induced by one of the
finitely many pure menu selectors and has the unique consistent
valuation \eqref{app:single:frozen}. Choose disjoint sufficiently
small neighborhoods around these finitely many distinct vectors,
and take the maximum of their finite sensitivity thresholds.
This supplies one $\widehat\beta(z)$ for all strict VE, as stated
in the main text.
\end{proof}

\begin{proposition}[Every ranking under monotone menu probabilities]
\label{app:single:all-rankings}
Under the menu-probability restrictions of
Proposition~\ref{app:single:total}, every strict total order is
represented by a strict VE whenever
\[
 z<\underline z_1:=-\frac{2M_0+1}{q^2}.
\]
For each fixed such $z$, all sufficiently large sensitivities
admit at least $n!$ distinct SVE locally exponentially stable
under both dynamics.
\end{proposition}
\begin{proof}
Fix an order $s_1\prec\cdots\prec s_n$, with $s_1$ lowest, and
let $\varphi$ choose the highest available class. Let
$\mathcal M_r$ be the menus won by $s_r$ and
$D_r=\sum_{\omega\in\mathcal M_r}p(\omega)$. For $r<n$, the
injection $\omega\mapsto\omega\cup\{s_{r+1}\}$ takes
$\mathcal M_r$ into $\mathcal M_{r+1}$, preserves or raises each
menu's probability, and misses the singleton $\{s_{r+1}\}$.
Thus $D_{r+1}\ge D_r+q$ and
\[
 \frac q{D_r}-\frac q{D_{r+1}}\ge q^2.
\]
The affine consistent valuations obey, uniformly over all orders,
\[
 v_{s_{r+1}}^\varphi(z)-v_{s_r}^\varphi(z)
 \ge q^2|z|-2M_0>1.
\]
They therefore reproduce the proposed strict order and define a
strict VE. Distinct orders give distinct valuation vectors.
Lemma~\ref{app:single:covered} and the maximum of the finitely
many sensitivity thresholds prove the claim about $n!$ stable
SVE.
\end{proof}

\subsection{Robustness to nearly equal singleton probabilities}
\label{app:single:near-equal-section}

The equal-singleton benchmark in
Propositions~\ref{app:single:total} and~\ref{app:single:all-rankings}
can be relaxed while retaining both conclusions. Maintain
$n=|\mathcal S|\ge2$, positive probability for every nonempty menu,
and inclusion monotonicity:
$p(\omega)\le p(\omega')$ whenever $\omega\subseteq\omega'$.
Allow unequal singleton probabilities and define
\begin{equation}
 q_s=p(\{s\}),\quad q_-=\min_s q_s,\quad q_+=\max_s q_s,\quad
 c=q_-^2-(q_+-q_-),\quad
 M_0=\max_{s,\omega\ni s}|\pi_s(\omega)|.
 \label{app:single:near-equal-margin}
\end{equation}
The payoffs entering $M_0$ are the unshifted payoffs; as before,
$\pi_s^z(\omega)=\pi_s(\omega)+z\mathbf1\{\omega=\{s\}\}$.

\begin{proposition}[Robustness of total indifference and every ranking]
\label{app:single:near-equal}
Suppose $c>0$.
\begin{enumerate}[label=(\roman*)]
\item If $z>2M_0/c$, every VE of the shifted tree has equal
valuations across all classes. If also $z>\bar z_0$, where
$\bar z_0$ is the threshold in Theorem~\ref{app:single:high}
for the given tree, all conclusions of that theorem hold and
its selected limiting VE has equal valuations and assigns positive
probability to every available class in every supported menu.
\item If $z<-2M_0/c$, every strict total order on the classes is
represented by a strict VE. For each fixed such $z$, there is
a finite $\widehat\beta(z)$ such that every
$\beta\ge\widehat\beta(z)$ admits at least $n!$ distinct SVE,
one near each of these strict VE, locally exponentially stable
under both the normalized and joint harmonic visit-count dynamics.
\end{enumerate}
Neither conclusion requires genericity.
\end{proposition}
\begin{proof}
\emph{A uniform gap in singleton shares.}
For any menu-choice profile $\sigma$, write
\[
 D_s=\sum_{\omega\ni s}p(\omega)\sigma_\omega(s),\qquad
 b_s=\frac{\sum_{\omega\ni s}p(\omega)\sigma_\omega(s)
                          \pi_s(\omega)}{D_s},\qquad
 \lambda_s=\frac{q_s}{D_s}.
\]
Singleton support gives $D_s\ge q_s>0$, and $|b_s|\le M_0$.
The unique valuations consistent with this profile are
$v_s=b_s+\lambda_s z$.

Suppose that, whenever $i$ is chosen with positive probability
in $\omega$, the menu excludes $j$ and $j$ is chosen with
probability one in $\omega\cup\{j\}$. Let
$\mathcal M_i=\{\omega:\sigma_\omega(i)>0\}$.
The map $\omega\mapsto\omega\cup\{j\}$ is injective on
$\mathcal M_i$, and its image excludes $\{j\}$ because every
image menu contains $i$. Inclusion monotonicity gives
\[
 D_j\ge q_j+\sum_{\omega\in\mathcal M_i}p(\omega\cup\{j\})
       \ge q_j+\sum_{\omega\in\mathcal M_i}p(\omega)
       \ge q_j+D_i.
\]
Since $D_i\le1$,
\[
 q_iD_j-q_jD_i
 \ge q_iq_j+(q_i-q_j)D_i
 \ge q_-^2-(q_+-q_-)=c.
\]
Dividing by $0<D_iD_j\le1$ yields the uniform bound
\begin{equation}
 \lambda_i-\lambda_j
 =\frac{q_iD_j-q_jD_i}{D_iD_j}\ge c>0.
 \label{app:single:near-equal-share-gap}
\end{equation}

\emph{Positive shifts.}
Fix an arbitrary VE $(\sigma,v)$ and suppose $v_i<v_j$.
Every menu in which $i$ is chosen with positive probability
contains only classes with valuations at most $v_i$.
It therefore excludes $j$, which is the unique maximizer after
being added. The preceding argument applies irrespective of
how ties within the original menu are resolved. Consistency and
\eqref{app:single:near-equal-share-gap} imply
\[
 v_i-v_j=(b_i-b_j)+z(\lambda_i-\lambda_j)
          \ge-2M_0+zc>0,
\]
contradicting $v_i<v_j$. Thus every VE has equal valuations.
Full menu support implies the singleton support and connectedness
required by Theorem~\ref{app:single:high}. For
$z>\max\{\bar z_0,2M_0/c\}$, that theorem supplies uniqueness and
global convergence at every finite sensitivity, almost-sure
convergence of harmonic learning, and convergence of the SVE
and choice profiles as $\beta\to\infty$. The limiting VE has equal
valuations by the preceding argument. If an available class $i$ had
zero limiting probability in some menu, the logit-odds argument in
the proof of Proposition~\ref{app:single:total} would give a class
$j$ that wins with probability one when added to any menu in which
$i$ has positive limiting probability. Each such menu excludes $j$.
Thus \eqref{app:single:near-equal-share-gap} applies to the limiting
profile and gives $\lambda_i-\lambda_j\ge c$. Payoff consistency
would then imply $0=v_i^\infty-v_j^\infty\ge -2M_0+zc>0$.
This proves full support of the selected limiting profile.

\emph{Negative shifts.}
Fix any proposed order $s_1\prec\cdots\prec s_n$, with $s_1$
lowest, and let $\varphi$ choose the highest available class
in each menu. Compute $D_s^\varphi$, $b_s^\varphi$, and
$\lambda_s^\varphi$ from this fixed selector, before assigning
valuations. For each adjacent pair $i=s_r$, $j=s_{r+1}$,
every menu won by $i$ excludes $j$ and is won by $j$ after
adjoining it. Hence
$\lambda_{s_r}^\varphi-\lambda_{s_{r+1}}^\varphi\ge c$.
The consistent valuations
$v_s^\varphi(z)=b_s^\varphi+\lambda_s^\varphi z$ satisfy
\[
 v_{s_{r+1}}^\varphi(z)-v_{s_r}^\varphi(z)
 \ge -2M_0+|z|c>0.
\]
They reproduce the proposed order, so $\varphi$ is optimal
in every menu and defines a strict VE. Distinct orders yield
distinct valuation vectors. For each fixed shifted tree,
choose disjoint sufficiently small neighborhoods of these
$n!$ vectors. Lemma~\ref{app:single:covered} gives a locally
unique, locally exponentially stable SVE in each neighborhood
at every sufficiently large sensitivity. Taking the maximum
of the finitely many sensitivity thresholds proves (ii).
\end{proof}

The condition $c>0$ is sufficient, not necessary, and is satisfied
near every equal-singleton benchmark. Explicitly, fix such a
benchmark with common singleton probability $q>0$ and fixed
unshifted payoffs. For any inclusion-monotone probability
distribution with
\[
 \varepsilon:=\max_s|q_s-q|<q^2/8,
\]
we have
\[
 c\ge(q-\varepsilon)^2-2\varepsilon>q^2/2,
\]
where $q\le1$. Thus $z>4M_0/q^2$ suffices uniformly for total
indifference at every VE, and $z<-4M_0/q^2$ suffices uniformly
for every ranking to be a strict VE. The same bounds apply to
payoff perturbations sharing the bound $M_0$. These uniform
statements concern the VE conclusions: the additional
global-convergence threshold $\bar z_0$ and the sensitivity
threshold $\widehat\beta(z)$ are taken for each fixed
perturbed environment.

\section{The three-menu cycling construction}
\label{app:cycle}

\subsection{Primitives and statements}

There are three similarity classes $A,B,C$ and three menus, sampled independently across periods with probabilities $1/3$. Their class-menu expected payoffs are
\[
\begin{array}{c|ccc}
\text{menu}&A&B&C\\\hline
AB&0&1&\text{--}\\
BC&\text{--}&-1&3/2\\
C&\text{--}&\text{--}&-3/4
\end{array}
\label{app:cycle:tab:base}
\]
These primitives are implemented by five objective alternatives: one in class $A$, two in $B$, and two in $C$. Each alternative occurs in its indicated menu and has a stationary reward law with the displayed mean. Throughout this section, rewards satisfy the maintained independence and uniform square-integrability assumptions. Thus menu-dependent class feedback does not require menu-dependent objective reward laws. The class co-occurrence graph is the path $A-B-C$.

Let $\ellb(z)=(1+e^{-\beta z})^{-1}$, where $0\leq\beta<\infty$. In a binary menu, classes are chosen by logit; the singleton is chosen with probability one. With positive prior counts $\kappa_s$ and arbitrary finite initial valuations, the learning rule is
\begin{equation}
v_{n+1,s}=v_{n,s}+\frac{\mathbf1\{S_{n+1}=s\}}{\kappa_s+N_{n,s}+1}
\bigl(r_{n+1}-v_{n,s}\bigr).
\label{app:cycle:eq:update}
\end{equation}
Here $N_{n,s}$ counts visits through period $n$. The leading coefficient in the harmonic gain is exactly one, and the benchmark update uses only the observed reward.

\begin{theorem}[Exact construction]\label{app:cycle:th:base}
For the primitives in \eqref{app:cycle:tab:base}:
\begin{enumerate}[label=(\roman*)]
\item At every finite $\beta\geq0$, there is exactly one SVE. Its valuations are $v^*=(0,0,0)$ and its class frequencies are $x^*=(1/6,1/3,1/2)$.
\item The joint visit-count mean field has an invariant, exponentially attracting affine plane. On this plane the equilibrium is a hyperbolic source for every $\beta>8$. In the full five-dimensional state space it has two unstable eigenvalues and three stable eigenvalues equal to $-1$.
\item For every $\beta>8$, every nonstationary trajectory on the plane has a periodic orbit as its omega-limit set. Every full-system trajectory outside the equilibrium's global stable manifold also has one periodic orbit as its omega-limit set.
\item A supercritical Hopf bifurcation occurs at $\beta=8$. For every $\beta>8$ sufficiently close to $8$, there is a hyperbolic attracting periodic orbit, also attracting in the full joint system.
\item For every fixed $\beta>8$, the stochastic process \eqref{app:cycle:eq:update}, from every finite initial valuation vector and positive prior counts, almost surely has one nonconstant mean-field periodic orbit as its limit set in joint sufficient statistics. Valuations and menu-choice probabilities do not converge. For the attracting Hopf cycle in (iv), the probability that it is the limit set is positive.
\end{enumerate}
\end{theorem}

The periodic orbit in (iii) or (v) may depend on the initial condition or sample path. Neither uniqueness nor hyperbolicity of periodic orbits is asserted for every $\beta>8$. Set convergence to a periodic orbit is different from convergence to a fixed phase, or periodicity in calendar time.

\subsection{Exact sample-mean coordinates and reduction}

Set $K=\sum_s\kappa_s$, $T_n=K+n$, and
\[
x_{n,s}=\frac{\kappa_s+N_{n,s}}{T_n},\qquad
y_{n,s}=x_{n,s}v_{n,s}.
\]
Then $\sum_sx_{n,s}=1$. With $H_{n+1}=(e_{S_{n+1}},r_{n+1}e_{S_{n+1}})$ and $z_n=(x_n,y_n)$, the exact recursion is
\begin{equation}
z_{n+1}=z_n+\frac1{T_n+1}\bigl[H_{n+1}-z_n\bigr].
\label{app:cycle:eq:sa}
\end{equation}
Its mean field is
\begin{equation}
\dot x=\Lambda(v)-x,\qquad \dot y=R(v)-y,\qquad v_s=y_s/x_s.
\label{app:cycle:eq:joint}
\end{equation}
Write $a(v)=\ellb(v_A-v_B)$ and $c(v)=\ellb(v_C-v_B)$. For \eqref{app:cycle:tab:base},
\[
\Lambda(v)=\tfrac13(a,2-a-c,1+c),\qquad
R(v)=\tfrac13(0,c-a,\tfrac32c-\tfrac34).
\]
The three affine discrepancies
\begin{equation}
e_A=y_A,\quad e_B=y_B-x_C+x_A+\tfrac13,\quad
e_C=y_C-\tfrac32x_C+\tfrac34
\label{app:cycle:eq:errors}
\end{equation}
satisfy $\dot e_s=-e_s$. Their common zero set is an invariant attracting plane. With $q=3x_A$ and $r=3x_C-1$, this plane is
\begin{align}
x&=\tfrac13(q,2-q-r,1+r),\nonumber\\
y&=\bigl(0,(r-q)/3,(r-1/2)/2\bigr).\label{app:cycle:eq:plane}
\end{align}
On its relevant domain $D=(0,1)^2$, valuations are
\begin{equation}
v_A=0,\qquad B(q,r):=v_B=\frac{r-q}{2-q-r},\qquad
C(r):=v_C=\frac{3(r-1/2)}{2(1+r)},
\label{app:cycle:eq:values}
\end{equation}
and the planar system is
\begin{equation}
\dot q=\ellb(-B(q,r))-q,\qquad
\dot r=\ellb(C(r)-B(q,r))-r.
\label{app:cycle:eq:planar}
\end{equation}
The variables $q,r$ are filtered menu-choice frequencies, not the current logit probabilities along a nonequilibrium trajectory. At equilibrium the two coincide.

For arbitrary initial conditions of the ODE, each $v_s$ stays in the interval spanned by its initial value and its class's expected payoffs: indeed
\[
\dot v_s=\frac{\Lambda_s(v)}{x_s}\bigl(g_s(v)-v_s\bigr).
\]
Hence valuations stay bounded. At finite $\beta$, binary-menu choice probabilities are then bounded away from zero and one. The linear filter equations for $x,q,r$ imply that $x_s$ is eventually bounded away from zero and $(q,r)$ eventually enters a compact subset of $D$. On that subset the full projected system differs from \eqref{app:cycle:eq:planar} by $O(e^{-t})$.

On $D=(0,1)^2$, $-1\leq B\leq1$ and $-3/4\leq C\leq3/8$. Thus both logit targets are bounded strictly between zero and one uniformly on $D$. A sufficiently small $\delta>0$ makes $[\delta,1-\delta]^2$ strictly inward pointing and absorbing. The denominator that vanishes at $(1,1)$ never affects recurrent dynamics.

\subsection{Uniqueness at every finite sensitivity}

We prove a slightly stronger fact: replace the competitive $C$ payoff by $h$ and its singleton payoff by $-h/2$. The equilibrium remains unique for $0<h\leq2$.

At $\beta=0$, both probabilities are $1/2$. At $\beta>0$, equilibrium requires
\begin{align}
\logit q&=\beta\frac{q-r}{2-q-r},\label{app:cycle:eq:eq1}\\
\logit r-\logit q&=\beta h\frac{r-1/2}{1+r}.\label{app:cycle:eq:eq2}
\end{align}
If either probability is $1/2$, both are. Otherwise monotonicity of logit implies that every noncentral equilibrium has either $q>1/2>r$ or $r>1/2>q$.

Let $L(t)=\log[(1/2+t)/(1/2-t)]$, $0<t<1/2$. This function is positive, increasing and strictly convex, with $L(0)=0$. If $q=1/2+a$, $r=1/2-b$, eliminating $\beta$ gives
\begin{equation}
h=\left(1+\frac{L(b)}{L(a)}\right)\frac{a+b}{b}
\frac{3/2-b}{1-a+b}.
\label{app:cycle:eq:h1}
\end{equation}
For $a\geq b$, the second factor is at least two and the last exceeds one, so $h>2$. For $a\leq b$, convexity gives $L(b)/L(a)\geq b/a$; the first two factors have product at least $(a+b)^2/(ab)\geq4$, while the last exceeds $2/3$. Hence $h>8/3$.

If $q=1/2-a$, $r=1/2+b$, the corresponding identity is
\begin{equation}
h=\left(1+\frac{L(b)}{L(a)}\right)\frac{a+b}{b}
\frac{3/2+b}{1+a-b}.
\label{app:cycle:eq:h2}
\end{equation}
For $a\geq b$,
\[
(a+b)(3/2+b)-3b(1+a-b)
=(a-b)(3/2-2b)+2b^2>0,
\]
so $h>3$. For $a\leq b$, the first two factors again have product at least four and the last exceeds $3/2$, giving $h>6$. All cases contradict $h\leq2$. Therefore $q=r=1/2$ is the unique equilibrium. Every joint equilibrium lies on \eqref{app:cycle:eq:plane}, proving Theorem~\ref{app:cycle:th:base}(i).

\subsection{Instability and the Hopf calculation}

For $h=3/2$, the planar Jacobian at $(1/2,1/2)$ is
\[
J_\beta=\begin{pmatrix}\beta/4-1&-\beta/4\\\beta/4&-1\end{pmatrix},\qquad
\lambda_\pm=-1+\beta/8\ \pm i\sqrt3\,\beta/8.
\]
The three discrepancy coordinates add three eigenvalues $-1$. The equilibrium is a planar source for $\beta>8$, but a saddle in the full five-dimensional system, with a three-dimensional stable manifold.

At $\beta=8$, use
\[
X=q-1/2,\qquad Y=\frac{2(r-1/2)-(q-1/2)}{\sqrt3}.
\]
The linear part is $\dot X=-\sqrt3Y$, $\dot Y=\sqrt3X$. Let $f,g$ denote the nonlinear terms. All derivatives below are evaluated at the origin. Direct differentiation yields
\begin{align*}
(f_{XX},f_{XY},f_{YY})&=(3,-\sqrt3,-3),\\
(g_{XX},g_{XY},g_{YY})&=(5\sqrt3/9,-7/3,-7\sqrt3/3),\\
(f_{XXX},f_{XYY},g_{XXY},g_{YYY})&=(11/2,-63/2,-49/6,-49/2).
\end{align*}
In this coordinate and time normalization the cubic radial coefficient is
\begin{align*}
a_1={}&\frac{f_{XXX}+f_{XYY}+g_{XXY}+g_{YYY}}{16}\\
&+\frac{f_{XY}(f_{XX}+f_{YY})-g_{XY}(g_{XX}+g_{YY})-f_{XX}g_{XX}+f_{YY}g_{YY}}{16\sqrt3}\\
={}&-\frac{11}{3}+\frac{2}{27}=-\frac{97}{27}<0.
\end{align*}
Since the real part crosses zero with derivative $1/8$, Hopf's theorem \citep{Kuznetsov2013} gives a hyperbolic attracting cycle for $\beta>8$ sufficiently close to $8$. In discrepancy coordinates the variational system is block triangular; its three transverse Floquet multipliers are $e^{-T}$, where $T$ is the cycle period. Thus the cycle attracts in the full joint system as well.

\subsection{Planar recurrence and the full mean field}

For $\beta>8$, a small Lyapunov neighborhood of the unique planar equilibrium is strictly repelling. No nonstationary planar trajectory has that equilibrium in its omega-limit set. Compact trapping and Poincar\'e--Bendixson therefore give a periodic omega-limit orbit for every nonstationary planar trajectory.

For the stochastic conclusion, we also need to rule out a continuum of cycles. The following argument supplies this extra step.

Lemma~\ref{app:cycle-normalized:no-annulus} below proves that every periodic orbit of this analytic source field is isolated; its proof includes the compactness and return-map argument needed to exclude a period annulus.

A bounded asymptotic pseudotrajectory has an internally chain transitive limit set. For the present planar field, Lemma~\ref{app:cycle-normalized:no-annulus} leaves only the equilibrium singleton or a single periodic orbit. To see why no additional equilibrium-containing case remains, the source is isolated in its chain component by its strictly repelling neighborhood, and no nonconstant orbit enters it forward in time. There can be no cyclic chain of connecting orbits involving it.

The projection of every full-system trajectory is an asymptotic pseudotrajectory of \eqref{app:cycle:eq:planar}, because its discrepancy forcing is $O(e^{-t})$ on an eventual compact interior domain. Thus its omega-limit is either the equilibrium or one periodic orbit. The former occurs precisely on the global stable manifold of the full hyperbolic equilibrium. This proves (iii) without treating attraction to a plane as, by itself, a Poincar\'e--Bendixson theorem for a five-dimensional system.

\subsection{The stochastic process}

\subsubsection*{Boundedness, persistent sampling and limit sets}
Let $\mathcal F_n$ be the history through period $n$ and let $\mathcal G_{n+1}$ additionally reveal the next menu and implemented alternative $a_{n+1}$, before its reward. With this shift of the indices in \eqref{app:model:conditional-rewards}, the maintained assumptions give
\begin{equation}
\E[r_{n+1}\mid\mathcal G_{n+1}]=\mu_{a_{n+1}},\qquad
\E[r_{n+1}^2\mid\mathcal G_{n+1}]\leq M_2.
\label{app:cycle:eq:reward-moments}
\end{equation}
Proposition~\ref{app:model:boundedness} implies that valuations are almost surely bounded and $\liminf_n\min_sx_{n,s}>0$. The valuation bound may depend on the sample path. At fixed finite $\beta$, it bounds every feasible menu-choice probability away from zero along that path. Applying the martingale strong law to the five menu-choice indicators shows that each has positive lower limiting frequency, while each menu's empirical frequency tends to $1/3$. Consequently $(3x_{n,A},3x_{n,C}-1)$ eventually stays in a compact subset of $D$.

Write $U_{n+1}=H_{n+1}-\E[H_{n+1}\mid\mathcal F_n]$. This is martingale-difference noise with
\[
\E[\|U_{n+1}\|^2\mid\mathcal F_n]\leq 1+M_2.
\]
Recursion \eqref{app:cycle:eq:sa} has harmonic step size $\gamma_{n+1}=1/(K+n+1)$, so $\sum\gamma_n=\infty$ and $\sum\gamma_n^2<\infty$. Proposition~\ref{app:model:apt} therefore applies, without bounded rewards or a moment of order greater than two.

For completeness, write the affine error vector in \eqref{app:cycle:eq:errors} as $e(z)=Lz+d$. The mean-field identity $\dot e=-e$ is equivalent to $L\E[H_{n+1}\mid\mathcal F_n]+d=0$. Hence, with $\zeta_{n+1}=LH_{n+1}+d$,
\begin{equation}
e(z_n)=\frac{Ke(z_0)+\sum_{j=1}^n\zeta_j}{K+n},\qquad
\E[\zeta_{n+1}\mid\mathcal F_n]=0,\qquad
\E[\|\zeta_{n+1}\|^2\mid\mathcal F_n]\leq C.
\label{app:cycle:eq:affine-martingale}
\end{equation}
The martingale strong law gives $e(z_n)\to0$ almost surely. Thus random rewards need not keep sample paths on the plane, but their distance from it vanishes. On the eventual compact interior domain, the projected drift differs from \eqref{app:cycle:eq:planar} by $O(\|e(z_n)\|)$, and the accumulated martingale noise vanishes on bounded learning-time windows. The projected interpolation is consequently an asymptotic pseudotrajectory of the planar field. Its limit set is internally chain transitive and, by the preceding section, is either the equilibrium or one periodic orbit.

\subsubsection*{Avoiding the unstable equilibrium}
We verify the needed unstable-direction noise condition explicitly. Put $\bar H_{n+1}=\E[H_{n+1}\mid\mathcal G_{n+1}]$ and $\bar U_{n+1}=\E[U_{n+1}\mid\mathcal G_{n+1}]$. In coordinates $(x_A,x_B,y_B,y_C)$, the five conditional mean outcome vectors $\bar H_{n+1}$, suppressing their zero $y_A$ coordinate, are
\[
(1,0,0,0),\ (0,1,1,0),\ (0,1,-1,0),\ (0,0,0,3/2),\ (0,0,0,-3/4).
\]
Their affine differences span $\R^4$. In the five-dimensional tangent state space $\sum_sdx_s=0$, this span is the invariant hyperplane $dy_A=0$. The quotient normal to this hyperplane has eigenvalue $-1$, so the unstable generalized eigenspace lies in the hyperplane. Let $P_+$ project onto that two-dimensional space parallel to the stable generalized eigenspace, at the equilibrium. The projected conditional mean noise $P_+\bar U_{n+1}$ spans the unstable space. All five menu-choice events have positive probability at equilibrium; by continuity there are a neighborhood of the equilibrium and $\delta>0$ on which $\E[\|P_+\bar U_{n+1}\|\mid\mathcal F_n]\geq\delta$. Conditional Jensen's inequality gives, on this neighborhood,
\begin{equation}
\E[\|P_+U_{n+1}\|\mid\mathcal F_n]
\geq\E[\|P_+\bar U_{n+1}\|\mid\mathcal F_n]
\geq\delta>0.
\label{app:cycle:eq:unstable-excitation}
\end{equation}
Thus centered reward noise cannot eliminate the unstable-direction excitation already supplied by random menus and choices.

The field is analytic locally, the noise has uniformly bounded conditional second moments, and the step size is harmonic. Together with the first-moment excitation condition \eqref{app:cycle:eq:unstable-excitation}, these properties verify the hypotheses of \citet[Theorem~1]{BD1996}, with noise coefficient equal to the step size and zero remainder. No $2+\delta$ reward moment or positive reward-noise variance is required. Thus
\[
\Pp(z_n\longrightarrow z^*)=0\qquad(\beta>8).
\]
This argument covers every finite initial valuation vector; it does not require $v_{0,A}=0$. Full ambient covariance may be singular, as in the deterministic special case; only excitation in the unstable directions is used.

Combining avoidance with the limit-set classification proves that almost surely the limit set is one nonconstant periodic orbit. The map from the plane to valuations is injective: $C(r)$ is strictly increasing, and at fixed $r<1$,
\[
\partial_q B(q,r)=-\frac{2(1-r)}{(2-q-r)^2}<0.
\]
Hence the valuation image of a cycle is nonconstant. The pair of logit menu-choice probabilities determines $B,C$ at finite positive $\beta$, so those probabilities also fail to converge.

\subsubsection*{Capture by the attracting Hopf cycle}
The Hopf cycle is an attractor of the full joint flow. For any point $w=(x,y)$ on the cycle, mix the five menu-choice reward laws with respective weights
\[
q/3,\quad(1-q)/3,\quad(1-r)/3,\quad r/3,\quad1/3.
\]
The resulting integrable outcome vector $H=(e_S,r e_S)$ has expectation $w$, so $w$ lies in the closed convex hull of the supported outcome vectors. The finite-word construction in the proof of Proposition~\ref{app:model:attainability} therefore makes every such point attainable at arbitrarily late dates. Specifically, approximate $w$ by a finite convex combination of supported outcomes, approximate its weights by empirical frequencies, and require rewards to lie in small bounded neighborhoods of the selected support points. Each prescribed finite menu-choice word and these reward events jointly have positive probability under adaptive logit. Exact averaging in \eqref{app:cycle:eq:sa} then brings the process arbitrarily close to $w$ once prior contributions are negligible. This argument does not require any reward to equal its mean.

The conditional martingale maximal inequality and trapping argument in that proof apply equally to a compact neighborhood of the hyperbolic attracting cycle contained in its basin. At a sufficiently late date, entry into a sufficiently small neighborhood therefore has positive conditional probability of subsequent capture, using only \eqref{app:cycle:eq:reward-moments}. The asymptotic-pseudotrajectory limit set then lies in the cycle and, as a nonempty invariant subset of a periodic orbit, equals the whole cycle. Thus the attracting Hopf cycle has positive capture probability from every finite prior. The natural mean-field clock is $\tau_n=\sum_{j\leq n}\gamma_j\sim\log n$.

\subsection{An open set with uniform eventual uniqueness}\label{app:cycle:sec:robust}

For the same three-menu support, write
\[
 a=p(AB),\qquad b=p(BC),\qquad c=p(C),
\]
and parameterize the five class-menu expected payoffs by
\[
 \pi_A(AB)=u,\quad \pi_B(AB)=v,\quad \pi_B(BC)=w,\quad
 \pi_C(BC)=z,\quad \pi_C(C)=t.
\]
Thus the parameter tuple and its space are
\begin{equation}
 \theta=(a,b,c,u,v,w,z,t)\in\Theta
 :=\{(a,b,c)\in(0,1)^3:a+b+c=1\}\times\R^5,
 \label{app:cycle:eq:parameters}
\end{equation}
with the relative Euclidean topology. Reward laws may otherwise vary
subject to the maintained assumptions. The centered construction has
$\theta_0=(1/3,1/3,1/3,0,1,-1,3/2,-3/4)$.

Let $Q_A=\ellb(v_A-v_B)$ and $Q_C=\ellb(v_C-v_B)$ denote
current binary choice probabilities. The mean-field targets are
\[
 \begin{aligned}
 \Lambda&=(aQ_A,\ a(1-Q_A)+b(1-Q_C),\ c+bQ_C),\\
 R&=(auQ_A,\ av(1-Q_A)+bw(1-Q_C),\ ct+bzQ_C).
 \end{aligned}
\]
Define the filtered-frequency coordinates
\begin{equation}
 q=\frac{x_A}{a},\qquad r=\frac{x_C-c}{b},\qquad
 \dot q=Q_A-q,\qquad \dot r=Q_C-r.
 \label{app:cycle:eq:general-coordinates}
\end{equation}
They need not equal the current choice probabilities away from equilibrium.
The affine discrepancies
\begin{equation}
 \begin{aligned}
 e_A&=y_A-u x_A,\\
 e_B&=y_B-v(a-x_A)-w(1-a-x_C),\\
 e_C&=y_C-tc-z(x_C-c)
 \end{aligned}
 \label{app:cycle:eq:general-errors}
\end{equation}
satisfy $\dot e_s=-e_s$ by direct substitution in
\eqref{app:cycle:eq:joint}. Their common zero set is therefore an
invariant, exponentially attracting affine plane, parameterized by
\begin{equation}
 \begin{aligned}
 x&=\bigl(aq,\ a(1-q)+b(1-r),\ c+br\bigr),\\
 y&=\bigl(auq,\ av(1-q)+bw(1-r),\ ct+bzr\bigr).
 \end{aligned}
 \label{app:cycle:eq:general-plane}
\end{equation}
On this plane, $v_A=u$ and
\begin{equation}
 B(q,r)=\frac{av(1-q)+bw(1-r)}{a(1-q)+b(1-r)},\qquad
 C(r)=\frac{ct+bzr}{c+br}.
 \label{app:cycle:eq:generalBC}
\end{equation}
The dependence of $B,C$ on $\theta$ is suppressed. The reduced dynamics are
\begin{equation}
 \dot q=\ellb(u-B(q,r))-q,\qquad
 \dot r=\ellb(C(r)-B(q,r))-r.
 \label{app:cycle:eq:general-planar}
\end{equation}
At each fixed finite sensitivity, bounded valuations keep $Q_A,Q_C$
bounded away from zero and one. The filters in
\eqref{app:cycle:eq:general-coordinates} therefore place $(q,r)$
eventually in a compact subset of $D=(0,1)^2$, even if it starts outside
$D$. The centered reduction \eqref{app:cycle:eq:errors}--\eqref{app:cycle:eq:planar}
is recovered at $\theta=\theta_0$.

\begin{theorem}[Robust high-sensitivity cycling]\label{app:cycle:th:robust}
Let $\theta_+$ have $a=b=c=1/3$ and $(u,v,w,z,t)=(1/100,1,-1,3/2,-3/4)$. There exist an open neighborhood $\mathcal U$ of $\theta_+$ in $\Theta$ and a finite $\widehat\beta$ such that, for every $\theta\in\mathcal U$ and every $\beta\geq\widehat\beta$:
\begin{enumerate}[label=(\roman*)]
\item there is exactly one SVE, which is a source in the reduced plane and linearly unstable in the full visit-count mean field;
\item every reduced nonstationary trajectory, and every full trajectory outside the equilibrium's stable manifold, has one periodic omega-limit orbit;
\item under the maintained reward assumptions, including uniform square-integrability, and harmonic visit-count learning, the stochastic joint process almost surely has one nonconstant periodic orbit as its limit set, from every finite initial valuation vector and positive prior counts. Neither valuations nor menu-choice probabilities converge.
\end{enumerate}
\end{theorem}

\begin{proof}
\emph{Step 1: strict inequalities.} In a neighborhood of $\theta_+$ we have
\begin{gather}
w<t<u<C(1)<v,\qquad k:=B(0,0)<u,\label{app:cycle:eq:strict}\\
B(0,r)>C(r)\quad\text{for every }r\in[0,1].\label{app:cycle:eq:edge}
\end{gather}
Also $C(r^*)=u$ and $B(q^*,r^*)=u$ have the unique interior solution
\begin{equation}
r^*=\frac{c(u-t)}{b(z-u)},\qquad
q^*=1-\frac{b(u-w)}{a(v-u)}(1-r^*).
\label{app:cycle:eq:mixed}
\end{equation}
At $\theta_+$,
\[
r^*=\frac{76}{149},\qquad q^*=\frac{7378}{14751}.
\]
All inequalities are strict. In particular, at the specified probabilities and $B,C$ payoffs,
\[
B(0,r)-C(r)=\frac{(5/2)r^2-(11/4)r+3/2}{(2-r)(1+r)}>0.
\]
The numerator has negative discriminant and positive leading coefficient. These properties persist on the closure of a sufficiently small neighborhood.

\emph{Step 2: exclude every boundary limit.} Consider any sequence of equilibria with $\beta_n\to\infty$ and primitives in that compact neighborhood, converging to $\bar\theta$. Pass to convergent probability subsequences $(q_n,r_n)\to(q,r)$. Away from $(1,1)$, the functions $B,C$ extend continuously. An interior limiting probability requires its limiting valuation gap to vanish; a zero probability requires a nonpositive gap; a probability one requires a nonnegative gap.

If $q=0$, \eqref{app:cycle:eq:edge} forces $r=0$; then the $q$ equation requires $u\leq k$, contradicting \eqref{app:cycle:eq:strict}. If $r=0$ and $q\in(0,1)$, then $B(q,0)=u$, whereas $B(q,0)\leq k<u$. At $(1,0)$, $B=w<t=C(0)$ contradicts $r=0$. If $q=1$ and $r<1$, then $B=w<C(r)$, forcing $r=1$, a contradiction. If $r=1$ and $q<1$, then $B=v>C(1)$, again a contradiction.

It remains to exclude the singular corner $(1,1)$. Along such a sequence the exact logits imply
\[
\frac{1-r_n}{1-q_n}
=\frac{1+e^{\beta_n(u_n-B_n)}}{1+e^{\beta_n(C_n-B_n)}}\longrightarrow0.
\]
Indeed $q_n\to1$ gives $\beta_n(u_n-B_n)\to\infty$, while $C_n-u_n$ is bounded below by a strictly positive constant. Formula \eqref{app:cycle:eq:generalBC} then gives $B_n\to v$, which contradicts $u<v$ and $q_n\to1$.

Every equilibrium limit is consequently interior and satisfies $B=u=C$. It must equal \eqref{app:cycle:eq:mixed}. By the sequential argument, this localization is uniform over the compact primitive neighborhood.

\emph{Step 3: a fixed-neighborhood implicit-function argument.} Put $\varepsilon=1/\beta$ and
\[
H(q,r,\varepsilon;\theta)=
\begin{pmatrix}u-B(q,r)-\varepsilon\logit q\\
C(r)-B(q,r)-\varepsilon\logit r\end{pmatrix}.
\]
At $(q^*,r^*,0)$,
\[
\det D_{q,r}H=(-B_q)C'>0,
\]
because $v>w$, $z>t$, and $q^*,r^*$ are interior. The parameter-dependent implicit function theorem gives one equilibrium branch in a fixed neighborhood, uniformly after shrinking the primitive neighborhood. Step 2 puts every equilibrium in that neighborhood for all sufficiently large $\beta$, with a common finite threshold. This proves uniform eventual uniqueness.

\emph{Step 4: instability.} At equilibrium the planar Jacobian is $J_\beta=\beta M-I$, where
\[
M=\begin{pmatrix}
q(1-q)(-B_q)&q(1-q)(-B_r)\\
r(1-r)(-B_q)&r(1-r)(C'-B_r)
\end{pmatrix}.
\]
Its determinant is $q(1-q)r(1-r)(-B_q)C'>0$. At the limiting point for $\theta_+$,
\[
\operatorname{tr}M=\frac{3240251}{13410000}>0,\qquad
(\operatorname{tr}M)^2-4\det M
=-\frac{33852981109031}{179828100000000}<0.
\]
Continuity and the uniform localization preserve strict positive real parts for the eigenvalues of $M$. Increasing the common threshold makes $J_\beta$ an unstable focus uniformly. The full system adds three eigenvalues $-1$.

\emph{Step 5: periodic and stochastic limits.} The finite expected payoffs again provide compact interior trapping for the ODE, and Proposition~\ref{app:model:boundedness} supplies almost-sure boundedness and persistent sampling for the stochastic process. The planar field is analytic on $D$. The preceding Poincar\'e--Bendixson, no-period-annulus, and chain-transitivity arguments apply unchanged. The affine discrepancies vanish almost surely by the martingale argument \eqref{app:cycle:eq:affine-martingale}, applied to the general discrepancies \eqref{app:cycle:eq:general-errors} defining the plane \eqref{app:cycle:eq:general-plane}. The five conditional mean outcome vectors have affine span of dimension four because $v\ne w$ and $z\ne t$. Their centered span is the invariant hyperplane $dy_A=u\,dx_A$, whose quotient has eigenvalue $-1$; its unstable projection therefore has rank two. Conditional Jensen's inequality again yields \eqref{app:cycle:eq:unstable-excitation}, so unstable-equilibrium avoidance and the stochastic limit-set argument hold under the maintained second-moment bound. Finally $C'>0$ and $B_q<0$ make the valuation map injective, proving nonconvergence as well.
\end{proof}

\subsubsection*{Why the positive perturbation is necessary for this robustness claim}
Keep all base primitives but set $u=-\eta$ with $\eta>0$ arbitrarily small. The corner $(q,r)=(0,0)$ has $B=0>u$ and $B=0>t=-3/4$. On a sufficiently small rectangle about that corner, both logit targets are exponentially small in $\beta$. For all sufficiently large $\beta$, the target map sends that rectangle to itself and has derivative norm less than one. It therefore has a locally attracting fixed point tending to the corner. The interior equilibrium branch near $(1/2,1/2)$ also persists for small $\eta$ by the implicit-function argument, and remains unstable at sufficiently large $\beta$. Thus arbitrary neighborhoods of the centered example contain models with an additional stable equilibrium. One cannot infer the open-set theorem by perturbing the centered example without this qualification.

\subsection{Corrective feedback fails, and extra classes do not remove the example}
For the functions in \eqref{app:cycle:eq:generalBC}, $B_q<0$ and
$C'>0$ throughout the cycling neighborhood. Using the current choice
probabilities $Q_A=\ell_\beta(v_A-v_B)$ and
$Q_C=\ell_\beta(v_C-v_B)$, the payoff targets are
$g_B(\mathbf v)=B(Q_A,Q_C)$ and $g_C(\mathbf v)=C(Q_C)$.
Thus, for $\beta>0$,
\[
 \partial_{v_A}g_B=B_q(Q_A,Q_C)\,\beta Q_A(1-Q_A)<0,\qquad
 \partial_{v_B}g_C=-C'(Q_C)\,\beta Q_C(1-Q_C)<0.
\]
Thus the nonnegative cross-derivative requirement equivalent to negative
payoff--crowding covariance fails throughout the cycling neighborhood.
This makes precise the claim that competition removes favorable
observations in the two indicated interactions.

\begin{corollary}[Any number of classes at least three]
\label{app:cycle:moreclasses}
For each $n\geq3$ there is a fixed menu support on $n$ classes and a
nonempty relatively open family of deterministic decision trees on that
support for which the conclusions of the persistent-cycling theorem
hold at all sufficiently large finite sensitivities.
\end{corollary}
\begin{proof}
Add $n-3$ classes, each occurring only in its own singleton menu, with
arbitrary deterministic payoffs. Assign positive probabilities to these
menus and total probability $h>0$ to the three original menus; conditional
on an original menu, choose its probabilities and payoffs in the open
family of Theorem~\ref{app:cycle:th:robust}. This describes an open family
relative to the enlarged fixed support.

Normalized dynamics for the original classes are unchanged and the
additional valuations solve independent stable linear equations. For the
joint dynamics let $z=\sum_{s\in\{A,B,C\}}x_s$. Then
$\dot z=h-z$. In the normalized core sufficient statistics
$\widetilde x_s=x_s/z$, $\widetilde y_s=y_s/z$, direct differentiation gives
\[
 \dot{\widetilde x}_s=\frac hz
       \bigl(\Lambda_s^{\rm core}(\widetilde y/\widetilde x)
                      -\widetilde x_s\bigr),\qquad
 \dot{\widetilde y}_s=\frac hz
       \bigl(R_s^{\rm core}(\widetilde y/\widetilde x)
                      -\widetilde y_s\bigr).
\]
This is exactly the original joint system with a positive time change
whose derivative tends to one. All extra frequencies converge to their
menu probabilities and all extra valuations to their payoffs. Hence
uniqueness, instability, and the periodic limit-set conclusion lift from
the three-class subsystem; the additional directions are stable.

For the stochastic process, restrict attention to the successive dates
at which an original menu occurs. Their conditional menus are i.i.d. with
the original distribution, and the class visit-count updates are exactly
the original three-class recursion. The number of these dates divided by
calendar time tends almost surely to $h$. The core periodic limit set
therefore lifts by multiplying its frequencies and sufficient rewards by
$h$; the extra coordinates converge by the ordinary strong law. This
proves the stochastic assertion as well.
\end{proof}

\subsection{An exact three-class stability criterion}
\label{app:cycle:three-class-stability-section}

The following result applies to any decision tree with three classes.
Theorem~\ref{app:eq:structure}(iii) shows that its hyperbolicity
assumption holds for generic $(p,\pi)$ at all but finitely many
sensitivities.

\begin{proposition}[Three-class stability and the equilibrium index]
\label{app:cycle:three-class-stability}
Fix a finite sensitivity $\beta\geq0$ and suppose $|\mathcal S|=3$.
Assume every SVE is hyperbolic. For each SVE $v^*$, define
\begin{equation}
 T(v^*):=-2 + \beta\sum_{s\ne r}
 \operatorname{Cov}_{w_s^{v^*}}
       \bigl(\pi_s(\omega),\sigma_\omega^r(v^*)\bigr),
 \qquad d(v^*):=\det\!\bigl(I_3-Dg(v^*)\bigr).
 \label{app:cycle:three-class-statistics}
\end{equation}
Under either the normalized or joint asynchronous dynamics, an SVE is
locally asymptotically stable if and only if $d(v^*)>0$ and $T(v^*)<0$.
If $d(v^*)<0$, it has exactly one unstable direction; if $d(v^*)>0$
and $T(v^*)>0$, it has exactly two unstable directions.
The equilibrium set $\mathcal V(\beta)$ is finite and satisfies
\begin{equation}
 \sum_{v^*\in\mathcal V(\beta)}\operatorname{sgn}d(v^*)=1.
 \label{app:cycle:three-class-degree}
\end{equation}
Consequently, there are no asymptotically stable SVE if and only if
every positive-determinant SVE has $T(v^*)>0$. In that case, at least
one SVE has two unstable relative-valuation directions.
\end{proposition}

\begin{proof}
Fix an SVE $v^*$ and write $L=Dg(v^*)$ and $J=L-I_3$.
Lemma~\ref{app:single:derivative} gives
\[
 L\mathbf 1=0,\qquad
 L_{ss}=\beta\sum_{r\ne s}
 \operatorname{Cov}_{w_s^{v^*}}
       \bigl(\pi_s(\omega),\sigma_\omega^r(v^*)\bigr),
 \qquad \operatorname{tr}L=T(v^*)+2.
\]
The common-shift direction therefore has eigenvalue $-1$ under $J$.
The induced linear map $\bar J$ on the two-dimensional quotient
$\mathbb R^3/\operatorname{span}\{\mathbf 1\}$ has trace
$T(v^*)$. Moreover,
$\det\bar J=-\det J=\det(-J)=d(v^*)$. Hence
\begin{equation}
 \det(\lambda I_3-J)
 =(\lambda+1)
       \bigl[\lambda^2-T(v^*)\lambda+d(v^*)\bigr].
 \label{app:cycle:three-class-polynomial}
\end{equation}
If $d(v^*)<0$, the two quotient eigenvalues are real and have opposite
signs. If $d(v^*)>0$, their real parts are both negative when
$T(v^*)<0$ and both positive when $T(v^*)>0$. Hyperbolicity excludes
$d(v^*)=0$ and the case $d(v^*)>0$, $T(v^*)=0$.
The linearization theorem proves the normalized classification,
including local exponential stability in the stable case.
Proposition~\ref{app:model:spectrum} adds two eigenvalues equal to $-1$
for the joint dynamics, so their classification is identical.

For the global index calculation, write
$m_s=\min_{\omega\in\Omega_s^+}\pi_s(\omega)$ and
$M_s=\max_{\omega\in\Omega_s^+}\pi_s(\omega)$.
Every SVE lies in the compact payoff box $K=\prod_s[m_s,M_s]$,
and the equilibrium set is closed because $g$ is continuous.
Hyperbolicity makes $I_3-Dg(v^*)$ nonsingular, so the inverse function
theorem makes every SVE isolated. Compactness then makes the
equilibrium set finite. To compute its total index, take the open box
\[
 Q=\prod_s(m_s-1,M_s+1)
\]
and any $c\in K$. Since $g(v)\in K$ for every $v\in\mathbb R^3$,
the homotopy
\[
 H_t(v)=v-\bigl[(1-t)c+t g(v)\bigr],\qquad 0\leq t\leq1,
\]
has no zero on $\partial Q$: its bracketed term belongs to
$K\subset Q$. Homotopy invariance of Brouwer degree gives
$\deg(H_1,Q,0)=\deg(H_0,Q,0)=1$.
Every zero is regular, so the sum of its local degrees is exactly
\eqref{app:cycle:three-class-degree}. This argument also covers classes
with constant payoffs, for which a coordinate interval of $K$ is a
singleton.

The degree identity guarantees a positive-determinant equilibrium.
The local classification shows that none is asymptotically stable
exactly when every positive-determinant equilibrium has $T(v^*)>0$.
At least one equilibrium must then have two unstable directions in
the quotient by common shifts. These are relative-valuation directions;
the common-shift direction and the additional joint frequency
directions remain stable.
\end{proof}

\subsection{The normalized cycling dynamics and the planar limit-set argument}
\label{app:cycle-normalized}

This subsection supplies the normalized-dynamics part of the persistent
cycling theorem and states the additional planar fact needed when a
projection of the joint process is only an asymptotic pseudo-trajectory.
The latter distinction matters: ordinary Poincar\'e--Bendixson alone
does not classify arbitrary asymptotic pseudo-trajectories.

\subsubsection{An exact planar quotient for normalized learning}

Use the three-menu support $\{AB,BC,C\}$. Write the positive menu
probabilities as $a,b,c$, with $a+b+c=1$, and the five payoffs as
\[
 \pi_A(AB)=a_0,\quad \pi_B(AB)=b_+,\quad
 \pi_B(BC)=b_-,\quad \pi_C(BC)=c_+,\quad \pi_C(C)=c_0.
\]
Thus $(a_0,b_+,b_-,c_+,c_0)=(u,v,w,z,t)$ in
\eqref{app:cycle:eq:parameters}, and the functions below are precisely
those in \eqref{app:cycle:eq:generalBC}. These alternative payoff labels
are used only within this section. Set $\ell_\beta(z)=(1+e^{-\beta z})^{-1}$ and
\[
 B(q,r)=\frac{ab_+(1-q)+bb_-(1-r)}{a(1-q)+b(1-r)},
 \qquad C(r)=\frac{cc_0+bc_+r}{c+br}.
\]
The strict inequalities $b_+>b_-$ and $c_+>c_0$ hold throughout the
open parameter neighborhood of the cycling construction. In particular,
$B_q<0$ and $C'>0$ on $(0,1)^2$.

Let $\xi=v_B-v_A$ and $\eta=v_C-v_A$. The current binary choice
probabilities are
\[
 q=\ell_\beta(-\xi),\qquad r=\ell_\beta(\eta-\xi).
\]
Since $g_A\equiv a_0$, the normalized system is exactly equivalent to
\begin{align}
 \dot\xi&=B\bigl(\ell_\beta(-\xi),
                 \ell_\beta(\eta-\xi)\bigr)-a_0-\xi,\notag\\
 \dot\eta&=C\bigl(\ell_\beta(\eta-\xi)\bigr)-a_0-\eta,\label{app:cycle-normalized:quotient}\\
 \dot v_A&=a_0-v_A.\notag
\end{align}
The first two equations are an autonomous analytic planar system on
$\mathbb R^2$. They are not the filtered-frequency equations governing
the joint invariant plane.

\begin{proposition}[Periodic limits under normalized learning]
\label{app:cycle-normalized:theorem}
Suppose the SVE is unique and is a hyperbolic source on the invariant
plane of the joint harmonic-visit dynamics. Then it is also a hyperbolic
source in the normalized quotient \eqref{app:cycle-normalized:quotient}.
Every normalized trajectory outside the equilibrium's stable manifold
has a non-constant periodic orbit as its omega-limit set. Its global
stable manifold is the affine line
\[
 \{v^*+d(1,1,1):d\in\mathbb R\}.
\]
Thus these conclusions apply to the same open neighborhood and the same
sufficiently large finite sensitivities as the joint cycling theorem.
They also apply to the centered example $a_0=0$ for every $\beta>8$.
\end{proposition}
\begin{proof}
Equilibria of the quotient are precisely SVE valuations expressed in
relative coordinates: an equilibrium $(\xi,\eta)$ lifts uniquely to
$v=(a_0,a_0+\xi,a_0+\eta)$. Thus uniqueness transfers.

The joint Jacobian consists of its two planar eigenvalues and three
eigenvalues $-1$. Proposition~\ref{app:model:spectrum} implies that the
normalized Jacobian has these same two planar eigenvalues and one
eigenvalue $-1$. In \eqref{app:cycle-normalized:quotient}, the last
eigenvalue belongs to the $v_A$ equation, so the quotient equilibrium is
a hyperbolic source. One can also verify this directly. Put
$K(q,r)=(B(q,r)-a_0,C(r)-a_0)$ and
$L=\left(\begin{smallmatrix}-1&0\\-1&1\end{smallmatrix}\right)$.
At an equilibrium, the planar joint and quotient Jacobians are
\[
 J_{\rm joint}=\beta D_\ell LDK-I_2,
 \qquad J_{\rm norm}=\beta DKD_\ell L-I_2,
 \quad D_\ell=\operatorname{diag}(q(1-q),r(1-r)).
\]
Because $\det DK=B_qC'\ne0$, these matrices are similar. In the centered
example their eigenvalues are
$-1+\beta/8\pm i\sqrt3\beta/8$.

The target $B(q,r)$ lies between $b_-$ and $b_+$, and $C(r)$ lies in the
finite interval with endpoints $C(0)$ and $C(1)$. Consequently a
sufficiently large rectangle in $(\xi,\eta)$ has strictly inward vector
field and contains the forward orbit of any prescribed initial point.
For example, enlarge each of the corresponding target intervals,
shifted by $-a_0$, by one unit at either endpoint; that rectangle is
absorbing. The quotient trajectory is therefore precompact.

A hyperbolic source admits a sufficiently small neighborhood with
strictly outward vector field. More explicitly, if $J$ is its
linearization, choose a positive definite $P$ satisfying
$J^{\mathsf T}P+PJ=I_2$; this is possible because $-J$ is Hurwitz.
The function $V(z)=(z-z^*)^{\mathsf T}P(z-z^*)$ has positive derivative
away from $z^*$ in a small neighborhood. A nonstationary trajectory
starting there exits a fixed small sublevel neighborhood in finite time
and cannot re-enter it. Thus its omega-limit set does not contain the
unique equilibrium. The Poincar\'e--Bendixson theorem now implies that
the omega-limit set is a single periodic orbit.

Finally $v_A(t)=a_0+(v_A(0)-a_0)e^{-t}$. Hence a quotient periodic orbit
lifts, at the limit, to the periodic orbit
$(a_0,a_0+\xi(t),a_0+\eta(t))$ of the normalized system. A quotient
trajectory converges to its source only if it is identically that
equilibrium. The corresponding full normalized trajectories are
exactly $v^*+d e^{-t}(1,1,1)$, proving the stable-manifold assertion.
\end{proof}

\subsubsection{Why a stochastic planar limit is one orbit, not an annulus}

\begin{lemma}[Analytic source fields have no period annulus]
\label{app:cycle-normalized:no-annulus}
Let $F$ be an analytic planar vector field on a simply connected open
domain $D$, with exactly one equilibrium $e$, which is a hyperbolic
source. Suppose the closed Jordan disk bounded by any periodic orbit
under consideration is a compact subset of $D$. Then no such periodic
orbit belongs to an open annulus foliated by periodic orbits. Every such
periodic orbit is isolated. A compact internally chain transitive set
in $D$ is either $\{e\}$ or a single periodic orbit.
\end{lemma}
\begin{proof}
Every periodic orbit encloses an equilibrium by the planar index theorem,
and therefore encloses $e$. Distinct periodic orbits are disjoint and
nested. Choose a closed disk-like source neighborhood $U$ with strictly
outward flow. A periodic orbit cannot intersect $U$: one starting in
$U\setminus\{e\}$ exits in finite time and cannot return. Hence $U$
lies in the interior of every periodic orbit.

Suppose a period annulus exists, fix a cycle in it, and continue the
annulus maximally inward. Choose nested cycles $\Gamma_n$ cofinal at its
inner end. Their closed Jordan disks $D_n$ decrease, contain $U$, and
are contained in one compact Jordan disk in $D$. Put
$D_\infty=\bigcap_nD_n$ and $K=\partial D_\infty$.
The cycles converge in Hausdorff distance to $K$. Indeed, an accumulation
point of the cycles lies in every $D_n$, but cannot lie in the interior
of $D_\infty$; conversely a boundary point of $D_\infty$ at distance
bounded away from some tail of the cycles would have a neighborhood
inside every $D_n$, a contradiction. Compactness makes these assertions
uniform. Thus $K$ is compact, connected, invariant, and has empty
interior. It is disjoint from the interior of $U$. It cannot meet
$\partial U$ either, because its backward orbit would immediately
enter that interior, contradicting invariance. Thus it is disjoint
from all equilibria.

The Hausdorff limit of these periodic orbits is internally chain
transitive. To verify this directly, fix an error tolerance and a lower
time bound $T$. Join nearby approximations of any two points of $K$
along a sufficiently close $\Gamma_n$. Subdivide its orbit arc into
segments with lengths in $[T,2T]$, traversing additional revolutions if
necessary. Approximate the segment endpoints by points of $K$.
Uniform continuity of the flow over $[T,2T]$ and Hausdorff convergence
make every resulting jump smaller than the chosen tolerance.

For any compact equilibrium-free internally chain transitive set
$L\subset D$, choose a compact neighborhood $W\subset D$ of $L$
without equilibria and a smooth cutoff $\chi\in C_c^\infty(D)$ equal
to one near $W$. Extend $\chi F$ by zero to $\mathbb R^2$. Its flow
is complete and leaves invariant a closed disk $B$ containing its
support in its interior. Restriction to the compact, locally path
connected space $B$ preserves the internal dynamics on $L$.
By \citet[Theorem~5.7(ii)]{Benaim1999}, there is an asymptotic
pseudotrajectory $X$ in $B$ with limit set $L$. Thus
$\operatorname{dist}(X(t),L)\to0$, and $X$ eventually remains in $W$.
Corollary~6.14 of the same reference makes $L$ a periodic orbit or a
cylinder of periodic orbits, also for the original field.
Applying this conclusion to $K$, whose interior is empty, shows that
$K$ is a periodic orbit.
Take a local transverse section to $K$. If its analytic return-map
displacement is not identically zero, its sign is fixed on each
sufficiently small punctured side of the fixed point. Points on each
side then converge to $K$ in forward or backward time. Thus a
sufficiently small tubular neighborhood of $K$ cannot meet any other
periodic orbit. This contradicts the accumulating cycles $\Gamma_n$.
The analytic identity theorem therefore makes the return map equal to
the identity on a full interval about its fixed point.
Thus periodic orbits fill a neighborhood on both sides of $K$,
extending the original annulus inward, a contradiction. Similarly,
nonisolation of any periodic orbit would make its analytic return map
identically the identity and produce a period annulus. This proves
isolation.

For the last assertion, let $L\subset D$ be compact and internally
chain transitive. If $e\in L$, time reversal preserves internal
chain transitivity and makes $e$ an attractor for the flow on $L$.
An internally chain transitive set admits no proper attractor
\citep[Proposition~5.3]{Benaim1999}, so $L=\{e\}$.
Otherwise the localized classification above makes $L$ a periodic
orbit or a cylinder of periodic orbits. Isolation of every periodic
orbit excludes the latter possibility.
\end{proof}

For the normalized quotient $D=\mathbb R^2$. For the joint invariant
plane in filtered-frequency coordinates, $D=(0,1)^2$; its convexity
ensures that every periodic orbit has its closed Jordan disk compactly
inside $D$. The compact interior trapping already established for that
plane supplies the precompactness required above. Hence any of the
projected asymptotic pseudo-trajectories used in the joint and stochastic
proofs has limit set either the equilibrium or a single periodic orbit.
Unstable-equilibrium avoidance is a separate stochastic step and is not
implied by this lemma. 

\subsection{Additional instability examples with a cyclic comparison graph}
\label{app:other-instability}

The following examples justify the auxiliary instability claim in the
main-text footnote. They establish linearly unstable SVE under both
mean fields; they do not require a conclusion about periodic orbits or
the absence of other equilibria.

There are three classes $A,B,C$. In both examples the binary menus have
the following deterministic class payoffs:
\begin{equation}
\begin{array}{c|rrr}
 \omega&A&B&C\\ \hline
 \{A,B\}&-1&1&\text{--}\\
 \{B,C\}&\text{--}&-1&1\\
 \{A,C\}&1&\text{--}&-1
\end{array}
\label{app:other:binary-table}
\end{equation}
Thus each class has a payoff advantage over one class and a disadvantage
against the other. Each supported class-menu payoff can be implemented
by a distinct objective alternative in that class, with a deterministic
reward law. The examples are consequently stationary decision problems
of the maintained kind.

\begin{proposition}[Intermediate singleton payoffs]
\label{app:other:singletons}
Support the three binary menus in \eqref{app:other:binary-table} and
the three singleton menus, each with probability $1/6$. Give every
singleton the payoff $z\in(-1,0)$. For every finite $\beta$, there is
an SVE
\[
 v^*=(z/2,z/2,z/2),\qquad x^*=(1/3,1/3,1/3).
\]
Its normalized eigenvalues are
\begin{equation}
 -1,\qquad
 -1-\frac{3z\beta}{16}\ \mathbin{\pm}\
 i\frac{\sqrt3\beta}{8}.
 \label{app:other:singleton-eigenvalues}
\end{equation}
The joint harmonic visit-count system has these eigenvalues together
with two additional eigenvalues $-1$. The equilibrium is therefore
linearly unstable for $\beta>-16/(3z)$.
\end{proposition}

\begin{proposition}[A ternary menu without singleton menus]
\label{app:other:ternary}
Support the three binary menus in \eqref{app:other:binary-table} and
the ternary menu $\{A,B,C\}$, each with probability $1/4$. In the
ternary menu give all three classes the same payoff $h>0$. There are
no singleton menus. For every finite $\beta$, there is an SVE
\[
 v^*=(h/4,h/4,h/4),\qquad x^*=(1/3,1/3,1/3).
\]
Its normalized eigenvalues are
\begin{equation}
 -1,\qquad
 -1+\frac{3h\beta}{64}\ \mathbin{\pm}\
 i\frac{3\sqrt3\beta}{16}.
 \label{app:other:ternary-eigenvalues}
\end{equation}
The joint system adds two eigenvalues $-1$, so this equilibrium is
linearly unstable for $\beta>64/(3h)$. For example, $h=1/2$ gives
the threshold $128/3$ and uses a ternary payoff strictly between
the binary payoffs $-1$ and $1$.
\end{proposition}

\begin{proof}[Proof of Propositions~\ref{app:other:singletons}
and~\ref{app:other:ternary}]
A common calculation covers both cases. Give each binary menu mass
$a$, each singleton mass $b$, and the ternary menu mass $c$, omitting
a menu type when its mass is zero. Thus $3a+3b+c=1$. At equal
valuations each binary choice probability is $1/2$ and each ternary
choice probability is $1/3$. For every class,
\[
 \Lambda_s=a+b+c/3=1/3,
 \qquad R_s=bz+ch/3,
 \qquad V:=g_s=3bz+ch.
\]
Hence $(V,V,V)$ is an SVE. Use the quotient-rule derivative identity
to define
\[
 \kappa=3\beta\left(\frac{aV}{4}
                         +\frac{c(V-h)}9\right),
 \qquad \lambda=\frac{3a\beta}{4}.
\]
In the coordinate order $(A,B,C)$ the Jacobian of $g$ is
\[
 Dg=
 \begin{pmatrix}
 -2\kappa&\kappa+\lambda&\kappa-\lambda\\
 \kappa-\lambda&-2\kappa&\kappa+\lambda\\
 \kappa+\lambda&\kappa-\lambda&-2\kappa
 \end{pmatrix}.
\]
Indeed each shared binary menu contributes
$3\beta(a/4)(V-\pi_s)$ to its cross derivative and the ternary
menu contributes $3\beta(c/9)(V-h)$. Translation invariance gives
the displayed diagonal entries. The eigenvalues are $0$ and
$-3\kappa\pm i\sqrt3\lambda$, so the normalized Jacobian
$Dg-I$ has eigenvalues $-1$ and
$-1-3\kappa\pm i\sqrt3\lambda$.

For $(a,b,c)=(1/6,1/6,0)$, $V=z/2$,
$\kappa=z\beta/16$, and $\lambda=\beta/8$. For
$(a,b,c)=(1/4,0,1/4)$, $V=h/4$,
$\kappa=-h\beta/64$, and $\lambda=3\beta/16$.
These give the two claimed spectra. Finally the harmonic joint
Jacobian at any equilibrium has blocks
\[
 \begin{pmatrix}Dg-I&0\\D\Lambda&-I\end{pmatrix}.
\]
Restriction to the frequency-simplex tangent space adds two copies
of $-1$, proving both joint conclusions.
\end{proof}

\begin{corollary}[Robustness on each fixed menu support]
\label{app:other:open}
Around either of the preceding specifications there is a nonempty
open set of reduced payoff and menu-probability parameters, with
that example's menu support fixed, and a common finite $\beta_0$
such that every tree in the open set has a linearly unstable SVE
for every $\beta\ge\beta_0$. The assertion is existence of an
unstable branch; it does not assert global uniqueness of SVE.
\end{corollary}
\begin{proof}
Denote the parameters by $\theta$ and their symmetric value by
$\theta_0$. Use scaled relative scores
$w=(\beta(v_A-v_C),\beta(v_B-v_C))$ and let $G_\theta(w)$ be the
conditional payoff vector when the logit score vector is
$(w_1,w_2,0)$. It is smooth jointly in $(w,\theta)$ near the
specified parameters, since all denominators are positive. Put
\[
 \Psi_\theta(w)=
 (G_{\theta,A}(w)-G_{\theta,C}(w),
  G_{\theta,B}(w)-G_{\theta,C}(w)).
\]
With $\varepsilon=1/\beta$, the relative equilibrium equation is
$\Psi_\theta(w)-\varepsilon w=0$. At $(\theta,w,\varepsilon)
=(\theta_0,0,0)$ it holds, and the two eigenvalues of
$D\Psi_{\theta_0}(0)$ have strictly positive real parts:
$-3z/16\pm i\sqrt3/8$ in the singleton example and
$3h/64\pm i3\sqrt3/16$ in the ternary example. In particular this
derivative is invertible. The implicit function theorem gives a
smooth local solution $w=W(\theta,\varepsilon)$ for $\theta$ in
one open neighborhood and $\varepsilon$ near zero. Define
\[
 v_C=G_{\theta,C}(W),\qquad
 v_A=v_C+\varepsilon W_1,\qquad
 v_B=v_C+\varepsilon W_2.
\]
The relative equilibrium equation makes this a genuine SVE.
By continuity, after shrinking the parameter neighborhood and
choosing $\varepsilon_0>0$, the eigenvalues of
$D\Psi_\theta(W(\theta,\varepsilon))-\varepsilon I$ retain
strictly positive real parts throughout $0\le\varepsilon\le
\varepsilon_0$. The normalized relative Jacobian is $\beta$
times this matrix. Its remaining valuation eigenvalue is $-1$;
the joint system adds two more copies of $-1$. Taking
$\beta_0=1/\varepsilon_0$ proves the uniform assertion.
\end{proof}

\section{Coarse representation and the welfare cost of learning}
\label{app:welfare}

This section establishes the claims in the main text's welfare discussion.
The within-class implementation rule remains fixed. Its expected payoff
is incorporated in $\pi_s(\omega)$, so changing that rule is outside the
policy comparison below. Sums and extrema are over the supported
set $\Omega^+$.

\subsection{The informed common-score benchmark}

For $\theta\in\mathbb R^{\mathcal S}$, write
\[
 q_s(\omega;\theta)=\frac{e^{\theta_s}}{\sum_{j\in\omega}e^{\theta_j}},\qquad
 W(q)=\sum_\omega p(\omega)\sum_{s\in\omega}
                  q_s(\omega)\pi_s(\omega).
\]
The unrestricted score index $\theta$ absorbs the product $\beta v$
at positive sensitivity.
Let $\mathcal L$ be the family of these menu-wise logit policies and let
$\overline{\mathcal L}$ be its closure in the finite product of menu
simplexes. For a strict class ranking $\succ$, let $q^\succ$ choose
$\max_\succ\omega$ from each menu and put $W_\succ=W(q^\succ)$.

\begin{proposition}[The informed coarse benchmark]
\label{app:welfare:oracle}
Allowing deterministic limits of logit choice, an optimal oracle policy
is attained by a strict class ranking, and its expected payoff is
\begin{equation}
 W^{\mathrm{coarse}}
 :=\max_{q\in\overline{\mathcal L}}W(q)
 =\max_\succ W_\succ
 =\max_\succ\sum_\omega p(\omega)
                      \pi_{\max_\succ\omega}(\omega).
 \label{app:welfare:benchmark}
\end{equation}
Every policy in $\overline{\mathcal L}$ is a mixture of strict-ranking
policies, so allowing further randomization cannot raise this value.
For generic reduced parameters $(p,\pi)$ on the fixed menu support,
the optimal policy on supported menus is unique and pure.

Suppose each supported menu has a unique payoff-maximizing class
$s^*(\omega)$. Define the efficient-comparison graph
$G^{\mathrm{eff}}$ on $\mathcal S$ by drawing an arrow $s\to r$ whenever
$s=s^*(\omega)$ and $r\in\omega\setminus\{s\}$ for some supported menu.
Then
\begin{equation}
 W^{\mathrm{coarse}}=W^{\mathrm{menu}}
 :=\sum_\omega p(\omega)\max_{s\in\omega}\pi_s(\omega)
 \quad\Longleftrightarrow\quad
 G^{\mathrm{eff}}\text{ is acyclic}.
 \label{app:welfare:acyclicity}
\end{equation}
\end{proposition}
\begin{proof}
Fix $\theta$. Independently for every class draw an exponential random variable
$T_s$ of rate $e^{\theta_s}$, and rank all classes by increasing $T_s$. Ties
have probability zero. For every menu,
\[
 \Pr\{T_s=\min_{j\in\omega}T_j\}
 =\int_0^\infty e^{\theta_s}
      \exp\left(-t\sum_{j\in\omega}e^{\theta_j}\right)\,dt
 =q_s(\omega;\theta).
\]
Thus a single random ranking, drawn independently of the menu, implements
the entire logit policy. If $\mu_\theta(\succ)$ is its distribution, then
$q(\theta)=\sum_\succ\mu_\theta(\succ)q^\succ$. There are finitely many rankings,
so their convex hull is closed. It therefore contains
$\overline{\mathcal L}$, and linearity gives
$W(q)\leq\max_\succ W_\succ$ throughout this closure. Conversely, fix
distinct scores ordered by $\succ$ and multiply them by $t\to\infty$.
Their logit policies tend to $q^\succ$, so every ranking policy belongs
to $\overline{\mathcal L}$. This proves equality and attainment.
Linearity also proves the assertion about further randomization.

For generic uniqueness, identify rankings that induce the same
selector on $\Omega^+$. Any two distinct induced selectors $f,h$
differ on some supported menu $\omega_0$. Their expected welfare
difference is
\[
 \sum_{\omega\in\Omega^+}p(\omega)
 \bigl[\pi_{f(\omega)}(\omega)-\pi_{h(\omega)}(\omega)\bigr].
\]
For each fixed $p$, the coefficient of $\pi_{f(\omega_0)}(\omega_0)$
is $p(\omega_0)>0$, so equality defines a proper hyperplane in the
payoff space. In the joint parameter space each tie set is relatively
closed by continuity, has empty interior by varying that payoff alone,
and has measure zero by Fubini's theorem. There are only finitely
many pairs. On the open dense, full-measure complement, one induced
selector uniquely maximizes welfare. An optimal mixture can then put
weight only on rankings that induce this selector, so every optimal
policy makes the same pure choice in each supported menu.

Every policy earns at most $W^{\mathrm{menu}}$. If $G^{\mathrm{eff}}$
is acyclic, a topological ordering gives a strict ranking with
$s\succ r$ for every arrow $s\to r$. It selects $s^*(\omega)$ in every
supported menu and therefore attains $W^{\mathrm{menu}}$.
Conversely, suppose $W^{\mathrm{coarse}}=W^{\mathrm{menu}}$ and choose
an optimal strict ranking, whose existence was just proved. Its
menu-wise payoff shortfalls are nonnegative. Since every supported
menu has positive probability and a unique maximizer, equality of
total welfare forces the ranking to select $s^*(\omega)$ in every
menu. It therefore respects every arrow. A directed cycle would
require a strict ranking to rank a class above itself, a contradiction.
\end{proof}

For precision, the benchmark preserves both the prescribed within-class
implementation and the use of a common score across menus. Allowing
arbitrary class choices separately in each reduced menu instead gives
$W^{\mathrm{menu}}\geq W^{\mathrm{coarse}}$, as defined in
\eqref{app:welfare:acyclicity}.
In the underlying alternative model, write $f(\psi)$ for the physical
menu distribution and $\mu_a$ for the expected reward of alternative
$a$. Unrestricted informed alternative choice yields
\[
 W^{\mathrm{full}}=\sum_\psi f(\psi)\max_{a\in\psi}\mu_a
 \ \geq\ W^{\mathrm{menu}}.
\]
Indeed, for any reduced-menu class rule its expected payoff in a physical
menu is an average of the available alternatives' means, hence no larger
than their maximum. The same conclusion holds with menu-dependent means
after replacing $\mu_a$ by $\mu_a(\psi)$.
Thus, for a policy $q\in\overline{\mathcal L}$, the identity
\begin{equation}
 W^{\mathrm{full}}-W(q)
 =\bigl(W^{\mathrm{full}}-W^{\mathrm{coarse}}\bigr)
  +\bigl(W^{\mathrm{coarse}}-W(q)\bigr)
 \label{app:welfare:decomposition}
\end{equation}
separates the loss imposed by the maintained representation from its
additional policy shortfall. The second term at a finite sensitivity
can include logit noise. Comparing high-sensitivity limits of long-run outcomes isolates losses that persist as logit noise vanishes.
\subsection{Efficient scores within own-payoff intervals}
\label{app:welfare:feasibility-section}

Throughout this and the next subsection, each supported menu has a
unique payoff-maximizing class, and $G^{\mathrm{eff}}$ is the graph in
Proposition~\ref{app:welfare:oracle}. Define
\begin{equation}
 \ell_s:=\min_{\omega\ni s}\pi_s(\omega),\qquad
 u_s:=\max_{\omega\ni s}\pi_s(\omega),\qquad
 I_s:=[\ell_s,u_s],\qquad K:=\prod_{s\in\mathcal S}I_s.
 \label{app:welfare:payoff-intervals}
\end{equation}
These intervals are well defined because every class has positive
availability. Call a score vector \emph{efficient} if
$v_{s^*(\omega)}>v_r$ for every supported menu and every
$r\in\omega\setminus\{s^*(\omega)\}$. Equivalently, holding these scores
fixed and letting logit sensitivity increase selects the unique
payoff-maximizing class in each menu with probability tending to one.

\begin{proposition}[Score feasibility]
\label{app:welfare:score-feasibility}
There exists an efficient score vector $v\in K$ if and only if
$G^{\mathrm{eff}}$ is acyclic and
\begin{equation}
 u_s>\ell_t
 \qquad\text{for every nontrivial directed path }s\leadsto t.
 \label{app:welfare:path-interval-condition}
\end{equation}
\end{proposition}
\begin{proof}
An efficient vector satisfies $v_s>v_r$ along every arrow, which
excludes directed cycles. Along any nontrivial directed path
$s\leadsto t$, it also satisfies $u_s\geq v_s>v_t\geq\ell_t$.
This proves necessity.

For sufficiency, assume the graph is acyclic and
\eqref{app:welfare:path-interval-condition} holds. Let $D(s)$ contain
$s$ and all its descendants. For $t\in D(s)$, let $L(s,t)$ be the
length of a longest directed path from $s$ to $t$, with $L(s,s)=0$.
All these lengths are finite. Since there are finitely many reachable
pairs and their endpoint inequalities are strict, choose
$\varepsilon>0$ such that
\[
 \ell_t+\varepsilon L(s,t)\leq u_s
 \quad\text{whenever }t\in D(s)\setminus\{s\}.
\]
If there are no arrows, any positive $\varepsilon$ suffices. Set
\[
 v_s:=\max_{t\in D(s)}
       \{\ell_t+\varepsilon L(s,t)\}.
\]
The term $t=s$ gives $v_s\geq\ell_s$, and every term is at most $u_s$,
so $v\in K$. If $s\to r$, then $D(r)\subseteq D(s)$ and
$L(s,t)\geq1+L(r,t)$ for every $t\in D(r)$. Hence
$v_s\geq v_r+\varepsilon$. The vector therefore satisfies every
required strict comparison.
\end{proof}

The criterion checks entire paths, including endpoints that never
appear together. In the main-text example, the menus $\{A,C\}$ and
$\{B,C\}$ have respective payoffs $(3/5,1)$ and $(2/5,0)$, so
\[
 B\longrightarrow C\longrightarrow A,\qquad
 I_B=\{2/5\},\quad I_C=[0,1],\quad I_A=\{3/5\}.
\]
Each individual arrow can be satisfied by scores in these intervals,
but the path cannot: it requires $v_B>v_A$, while $u_B<\ell_A$.
The proposition concerns feasibility of efficient score comparisons.

\subsection{Welfare bounds from conflicting paths}
\label{app:welfare:path-section}

A \emph{common monotone binary rule} is a measurable function
$Q:\mathbb R^2\to[0,1]$ that is weakly increasing in its first
argument, weakly decreasing in its second, and satisfies
$Q(x,y)+Q(y,x)=1$. In a binary menu $\{s,r\}$ it selects $s$ with
probability $Q(v_s,v_r)$. The same function is used across classes
and menus. Logit is one such rule; continuity is not required.

\begin{lemma}[Errors along a path]
\label{app:welfare:path-errors}
For any common monotone binary rule, any integer $d\geq2$, and any
real scores $a_0,\ldots,a_d$ with $a_0\leq a_d$,
\begin{equation}
 \sum_{i=1}^d Q(a_i,a_{i-1})\geq1.
 \label{app:welfare:error-sum}
\end{equation}
\end{lemma}
\begin{proof}
First, for all real $a,b,c$,
\[
 Q(b,a)+Q(c,b)\geq Q(c,a).
\]
If $b\geq c$, the first term alone is at least the right side. If
$b\leq a$, the second term alone is at least the right side.
In the remaining case $a<b<c$, both terms are at least $1/2$,
because $Q(x,x)=1/2$ and the rule is monotone. Thus their sum is at
least one and again dominates the right side.
Repeatedly applying this inequality to the first $d-1$ terms gives
\[
 \sum_{i=1}^d Q(a_i,a_{i-1})
 \geq Q(a_{d-1},a_0)+Q(a_d,a_{d-1})
 \geq Q(a_{d-1},a_0)+Q(a_0,a_{d-1})=1,
\]
where the last inequality uses $a_d\geq a_0$.
\end{proof}

\begin{proposition}[A path-based welfare bound]
\label{app:welfare:path-bound}
Suppose $G^{\mathrm{eff}}$ is acyclic and contains a directed path
of distinct classes $s_0\to s_1\to\cdots\to s_d$, with $d\geq2$.
For each edge choose a supported menu $\omega_i$ whose unique
payoff-maximizing class is $s_{i-1}$ and which contains $s_i$.
Put
\begin{equation}
 \Delta_i:=\pi_{s_{i-1}}(\omega_i)
       -\max_{r\in\omega_i\setminus\{s_{i-1}\}}\pi_r(\omega_i)>0,
 \qquad c:=\min_{1\leq i\leq d}p(\omega_i)\Delta_i>0.
 \label{app:welfare:path-cost}
\end{equation}
Every common-score logit policy with finite sensitivity $\beta\geq0$
and $v_{s_0}\leq v_{s_d}$ satisfies
\begin{equation}
 W(q)\leq W^{\mathrm{coarse}}-c.
 \label{app:welfare:path-gap}
\end{equation}
If all the witnessing menus $\omega_i$ are binary, the same bound
holds for any common monotone binary rule on those menus, irrespective
of choices in the other menus.
\end{proposition}
\begin{proof}
The menus $\omega_i$ are distinct: otherwise the same menu would
have two different unique payoff maximizers. Let $E_i$ be the
probability of failing to select $s_{i-1}$ in $\omega_i$. Since the
graph is acyclic, Proposition~\ref{app:welfare:oracle} gives
$W^{\mathrm{coarse}}=W^{\mathrm{menu}}$. Menu-wise welfare shortfalls
are nonnegative, so
\[
 W^{\mathrm{coarse}}-W(q)
 \geq\sum_{i=1}^d p(\omega_i)\Delta_i E_i
 \geq c\sum_{i=1}^d E_i.
\]
Under logit, comparison with the two-class menu gives
\[
 E_i
 =1-\frac{e^{\beta v_{s_{i-1}}}}
             {\sum_{r\in\omega_i}e^{\beta v_r}}
 \geq
 \frac{e^{\beta v_{s_i}}}
      {e^{\beta v_{s_{i-1}}}+e^{\beta v_{s_i}}}
 =Q_\beta(v_{s_i},v_{s_{i-1}}).
\]
Lemma~\ref{app:welfare:path-errors} implies $\sum_i E_i\geq1$,
proving \eqref{app:welfare:path-gap}. When the witnessing menus are
binary and use a common monotone rule, $E_i=Q(v_{s_i},v_{s_{i-1}})$,
and the same lemma proves the second assertion.
\end{proof}

For a binary witnessing menu,
$p(\omega_i)\Delta_i$ is its probability-weighted payoff advantage.
In the main-text example both advantages equal
$(1/2)(2/5)=1/5$, giving the ceiling $7/10-1/5=1/2$.
A failed interval comparison cannot arise along a single arrow:
if $s$ is best in a menu containing $t$, then
$u_s\geq\pi_s(\omega)>\pi_t(\omega)\geq\ell_t$.
Thus any violation of \eqref{app:welfare:path-interval-condition}
in an acyclic graph occurs along a path of length at least two,
as required by Proposition~\ref{app:welfare:path-bound}.

For the online version, let $\mathcal H_k$ contain the observed
history and private randomization used before period $k$'s menu is
drawn. The current menu is independent of $\mathcal H_k$ and has
distribution $p$. Scores and the choice rule (including any
sensitivity) are $\mathcal H_k$-measurable, and the current menu,
choice, and reward belong to $\mathcal H_{k+1}$. Assume
\[
 \E[r_k\mid\mathcal H_k,\omega_k,s_k]=\pi_{s_k}(\omega_k),
 \qquad \E[r_k^2\mid\mathcal H_k]\leq M_2<\infty.
\]
These reward conditions hold under the maintained finite collection
of stationary independent reward laws and fixed within-class rule.

\begin{corollary}[A persistent loss under endpoint calibration]
\label{app:welfare:path-online}
Under the hypotheses of Proposition~\ref{app:welfare:path-bound},
suppose
\begin{equation}
 u_{s_0}<\ell_{s_d},\qquad
 \operatorname{dist}(v_k(s),I_s)\longrightarrow0
 \quad(s=s_0,s_d)\quad\text{almost surely}.
 \label{app:welfare:endpoint-calibration}
\end{equation}
In the preceding online setting, use either logit with any finite
history-dependent sensitivities $\beta_k\geq0$, or, when the
witnessing menus are binary, any history-dependent common monotone
binary rules $Q_k$ on those menus. Then, almost surely,
\begin{equation}
 \limsup_{k\to\infty}\E[r_k\mid\mathcal H_k]
 \leq W^{\mathrm{coarse}}-c,\qquad
 \limsup_{K\to\infty}\frac1K\sum_{k=0}^{K-1}r_k
 \leq W^{\mathrm{coarse}}-c.
 \label{app:welfare:path-online-bound}
\end{equation}
\end{corollary}
\begin{proof}
Strict separation of the endpoint intervals and their asymptotic
calibration give $v_k(s_0)<v_k(s_d)$ eventually almost surely.
Conditional on $\mathcal H_k$, the scores and choice rule are fixed
before the independent menu is drawn. Proposition~\ref{app:welfare:path-bound}
therefore bounds the conditional expected reward by
$W^{\mathrm{coarse}}-c$ eventually, proving the first assertion.

Let $D_k=r_k-\E[r_k\mid\mathcal H_k]$. These are martingale differences
with $\E[D_k^2\mid\mathcal H_k]\leq M_2$. The martingale
$\sum_{k<K}D_k/(k+1)$ has uniformly bounded second moments because
$\sum_k M_2/(k+1)^2<\infty$, and hence converges almost surely.
Kronecker's lemma gives $K^{-1}\sum_{k<K}D_k\to0$ almost surely.
The eventual conditional reward bound then gives the asserted
bound on realized average rewards.
\end{proof}

If class-menu rewards are deterministic, endpoint calibration follows
from any nonnegative own-payoff averaging rule with eventually positive
denominators and vanishing initial-score shares for the two endpoints.
Indeed, its endpoint score can be written
$v_k(s)=\lambda_{s,k}v_0(s)+(1-\lambda_{s,k})z_{s,k}$ with
$z_{s,k}\in I_s$ and $\lambda_{s,k}\to0$; compactness of $I_s$ gives
the calibration in \eqref{app:welfare:endpoint-calibration}. This includes arbitrary
history- and menu-dependent weights. For noisy rewards, unrestricted
reward-dependent weights need not approach the interval of conditional
means; endpoint calibration must instead be assumed or established.
Neither intermediate-score convergence nor convergence of the choice
rule is required.

For fixed menu support, strict endpoint separation and the unique
menu-wise payoff maximizers persist under all sufficiently small
perturbations of payoffs and positive menu probabilities. The constant
$c$ remains positive. Each perturbed problem therefore has the same
type of welfare bound for learners satisfying its endpoint calibration.

Strict separation in \eqref{app:welfare:endpoint-calibration} is
material when sensitivity can grow. With equally likely menus
$\{A,C\}$ and $\{B,C\}$ of payoffs $(0,1)$ and $(0,-1)$, the graph is
$B\to C\to A$, but $I_A=I_B=\{0\}$ and $I_C=[-1,1]$.
No efficient vector lies in $K$. Nevertheless, for $k\geq1$ the
scores $v_k(A)=-1/k$, $v_k(C)=0$, $v_k(B)=1/k$ approach $K$, while
logit sensitivity $\beta_k=k^2$ selects the efficient class in each
menu with probability tending to one. Thus equality of the endpoint
bounds rules out exact score feasibility but does not by itself
imply a persistent loss under asymptotic calibration.

\subsection{Singleton shifts and the two-class illustrations}

\begin{lemma}[Singleton payoffs do not change the oracle's ranking]
\label{app:welfare:singletons}
Replace each supported singleton payoff $\pi_s(\{s\})$ by
$\pi_s(\{s\})+\delta_s$, leaving all other primitives unchanged.
The welfare of every fixed policy changes by the same amount
$D=\sum_s p(\{s\})\delta_s$. Hence
$W^{\mathrm{coarse}}$ changes by $D$, its maximizing rankings are
unchanged, and each fixed policy's gap to the benchmark is unchanged.
\end{lemma}
\begin{proof}
Every feasible policy selects $s$ with probability one from $\{s\}$.
Its payoff change is therefore exactly $D$, independently of all choices
on competitive menus. Taking the maximum gives the assertions.
\end{proof}

This invariance concerns a fixed policy. Learning may produce a different
policy after the shift. In the Bordeaux--Chianti illustrations, the
menus $\{B,C\},\{B\},\{C\}$ each have probability $1/3$. Competitive
payoffs are $(2,1)$ and both singleton payoffs equal the shift $z$. If
$q$ is the probability of Bordeaux in the competitive menu, then
\begin{equation}
 W_z(q)=\frac{1+q+2z}{3},\qquad
 W_z^{\mathrm{coarse}}=\frac{2+2z}{3},\qquad
 W_z^{\mathrm{coarse}}-W_z(q)=\frac{1-q}{3}.
 \label{app:welfare:wines}
\end{equation}
The convergence and high-sensitivity calculations are proved in
Section~\ref{app:potential}. At $z=0$, the Chianti-favoring strict branch
has $q\to0$, so its limiting loss is $1/3$. At $z=11/3$, the unique
branch has $q\to2/3$ and both valuations tend to $3$, so its limiting
loss is $1/9$. To check the latter numbers directly, its limiting
conditional means satisfy
\[
 \frac{2q+11/3}{1+q}
 =\frac{1-q+11/3}{2-q},
\]
whose only solution in $[0,1]$ is $q=2/3$; both sides then equal $3$.
These losses are measured after taking the long-run outcome at each
fixed sensitivity and then letting sensitivity increase.

\subsection{An obstruction to calibrated payoff aggregation}

The example admits a useful open family. There are two menus,
$\{A,C\}$ and $\{B,C\}$, of probabilities $p$ and $1-p$, with
deterministic payoffs $(a,h)$ and $(b,\ell)$, respectively, where
\begin{equation}
 0<p<1,\qquad h>a>b>\ell.
 \label{app:welfare:family}
\end{equation}
The maintained within-class implementation has no choice to make in
either menu.

\begin{lemma}[Realization by objective alternatives]
\label{app:welfare:realization}
The reduced tree in \eqref{app:welfare:family} is induced by four
alternatives with menu-independent deterministic rewards. Moreover,
\[
 W^{\mathrm{full}}=W^{\mathrm{coarse}}=ph+(1-p)b,
\]
and the ranking $B\succ C\succ A$ attains this value.
\end{lemma}
\begin{proof}
Take classes $A=\{a_0\}$, $B=\{b_0\}$, and
$C=\{c_H,c_L\}$, with respective alternative rewards $a,b,h,\ell$.
Let the physical menus be $\{a_0,c_H\}$ and $\{b_0,c_L\}$, with
probabilities $p$ and $1-p$. Each available class has exactly one
available alternative. The stated ranking chooses $c_H$ in the first
menu and $b_0$ in the second, which are precisely the two menu-wise
payoff maximizers. No policy can improve on their expected payoff.
\end{proof}

\begin{proposition}[Welfare bound under the calibrated ordering]
\label{app:welfare:static}
In \eqref{app:welfare:family}, every common monotone binary rule $Q$
and every common score vector satisfying $v_A\geq v_B$, regardless of
$v_C$, satisfy
\begin{align}
 W_Q(v)&\leq\overline W
 :=\max\{pa+(1-p)b,\;ph+(1-p)\ell\},
 \label{app:welfare:ceiling}\\
 W^{\mathrm{coarse}}-\overline W
 &=\min\{p(h-a),(1-p)(b-\ell)\}>0.
 \label{app:welfare:positive-gap}
\end{align}
This bound is sharp within the logit family as sensitivity increases,
even when each score lies in the convex hull of that class's own
possible rewards.
\end{proposition}
\begin{proof}
Let $u$ and $w$ be the probabilities of choosing $C$ in $\{A,C\}$ and
$\{B,C\}$, respectively. Monotonicity and $v_A\geq v_B$ give
\[
 u=Q(v_C,v_A)\leq Q(v_C,v_B)=w.
\]
Put $B_0=pa+(1-p)b$, $X=p(h-a)>0$, and
$Y=(1-p)(b-\ell)>0$. Welfare equals
\[
 W_Q(v)=B_0+Xu-Yw
 \leq B_0+(X-Y)u
 \leq B_0+\max\{0,X-Y\}.
\]
This last value equals $\overline W$, while
$W^{\mathrm{coarse}}=B_0+X$, proving both formulas.
For sharpness set $(v_A,v_B)=(a,b)$. Taking $v_C=\ell$ and
$\beta\to\infty$ yields the first value in the maximum; taking
$v_C=h$ yields the second. These scores belong to the respective
own-payoff hulls.
\end{proof}

The ordering $v_A\geq v_B$ is precisely what constant-payoff calibration
eventually implies. Use the online setting preceding
Corollary~\ref{app:welfare:path-online}, with scores $v_k$ and a common
monotone binary rule $Q_k$ determined from $\mathcal H_k$ before the
current independent menu is drawn. This includes logit with arbitrary
history-dependent finite sensitivities $\beta_k\geq0$. Neither scores
nor the choice rule may be selected separately for different current
menus.

\begin{theorem}[The long-run aggregation bound]
\label{app:welfare:online}
Suppose $v_k(A)\to a$ and $v_k(B)\to b$ almost surely in the environment
\eqref{app:welfare:family}. No convergence or calibration condition is
imposed on $v_k(C)$, and $Q_k$ need not converge. Then almost surely
\begin{equation}
 \limsup_{k\to\infty}\E[r_k\mid\mathcal H_k]\leq\overline W,
 \qquad
 \limsup_{K\to\infty}\frac1K\sum_{k=0}^{K-1}r_k
 \leq\overline W.
 \label{app:welfare:online-bound}
\end{equation}
In particular, these conclusions hold for any nonnegative own-payoff
aggregation rule
\begin{equation}
 v_k(s)=
 \frac{a_{s,k}v_0(s)+\sum_{j<k:s_j=s}w_{s,k,j}r_j}
      {a_{s,k}+\sum_{j<k:s_j=s}w_{s,k,j}},
 \qquad a_{s,k},w_{s,k,j}\geq0,
 \label{app:welfare:weights}
\end{equation}
provided that for $s=A,B$ the denominators are eventually positive and
\begin{equation}
 \lambda_{s,k}:=
 \frac{a_{s,k}}{a_{s,k}+\sum_{j<k:s_j=s}w_{s,k,j}}
 \longrightarrow0\quad\hbox{almost surely}.
 \label{app:welfare:prior-share}
\end{equation}
The finite nonnegative weights can depend arbitrarily on the observed
history, including past menus and rewards, and can be chosen with full
knowledge of the primitives.
\end{theorem}
\begin{proof}
Since $a>b$, calibration implies $v_k(A)>v_k(B)$ eventually on an event
of probability one. Applying Proposition~\ref{app:welfare:static}
conditional on $\mathcal H_k$ gives the eventual bound
$\E[r_k\mid\mathcal H_k]\leq\overline W$. Independence of the current
menu is what preserves its probabilities in this conditional
calculation. Let $D_k=r_k-\E[r_k\mid\mathcal H_k]$. The variables
$D_k$ are bounded martingale differences. Hence the martingale series
$\sum_{k\geq0}D_k/(k+1)$ converges almost surely, since its increment
variances have finite sum. Kronecker's lemma gives
$K^{-1}\sum_{k<K}D_k\to0$ almost surely. The Ces\`aro averages of the
conditional rewards have limit superior at most $\overline W$; their
finitely many initial terms vanish from the average. This proves
\eqref{app:welfare:online-bound}.

Under \eqref{app:welfare:weights}, every observation of $A$ equals $a$
and every observation of $B$ equals $b$. Consequently,
\[
 v_k(A)=a+\lambda_{A,k}(v_0(A)-a),\qquad
 v_k(B)=b+\lambda_{B,k}(v_0(B)-b)
\]
whenever their denominators are positive. Condition
\eqref{app:welfare:prior-share} implies the required calibration,
irrespective of how the observation weights were chosen.
\end{proof}

For the main-text example,
\[
 p=\frac12,\quad h=1,\quad a=\frac35,\quad
 b=\frac25,\quad\ell=0,
\]
so $W^{\mathrm{coarse}}=7/10$, $\overline W=1/2$, and the persistent
shortfall is at least $1/5$. The same obstruction applies on the open
set \eqref{app:welfare:family}; it requires neither singleton menus nor
convergence of the full score vector. Vanishing prior influence is an
assumption on the feedback retained by the aggregation rule, not an
automatic consequence of positive logit choice probabilities. A rule
that discards its feedback need not satisfy it. Subject to calibration,
however, even unrestricted reweighting of $C$'s observations cannot
reverse the ordering of its two choice probabilities that drives the
bound.

\subsection{Convergence in the example}
\label{app:welfare:example-dynamics-section}

\begin{lemma}[Global convergence with a persistent welfare loss]
\label{app:welfare:example-dynamics}
Consider equally likely menus $\{A,C\}$ and $\{B,C\}$ with deterministic
payoffs $(3/5,1)$ and $(2/5,0)$, respectively. For every fixed finite
$\beta\geq0$, there is a unique SVE $v^\beta=(3/5,2/5,c_\beta)$.
It is globally asymptotically stable and locally exponentially stable
under the normalized dynamics, as is
$(v^\beta,\Lambda(v^\beta))$ under the joint dynamics on
$\mathbb R^3\times\operatorname{ri}\Delta(\mathcal S)$.
From every finite initial valuation vector and positive prior counts,
harmonic visit-count learning converges to this joint equilibrium almost
surely. Moreover,
\[
 v^\beta\longrightarrow(3/5,2/5,0)
 \qquad\text{as }\beta\to\infty.
\]
This strict VE selects $A$ and $B$ in their respective menus. Equilibrium
welfare therefore tends to $1/2$, below $W^{\mathrm{coarse}}=7/10$.
\end{lemma}
\begin{proof}
The conditional payoff targets satisfy $g_A\equiv3/5$ and
$g_B\equiv2/5$, so every SVE has these two coordinates. Write
$c=v_C$, $U=e^{3\beta/5}$, $V=e^{2\beta/5}$, and $t=e^{\beta c}$.
At these fixed values of $v_A,v_B$, the probabilities of choosing $C$
are $t/(t+U)$ and $t/(t+V)$. Its conditional payoff target is consequently
\[
 H_\beta(c):=g_C(3/5,2/5,c)=\frac{t+V}{2t+U+V},
 \qquad F_\beta(c):=H_\beta(c)-c.
\]
At $\beta=0$, $H_0\equiv1/2$, so $c_0=1/2$ and $F_0'=-1$.
For $\beta>0$, $U>V$ and $0<H_\beta(c)<1/2$ for every $c$.
Thus every zero of $F_\beta$ lies in $(0,1/2)$, and
$F_\beta(0)>0>F_\beta(1/2)$. Direct differentiation gives
\[
 H_\beta'(c)=\frac{\beta t(U-V)}{(2t+U+V)^2},
 \qquad
 H_\beta''(c)=
 \frac{\beta^2t(U-V)(U+V-2t)}{(2t+U+V)^3}.
\]
For $0\leq c\leq1/2$, $2t\leq2e^{\beta/2}=2\sqrt{UV}\leq U+V$,
so $F_\beta$ is convex on this interval. Its endpoint signs give a zero;
two distinct zeros would, by monotonicity of convex secant slopes,
force $F_\beta(1/2)\geq0$. The zero $c_\beta$ is therefore unique, and
\[
 F_\beta'(c_\beta)
 \leq\frac{F_\beta(1/2)-F_\beta(c_\beta)}{1/2-c_\beta}<0.
\]

The normalized Jacobian is lower triangular, with eigenvalues
$-1,-1,F_\beta'(c_\beta)$. Proposition~\ref{app:model:spectrum}
adds two eigenvalues equal to $-1$ for the joint dynamics on the
frequency simplex. Both equilibria are thus locally exponentially stable.
Because only $C$ has a nonconstant conditional payoff target,
Lemma~\ref{app:potential:oneheterogeneous} gives convergence of every
trajectory of either system to an equilibrium. Uniqueness and local
stability therefore establish global asymptotic stability for both
systems. The deterministic rewards satisfy the maintained moment
assumption, so Corollary~\ref{app:model:apt-consequences} gives the
almost-sure convergence of harmonic visit-count learning.

Finally, monotonicity of $H_\beta$ gives, for $\beta>0$,
\[
 0<c_\beta=H_\beta(c_\beta)
 \leq H_\beta(1/2)=\frac{1}{1+e^{\beta/10}}\longrightarrow0.
\]
The two limiting valuation gaps against $C$ are strictly positive,
so its choice probabilities vanish in both menus. The limiting policy
chooses $A$ and $B$, whose constant payoffs equal their valuations;
$C$ is unchosen. Hence the limit is a strict VE, and equilibrium welfare
tends to $\tfrac12(3/5)+\tfrac12(2/5)=1/2$.
\end{proof}

\section{Algorithmic prominence and consumer learning}
\label{app:platform}

This section proves the two-state illustration, the propositions on
optimal deterministic and randomized prominence, and the associated
welfare and learning claims, including extensions to heterogeneous
consumer partitions and prominence caps. The maintained specification has
no commission on competitor sales, and uses a fixed
stationary prominence rule. Prominence has no direct cost except in
Subsection~\ref{app:platform:costly-section}. There is no cap except in
Subsection~\ref{app:platform:caps-section}.
A product state $i\in\{1,\ldots,J\}$ is drawn independently
and uniformly each period. The in-house product has expected payoff $\mu_i$
and margin $m_i\geq0$; the competitor and outside option have respective
expected payoffs $\mu_C>\mu_O$. Their reward laws satisfy the maintained
moment assumptions. The platform chooses prominence with probability
$\tau_i\in[0,1]$ in state $i$, using fresh draws across periods that are
independent of reward shocks and choice randomization. The resulting
class menu is $\{H,O\}$ under
prominence and $\{H,C,O\}$ otherwise. Consumers use their common class
valuations and harmonic visit-count learning. Among policies yielding
the same limiting profit, the platform prefers lower prominence frequency
$T(\tau):=J^{-1}\sum_i\tau_i$. More precisely, among attained primary maximizers, lower prominence frequency is preferred. A minimum-frequency maximizer is selected whenever that secondary minimum exists. This convention selects among profit
maximizers; it does not change their profit values.

Throughout this section,
\begin{equation}
 \bar\mu:=\frac1J\sum_i\mu_i<\mu_O<\mu_C,
 \qquad \max_i\mu_i>\mu_O.
 \label{app:platform:maintained}
\end{equation}
Sensitivity is fixed during learning. All high-sensitivity conclusions
first take the long-run limit at a fixed finite $\beta$, and only then let
$\beta\to\infty$. Approximations to a profit supremum by different fixed
rules are a further limit; no interchange of these limits is assumed.
Classes that are never available are omitted, as in the standing convention.

\subsection{Reduction, convergence, and long-run payoffs}
\label{app:platform:reduction-section}

For a fixed rule put
\[
 T=\frac1J\sum_i\tau_i.
\]
If $0<T<1$, define the promoted and unpromoted means
\begin{equation}
 a=\frac{\sum_i\tau_i\mu_i}{JT},\qquad
 b=\frac{\sum_i(1-\tau_i)\mu_i}{J(1-T)},\qquad
 Ta+(1-T)b=\bar\mu.
 \label{app:platform:menu-means}
\end{equation}
The class-menu probabilities are $T$ and $1-T$, and the conditional means
of $H$ are $a$ and $b$. Indeed, conditional on the class menu, the
probability of choosing $H$ does not depend further on the physical state.
Thus choosing $H$ does not change the conditional distribution of states
within that menu. This is exactly the reduction in
Lemma~\ref{app:model:reduction}; it gives Alice neither the state payoffs
nor their frequencies.

At an SVE the valuations of $C$ and $O$ are their constant targets.
Writing $v$ for the valuation of $H$, define
\begin{equation}
 A_\beta(v)=\frac{e^{\beta v}}{e^{\beta v}+e^{\beta\mu_O}},\qquad
 B_\beta(v)=\frac{e^{\beta v}}
 {e^{\beta v}+e^{\beta\mu_O}+e^{\beta\mu_C}},\qquad
 R_\beta(v)=\frac{B_\beta(v)}{A_\beta(v)}.
 \label{app:platform:choice}
\end{equation}
These are the in-house choice probabilities under prominence, without
prominence, and their ratio. In particular $0<R_\beta(v)<1$. The scalar
conditional payoff target is
\begin{align}
 g_\beta(v)
 &=\frac{Ta+(1-T)bR_\beta(v)}{T+(1-T)R_\beta(v)} \notag\\
 &=\bar\mu+
 \frac{T(a-\bar\mu)e^{\beta\mu_C}}
 {e^{\beta v}+e^{\beta\mu_O}+T e^{\beta\mu_C}}.
 \label{app:platform:target}
\end{align}
It lies between $a$ and $\bar\mu$. Its derivative has the sign of
$\bar\mu-a$ (and is zero at $\beta=0$).

\begin{proposition}[Learning under a fixed prominence rule]
\label{app:platform:convergence}
Under \eqref{app:platform:maintained}, every fixed stationary prominence
rule admits a unique SVE at every finite sensitivity. Its joint equilibrium
is globally asymptotically stable and locally exponentially stable on the
active-class state space. Harmonic visit-count learning converges to it
almost surely from every finite initial valuation vector and positive
initial pseudo-counts.
\end{proposition}
\begin{proof}
Suppose first that $0<T<1$. The equilibrium equation for $H$ is
$v=g_\beta(v)$. If $a\geq\bar\mu$, then $g_\beta$ is non-increasing, so
$g_\beta(v)-v$ is strictly decreasing and has exactly one zero; existence
also follows from the payoff-box fixed-point argument. At that zero,
$g_\beta'(v)-1\leq-1$.

If $a<\bar\mu$, every zero lies in $(a,\bar\mu)\subset(-\infty,\mu_O)$.
For $\beta>0$, differentiating \eqref{app:platform:target} and using the
fixed-point equation gives, at each such zero,
\[
 g_\beta'(v)
 =\frac{\beta(\bar\mu-v)e^{\beta v}}
 {e^{\beta v}+e^{\beta\mu_O}+T e^{\beta\mu_C}}
 <\frac{u}{1+e^u}<1,
 \qquad u=\beta(\mu_O-v)>0.
\]
At $\beta=0$ the derivative is zero. Hence every zero of $g_\beta-v$ is
a strict downward crossing. The function is analytic and is not
identically zero, since $g_\beta$ is bounded whereas $v$ is not.
Its zeros in the compact payoff interval are therefore finite; two
downward crossings would require an intervening zero that is not a
downward crossing. Thus the zero is unique in this case as well, and
$g_\beta'(v)-1<0$ there.

All other class targets are constant. Lemma~\ref{app:potential:oneheterogeneous}
therefore gives convergence of every joint trajectory to an equilibrium,
which must be the unique one just identified. In coordinates ordered as
$H,C,O$, the normalized valuation Jacobian is upper triangular, with
eigenvalues $g_\beta'(v)-1,-1,-1$. Proposition~\ref{app:model:spectrum}
adds only negative frequency eigenvalues to the joint linearization.
The joint equilibrium is thus locally exponentially stable and globally
attracting, hence globally asymptotically stable. Almost-sure convergence
follows from Corollary~\ref{app:model:apt-consequences}.

If $T=0$, every state generates the same full menu and $g_H\equiv\bar\mu$.
If $T=1$, every state generates $\{H,O\}$, the unavailable competitor
is omitted, and again $g_H\equiv\bar\mu$. All active targets are constant
in either case, so Proposition~\ref{app:model:constant-target} applies.
\end{proof}

Let $v_\beta$ denote the resulting equilibrium valuation of $H$ and write
\begin{align}
 d_i^\beta&=\tau_iA_\beta(v_\beta)+(1-\tau_i)B_\beta(v_\beta),\notag\\
 c_\beta&=(1-T)\frac{e^{\beta\mu_C}}
 {e^{\beta v_\beta}+e^{\beta\mu_O}+e^{\beta\mu_C}},\notag\\
 o_\beta&=T\bigl(1-A_\beta(v_\beta)\bigr)
 +(1-T)\frac{e^{\beta\mu_O}}
 {e^{\beta v_\beta}+e^{\beta\mu_O}+e^{\beta\mu_C}}.
 \label{app:platform:shares}
\end{align}
Here $d_i^\beta$ is conditional on state $i$, while $c_\beta,o_\beta$ are
unconditional market shares. These expressions remain valid at $T=0,1$
by ignoring the probability of choice in a menu that never occurs.
The limiting calendar-time profit and consumer payoff at fixed sensitivity
are, respectively,
\begin{equation}
 \Pi_\beta(\tau)=\frac1J\sum_i m_i d_i^\beta,\qquad
 W_\beta(\tau)=\frac1J\sum_i\mu_i d_i^\beta
                  +\mu_Cc_\beta+\mu_Oo_\beta.
 \label{app:platform:payoffs}
\end{equation}
These are also the almost-sure limits of realized time averages. To see
this, subtract the conditional expectation from each state-choice
indicator and from each period's reward or profit. Indicator and profit
martingale differences are bounded, and reward differences have bounded
conditional second moments. The martingale series obtained by dividing
the period-$k$ difference by $k+1$ converges, so Kronecker's lemma makes
its Ces\`aro average vanish. Convergence of valuations and continuity of
the finite-$\beta$ logits make the corresponding conditional expectations
converge to \eqref{app:platform:shares}--\eqref{app:platform:payoffs}.
In particular, no identification of a deterministic time average with a
stochastic time average is being assumed without this argument.

\subsection{The high-sensitivity limit of any fixed rule}
\label{app:platform:limits-section}

\begin{lemma}[Fixed-rule purchase limits]
\label{app:platform:limits}
Suppose $0<T<1$, and use \eqref{app:platform:menu-means}. Equilibrium
valuations and purchase probabilities have limits as $\beta\to\infty$.
They are given by
\begin{equation}
\begin{array}{c|c|c|c}
 \text{promoted mean} & \lim v_\beta
 & \lim A_\beta(v_\beta) & \lim B_\beta(v_\beta)\\ \hline
 a<\mu_O & a &0&0\\
 a=\mu_O &\mu_O&1/2&0\\
 \mu_O<a\leq\mu_C&a&1&0\\
 a>\mu_C&\mu_C&1&q_\tau
\end{array}
\label{app:platform:limit-table}
\end{equation}
In the last row,
\begin{equation}
 q_\tau=
 \frac{T(a-\mu_C)}{(1-T)(\mu_C-b)}
 =\frac{\sum_i\tau_i(\mu_i-\mu_C)}
 {\sum_i(1-\tau_i)(\mu_C-\mu_i)}\in(0,1).
 \label{app:platform:q-general}
\end{equation}
Thus the limiting sales vector is respectively $0$, $\tau/2$, $\tau$,
or $(\tau_i+(1-\tau_i)q_\tau)_i$. At $T=0$ and $T=1$ it is zero,
because $v_\beta=\bar\mu<\mu_O$. Consequently
$\Pi^\infty(\tau):=\lim_{\beta\to\infty}\Pi_\beta(\tau)$ and
$W^\infty(\tau):=\lim_{\beta\to\infty}W_\beta(\tau)$ are well defined
for every fixed stationary rule.
\end{lemma}
\begin{proof}
Write $\rho_\beta=R_\beta(v_\beta)$. If $a<\mu_O$, every root is at most
$u_0:=\max\{a,\bar\mu\}<\mu_O$. Therefore
\[
 0<\rho_\beta
 =\frac{e^{\beta v_\beta}+e^{\beta\mu_O}}
 {e^{\beta v_\beta}+e^{\beta\mu_O}+e^{\beta\mu_C}}
 \leq2e^{-\beta(\mu_C-\mu_O)}\longrightarrow0.
\]
Equation~\eqref{app:platform:target} gives $v_\beta\to a$, and the two
purchase probabilities tend to zero.

If $a\geq\mu_O$, then $a>\bar\mu$ and $b<\bar\mu<\mu_O$.
For $a=\mu_O$ the equilibrium satisfies $v_\beta<a$, and
\[
 0\leq a-v_\beta
 =\frac{(1-T)(a-b)\rho_\beta}{T+(1-T)\rho_\beta}
 \leq\frac{2(1-T)(a-b)}{T}
         e^{-\beta(\mu_C-\mu_O)}.
\]
Hence $\beta(v_\beta-\mu_O)\to0$, proving the limiting promoted
purchase probability $1/2$. This is a rate calculation, rather than an
inference from equality of limiting valuations alone.

If $\mu_O<a<\mu_C$, then $v_\beta\leq a<\mu_C$, which again forces
$\rho_\beta\to0$ and $v_\beta\to a$. The limits are now $A_\beta\to1$
and $B_\beta\to0$. If $a=\mu_C$, a subsequential valuation limit strictly
below $\mu_C$ would force $\rho_\beta\to0$ and hence $v_\beta\to a$,
a contradiction. Since $v_\beta\leq a$, it follows that
$v_\beta\to\mu_C$. Rearranging the fixed-point equation gives
\begin{equation}
 \rho_\beta=
 \frac{T(a-v_\beta)}{(1-T)(v_\beta-b)}.
 \label{app:platform:rho-equilibrium}
\end{equation}
Its denominator stays positive in the limit and its numerator tends to
zero. Thus $B_\beta\to0$ even though $v_\beta\to\mu_C$.

Finally suppose $a>\mu_C$. A subsequential valuation limit below $\mu_C$
would force $\rho_\beta\to0$ and $v_\beta\to a>\mu_C$. A limit above
$\mu_C$ would force $\rho_\beta\to1$ and $v_\beta\to\bar\mu<\mu_O$.
Compactness therefore implies $v_\beta\to\mu_C$. Equation
\eqref{app:platform:rho-equilibrium} gives $\rho_\beta\to q_\tau$;
$A_\beta\to1$ then gives $B_\beta\to q_\tau$. Its numerator is positive,
and the denominator exceeds it by $\mu_C-\bar\mu>0$, proving
$0<q_\tau<1$. The remaining assertions follow from
\eqref{app:platform:shares}--\eqref{app:platform:payoffs}.
\end{proof}

\subsection{The two-state illustration}
\label{app:platform:two-section}

The following lemma collects the claims in the main text's two-state
illustration and its welfare comparison.
\begin{lemma}[Two-state prominence and welfare]
\label{app:platform:two}
Suppose $J=2$,
\[
 \mu_h>\mu_C>\mu_O>\mu_\ell,\qquad
 (\mu_h+\mu_\ell)/2<\mu_O,\qquad m_h,m_\ell>0.
\]
High-only prominence uniquely maximizes limiting profit among deterministic
rules. Its limiting valuation is $\mu_C$, its sales are $(1,q_2)$, and
\begin{equation}
 q_2=\frac{\mu_h-\mu_C}{\mu_C-\mu_\ell}\in(0,1),\qquad
 \Pi^\infty=\frac{m_h+q_2m_\ell}{2},\qquad W^\infty=\mu_C.
 \label{app:platform:two-values}
\end{equation}
For that same prominence rule, the informed coarse oracle earns consumer
payoff $(\mu_h+\mu_C)/2$.
\end{lemma}
\begin{proof}
There are four deterministic rules. No prominence, low-only prominence,
and universal prominence have zero limiting sales by
Lemma~\ref{app:platform:limits}. For high-only prominence, $T=1/2$,
$a=\mu_h>\mu_C$ and $b=\mu_\ell$, giving \eqref{app:platform:two-values}
and strictly positive profit. Its welfare equals
\[
 \frac{\mu_h+q_2\mu_\ell+(1-q_2)\mu_C}{2}=\mu_C,
\]
because $\mu_h-\mu_C=q_2(\mu_C-\mu_\ell)$. No prominence also gives
welfare $\mu_C$, since the competitor is then chosen in every state.

With high-only prominence, the highest available expected payoff is
$\mu_h$ in state $h$ and $\mu_C$ in state $\ell$. The common ranking
$C\succ H\succ O$ implements both choices, so it attains the menu-wise
upper bound and is the coarse oracle optimum of
Proposition~\ref{app:welfare:oracle}. The welfare gap is
$(\mu_h-\mu_C)/2$. This efficient ranking cannot be a self-confirming
CQL limit: its positive in-house purchases occur only in state $h$ and
therefore calibrate $H$ to $\mu_h>\mu_C$, contrary to its ranking below
$C$. Under complete information prominence does not change sales in
either state, since the high product exceeds $\mu_C$ and the low product
is below $\mu_O$. The platform's tie-breaking convention therefore selects
no prominence under complete information.
\end{proof}

\subsection{Optimal deterministic prominence}
\label{app:platform:three-section}

This is the main text's deterministic-prominence proposition, with its
maintained restriction on the mean payoff of the two promoted products.
\begin{proposition}[Optimal deterministic prominence]
\label{app:platform:three}
Suppose \eqref{app:platform:maintained} holds, $J=3$, and
\begin{equation}
 \mu_h>\mu_C>\mu_m>\mu_O>\mu_\ell,\qquad
 \frac{\mu_h+\mu_m}{2}\leq\mu_C,\qquad
 m_h,m_m,m_\ell>0.
 \label{app:platform:three-assumptions}
\end{equation}
Define
\begin{equation}
 q_3:=\frac{\mu_h-\mu_C}{2\mu_C-\mu_m-\mu_\ell}\in(0,1).
 \label{app:platform:three-probabilities}
\end{equation}
Only high-only and high-and-medium prominence can maximize limiting profit
among deterministic rules. Their purchase probabilities in states
$(h,m,\ell)$ and their profits are
\begin{align}
 d^h&=(1,q_3,q_3),&
 \Pi_h^\infty&=\frac{m_h+q_3(m_m+m_\ell)}3,\notag\\
 d^{hm}&=(1,1,0),&
 \Pi_{hm}^\infty&=\frac{m_h+m_m}3.
 \label{app:platform:three-values}
\end{align}
High-only prominence uniquely maximizes profit if and only if
\begin{equation}
 q_3m_\ell>(1-q_3)m_m.
 \label{app:platform:three-comparison}
\end{equation}
The reverse strict inequality selects high-and-medium prominence. At
equality, both rules maximize profit, and the tie-breaking convention
selects high-only prominence.
\end{proposition}
\begin{proof}
No prominence has zero limiting sales. Every promoted subset containing
$\ell$ also has mean below $\mu_O$: for example,
\[
 \mu_h+\mu_\ell<3\mu_O-\mu_m<2\mu_O,
 \qquad
 \mu_m+\mu_\ell<3\mu_O-\mu_h<2\mu_O.
\]
This covers both pairs; the singleton and the full set satisfy the same
claim by assumption. These four rules have zero limiting sales by
Lemma~\ref{app:platform:limits}.

Medium-only prominence has promoted mean $\mu_m\in(\mu_O,\mu_C)$ and
therefore sells only in state $m$. High-only prominence has $a=\mu_h$
and $b=(\mu_m+\mu_\ell)/2$, giving $d^h$ in
\eqref{app:platform:three-values} and limiting brand valuation $\mu_C$.
High-and-medium prominence has
$a=(\mu_h+\mu_m)/2\in(\mu_O,\mu_C]$ and $b=\mu_\ell$.
Lemma~\ref{app:platform:limits} gives limiting valuation $a$ and sales
$(1,1,0)$. This includes $a=\mu_C$: the lemma establishes that the
unpromoted purchase probability tends to zero even at that boundary.
Since $m_h>0$, high-and-medium prominence strictly dominates medium-only
prominence. Both candidate rules earn positive profit. The eight
deterministic rules have now all been covered.

The numerator defining $q_3$ is positive, and its denominator exceeds
the numerator by
$3\mu_C-\mu_h-\mu_m-\mu_\ell=3(\mu_C-\bar\mu)>0$.
Hence $q_3\in(0,1)$.
Finally,
\[
 3(\Pi_h^\infty-\Pi_{hm}^\infty)
 =q_3m_\ell-(1-q_3)m_m,
\]
which proves both strict comparisons and the equality case. At equality,
the prominence frequencies are $1/3$ and $2/3$, so tie-breaking selects
high-only prominence. In particular $m_h$ cancels; no ordering of high-
and low-state margins has been imposed.
\end{proof}

Under complete information, high-quality purchases occur irrespective of
prominence, low-quality purchases never occur, and medium-quality
purchases occur precisely under prominence. Positive $m_m$ therefore
requires promoting $m$ for profit maximization; the tie-breaking convention
selects medium-only prominence. Its CQL outcome instead has
valuation $\mu_m$ and purchases only in state $m$, as just proved.
The comparison is thus between policies under different consumer models,
not an assertion that the coarse learner can implement informed choices.

High-only prominence gives welfare $\mu_C$, because
\[
 \mu_h-\mu_C
 =q_3(2\mu_C-\mu_m-\mu_\ell).
\]
Its conditional coarse-oracle welfare is $(\mu_h+2\mu_C)/3$, implemented
by $C\succ H\succ O$; the gap is $(\mu_h-\mu_C)/3$. As in
Lemma~\ref{app:platform:two}, this ranking cannot be sustained by CQL:
its brand purchases occur only in state $h$, calibrating the brand to
$\mu_h>\mu_C$. No prominence gives welfare $\mu_C$, because the
competitor is then chosen in every state. With high-and-medium prominence,
welfare is $(\mu_h+\mu_m+\mu_C)/3\leq\mu_C$.
The competitor's shares under the two candidate policies are
$2(1-q_3)/3$ and $1/3$, both strictly positive.

\subsection{Optimal randomized prominence}
\label{app:platform:random-section}

\begin{proposition}[Optimal randomized prominence]
\label{app:platform:randomized}
Under \eqref{app:platform:maintained}, with $m_i\geq0$, the supremum of
limiting profit over fixed stationary prominence rules is
\begin{equation}
 \sup_{\tau\in[0,1]^J}\Pi^\infty(\tau)
 =\max_{d\in[0,1]^J}
 \left\{\frac1J\sum_i m_i d_i:
                  \sum_i(\mu_i-\mu_O)d_i\geq0\right\}.
 \label{app:platform:lp}
\end{equation}
The bound can be approached by a sequence of stationary rules that
randomize in at most one state. Along this sequence, limiting brand
purchases occur only under prominence and the limiting brand valuation
approaches $\mu_O$ from above. If every margin is strictly positive,
no fixed stationary rule attains this supremum.
More generally, every stationary rule can be replaced by one randomizing
in at most one state with weakly higher limiting profit, although its
prominence frequency may change.
\end{proposition}
\begin{proof}
\emph{Upper bound.} For every fixed rule, let $d_i=\lim_\beta d_i^\beta$,
which exists by Lemma~\ref{app:platform:limits}. At finite sensitivity,
the fixed-point equation gives the exact calibration identity
\begin{equation}
 v_\beta\sum_i d_i^\beta=\sum_i\mu_i d_i^\beta.
 \label{app:platform:calibration}
\end{equation}
If $\sum_i d_i>0$, positive limiting purchases require the limiting
valuation to be at least $\mu_O$. Passing to the limit in
\eqref{app:platform:calibration} yields
$\sum_i(\mu_i-\mu_O)d_i\geq0$. If all $d_i$ vanish, this inequality
holds as equality. Since $0\leq d_i\leq1$, every realized limiting profit
is bounded by the compact linear program on the right of
\eqref{app:platform:lp}.

\emph{Solving the linear program.} Put
\[
 I_+=\{i:\mu_i\geq\mu_O\},\quad
 I_-=\{i:\mu_i<\mu_O\},\quad
 S=\sum_{i\in I_+}(\mu_i-\mu_O)>0,\quad
 w_i=\mu_O-\mu_i\quad(i\in I_-).
\]
Increasing any $d_i$, $i\in I_+$, weakly increases the objective and
preserves feasibility, so an optimizer can be chosen with these
coordinates equal to one. The remaining problem is
\begin{equation}
 \max\left\{\sum_{i\in I_-}m_i d_i:
       0\leq d_i\leq1,\quad \sum_{i\in I_-}w_i d_i\leq S\right\}.
 \label{app:platform:knapsack}
\end{equation}
Because $\bar\mu<\mu_O$, the total weight $\sum_{I_-}w_i$ exceeds $S$.
With nonnegative margins, a maximizing vector can be chosen to use the
entire budget: any unused budget can be filled in an unsaturated
coordinate without reducing the objective, even when that margin is zero.

Order the ratios $m_i/w_i$ non-increasingly, fill coordinates in that
order, and use a fraction of the first coordinate whose full weight
would exceed the remaining budget. This is optimal. Indeed, moving an
equal amount of weight from an included lower-ratio coordinate to an
unfilled higher-ratio coordinate weakly increases the objective. Applying
these exchanges successively transforms any budget-exhausting feasible
vector into the stated allocation without decreasing profit. The construction uses at
most one fractional coordinate; ties may give additional optimizers.
This establishes the asserted existence statements without claiming
that every optimizer has a binding constraint when some margins vanish.

\emph{Approximate implementation with randomization in one state.}
Choose a binding optimizer $d^*$ of the preceding form: every coordinate
in $I_+$ equals one, and at most one coordinate in $I_-$ is fractional.
Since $\sum_{i\in I_-}w_i d_i^*=S>0$, at least one below-outside-option
product is included. If a fractional coordinate exists, let $j$ be that
state; otherwise choose any $j\in I_-$ with $d_j^*=1$. In either case,
$d_j^*>0$ and $d_i^*\in\{0,1\}$ for every $i\neq j$.
For $0<\varepsilon<d_j^*$, set
\begin{equation}
 \tau_i^\varepsilon=
 \begin{cases}
  d_j^*-\varepsilon,&i=j,\\
  d_i^*,&i\neq j.
 \end{cases}
 \label{app:platform:implementation}
\end{equation}
Every rule randomizes only in state $j$, which is fixed along the
sequence. Since $d^*$ binds the constraint,
\[
 \sum_i(\mu_i-\mu_O)\tau_i^\varepsilon
 =\varepsilon w_j>0.
\]
The rule uses both menus: prominence is certain in every state in
$I_+$, while $\tau_j^\varepsilon<1$. Its promoted mean satisfies
\[
 a_\varepsilon
 =\mu_O+\frac{\varepsilon w_j}{\sum_i\tau_i^\varepsilon}
 \longrightarrow\mu_O\quad\text{from above}.
\]
For all sufficiently small $\varepsilon$, it lies below $\mu_C$.
Lemma~\ref{app:platform:limits} then gives limiting sales exactly
$d_i^\varepsilon=\tau_i^\varepsilon$. Their profit shortfall from the
linear program's value is exactly $\varepsilon m_j/J$, which tends to
zero (and is identically zero when $m_j=0$).
Proposition~\ref{app:platform:convergence} supplies the
unique globally attracting joint equilibrium and almost-sure learning
convergence for each of these fixed rules and each finite sensitivity.
This proves the lower bound in \eqref{app:platform:lp} with the stated
order of limits.

\emph{Non-attainment with positive margins.} If all $m_i>0$, every
optimizer has $d_i=1$ on $I_+$: increasing such a coordinate strictly
improves the objective while preserving feasibility. Every optimizer
also binds the constraint, since otherwise an unsaturated coordinate
could be increased; the all-one vector is infeasible. Thus its positive
sales have purchase-weighted mean exactly $\mu_O$. For any fixed
stationary rule attaining this objective, calibration would force its
limiting valuation to equal $\mu_O$. Lemma~\ref{app:platform:limits}
then forces $a=\mu_O$ and $d_i=\tau_i/2\leq1/2$, contradicting
$d_i=1$ on the nonempty set $I_+$. Hence the value is a supremum rather
than a maximum over fixed rules. In particular, simply setting
$\tau=d^*$ gives promoted mean $\mu_O$ and halves the intended
promoted purchases, rather than implementing $d^*$.
The final, policy-by-policy assertion follows from
Lemma~\ref{app:platform:policy-reduction}(ii).
\end{proof}

\subsection{The illustrative randomized optima and their welfare}
\label{app:platform:random-examples-section}

\begin{corollary}[Two and three states under randomized prominence]
\label{app:platform:random-examples}
In the two-state environment of Lemma~\ref{app:platform:two}, put
\[
 t_2^*=\frac{\mu_h-\mu_O}{\mu_O-\mu_\ell}\in(q_2,1).
\]
The unique optimizing sales vector in \eqref{app:platform:lp} is
$(1,t_2^*)$. Its profit supremum is $(m_h+t_2^*m_\ell)/2$, approached by
$(\tau_h,\tau_\ell)=(1,t)$ as $t\uparrow t_2^*$.

In the three-state environment of Proposition~\ref{app:platform:three}, put
\[
 t_3^*=\frac{\mu_h+\mu_m-2\mu_O}{\mu_O-\mu_\ell}\in(q_3,1).
\]
The unique optimizing sales vector is $(1,1,t_3^*)$. Its profit supremum
is $(m_h+m_m+t_3^*m_\ell)/3$, approached by
$(\tau_h,\tau_m,\tau_\ell)=(1,1,t)$ as $t\uparrow t_3^*$.
Each supremum strictly exceeds every deterministic policy's profit.

For every sequence of stationary rules whose limiting profits approach
the corresponding supremum, limiting consumer welfare and competitor
share converge, respectively, to
\begin{equation}
 W_j^*=\mu_C-\frac{\mu_C-\mu_O}{j}(j-1+t_j^*),\qquad
 c_j^*=\frac{1-t_j^*}{j}>0,\qquad j=2,3.
 \label{app:platform:optimal-welfare}
\end{equation}
In each case the limiting brand valuation approaches $\mu_O$, and
$W_j^*<\mu_C$, the no-prominence welfare level.
\end{corollary}
\begin{proof}
In either example there is a single state below $\mu_O$. Positive
margins force every above-outside-option coordinate to be one and the
remaining coordinate to exhaust the feedback constraint. This proves
uniqueness and the formulas for $t_2^*,t_3^*$. Positivity of their
numerators and the strict unconditional-mean assumption put both in
$(0,1)$. The implementations follow from the same promoted-mean argument
as \eqref{app:platform:implementation}.

For two states, the function $(\mu_h-z)/(z-\mu_\ell)$ is strictly
decreasing on $(\mu_\ell,\mu_h)$, so $t_2^*>q_2$.
For three states the high-only sales vector $(1,q_3,q_3)$ has weighted
mean $\mu_C>\mu_O$. Increasing its medium-state coordinate from $q_3$
to one adds $(1-q_3)(\mu_m-\mu_O)>0$ to its strict feedback surplus.
The vector $(1,1,q_3)$ therefore still has positive surplus, so the
binding low-state coordinate $t_3^*$ is strictly larger than $q_3$.
Thus $(1,1,t_3^*)$ weakly increases every coordinate of both deterministic
candidate vectors $(1,q_3,q_3)$ and $(1,1,0)$, and strictly increases at
least one coordinate of each. Strictly positive margins prove the strict
profit comparisons.

It remains to establish welfare for \emph{every} approximately optimal
sequence, rather than only the displayed implementation. Let $\tau^n$
be such a sequence and let $d^n$ be its limiting sales vector. Each
$d^n$ is feasible for the compact linear program. Uniqueness of its
optimizer implies $d^n\to d^*$: otherwise a convergent subsequence
bounded away from $d^*$ would have an equally profitable feasible limit.
In particular, $d_h^n\to1$.

Let $A^n$ be the high-sensitivity limit of the promoted purchase
probability for rule $\tau^n$, defining the counterfactual probability
when the promoted menu is absent. From $B_\beta\leq A_\beta$ we have
$d_h^n\leq A^n\leq1$, so $A^n\to1$. The outside share in unpromoted
encounters is at most $e^{-\beta(\mu_C-\mu_O)}$, uniformly over scores
and rules. In promoted encounters it is $T^n(1-A^n)$. Consequently the
limiting outside shares $o^n$ tend to zero, and the competitor shares
converge to $1-J^{-1}\sum_i d_i^*$.

The limiting calibration identity gives
$v^n=(\sum_i\mu_i d_i^n)/(\sum_i d_i^n)$ whenever the denominator is
positive; it tends to $\mu_O$, since the optimal feedback constraint
binds. Finally,
\[
 W^\infty(\tau^n)
 =\frac1J\sum_i\mu_i d_i^n+\mu_Cc^n+\mu_Oo^n
 \longrightarrow
 \mu_O\frac1J\sum_i d_i^*
 +\mu_C\left(1-\frac1J\sum_i d_i^*\right).
\]
Substituting the two optimizing sales vectors proves
\eqref{app:platform:optimal-welfare} and all the remaining assertions.
\end{proof}

For $J$ states, every binding maximizing allocation constructed in
Proposition~\ref{app:platform:randomized} leaves a strictly positive
competitor share in its approximate implementation: the all-one vector
is infeasible under $\bar\mu<\mu_O$, so $J^{-1}\sum_i d_i^*<1$.
Its welfare tends to
$\mu_C-(\mu_C-\mu_O)J^{-1}\sum_i d_i^*$.
When there are several below-outside-option states, tied margin-to-shortfall
ratios can give several optimizing sales vectors and different welfare
limits. No general uniqueness claim is needed for the profit proposition.

The positive competitor share in these randomized implementations is
distinct from consumer mixing within a menu. Along the approximating
rules, $\mu_O<v_H^\infty<\mu_C$, so the brand is purchased under
prominence and the competitor without it. It is the platform's
state-contingent prominence that randomizes between these encounters.
By contrast, the deterministic high-only examples have
$v_H^\infty=\mu_C$ and strictly interior brand choice in unpromoted menus.
The welfare decline in \eqref{app:platform:optimal-welfare} is relative
to no prominence. For each displayed approximate implementation, the
ranking $C\succ H\succ O$ already implements the limiting choices and is
optimal among common-score rankings at that same prominence rule.
It is therefore a different comparison from the additional learning
loss, holding prominence fixed, in the deterministic illustrations.

\subsection{Heterogeneous consumer partitions}
\label{app:platform:heterogeneous-section}

Let a fraction $\lambda\in[0,1]$ of consumers use the coarse classes
$H,C,O$, and let the remainder distinguish the $J$ in-house products,
using classes $H_1,\ldots,H_J,C,O$. Both types learn from their own
repeated experiences using harmonic visit-count updates and the maintained
logit rule. State probabilities, product payoffs, and margins are the same
for both types. The platform knows $\lambda$ and applies one fixed
stationary prominence rule independently of consumer type and individual
learning histories. Profits and welfare are population-weighted averages
of the two types' long-run outcomes. We retain
\eqref{app:platform:maintained} and the order of limits specified above.

For a finest-partition consumer, define
\begin{equation}
 a_i^{\mathrm f}=
 \begin{cases}1&\mu_i>\mu_O,\\1/2&\mu_i=\mu_O,\\0&\mu_i<\mu_O,\end{cases}
 \qquad
 b_i^{\mathrm f}=
 \begin{cases}1&\mu_i>\mu_C,\\1/2&\mu_i=\mu_C,\\0&\mu_i<\mu_C.\end{cases}
 \label{app:platform:heterogeneous-fine-rates}
\end{equation}
These are the high-sensitivity brand-purchase probabilities with and
without prominence. Write $\Pi_{\mathrm c}^\infty(\tau)$ for the
coarse-consumer profit studied above, and put
\[
 \Pi_{\mathrm f}^\infty(\tau)
 =\frac1J\sum_i m_i\bigl[\tau_i a_i^{\mathrm f}
                    +(1-\tau_i)b_i^{\mathrm f}\bigr],\qquad
 \Pi_\lambda^\infty(\tau)
 =\lambda\Pi_{\mathrm c}^\infty(\tau)
 +(1-\lambda)\Pi_{\mathrm f}^\infty(\tau).
\]
Let $V_{\mathrm c}$ be the linear-program value in
\eqref{app:platform:lp}, and define
\begin{equation}
 V_{\mathrm f}
 =\frac1J\left(\sum_{\mu_i>\mu_O}m_i
                  +\frac12\sum_{\mu_i=\mu_O}m_i\right).
 \label{app:platform:heterogeneous-fine-value}
\end{equation}

\begin{proposition}[Randomized prominence with heterogeneous consumers]
\label{app:platform:heterogeneous}
After omitting unavailable classes, for every fixed stationary rule and
finite sensitivity both consumer types have unique globally
asymptotically stable joint equilibria, and
their harmonic learning processes converge almost surely. Moreover,
\begin{equation}
 \sup_{\tau\in[0,1]^J}\Pi_\lambda^\infty(\tau)
 =\lambda V_{\mathrm c}+(1-\lambda)V_{\mathrm f}.
 \label{app:platform:heterogeneous-value}
\end{equation}
The policies in \eqref{app:platform:implementation}, constructed from an
optimizer chosen as in Proposition~\ref{app:platform:randomized},
randomize in at most one state and approach this bound
simultaneously for every $\lambda\in[0,1]$. For each fixed $\lambda$,
every stationary rule also admits a replacement randomizing in at most
one state with weakly higher limiting population profit; its prominence
frequency may change. If all margins are strictly
positive and $\lambda>0$, no fixed rule attains the bound. At $\lambda=0$
the bound is attained; among its maximizers, the unique rule minimizing
the frequency of prominence is
\begin{equation}
 \tau_i=\mathbf1\{m_i>0,\ \mu_O\leq\mu_i\leq\mu_C\}.
 \label{app:platform:heterogeneous-fine-policy}
\end{equation}
\end{proposition}
\begin{proof}
For finest-partition consumers, the payoff target of every active class
is constant: it is $\mu_i$ for $H_i$, $\mu_C$ for $C$, and $\mu_O$ for
$O$. Every $H_i$ is available with probability $1/J$; $C$ is omitted
only if it is never available. Proposition~\ref{app:model:constant-target}
therefore gives uniqueness, global asymptotic stability, and almost-sure
convergence at every finite sensitivity. At that limit the finest
consumer's brand-choice probabilities in state $i$ are
\[
 \frac{e^{\beta\mu_i}}{e^{\beta\mu_i}+e^{\beta\mu_O}},\qquad
 \frac{e^{\beta\mu_i}}
 {e^{\beta\mu_i}+e^{\beta\mu_O}+e^{\beta\mu_C}}.
\]
Their high-sensitivity limits are
\eqref{app:platform:heterogeneous-fine-rates}. Thus finest-partition
learning gives informed choices whenever payoff comparisons are strict;
at a tie it gives the corresponding equal-probability logit limit.
The martingale averaging argument following
\eqref{app:platform:payoffs} applies unchanged, so these probabilities
also determine almost-sure realized time-average payoffs. The coarse
type's convergence follows from
Proposition~\ref{app:platform:convergence}.

Because $a_i^{\mathrm f}\geq b_i^{\mathrm f}$ and $m_i\geq0$, maximizing
the fine type's profit amounts to setting $\tau_i=1$ whenever
$m_i(a_i^{\mathrm f}-b_i^{\mathrm f})>0$. Its maximal value is
$J^{-1}\sum_i m_i a_i^{\mathrm f}=V_{\mathrm f}$. The strict inequality
holds exactly when $m_i>0$ and $\mu_O\leq\mu_i\leq\mu_C$.
All other coordinates leave profit unchanged, so setting them to zero
uniquely minimizes $J^{-1}\sum_i\tau_i$. This proves
\eqref{app:platform:heterogeneous-fine-policy}.

For every common policy, each type's profit is bounded by its individual
supremum, giving the upper bound in
\eqref{app:platform:heterogeneous-value}. The implementing policies
\eqref{app:platform:implementation} set $\tau_i=1$ on every state with
$\mu_i\geq\mu_O$. They therefore attain $V_{\mathrm f}$ exactly for
every fixed $\varepsilon$. Changing prominence on states below $\mu_O$
does not affect fine-consumer brand sales, since both probabilities in
\eqref{app:platform:heterogeneous-fine-rates} vanish there. Meanwhile,
their coarse-consumer profit tends to $V_{\mathrm c}$ by
Proposition~\ref{app:platform:randomized}. This proves the lower bound
with a policy family independent of $\lambda$.
The policy-by-policy strengthening follows from
Lemma~\ref{app:platform:policy-reduction}(ii); that replacement may
depend on $\lambda$.

If all margins are positive and $\lambda>0$, an attaining common policy
would have to attain $V_{\mathrm c}$, since
\begin{equation}
 \lambda V_{\mathrm c}+(1-\lambda)V_{\mathrm f}
       -\Pi_\lambda^\infty(\tau)
 =\lambda\bigl(V_{\mathrm c}-\Pi_{\mathrm c}^\infty(\tau)\bigr)
 +(1-\lambda)\bigl(V_{\mathrm f}-\Pi_{\mathrm f}^\infty(\tau)\bigr),
 \label{app:platform:heterogeneous-gap}
\end{equation}
and both terms are nonnegative. This contradicts the non-attainment part
of Proposition~\ref{app:platform:randomized}. The endpoint $\lambda=0$
was established directly above.
\end{proof}

For any fixed $\lambda>0$, identity
\eqref{app:platform:heterogeneous-gap} also implies that every sequence
approaching the mixed-population bound approaches $V_{\mathrm c}$ for
coarse consumers; if $\lambda<1$, it approaches $V_{\mathrm f}$ for
fine consumers as well. The result identifies a common approximately
optimal policy family, not a unique family for general $J$. Tied
margin-to-shortfall ratios can generate several maximizing coarse sales
allocations with different welfare, as discussed above. The assumption
of a common rule is substantive for welfare: if types were observed and
prominence could be personalized, the two individual bounds would still
bound total profit, but unnecessary low-product prominence could be
withheld from finest-partition consumers.

\begin{corollary}[Heterogeneous consumers in the illustrations]
\label{app:platform:heterogeneous-examples}
Maintain the respective two- and three-state assumptions in
Lemma~\ref{app:platform:two} and
Proposition~\ref{app:platform:three}, including strictly positive margins.
Let $q_j$ and $t_j^*$ be as defined there and in
Corollary~\ref{app:platform:random-examples}.

For every $\lambda>0$, high-only prominence remains the unique
profit-maximizing deterministic rule in the two-state illustration,
with profit $(m_h+\lambda q_2m_\ell)/2$. In the three-state illustration,
only high-only and high-and-medium prominence can maximize deterministic
profit, with respective values
\begin{equation}
 \Pi_{\lambda,h}^\infty
 =\frac{m_h+\lambda q_3(m_m+m_\ell)}3,
 \qquad
 \Pi_{\lambda,hm}^\infty=\frac{m_h+m_m}3.
 \label{app:platform:heterogeneous-deterministic}
\end{equation}
With ties resolved in favor of less frequent prominence, high-only
prominence is selected precisely when
\begin{equation}
 \lambda q_3(m_m+m_\ell)\geq m_m;
 \label{app:platform:heterogeneous-threshold}
\end{equation}
otherwise high-and-medium prominence is selected. At $\lambda=0$,
the least-prominence profit maximizers are no prominence in the
two-state illustration and medium-only prominence in the three-state
illustration.

For every fixed $\lambda>0$, the randomized profit suprema are
\begin{equation}
 V_{\lambda,2}=\frac{m_h+\lambda t_2^*m_\ell}{2},\qquad
 V_{\lambda,3}=\frac{m_h+m_m+\lambda t_3^*m_\ell}{3}.
 \label{app:platform:heterogeneous-example-values}
\end{equation}
Each strictly exceeds every deterministic rule's profit. Every sequence
of stationary rules approaching the respective supremum satisfies
\begin{equation}
 \tau^n\longrightarrow
 \begin{cases}(1,t_2^*),&j=2,\\(1,1,t_3^*),&j=3.\end{cases}
 \label{app:platform:heterogeneous-policy-limits}
\end{equation}
Along every such sequence, aggregate brand, competitor, and outside
shares converge, respectively, to
\begin{equation}
 \frac{j-1+\lambda t_j^*}{j},\qquad
 \frac{1-t_j^*}{j},\qquad
 \frac{(1-\lambda)t_j^*}{j},\qquad j=2,3.
 \label{app:platform:heterogeneous-example-shares}
\end{equation}
Welfare converges to
$W_{\lambda,j}^*=\lambda W_j^*+(1-\lambda)W_{\mathrm f,j}^*$,
where $W_j^*$ is given by \eqref{app:platform:optimal-welfare} and
\begin{equation}
 W_{\mathrm f,2}^*
 =\frac{\mu_h+t_2^*\mu_O+(1-t_2^*)\mu_C}{2},\qquad
 W_{\mathrm f,3}^*
 =\frac{\mu_h+\mu_m+t_3^*\mu_O+(1-t_3^*)\mu_C}{3}.
 \label{app:platform:heterogeneous-fine-welfare}
\end{equation}
Relative to no prominence, aggregate welfare falls by
\begin{equation}
 \Delta_{\lambda,2}
 =\frac{(\mu_C-\mu_O)(\lambda+t_2^*)}{2},\qquad
 \Delta_{\lambda,3}
 =\frac{(\mu_C-\mu_O)(2\lambda+t_3^*)
              +(1-\lambda)(\mu_C-\mu_m)}{3}.
 \label{app:platform:heterogeneous-welfare-loss}
\end{equation}
\end{corollary}
\begin{proof}
In both illustrations, a fine consumer buys the high product regardless
of prominence and never buys the low product. With three states, the
medium product is purchased exactly under prominence. Thus two-state
fine profit is $m_h/2$ under every deterministic rule. Combining this
with Lemma~\ref{app:platform:two} proves the two-state deterministic
claim for $\lambda>0$.

With three states, fine profit equals
$(m_h+\tau_m m_m)/3$. Coarse profit under the eight deterministic rules
was computed in Proposition~\ref{app:platform:three}. High-and-medium
prominence strictly dominates medium-only prominence by the additional
coarse high-state sales, worth $\lambda m_h/3>0$. It also strictly
dominates every rule whose coarse sales vanish: their fine profits are
at most $(m_h+m_m)/3$, whereas high-and-medium prominence produces
positive coarse sales and the same maximal fine profit. This leaves
only the stated pair. Their weighted profits are
\eqref{app:platform:heterogeneous-deterministic}, with difference
$[\lambda q_3(m_m+m_\ell)-m_m]/3$. The high-only rule has prominence
frequency $1/3$, as against $2/3$, so the less-prominence convention
selects it at equality. The $\lambda=0$ conclusions follow from
\eqref{app:platform:heterogeneous-fine-policy}.

The randomized values follow from
\eqref{app:platform:heterogeneous-value} because
$V_{\mathrm f}=m_h/2$ or $(m_h+m_m)/3$, respectively. The difference
between the two-state randomized bound and high-only profit is
$\lambda(t_2^*-q_2)m_\ell/2>0$. In three states its difference from
high-and-medium profit is $\lambda t_3^*m_\ell/3>0$, and from high-only
profit is
\[
 \frac{(1-\lambda q_3)m_m
              +\lambda(t_3^*-q_3)m_\ell}{3}>0.
\]
These comparisons cover all deterministic candidates.

Fix $\lambda>0$ and an approximately optimal sequence $\tau^n$.
Equation~\eqref{app:platform:heterogeneous-gap} forces its coarse
profits to converge to $V_{\mathrm c}$. Let $d^n$ be its limiting
coarse sales vector, and $v^n$ the limiting coarse brand valuation.
The proof of Corollary~\ref{app:platform:random-examples} gives
$d^n\to d^*$ and $v^n\to\mu_O$, where $d^*$ is the unique optimizer
$(1,t_2^*)$ or $(1,1,t_3^*)$. In particular, $d_h^n\to1$.
For all sufficiently large $n$, $d_h^n>1/2$ and $v^n<\mu_C$.
The complete classification in Lemma~\ref{app:platform:limits} now
forces the promoted mean $a_n$ to satisfy
$\mu_O<a_n<\mu_C$: absent or universal prominence and $a_n<\mu_O$
give zero sales; $a_n=\mu_O$ gives $d_h^n\leq1/2$; and
$a_n\geq\mu_C$ gives $v^n=\mu_C$. Hence $d^n=\tau^n$ for all
sufficiently large $n$, proving
\eqref{app:platform:heterogeneous-policy-limits}. This step is needed
to identify the fine type's outcomes for every approximately optimal
sequence, rather than only for the displayed implementation.

Eventually, coarse consumers therefore buy the brand under prominence
and the competitor otherwise. Fine consumers always buy $h$, buy $m$
under prominence when present, and choose $O$ under low-state prominence
or $C$ without it. Taking the policy limits gives the shares in
\eqref{app:platform:heterogeneous-example-shares} and the fine welfare
in \eqref{app:platform:heterogeneous-fine-welfare}. Coarse welfare tends
to $W_j^*$ by Corollary~\ref{app:platform:random-examples}; weighting
the two types gives $W_{\lambda,j}^*$.

Without prominence, coarse welfare is $\mu_C$ and fine welfare is
$[\mu_h+(j-1)\mu_C]/j$. Subtracting the preceding limits from their
population-weighted average gives
\eqref{app:platform:heterogeneous-welfare-loss}. In particular the
fine type's own losses are $t_2^*(\mu_C-\mu_O)/2$ and
$[\mu_C-\mu_m+t_3^*(\mu_C-\mu_O)]/3$, both positive.
These consumers avoid inferior brand purchases but cannot restore a
competitor excluded by the common policy. All sequence conclusions fix
$\lambda>0$; at $\lambda=0$, profit maximization with the tie-breaking
convention selects the different endpoint policies already stated.
\end{proof}

\subsection{Prominence caps}
\label{app:platform:caps-section}

Return to the all-coarse benchmark and maintain
\eqref{app:platform:maintained}. A prominence cap restricts a stationary
rule to
\[
 T(\tau):=\frac1J\sum_i\tau_i\leq\rho,\qquad 0\leq\rho\leq1.
\]
The cap concerns the frequency of prominence, rather than purchases.
As before, learning first reaches its long-run outcome at fixed
sensitivity, sensitivity then increases without bound, and any sequence
of different stationary rules is considered only thereafter. Put
$c=J\rho$ and $M(d)=J^{-1}\sum_i m_i d_i$.
The first program below describes sales supported by purchases only
under prominence, including their limits. The second allows purchases
when the competitor remains available, with the brand calibrated to
$\mu_C$.

\begin{proposition}[Optimal randomized prominence under a cap]
\label{app:platform:caps}
Define the two linear-program values
\begin{equation}
\begin{aligned}
 A_\rho:=\max_{d\in[0,1]^J}\quad &M(d)\\
 \text{subject to}\quad &\sum_i d_i\leq c,\\
 &\sum_i(\mu_i-\mu_O)d_i\geq0,\qquad
   \sum_i(\mu_i-\mu_C)d_i\leq0.
\end{aligned}
\label{app:platform:caps-A}
\end{equation}
\begin{equation}
\begin{aligned}
 B_\rho:=\max_{d\in[0,1]^J}\quad &M(d)\\
 \text{subject to}\quad &\sum_i(\mu_i-\mu_C)d_i=0,\\
 &\sum_jd_j-(J-c)d_i\leq c\quad\text{for every }i.
\end{aligned}
\label{app:platform:caps-B}
\end{equation}
Then
\begin{equation}
 P_\rho:=\sup_{T(\tau)\leq\rho}\Pi^\infty(\tau)
 =\max\{A_\rho,B_\rho\}.
 \label{app:platform:caps-value}
\end{equation}
Every feasible allocation in \eqref{app:platform:caps-B} is implemented
by setting
\begin{equation}
 q=\min_i d_i,\qquad
 \tau_i=\frac{d_i-q}{1-q}.
 \label{app:platform:caps-recovery}
\end{equation}
An allocation feasible for \eqref{app:platform:caps-A} with
$\sum_i(\mu_i-\mu_O)d_i>0$ is implemented by $\tau=d$; every other
feasible allocation is a limit of implementable allocations when
$\rho>0$. Consequently, a profit-maximizing stationary rule exists unless
$A_\rho>B_\rho$ and every optimizer of
\eqref{app:platform:caps-A} satisfies
$\sum_i(\mu_i-\mu_O)d_i=0$. In precisely that case, the value is a
supremum that is not attained. At $\rho=0$, the sole feasible rule is
no prominence and $P_0=0$.

Every stationary rule has a replacement with the same prominence
frequency and weakly higher limiting profit that randomizes in at most
two states. Hence an attained $P_\rho$ has such a maximizer, and an
unattained $P_\rho$ can be approached by such rules. The reduction also
holds for any fixed mixture of coarse and finest-partition consumers,
although \eqref{app:platform:caps-value} concerns the all-coarse case.
The two-state bound is sharp even with a positive finest-partition share.
\end{proposition}

\begin{proof}
Both programs have nonempty compact feasible sets and continuous linear
objectives, so their maxima exist. First suppose $\rho>0$.

\emph{Upper bound.} Consider a fixed feasible rule. The zero-sales
cases in Lemma~\ref{app:platform:limits} yield zero profit. If its
promoted mean satisfies $\mu_O<a\leq\mu_C$, its limiting sales are
$d=\tau$, which is feasible for \eqref{app:platform:caps-A}. If
$a=\mu_O$, the limiting sales $d=\tau/2$ are also feasible for that
program. If $a>\mu_C$, the sales satisfy
\[
 d_i=q+(1-q)\tau_i,\qquad 0<q<1,\qquad
 \sum_i(\mu_i-\mu_C)d_i=0.
\]
Writing $D=\sum_i d_i$, the cap is equivalent to
$D\leq c+(J-c)q$. Since $q\leq d_i$ for every $i$ and $J-c\geq0$,
the constraints of \eqref{app:platform:caps-B} follow. This proves
the upper bound in \eqref{app:platform:caps-value}.

\emph{Implementation of the second program.} Let $d$ be feasible
for \eqref{app:platform:caps-B} and take $q$ as in
\eqref{app:platform:caps-recovery}. Calibration excludes $d=(1,\ldots,1)$,
because $\bar\mu<\mu_C$, so $q<1$. The recovered rule belongs to
$[0,1]^J$. The constraint corresponding to an index with $d_i=q$ gives
\[
 T(\tau)=\frac{D-Jq}{J(1-q)}\leq\rho.
\]
If $d=0$, this is no prominence. Otherwise $0<T(\tau)<1$: equality
$T=0$ would make every $d_i=q$, and calibration would then force
$q=0$; equality $T=1$ would force every $d_i=1$. Moreover,
\[
 \sum_i(\mu_i-\mu_C)\tau_i
 =\frac{qJ(\mu_C-\bar\mu)}{1-q}.
\]
Thus $q=0$ gives promoted mean $a=\mu_C$, and
Lemma~\ref{app:platform:limits} implements $d=\tau$. If $q>0$, the
promoted mean exceeds $\mu_C$, and substitution in
\eqref{app:platform:q-general} gives $q_\tau=q$; the same lemma
implements $d_i=q+(1-q)\tau_i$. This proves exact implementation.
For a given nonzero sales allocation, the prominence frequency
$(D-Jq)/(J(1-q))$ decreases with $q$, since $D<J$. Hence choosing
$q=\min_i d_i$ also minimizes prominence among implementations of
that allocation. This conditional statement does not select among
different profit-maximizing allocations.

\emph{Implementation of the first program.} If its outside-option
constraint is strict, $d\neq0$ and the rule $\tau=d$ has promoted
mean in $(\mu_O,\mu_C]$. Its frequency is strictly below one,
because the all-one vector violates that constraint under
$\bar\mu<\mu_O$. Lemma~\ref{app:platform:limits} therefore gives
limiting sales exactly $d$.
To handle the boundary, choose states $h,\ell$ with
$\mu_h>\mu_O>\mu_\ell$, which exist under
\eqref{app:platform:maintained}. A sufficiently small positive vector
$r$, supported on these states, can have
\[
 \sum_i r_i<c,\qquad
 \mu_O<\frac{\sum_i\mu_i r_i}{\sum_i r_i}<\mu_C.
\]
Indeed, choose its mean strictly between $\mu_O$ and
$\min\{\mu_C,\mu_h\}$ and scale its coordinates downward.
For every feasible $d$, the vectors
$d^\varepsilon=(1-\varepsilon)d+\varepsilon r$ remain feasible,
have strict outside-option constraint, and tend to $d$.
Their exact implementations prove the lower bound in
\eqref{app:platform:caps-value}. The learning conclusions of
Proposition~\ref{app:platform:convergence} apply separately to each
implementing rule.

\emph{Attainment.} If $B_\rho\geq A_\rho$, the second program
provides an exact maximizer. If $A_\rho>B_\rho$ and the first
program has an optimizer with strict outside-option constraint,
$\tau=d$ attains the value. In the remaining case,
$P_\rho=A_\rho>B_\rho\geq0$. An attaining rule could not belong
to the competitor-indifference regime, whose profit is at most
$B_\rho$, or have zero sales. It also could not have
$a=\mu_O$: then $\tau$ itself is feasible for
\eqref{app:platform:caps-A}, whereas sales are $\tau/2$, so profit
is at most $A_\rho/2<P_\rho$. Finally, a rule with
$\mu_O<a\leq\mu_C$ would produce an optimizer of the first
program with strict outside-option constraint, contrary to the
hypothesis. This proves the asserted equivalence.

\emph{Randomization in at most two states.}
Apply Lemma~\ref{app:platform:policy-reduction}(i) to an exact maximizer
or to each rule in a profit-approaching sequence. The replacement
preserves prominence frequency, hence feasibility, and weakly raises
profit. Its profit remains bounded by $P_\rho$, proving both assertions.
The lemma applies to each fixed consumer mix as well. Sharpness, including
with heterogeneous consumers, follows from
Proposition~\ref{app:platform:caps-example}.

At $\rho=0$, the first program has only $d=0$. The second program
requires $D\leq Jd_i$ for every $i$, hence all coordinates equal
$D/J$; calibration and $\bar\mu<\mu_C$ force $D=0$. Both values
therefore agree with the unique feasible policy's zero profit, and
that policy uses no randomization.
\end{proof}

\begin{proposition}[Welfare under a prominence cap]
\label{app:platform:caps-welfare}
No feasible stationary rule raises limiting consumer welfare above
the no-prominence benchmark $\mu_C$. A welfare-neutral
profit-maximizing rule exists if and only if $B_\rho\geq A_\rho$.

Fix $\rho>0$ with $P_\rho>0$. Define
\begin{equation}
 G_\rho:=\max\left\{
 \frac1J\sum_i(\mu_C-\mu_i)d_i:
 d\text{ is feasible for \eqref{app:platform:caps-A}},\quad
 M(d)=P_\rho\right\},
 \label{app:platform:caps-gap}
\end{equation}
with $G_\rho=0$ if the displayed set is empty. Every sequence of
feasible stationary rules whose limiting profits tend to $P_\rho$
satisfies
\[
 \liminf_n W^\infty(\tau^n)\geq\mu_C-G_\rho,
\]
and a profit-approaching sequence randomizing in at most two states
attains equality with convergent welfare. A welfare-neutral exact
maximizer, whenever one exists, can also be chosen with at most two
randomized states. Thus welfare converges to the no-prominence benchmark
along every profit-approaching sequence if and only if $G_\rho=0$. In particular, this
holds when $B_\rho>A_\rho$, whereas $A_\rho>B_\rho$ implies a welfare
loss bounded away from zero along every profit-approaching sequence.
\end{proposition}

\begin{proof}
With no prominence welfare is $\mu_C$. With universal prominence it
is $\mu_O$, because $\bar\mu<\mu_O$. For $0<T<1$, apply
Lemma~\ref{app:platform:limits}. If $a<\mu_O$, consumers choose the
outside option under prominence and the competitor otherwise,
giving welfare $\mu_C-T(\mu_C-\mu_O)$. The same expression holds
at $a=\mu_O$, because the conditional expected payoffs of brand
purchases and the outside option then coincide. If
$\mu_O<a\leq\mu_C$, sales equal $d=\tau$ and
\begin{equation}
 W^\infty(\tau)
 =\mu_C-\frac1J\sum_i(\mu_C-\mu_i)d_i.
 \label{app:platform:caps-welfare-identity}
\end{equation}
If $a>\mu_C$, consumers never take the outside option in the limit;
calibration makes the purchase-weighted brand payoff equal to
$\mu_C$, so welfare is again $\mu_C$. These formulas prove the
upper bound. They also show that the welfare-neutral rules are
exactly no prominence and the positive-frequency rules with
$a\geq\mu_C$. Their sales are precisely allocations implemented
by the second program. Consequently their maximal profit is
$B_\rho$, proving the characterization of a neutral maximizer.

Now consider a profit-approaching sequence at a fixed cap with
$P_\rho>0$. Eventually it excludes all zero-sales rules and all
rules with $a=\mu_O$: as in the preceding proof, the latter have
profit at most $A_\rho/2\leq P_\rho/2$. Every remaining rule is
either welfare-neutral with sales feasible for the second program,
or has $\mu_O<a\leq\mu_C$ and sales feasible for the first.
In any subsequence of the latter type, compactness gives a further
subsequence of sales converging to an allocation feasible for
\eqref{app:platform:caps-A} with profit $P_\rho$. The continuous
identity \eqref{app:platform:caps-welfare-identity} then bounds its
welfare limit below by $\mu_C-G_\rho$. Neutral subsequences satisfy
the same bound because $G_\rho\geq0$.

If the set defining $G_\rho$ is nonempty, choose an allocation
attaining that maximum and use the strict-constraint approximations
from Proposition~\ref{app:platform:caps}. Their profits tend to
$P_\rho$, and their welfare tends to $\mu_C-G_\rho$ by
\eqref{app:platform:caps-welfare-identity}. If the set is empty,
$A_\rho<P_\rho=B_\rho$, and a constant sequence of neutral
maximizers attains the bound instead. Apply
Lemma~\ref{app:platform:policy-reduction}(i) to each rule in the resulting
sequence: its frequency and coarse welfare are unchanged, its profit
weakly increases and still tends to $P_\rho$, and it randomizes in at
most two states. The same argument applies to a neutral exact maximizer.
This proves the stated sharpness and the if-and-only-if claim, using
also $W^\infty\leq\mu_C$.

Finally, if $A_\rho>B_\rho$, no optimizer of the first program
can have $\sum_i(\mu_C-\mu_i)d_i=0$: such an allocation would also
be feasible for the second program. Compactness of the first
program's optimal set therefore gives a strictly positive minimum
of this welfare gap. Rules from the second program cannot approach
$P_\rho$, and the preceding subsequence argument bounds every
profit-approaching sequence's welfare uniformly below $\mu_C$.
When $A_\rho=B_\rho$, a neutral exact maximizer exists, but the
condition $G_\rho=0$ remains necessary for preservation along all
approximately optimal sequences.
\end{proof}

\paragraph{A positive cap can preserve welfare.}
The characterization does not imply that every positive cap fails
in every environment. For example, take
\[
 J=2,\quad (\mu_h,\mu_\ell)=(3/2,0),\quad
 \mu_O=1,\quad\mu_C=6/5,\quad
 (m_h,m_\ell)=(1,1/100),\quad\rho=1/10.
\]
The first program requires $d_h+d_\ell\leq1/5$ and
$d_\ell\geq d_h/4$. For a given total $D$, its objective is at most
$[D-(99/100)(D/5)]/2$; equality is feasible at
$(d_h,d_\ell)=(4/25,1/25)$. Hence $A_\rho=401/5000$.
In the second program, calibration gives $d_\ell=d_h/4$,
and the low-state cap inequality gives $d_h\leq1/4$.
The upper endpoint is feasible, so $B_\rho=401/3200>A_\rho$.
It is implemented by $\tau=(1/5,0)$, with unpromoted purchase
probability $q=1/16$. Thus this positive cap admits a neutral
profit maximizer, and every profit-approaching sequence preserves
welfare. Whether a cap protects welfare depends on the payoff
and margin configuration, as well as its size.

\subsubsection{A cap that reduces but does not eliminate the welfare loss}
\label{app:platform:caps-example-section}

The two-program characterization permits an exact solution in a
three-state example. Unlike the main-text illustration, both the medium
and low products here yield less than the outside option. Let the three
states be equally likely, with
\begin{equation}
\begin{gathered}
 (\mu_h,\mu_m,\mu_\ell)=(1.202,0.8,0.2),\qquad
 \mu_O=1,\quad \mu_C=1.1,\\
 (m_h,m_m,m_\ell)=(0.001,0.001,1).
\end{gathered}
 \label{app:platform:caps-example-primitives}
\end{equation}
The platform faces the prominence-frequency constraint
$T(\tau)=(\tau_h+\tau_m+\tau_\ell)/3\leq\rho$.
All policies below are stationary. Limiting profits and welfare first
take the long-run limit at fixed sensitivity and then let sensitivity
tend to infinity, as in Lemma~\ref{app:platform:limits}.

\begin{proposition}[The optimal response to a prominence cap]
\label{app:platform:caps-example}
Under \eqref{app:platform:caps-example-primitives}, define
\[
 \rho_*:=\frac{167}{400},\qquad r:=\min\{\rho,\rho_*\}.
\]
For every $\rho\in[0,1]$, the supremum of limiting platform profit is
\begin{equation}
 V(\rho):=\sup_{\tau\in[0,1]^3:T(\tau)\leq\rho}
              \Pi^\infty(\tau)
 =\frac{169}{835}\,r.
 \label{app:platform:caps-example-value}
\end{equation}
For $\rho>0$, no fixed rule attains this supremum. Every feasible sequence
whose profits approach it converges to the same policy vector,
\begin{equation}
 \tau^*(\rho)=\left(\frac{400r}{167},\,0,\,
                            \frac{101r}{167}\right),
 \label{app:platform:caps-example-policy}
\end{equation}
and its consumer welfare converges to
\begin{equation}
 W^*(\rho)=\frac{11}{10}-\frac{r}{10}<\frac{11}{10}.
 \label{app:platform:caps-example-welfare}
\end{equation}
One implementing sequence is
$\tau^\epsilon=(\tau_h^*,0,(1-\epsilon)\tau_\ell^*)$
as $\epsilon\downarrow0$. At $\rho=0$, no prominence is the only feasible
rule; it yields zero platform profit and welfare $11/10$.

If $0<\rho<\rho_*$, every sequence approaching $V(\rho)$ eventually
randomizes in both $h$ and $\ell$. Restricting randomization to at most
one state therefore leaves a strictly positive profit gap.

For a fixed coarse share $\lambda\in(0,1]$, the population profit
supremum is
\[
 V_\lambda(\rho)=\lambda\frac{169r}{835}
                         +(1-\lambda)\frac1{3000}.
\]
For $\rho>0$, the same non-attainment, policy-limit and randomization
conclusions hold. Relative to population welfare without prominence,
\[
 W_\lambda^0=\lambda\frac{11}{10}+(1-\lambda)\frac{567}{500},
\]
every profit-approaching sequence has limiting welfare
\[
 W_\lambda^0-r\left[\frac{\lambda}{10}
                       +(1-\lambda)\frac{101}{5010}\right]
 <W_\lambda^0.
\]
At $\lambda=0$, all stationary rules give the same limiting profit and
the frequency tie-break selects no prominence.
\end{proposition}

\begin{proof}
Put $R:=3\rho$. Write $(x,y,z)$ for state-specific brand purchase
probabilities and
\[
 F(x,y,z):=\frac{x+y}{1000}+z
\]
for three times platform profit. We solve the two programs in
\eqref{app:platform:caps-A}--\eqref{app:platform:caps-B} and then check
fixed-rule implementation and all approximately optimal sequences.

\emph{Purchases confined to promoted encounters.}
The first program requires
\begin{equation}
 0\leq x,y,z\leq1,\qquad x+y+z\leq R,\qquad
 D:=\frac{101}{500}x-\frac15y-\frac45z\geq0,
 \label{app:platform:caps-example-A-constraints}
\end{equation}
as well as the upper promoted-payoff bound
$\frac{51}{500}x-\frac3{10}y-\frac9{10}z\leq0$.
For the upper bounds that follow, the latter restriction can be omitted.
Let
\[
 \alpha:=\frac{169}{835},\qquad
 \eta:=\frac{333}{334},\qquad
 \delta:=\alpha+\frac{\eta}{5}-\frac1{1000}>0.
\]
Direct substitution gives the identity
\[
 F(x,y,z)=\alpha(x+y+z)-\eta D-\delta y
             \leq\alpha R.
\]
Also, $D\geq0$ implies
$z\leq101x/400-y/4$, and therefore
\[
 F(x,y,z)\leq\frac{507}{2000}x-\frac{249}{1000}y
                 \leq\frac{507}{2000}.
\]
The smaller of these two bounds is attained in this linear program by
\[
 (x,y,z)=\left(\frac{400r}{167},0,\frac{101r}{167}\right).
\]
This vector has purchase-weighted payoff exactly $\mu_O=1$, so it also
satisfies the omitted upper bound. Consequently the first program's value is
\begin{equation}
 V_A(\rho)=\min\left\{\frac{169}{835}\rho,
                              \frac{169}{2000}\right\}
             =\frac{169}{835}r.
 \label{app:platform:caps-example-A-value}
\end{equation}
Its optimizer is unique. When $0<\rho\leq\rho_*$, equality in the first
bound requires $x+y+z=R$, $D=0$, and $y=0$, which determine the displayed
vector. When $\rho\geq\rho_*$, equality in the second bound requires
$x=1$, $y=0$, and $z=101/400$. These conditions agree at $\rho=\rho_*$.

\emph{Purchases also possible without prominence.}
In the second program, calibration at $\mu_C=1.1$ gives
\[
 \frac{51}{500}x=\frac3{10}y+\frac9{10}z,
 \qquad
 z=\frac{17}{150}x-\frac13y,
 \qquad
 F(x,y,z)=\frac{343x-997y}{3000}.
\]
Recover the common nonpromoted purchase probability
$q=\min\{x,y,z\}$ as in \eqref{app:platform:caps-recovery}.
Then $q\leq y$ and
$x+y+z\leq R+(3-R)q$. Since $0\leq R\leq3$, it follows that
\[
 \frac{167}{150}x
 \leq R+\left(\frac73-R\right)y.
\]
Substituting this inequality into the formula for $F$ yields
\[
 F\leq\frac{343R}{3340}
       -\frac{46449+51450R}{501000}\,y
       \leq\frac{343R}{3340}.
\]
The bound $x\leq1$ also gives $F\leq343/3000$. The smaller of these two bounds is attained by taking $q=0$ and
\[
 (x,y,z)=\left(u,0,\frac{17u}{150}\right),\qquad
 u:=\min\left\{\frac{450\rho}{167},1\right\}.
\]
For $\rho>0$, the corresponding prominence rule has promoted mean
exactly $\mu_C$, and Lemma~\ref{app:platform:limits} implements these
purchase probabilities. Thus the second program has value
\begin{equation}
 V_B(\rho)=\min\left\{\frac{343}{3340}\rho,
                               \frac{343}{9000}\right\}
       =\frac{343}{3340}\min\left\{\rho,\frac{167}{450}\right\}.
 \label{app:platform:caps-example-B-value}
\end{equation}
Because $343/3340<169/835$ and $167/450<167/400$,
$V_B(\rho)<V_A(\rho)$ whenever $\rho>0$.

\emph{Implementation and non-attainment.}
Fix $\rho>0$. Reducing the low component of
\eqref{app:platform:caps-example-policy} gives the feasible rule
$\tau^\epsilon=(\tau_h^*,0,(1-\epsilon)\tau_\ell^*)$.
For $0<\epsilon<167/303$, its promoted mean is strictly between
$\mu_O$ and $\mu_C$. Lemma~\ref{app:platform:limits} therefore gives
limiting sales $d=\tau^\epsilon$, and its profit tends to $V_A(\rho)$.
This proves \eqref{app:platform:caps-example-value}, either directly
or by the general reduction \eqref{app:platform:caps-value}.

For completeness, a fixed rule with promoted mean below $\mu_O$ has
zero limiting sales. A fixed rule with promoted mean exactly $\mu_O$
has limiting sales $\tau/2$; its vector $\tau$ is feasible for the
first program, so its profit is at most $V_A(\rho)/2$. A fixed rule
with promoted mean at least $\mu_C$ has profit at most $V_B(\rho)$.
The endpoint rules with $T=0$ or $T=1$ have zero brand sales.
Any rule attaining $V_A(\rho)$ would consequently have promoted mean
strictly between $\mu_O$ and $\mu_C$, with sales equal to $\tau$.
Uniqueness in the first program would force
$\tau=\tau^*(\rho)$, whose promoted mean is exactly $\mu_O$, a
contradiction. The supremum is therefore unattained.

\emph{Every approximately optimal sequence.}
Let $\tau^n$ be feasible with $\Pi^\infty(\tau^n)\to V_A(\rho)$.
The strict bounds
\[
 \max\{V_B(\rho),V_A(\rho)/2,0\}<V_A(\rho)
\]
show that all sufficiently large $n$ have promoted means strictly between
$\mu_O$ and $\mu_C$. Hence their sales equal $\tau^n$, which are feasible
vectors in the first program. Compactness of that feasible set and
uniqueness of its optimizer imply $\tau^n\to\tau^*(\rho)$: otherwise
a subsequence bounded away from the optimizer would have a convergent
subsequence whose limit is a different optimizer.

For these policies consumers buy the brand under prominence and the
competitor otherwise. Consequently
\[
 W^\infty(\tau^n)
 =\frac13\sum_{i\in\{h,m,\ell\}}
       \left[\mu_i\tau_i^n+\mu_C(1-\tau_i^n)\right].
\]
At the limiting vector, $\sum_i\mu_i\tau_i^*=\mu_O\sum_i\tau_i^*$
and $\sum_i\tau_i^*/3=r$. Taking limits gives
$W^*(\rho)=\mu_C-(\mu_C-\mu_O)r$, proving
\eqref{app:platform:caps-example-welfare}. At $\rho=0$ the only feasible
rule is $\tau=0$; since the unconditional brand mean is below $\mu_O$,
consumers choose the competitor in the high-sensitivity limit.

\emph{Sharpness of the two-state bound.}
Fix $0<\rho<\rho_*$. Both $\tau_h^*(\rho)=400\rho/167$ and
$\tau_\ell^*(\rho)=101\rho/167$ lie strictly between zero and one.
The convergence $\tau^n\to\tau^*(\rho)$ just proved therefore forces
both coordinates of every profit-approaching sequence to be interior
for all sufficiently large $n$. Consequently,
\[
 \sup_{\substack{\tau\in[0,1]^3,\ T(\tau)\leq\rho\\
                  |\{i:0<\tau_i<1\}|\leq1}}
       \Pi^\infty(\tau)<V(\rho).
\]
If equality held, a sequence satisfying the restriction would approach
$V(\rho)$, contradicting that eventual interiority. The displayed
implementing sequence randomizes only in $h$ and $\ell$, establishing
sharpness of the bound in Proposition~\ref{app:platform:caps}.

\emph{Heterogeneous consumers.}
Finest-partition consumers buy $h$ with or without prominence and never
buy $m$ or $\ell$ in the high-sensitivity limit. Thus, for every rule,
\[
 \Pi_{\mathrm f}^\infty(\tau)=\frac1{3000},\qquad
 W_{\mathrm f}^\infty(\tau)
 =\frac{\mu_h+2\mu_C}{3}
       -\frac{\mu_C-\mu_O}{3}(\tau_m+\tau_\ell)
 =\frac{567}{500}-\frac{\tau_m+\tau_\ell}{30}.
\]
For fixed $\lambda>0$, population profit approaches its supremum if and
only if coarse-consumer profit approaches $V(\rho)$. The policy limit,
non-attainment and positive profit gap under one-state randomization
therefore carry over directly; the latter gap is multiplied by
$\lambda$. Substituting \eqref{app:platform:caps-example-policy} into
the fine welfare formula gives a loss of $101r/5010$. Combining it with
the coarse welfare loss $r/10$ proves the population formula. When
$\lambda=0$, profit is constant across rules and no prominence uniquely
minimizes frequency.
\end{proof}

For example, at $\rho=1/3$ the limiting policy and the corresponding
profit and welfare limits are
\[
 \tau^*=\left(\frac{400}{501},0,\frac{101}{501}\right),\qquad
 V(1/3)=\frac{169}{2505}\simeq0.067465,
 \qquad W^*(1/3)=\frac{16}{15}.
\]
The feasible rule $(0.8,0,0.2)$ yields profit approximately $0.066933$
and welfare $1.0672$, close to these limits. For caps below $\rho_*$, tightening
the cap scales down the profitable combination of high- and low-quality
prominence without changing its proportions. Welfare loss falls with
the cap but remains positive at every positive cap. Above $\rho_*$,
the cap does not reduce the profit supremum or its limiting welfare loss.

\subsection{General policy reduction and bounds on randomization}
\label{app:platform:policy-reduction-section}

The preceding randomization bounds follow from the number of policy
moments that determine coarse feedback. Put
\[
 U(\tau):=\frac1J\sum_i\mu_i\tau_i.
\]
For a coarse share $\lambda\in[0,1]$, let
$\Pi_{\lambda,\beta}$ denote long-run population profit at finite
sensitivity, and write $\Pi_{\lambda,\infty}:=\Pi_\lambda^\infty$ for
its fixed-rule high-sensitivity limit. At finite sensitivity the fine
type's conditional purchase probabilities are
\[
 a_{i,\beta}^{\mathrm f}=A_\beta(\mu_i),\qquad
 b_{i,\beta}^{\mathrm f}=B_\beta(\mu_i);
\]
at $\beta=\infty$ use
\eqref{app:platform:heterogeneous-fine-rates}. Thus fine-consumer profit
is affine in the policy at every sensitivity considered here.

\begin{lemma}[Reduction of stationary prominence policies]
\label{app:platform:policy-reduction}
Maintain \eqref{app:platform:maintained}. Fix a consumer share
$\lambda\in[0,1]$ and an evaluation sensitivity
$\beta\in[0,\infty]$. For every stationary rule $\tau$:
\begin{enumerate}[label=(\roman*),leftmargin=*]
\item There is a stationary rule $\widehat\tau$ randomizing in at most
two states such that
\[
 T(\widehat\tau)=T(\tau),\qquad U(\widehat\tau)=U(\tau),\qquad
 \Pi_{\lambda,\beta}(\widehat\tau)\geq\Pi_{\lambda,\beta}(\tau).
\]
The replacement preserves coarse consumers' deterministic mean dynamics
at every finite sensitivity, and their equilibrium valuations, aggregate
choice shares and welfare at every finite sensitivity and in the
high-sensitivity limit. In the all-coarse case, one replacement can be
chosen to weakly raise profit simultaneously at every sensitivity.
\item If prominence frequency need not be preserved, there is a
stationary rule $\widetilde\tau$ randomizing in at most one state with
$\Pi_{\lambda,\beta}(\widetilde\tau)
\geq\Pi_{\lambda,\beta}(\tau)$.
\end{enumerate}
\end{lemma}
\begin{proof}
We use an elementary fact about a box intersected with affine equalities.
If the equalities have rank at most $r$, each vertex has at most $r$
coordinates in $(0,1)$. Otherwise a nonzero perturbation supported on
the interior coordinates lies in the null space of the equalities;
small positive and negative multiples remain feasible, contradicting
extremality. A linear objective on a nonempty compact polytope attains
its maximum at a vertex.

\emph{(i) Preserving the two moments.}
Consider the nonempty compact polytope
\begin{equation}
 \mathcal P(\tau):=\left\{z\in[0,1]^J:
       \sum_i z_i=\sum_i\tau_i,\quad
       \sum_i\mu_i z_i=\sum_i\mu_i\tau_i\right\}.
 \label{app:platform:moment-polytope}
\end{equation}
For $0<T<1$, all its policies have the same menu probabilities and
conditional brand means $a=U/T$ and $b=(\bar\mu-U)/(1-T)$.
Consequently their reduced mean dynamics coincide. At $T=0$ or $T=1$
the polytope is a singleton and the claim is immediate.
Uniqueness in Proposition~\ref{app:platform:convergence} gives the same
coarse equilibrium valuation $v_\beta$ for all $z\in\mathcal P(\tau)$.
At the evaluation sensitivity write $A,B$ for the corresponding coarse
conditional purchase probabilities, using their limits when
$\beta=\infty$. Population profit on the polytope is
\begin{equation}
 \Pi_{\lambda,\beta}(z)
 =\frac1J\sum_i m_i\left\{
   \lambda[B+(A-B)z_i]
   +(1-\lambda)
       [b_{i,\beta}^{\mathrm f}
        +(a_{i,\beta}^{\mathrm f}-b_{i,\beta}^{\mathrm f})z_i]
                         \right\}.
 \label{app:platform:reduction-objective}
\end{equation}
It is affine in $z$, so an optimal vertex gives the profit inequality
with at most two interior coordinates.

The total coarse brand share is $TA+(1-T)B$ and its contribution to
welfare is $AU+B(\bar\mu-U)$. The competitor and outside-option shares
depend only on $T$ and the common equilibrium valuations, by
\eqref{app:platform:shares}. All aggregate shares and coarse welfare
are therefore unchanged. This argument holds at every finite sensitivity;
Lemma~\ref{app:platform:limits} gives the same conclusion in the limit.
If $\lambda=1$, maximize $\sum_i m_i z_i$ on
\eqref{app:platform:moment-polytope}. Since $A_\beta-B_\beta>0$ at
finite sensitivity and its limit is nonnegative, a single optimal
vertex weakly improves profit at all sensitivities.

\emph{(ii) One equality without a frequency restriction.}
First let $\beta<\infty$, fix the original equilibrium brand valuation
$v$, and put $A=A_\beta(v)$ and $B=B_\beta(v)$. On the nonempty polytope
\begin{equation}
 \left\{z\in[0,1]^J:
    \sum_i(\mu_i-v)[B+(A-B)z_i]=0\right\},
 \label{app:platform:single-calibration}
\end{equation}
the same $v$ satisfies the brand calibration equation. Its sampling
denominator is positive. Uniqueness makes $v$ the equilibrium valuation
under every feasible rule, including $T=0,1$ when these are feasible.
Formula~\eqref{app:platform:reduction-objective} is again the actual
population profit. Maximizing it at a vertex gives at most one interior
coordinate.

Now let $\beta=\infty$. Endpoint rules with $T=0$ or $T=1$ are already
deterministic, so retain them. For $0<T<1$ and promoted mean
$a\leq\mu_C$, define
\[
 k(a)=\begin{cases}
       0,&a<\mu_O,\\
       1/2,&a=\mu_O,\\
       1,&\mu_O<a\leq\mu_C.
      \end{cases}
\]
On the polytope $\{z\in[0,1]^J:\sum_i(\mu_i-a)z_i=0\}$, every
non-endpoint rule has the same promoted mean and coarse sales $k(a)z_i$.
This sales formula also holds at $z=0$. The other endpoint $z=\mathbf1$
is feasible only if $a=\bar\mu<\mu_O$, when it too has zero coarse
sales. Hence population profit is affine on the entire polytope, and
an optimal vertex has at most one interior coordinate. In particular,
the boundary probability $1/2$ at $a=\mu_O$ is preserved whenever the
replacement uses prominence.

Finally, if $a>\mu_C$, put $H=\sum_i\tau_i(\mu_i-\mu_C)>0$ and impose
the single equality $\sum_i z_i(\mu_i-\mu_C)=H$. Neither endpoint is
feasible. Every feasible policy has promoted mean above $\mu_C$ and
the same non-promoted brand-purchase probability
\[
 q=\frac{H}{J(\mu_C-\bar\mu)+H}\in(0,1).
\]
Its coarse sales are $q+(1-q)z_i$, so population profit is affine once
more. An optimal vertex completes the proof.
\end{proof}

Part (i) preserves expected feedback and the resulting mean dynamics,
not necessarily the reward distributions or stochastic learning paths.
It does not generally preserve finest-partition or population welfare.
The profit-improving replacement may depend on consumer composition and
the evaluation sensitivity, except for the simultaneous all-coarse
conclusion. Part (ii) imposes no preservation of frequency or welfare.

\begin{corollary}[Randomization bounds for optimal policies]
\label{app:platform:sparse-optima}
For each fixed consumer share and finite sensitivity, a profit-maximizing
stationary rule exists. One can be chosen to randomize in at most one
state without a cap, and at most two states under any cap. In the
high-sensitivity limit the same bounds hold for attained optima and for
sequences approaching a profit supremum. If all consumers use the finest
partition, one randomized state suffices even under a cap. In the
high-sensitivity limit, the capped two-state bound is sharp for every
positive coarse share.
\end{corollary}
\begin{proof}
At finite sensitivity, the equilibrium brand valuation is continuous in
the policy. Indeed, each valuation lies in the compact interval
$[\min_i\mu_i,\max_i\mu_i]$. For any convergent policy sequence, every
subsequential valuation limit solves the limiting policy's continuous
calibration equation, whose sampling denominator is positive and whose
solution is unique. Thus all subsequential limits coincide with that
policy's equilibrium valuation, including at $T=0,1$. Profit is therefore
continuous on the policy cube. The cube and each capped subset are
compact, so maxima exist. Apply part (ii) of the lemma without a cap and
part (i) with a cap. At high sensitivity, apply the same parts to an
attained optimizer or separately to each member of a profit-approaching
sequence; no exchange of limits is involved.

For a pure finest-partition population, profit is affine in the policy.
An optimum over $[0,1]^J$ with $\sum_i\tau_i\leq J\rho$ has at most one
interior coordinate at a vertex: if the cap binds, perturb two interior
coordinates in opposite directions; if it is slack, perturb one. This
holds at finite sensitivity and in the high-sensitivity limit. If an
attained optimizer is selected by the frequency tie-break, start from
that rule and, where frequency is not already fixed, minimize it on the
relevant optimal face before choosing a vertex. The replacement cannot
have a lower frequency than the selected optimizer and therefore
preserves the tie-break. Finally,
Proposition~\ref{app:platform:caps-example} proves high-sensitivity
sharpness for every fixed positive coarse share.
\end{proof}

\subsection{Optimal costly stationary prominence}
\label{app:platform:costly-section}

This subsection retains the all-coarse, uncapped specification and
\eqref{app:platform:maintained}, but charges $c\geq0$ whenever prominence
occurs, irrespective of the ensuing purchase. The platform maximizes
\[
 N_c(\tau):=\Pi^\infty(\tau)-cT(\tau).
\]
Consumer rewards and choice
menus are unchanged, so Proposition~\ref{app:platform:convergence} and
Lemma~\ref{app:platform:limits} continue to apply to each fixed rule.
The realized prominence frequency converges almost surely to $T(\tau)$
by the strong law for the independent state and prominence draws.
Consequently subtracting $cT(\tau)$ also gives the limit of realized
average net profit, with the same order of limits as above. If a
tie-break is needed, lower prominence frequency is preferred among
rules with equal net profit.

For $d\in[0,1]^J$, put
\[
 L_O(d):=\sum_i(\mu_i-\mu_O)d_i,\qquad
 L_C(d):=\sum_i(\mu_i-\mu_C)d_i,\qquad
 M(d):=\frac1J\sum_i m_i d_i,
\]
and define
\begin{equation}
 \mathcal D:=\{d\in[0,1]^J:L_O(d)\geq0,\ L_C(d)\leq0\},
 \qquad \mathcal E:=\{d\in\mathcal D:L_C(d)=0\}.
 \label{app:platform:costly-allocations}
\end{equation}
Writing $\bar d=J^{-1}\sum_i d_i$ and $\underline d=\min_i d_i$, set
\begin{equation}
 T_{\min}(d):=
 \begin{cases}
  \bar d,&L_C(d)<0,\\[5pt]
  \displaystyle\frac{\bar d-\underline d}{1-\underline d},&L_C(d)=0.
 \end{cases}
 \label{app:platform:costly-frequency}
\end{equation}
The denominator is positive: $\mathbf1 \notin \mathcal D$, since
$L_O(\mathbf1)=J(\bar\mu-\mu_O)<0$.

\begin{lemma}[Feasible sales and least prominence frequency]
\label{app:platform:costly-implementation}
The closure of the limiting sales vectors of fixed stationary rules is
$\mathcal D$. For each $d\in\mathcal D$,
\begin{equation}
 T_{\min}(d)=
 \inf_{\{\tau^n:\ d(\tau^n)\to d\}}
       \liminf_{n\to\infty}T(\tau^n),
 \label{app:platform:costly-minimum}
\end{equation}
where $d(\tau)$ denotes the fixed-rule high-sensitivity sales vector.
The infimum is achieved by a constant sequence if $d=0$ or $L_O(d)>0$.
In the latter case a least-frequency implementation is
\begin{equation}
 \tau_i=\frac{d_i-q}{1-q},\qquad
 q=\begin{cases}0,&L_C(d)<0,\\ \underline d,&L_C(d)=0.
 \end{cases}
 \label{app:platform:costly-implementation-rule}
\end{equation}
For nonzero $d$ with $L_O(d)=0$, the infimum is approached by rules
whose promoted means lie strictly between $\mu_O$ and $\mu_C$.
\end{lemma}
\begin{proof}
For a rule with zero limiting sales, $d=0\in\mathcal D$ and
$T_{\min}(d)=0\leq T(\tau)$. Otherwise $0<T<1$, and
Lemma~\ref{app:platform:limits} gives three possibilities. If
$a=\mu_O$, then $d=\tau/2$, so $L_O(d)=0$, $L_C(d)<0$, and
$T_{\min}(d)=\bar d=T/2\leq T$. If
$\mu_O<a\leq\mu_C$, then $d=\tau$ and $d\in\mathcal D$;
for $a<\mu_C$ one has $T_{\min}(d)=T$. If $a>\mu_C$, then
$d_i=q+(1-q)\tau_i$ for $q=q_\tau\in(0,1)$, and calibration
gives $L_C(d)=0$. Both cases with $L_C(d)=0$ have the same
representation, with $q=0$ when $a=\mu_C$. Necessarily
$0\leq q\leq\underline d$ and
\[
 T=\frac{\bar d-q}{1-q}\geq
       \frac{\bar d-\underline d}{1-\underline d}=T_{\min}(d),
\]
because the displayed fraction is decreasing in $q$ when $\bar d<1$.
Thus every actual sales vector belongs to $\mathcal D$ and uses at
least $T_{\min}(d)$ prominence.

Conversely, implement $d=0$ by no prominence. If $L_O(d)>0$ and
$L_C(d)<0$, take $\tau=d$. It has $0<T<1$, and its promoted mean
lies strictly between $\mu_O$ and $\mu_C$. Hence its sales are $d$
and its frequency is $\bar d$. If $d\in\mathcal E\setminus\{0\}$,
put $q=\underline d$ and use
\eqref{app:platform:costly-implementation-rule}. A positive constant
vector cannot belong to $\mathcal E$, because $\bar\mu<\mu_C$.
Thus this policy has $0<T<1$. Writing
$\Delta:=J(\mu_C-\bar\mu)>0$, one obtains
\[
 \sum_i(\mu_i-\mu_C)\tau_i=\frac{q\Delta}{1-q}.
\]
If $q=0$, its promoted mean is $\mu_C$ and its sales are $\tau=d$.
If $q>0$, its promoted mean exceeds $\mu_C$, and
\eqref{app:platform:q-general} gives competitive purchase probability
$q$. Its sales are again $d$, with frequency $T_{\min}(d)$.
Notice that $L_O(d)=(\mu_C-\mu_O)\sum_i d_i>0$ on
$\mathcal E\setminus\{0\}$.

It remains to approximate a nonzero $d$ with $L_O(d)=0$.
Choose a state $h$ with $\mu_h>\mu_O$ and put
$d^\varepsilon=(1-\varepsilon)d+\varepsilon e_h$, where $e_h$
is its coordinate unit vector. Then
\[
 L_O(d^\varepsilon)=\varepsilon(\mu_h-\mu_O)>0,
 \qquad L_C(d)=-(\mu_C-\mu_O)\sum_i d_i<0.
\]
For all sufficiently small positive $\varepsilon$,
$L_C(d^\varepsilon)<0$. The policies $\tau^\varepsilon=d^\varepsilon$
therefore generate sales $d^\varepsilon\to d$, with
$T(\tau^\varepsilon)\to\bar d=T_{\min}(d)$. This proves that
$\mathcal D$ is exactly the closure of feasible sales and constructs
sequences attaining the right side of
\eqref{app:platform:costly-minimum}.

For its reverse inequality, $T_{\min}$ is lower semicontinuous on
$\mathcal D$. Indeed, $\underline d$ is bounded away from one on
this compact set. The expression on $\mathcal E$ is continuous and
does not exceed $\bar d$, while outside $\mathcal E$ the function
equals the continuous function $\bar d$. At $\mathcal E$ the only
possible jump is therefore downward. For any sequence with
$d(\tau^n)\to d$, the already established frequency bound gives
\[
 \liminf_nT(\tau^n)\geq
 \liminf_nT_{\min}(d(\tau^n))\geq T_{\min}(d).
\]
This proves \eqref{app:platform:costly-minimum}.
\end{proof}

\begin{proposition}[Optimal costly stationary prominence]
\label{app:platform:costly}
Under \eqref{app:platform:maintained}, with $m_i\geq0$ and $c\geq0$,
\begin{equation}
 V_c:=\sup_{\tau\in[0,1]^J}N_c(\tau)
 =\max_{d\in\mathcal D}\{M(d)-cT_{\min}(d)\}.
 \label{app:platform:costly-value}
\end{equation}
If $V_c>0$, $V_c$ is attained by a stationary rule if and only if some
maximizing sales allocation has $L_O(d)>0$. Each such allocation
is implemented at value $V_c$ by
\eqref{app:platform:costly-implementation-rule}. If $V_c=0$,
no prominence is optimal. An optimal rule, when one exists, or a
sequence approaching $V_c$ can be chosen to randomize in at most
one state.
\end{proposition}
\begin{proof}
The set $\mathcal D$ is nonempty and compact, and
$M-cT_{\min}$ is upper semicontinuous by
Lemma~\ref{app:platform:costly-implementation} and $c\geq0$.
It therefore has a maximum, denoted $F_c^*\geq0$.
Every rule satisfies
\[
 N_c(\tau)=M(d(\tau))-cT(\tau)
 \leq M(d(\tau))-cT_{\min}(d(\tau))\leq F_c^*.
\]
For any maximizing allocation, the implementation or approximation
constructed in that lemma has sales converging to $d$ and frequency
converging to $T_{\min}(d)$. Its net profit converges to $F_c^*$.
This proves \eqref{app:platform:costly-value}. The approximation is
over fixed policies after their respective long-run and high-sensitivity
limits; no interchange of limits is used.

\emph{Attainment.} A maximizing $d$ with $L_O(d)>0$ has an exact
least-frequency implementation, which attains $V_c$. If $V_c=0$,
the zero policy attains it. Conversely, suppose a rule attains a
positive $V_c$. It has positive sales and $0<T<1$, so its promoted
mean is at least $\mu_O$. It cannot have $a=\mu_O$: in that case
its actual sales are $\tau/2$, whereas the vector $\tau$ itself
belongs to $\mathcal D$, satisfies $L_C(\tau)<0$, and obeys
\[
 M(\tau)-cT_{\min}(\tau)
 =M(\tau)-cT(\tau)
 =N_c(\tau)+\frac12M(\tau)>V_c.
\]
The last inequality follows from $N_c(\tau)>0$, which implies
$M(\tau)>0$, and contradicts \eqref{app:platform:costly-value}.
Thus $a>\mu_O$ and its actual sales $d$ have $L_O(d)>0$.
The upper-bound inequalities above must both be equalities, so $d$
is a maximizing allocation. This proves the asserted equivalence,
including cases with zero margins or zero cost. In particular,
the exact half-purchase rule at $a=\mu_O$ cannot attain a positive
supremum, even though boundary allocations may approximate that value.

\emph{One-state randomization.} We show more generally that every
fixed rule admits a rule randomizing in at most one state with weakly
higher net profit. A rule with zero sales has $N_c\leq0$ and can be
replaced by no prominence. For $0<T<1$ and
$\mu_O\leq a\leq\mu_C$, put $k=1/2$ if $a=\mu_O$ and $k=1$
otherwise. On the nonempty compact polytope
\[
 \mathcal P_a:=\left\{z\in[0,1]^J:
                   \sum_i(\mu_i-a)z_i=0\right\},
\]
every nonzero rule has the same promoted mean $a$, and hence sales
$kz$. The rule $z=\mathbf1$ is infeasible since $a\geq\mu_O>\bar\mu$;
the zero rule also has sales $kz=0$. Thus actual net profit throughout
$\mathcal P_a$ is the linear function
$J^{-1}\sum_i(km_i-c)z_i$. It attains its maximum at a vertex.

If instead $a>\mu_C$, hold fixed
$H:=\sum_i(\mu_i-\mu_C)\tau_i>0$ and maximize over
\[
 \mathcal P_H:=\left\{z\in[0,1]^J:
                  \sum_i(\mu_i-\mu_C)z_i=H\right\}.
\]
Neither endpoint rule belongs to this polytope. Every feasible rule
has promoted mean above $\mu_C$ and the same competitive purchase
probability $q=H/(\Delta+H)$, where
$\Delta=J(\mu_C-\bar\mu)>0$. Its actual net profit is affine:
\[
 \frac{q}{J}\sum_i m_i+
 \frac1J\sum_i\bigl((1-q)m_i-c\bigr)z_i.
\]
Again choose an optimal vertex. A box intersected with one affine
equality has at most one coordinate in $(0,1)$ at each vertex:
two interior coordinates would permit a nonzero small two-sided
perturbation preserving the equality, contradicting extremality.
This is the same reduction as in the proof of
Lemma~\ref{app:platform:policy-reduction}, with the affine cost
subtracted from the objective. Apply it to an attained optimizer or
separately to each rule in a sequence approaching $V_c$.
If an attained optimizer is selected by the frequency tie-break,
minimize frequency on the relevant optimal face before choosing a
vertex; the same bound holds for a selected optimizer.
\end{proof}

\paragraph{Competitive mixing, deterministic implementation, and welfare.}
For any nonzero $d\in\mathcal E$, the least-frequency implementation
has $q=\underline d$. If $q>0$ and every $d_i$ is either $q$ or one,
\eqref{app:platform:costly-implementation-rule} is deterministic:
it promotes exactly the states with $d_i=1$. If such a $d$ maximizes
\eqref{app:platform:costly-value}, this deterministic policy is optimal
even among unrestricted stationary randomized rules. More generally,
every rule sustaining valuation $\mu_C$ has no limiting outside-option
purchases, and its brand sales satisfy $L_C(d)=0$. Consequently
\begin{equation}
 W^\infty=\frac1J\sum_i\mu_i d_i+
          \mu_C(1-\bar d)=\mu_C.
 \label{app:platform:costly-welfare}
\end{equation}
This equals no-prominence welfare under
\eqref{app:platform:maintained}. Costly prominence can therefore
select a welfare-preserving policy through competitive mixing. Neither
deterministic optimality nor monotonicity of optimal consumer welfare
in $c$ is asserted for general primitives.

\subsubsection{The two-state costly-prominence example}
\label{app:platform:costly-example-section}

\begin{proposition}[The exact cost range for high-only prominence]
\label{app:platform:costly-example}
Suppose the two states are equally likely and
\begin{equation}
 (\mu_h,\mu_\ell)=(4,0),\qquad
 \mu_O=\frac52,\qquad \mu_C=3,\qquad m_h=m_\ell=1.
 \label{app:platform:costly-example-primitives}
\end{equation}
High-only prominence is optimal among all stationary policies if and
only if $c\in[4/9,8/9]$. Throughout this closed interval it is the
unique optimal policy; it induces sales $(1,1/3)$, competitive purchase
probability $1/3$, net profit $2/3-c/2$, and consumer welfare $3$.
At $c=2/3$ its net profit is $1/3$. The policy $(1,1/3)$ generates
the same sales but net profit $2/9$. The costless-optimal sequence
$(1,t)$, with $t\uparrow3/5$, has net profit tending to $4/15$ at
$c=2/3$ and consumer welfare tending to $13/5$.
\end{proposition}
\begin{proof}
Write a rule as $\tau=(x,y)\in[0,1]^2$ and put
\[
 N_c(x,y)=\Pi^\infty(x,y)-\frac c2(x+y).
\]
No prominence yields zero net profit, whereas universal prominence
yields $-c$, since $\bar\mu=2<\mu_O$. For every other nonzero
rule the promoted mean is $a=4x/(x+y)$. We apply
Lemma~\ref{app:platform:limits} to each purchase regime.

First let $a>\mu_C$, or equivalently $x>3y$. The competitive
purchase probability and state-specific sales are
\[
 q=\frac{x-3y}{2+x-3y},\qquad
 d_h=\frac{3(x-y)}{2+x-3y},\qquad
 d_\ell=\frac{x-y}{2+x-3y}.
\]
Set $z=(x-y)/(1-y)\in(0,1]$. The rule $(z,0)$ generates exactly
the same sales. If $y>0$, it uses strictly less prominence, because
\[
 x+y-z=\frac{y(2-x-y)}{1-y}>0.
\]
For $c>0$, any competitive-mixing rule with $y>0$ is therefore
strictly dominated by $(z,0)$. The latter rule has
\begin{align}
 N_c(z,0)&=\frac{2z}{2+z}-\frac c2z,\notag\\
 \frac{d}{dz}N_c(z,0)&=\frac4{(2+z)^2}-\frac c2,
 &\frac{d^2}{dz^2}N_c(z,0)&=-\frac8{(2+z)^3}<0.
 \label{app:platform:costly-example-mixing}
\end{align}
Thus $(1,0)$ uniquely maximizes net profit within the competitive-mixing
regime for $0<c\leq8/9$, with value $2/3-c/2$.

Next let $\mu_O<a\leq\mu_C$, or equivalently
$x/3\leq y<3x/5$. Purchases occur only under prominence, and for
$c<1$,
\begin{equation}
 N_c(x,y)=\frac{1-c}{2}(x+y)<\frac45(1-c).
 \label{app:platform:costly-example-direct}
\end{equation}
The bound is approached by $(1,t)$ with $t\uparrow3/5$.
At $a=\mu_O$, one has $y=3x/5$ and $T=(x+y)/2\leq4/5$.
The promoted purchase probability is $1/2$, so
\[
 N_c(x,y)=(1/2-c)T<(1-c)T\leq\frac45(1-c)
 \qquad(c<1).
\]
If $a<\mu_O$, sales vanish and net profit is nonpositive.

For $c\in[4/9,8/9]$, high-only net profit is positive and
\begin{equation}
 N_c(1,0)-\frac45(1-c)=\frac{9c-4}{30}\geq0.
 \label{app:platform:costly-example-threshold}
\end{equation}
The strict bounds in the nonmixing regimes and uniqueness in the
competitive-mixing regime prove global uniqueness. At $c=4/9$, the
competing supremum in \eqref{app:platform:costly-example-direct}
is not attained. At $c=8/9$, the derivative in
\eqref{app:platform:costly-example-mixing} is zero only at $z=1$
and strictly positive below it. Hence both endpoints belong to the
range of unique optimality.

Conversely, if $0\leq c<4/9$, the sequence $(1,t)$ with
$t\uparrow3/5$ has net profit tending to a value strictly above
$N_c(1,0)$, and some member strictly improves on high-only prominence.
If $c>8/9$, the derivative of $N_c(z,0)$ at $z=1$ is negative,
so reducing high-state prominence slightly strictly raises net profit.
This proves the if-and-only-if cost range.

For $(1,0)$, $q=1/3$, sales are $(1,1/3)$, gross profit is $2/3$,
and $T=1/2$. At $c=2/3$, net profit is therefore $1/3$.
The rule $(1,1/3)$ has promoted mean $3$, the same sales, and
$T=2/3$, giving net profit $2/9$. Under either rule welfare is
\[
 W^\infty=\frac12\,4+
          \frac12\left(\frac13\,0+\frac23\,3\right)=3,
\]
also following from \eqref{app:platform:costly-welfare}. Without
prominence, the brand valuation is $2<\mu_O$, so all limiting
purchases go to the competitor and welfare is also $3$.

Finally, for $1/3<t<3/5$, the rule $(1,t)$ has promoted mean
strictly between $\mu_O$ and $\mu_C$. Its gross profit and prominence
frequency both equal $(1+t)/2$, while
\[
 W^\infty(1,t)=2+\frac32(1-t).
\]
At $c=2/3$, its net profit tends to $4/15$ and its welfare to
$13/5$. At $c=0$, its gross profit tends to the global supremum
$4/5$: direct-purchase rules have profit below $4/5$, competitive-mixing
rules have profit at most $2/3$ by the same-sales replacement and
\eqref{app:platform:costly-example-mixing}, exact outside-option-boundary
rules have profit at most $2/5$, and all other rules have zero sales.
\end{proof}

\subsubsection{Costly prominence in the three-state illustration}
\label{app:platform:costly-three-section}

The next result revisits the three-state illustration with equal margins.
Costs can change which products receive prominence, while making a
deterministic rule uniquely optimal even among randomized rules.

\begin{proposition}[High-only prominence with three states]
\label{app:platform:costly-three}
Suppose \eqref{app:platform:maintained} and
\eqref{app:platform:three-assumptions} hold, and
$m_h=m_m=m_\ell=m>0$. Let $q_3$ be given by
\eqref{app:platform:three-probabilities}. Then $0<q_3<1/2$, and
high-only prominence $(1,0,0)$ uniquely maximizes $N_c$ over all
stationary rules $\tau\in[0,1]^3$ whenever
\begin{equation}
 m\leq c\leq\overline c_3:=m(1-q_3)(1+2q_3).
 \label{app:platform:costly-three-range}
\end{equation}
Throughout this closed interval, the optimal rule induces sales
$(1,q_3,q_3)$, brand valuation $\mu_C$, and
\begin{equation}
 N_c(1,0,0)=\frac{m(1+2q_3)-c}{3},
 \qquad W^\infty(1,0,0)=\mu_C.
 \label{app:platform:costly-three-values}
\end{equation}
If $c>\overline c_3$, high-only prominence is not optimal.
\end{proposition}

The interval in \eqref{app:platform:costly-three-range} has positive
length and includes $c=m$, at which optimal net profit is $2mq_3/3>0$.
Its lower bound is sufficient, not asserted to be necessary. In contrast,
when prominence is costless, Proposition~\ref{app:platform:three} selects
high-and-medium prominence uniquely among deterministic rules, since
$q_3<1/2$. Thus costs can make concentrating prominence on the best
product optimal even when all states have the same margin.

\begin{proof}
Write
\[
 a_h:=\mu_h-\mu_C,\qquad
 a_m:=\mu_C-\mu_m,\qquad
 a_\ell:=\mu_C-\mu_\ell,\qquad
 \Delta:=a_m+a_\ell-a_h=3(\mu_C-\bar\mu)>0.
\]
The three-state assumptions imply $0<a_h\leq a_m<a_\ell$, so
$q_3=a_h/(a_m+a_\ell)<1/2$. In particular,
$\overline c_3-m=mq_3(1-2q_3)>0$.

Fix $c\geq m$. For a stationary rule $\tau$, let
$S:=\tau_h+\tau_m+\tau_\ell=3T(\tau)$ and
\[
 H:=a_h\tau_h-a_m\tau_m-a_\ell\tau_\ell.
\]
First suppose $H>0$. This is precisely the competitive-mixing regime:
the promoted mean exceeds $\mu_C$, and neither no prominence nor
universal prominence belongs to this regime. By
\eqref{app:platform:q-general}, its competitive purchase probability is
$q(H)=H/(\Delta+H)$. Since promoted purchases occur with probability
one, sales satisfy $d_i=q(H)+(1-q(H))\tau_i$, giving
\begin{equation}
 3N_c(\tau)=3mq(H)+\bigl[m(1-q(H))-c\bigr]S.
 \label{app:platform:costly-three-fixed-feedback}
\end{equation}
The coefficient of $S$ is strictly negative, including when $c=m$.
For a fixed $H$, therefore, net profit is maximized by minimizing
prominence frequency. Because $0<H\leq a_h$, the rule
$\widehat\tau=(H/a_h,0,0)$ is feasible, has the same $H$, and satisfies
\[
 S-\frac{H}{a_h}
 =\left(1+\frac{a_m}{a_h}\right)\tau_m
  +\left(1+\frac{a_\ell}{a_h}\right)\tau_\ell\geq0,
\]
with equality only if $\tau_m=\tau_\ell=0$. Thus every rule in this
regime that promotes either lower-quality product is strictly dominated
by a rule promoting only the high-quality product.

It remains to choose its prominence probability $x\in(0,1]$.
Writing $n_c(x):=N_c(x,0,0)$ and extending continuously to $n_c(0)=0$,
\begin{align}
 3n_c(x)&=\frac{m(\Delta+3a_h)x}{\Delta+a_hx}-cx,\notag\\
 3n_c'(x)&=\frac{m\Delta(\Delta+3a_h)}{(\Delta+a_hx)^2}-c,\notag\\
 3n_c''(x)&=-\frac{2ma_h\Delta(\Delta+3a_h)}{(\Delta+a_hx)^3}<0.
 \label{app:platform:costly-three-high-only}
\end{align}
Since $a_h/(\Delta+a_h)=q_3$,
\[
 3n_c'(1)=m(1-q_3)(1+2q_3)-c=\overline c_3-c.
\]
Strict concavity implies $n_c'(x)>n_c'(1)\geq0$ for $x<1$ whenever
$c\leq\overline c_3$. Hence $x=1$ is the unique maximizer throughout
\eqref{app:platform:costly-three-range}, including its upper endpoint.

To verify global optimality, consider the remaining purchase regimes.
For $0<T<1$ and promoted mean at most $\mu_C$,
Lemma~\ref{app:platform:limits} gives sales $0$, $\tau/2$, or $\tau$.
Gross profit is consequently at most $mT$, and $N_c\leq(m-c)T\leq0$.
This also covers the exact outside-option boundary. No prominence has
net profit zero, while universal prominence has zero limiting sales
and net profit $-c$, since $\bar\mu<\mu_O$.
High-only net profit, by contrast, is strictly positive over the entire
stated interval: at its upper endpoint,
\[
 3n_{\overline c_3}(1)
 =m(1+2q_3)-\overline c_3=mq_3(1+2q_3)>0.
\]
The preceding strict comparisons establish uniqueness among all
stationary rules without invoking the frequency tie-break.
If $c>\overline c_3$, then $n_c'(1)<0$, so reducing $x$ slightly below
one strictly improves net profit. This proves the last assertion.

Finally, Lemma~\ref{app:platform:limits} gives the high-only sales
$(1,q_3,q_3)$ and brand valuation $\mu_C$. Its frequency is $1/3$,
yielding the net profit in \eqref{app:platform:costly-three-values}.
Equation~\eqref{app:platform:costly-welfare} gives welfare $\mu_C$,
equal to its no-prominence level. The costless deterministic comparison
follows from
$3(\Pi_{hm}^\infty-\Pi_h^\infty)=m(1-2q_3)>0$ and
Proposition~\ref{app:platform:three}.
\end{proof}

\subsection{The comparative prediction about experience}

The main-text prediction concerns the direction of the learning response,
not an additional equilibrium-selection claim. Conditional on the same
pre-update valuation, visit count, and selected class $H$, two realized
rewards $r'>r$ produce next valuations differing by
$(r'-r)/(\kappa_H+N(H)+1)>0$. Holding rival valuations and the subsequent
menu fixed, the logit probability of choosing $H$ is increasing in this
valuation (strictly for $\beta>0$). Thus a more favorable brand experience
raises its next choice probability in a competitive encounter as well as
in a promoted one. The policy and equilibrium results above establish how
stationary prominence can sustain this transmission across product states;
they do not assert that every stochastic sample path is monotone after
an experience.

\section{Bayesian learning and coarse misspecification}
\label{app:fls-connection}

This section connects an argmax counterpart of harmonic coarse learning
to \citet{Lanzani}, hereafter FLS. We impose finite reward support and
singleton support, but allow different reward
supports across class--menu pairs. The resulting Bayesian embedding is
exact along every realized history. Its stability conclusions concern
argmax learning and suitably concentrated initial beliefs; they are
distinct from the fixed-finite-sensitivity logit results proved above.

\subsection{A finite-outcome Bayesian embedding}

Let $Y\subset\R$ be a finite reward set and let
$\mathsf R_{s\omega}\in\Delta(Y)$ be the conditional reward law after
selecting $s\in\omega$, obtained by integrating out the fixed within-class
implementation. Conditional reward draws have these stationary laws and
are independent across dates. As elsewhere, $\Omega^+$ is the finite
set of supported class menus, $p(\omega)>0$ for $\omega\in\Omega^+$,
and menus are iid. Throughout this section assume
\begin{equation}
 p(\{s\})>0\qquad(s\in\mathcal S).
 \label{app:fls:singletons}
\end{equation}
The decision maker's subjective parameter is
$q=(q_s)_{s\in\mathcal S}\in Q:=\prod_s\Delta(Y)$. She regards
rewards from class $s$ as iid with law $q_s$, independently of the
realized menu, and has a product-Dirichlet prior
\begin{equation}
 \nu=\bigotimes_{s\in\mathcal S}
       \operatorname{Dirichlet}((a_{sy})_{y\in Y}),
 \qquad a_{sy}>0.
 \label{app:fls:prior}
\end{equation}
Put $\kappa_s:=\sum_y a_{sy}$ and
$v_0(s):=\kappa_s^{-1}\sum_y ya_{sy}$.\footnote{Any finite initial
valuations and positive pseudo-counts can be represented after enlarging
$Y$, if necessary, by two hypothetical rewards strictly bracketing those
valuations and the actual reward support. Assigning these additional
rewards zero objective probability leaves the observed process
unchanged and is allowed by the weaker support condition used below.}
If $N_k(s,y)$ counts
observations of reward $y$ from class $s$ strictly before date $k$,
the posterior parameter is $(a_{sy}+N_k(s,y))_y$, and its mean payoff is
\begin{equation}
 v_k(s)=
 \frac{\kappa_sv_0(s)+\sum_{t<k}\mathbf1\{s_t=s\}r_t}
      {\kappa_s+N_k(s)}.
 \label{app:fls:harmonic}
\end{equation}
Thus Bayesian updating coincides exactly with the maintained harmonic
recursion. Here choices maximize posterior mean payoff in the current
menu, with a fixed measurable tie-breaking rule, rather than using logit.
To embed this process in the finite-action model of FLS, define the finite
set of feasible menu plans
\[
 \Phi:=\{\phi:\Omega^+\to\mathcal S:
                 \phi(\omega)\in\omega\ \text{for all }\omega\}.
\]
A plan is an auxiliary action, and its observed outcome is only the
realized pair $(\omega,y)$. Fix any full-support subjective menu law
$\widetilde p\in\Delta(\Omega^+)$. The objective and subjective
outcome laws associated with plan $\phi$ are
\begin{equation}
 P_\phi^*(\omega,y)
   =p(\omega)\mathsf R_{\phi(\omega),\omega}(y),
 \qquad
 P_\phi^q(\omega,y)
   =\widetilde p(\omega)q_{\phi(\omega)}(y),
 \label{app:fls:plan-laws}
\end{equation}
and utility is $y$. Because every $\widetilde p(\omega)$ is positive,
maximizing a plan's expected payoff is equivalent to maximizing the
posterior mean separately in each menu. Selecting the counterfactual
plan before the menu is drawn therefore implements the same myopic
choices as selecting after observing it. The likelihood factors
$\widetilde p(\omega)$ cancel when updating $q$, so neither choices
nor class valuations require knowledge of the true menu probabilities.
Plans are an analytical device: only one distribution per class is
learned, and no counterfactual rewards are observed.

For any feasible behavioral rule $\sigma=(\sigma(s\mid\omega))$,
define
\begin{align}
 D_s(\sigma)&:=\sum_{\omega\ni s}p(\omega)\sigma(s\mid\omega),
 &q_s^\sigma&:=
  \frac{\sum_{\omega\ni s}p(\omega)\sigma(s\mid\omega)
                       \mathsf R_{s\omega}}{D_s(\sigma)},
 \notag\\
 v_s^\sigma&:=\sum_{y\in Y}yq_s^\sigma(y).
 \label{app:fls:pooled-law}
\end{align}
Singleton support gives $D_s(\sigma)\geq p(\{s\})>0$. The pair
$(v^\sigma,\sigma)$ is a VE exactly when
\begin{equation}
 \operatorname{supp}\sigma(\cdot\mid\omega)
 \subseteq\arg\max_{s\in\omega}v_s^\sigma
 \qquad(\omega\in\Omega^+).
 \label{app:fls:ve}
\end{equation}
For a pure plan $\phi$, write $q^\phi$ and $v^\phi$ for these objects
under its deterministic behavioral rule. Call it \emph{menu-wise strict}
when $v_{\phi(\omega)}^\phi>v_s^\phi$ for every supported menu and
every available $s\neq\phi(\omega)$. No ordering between winners that
never compete is needed here; in particular, the strict VE definition
used earlier implies the strictness required below.

\begin{proposition}[Argmax learning and Berk--Nash equilibrium]
\label{prop:fls-connection}
Under the finite-outcome assumptions and
\eqref{app:fls:singletons}--\eqref{app:fls:prior}:
\begin{enumerate}[label=(\roman*),leftmargin=*]
\item A behavioral rule $\sigma$ satisfies \eqref{app:fls:ve} if and
only if any distribution over feasible plans with menu-wise marginals
$\sigma$ is a mixed Berk--Nash equilibrium of the embedded problem.
Every such equilibrium is \emph{strongly uniform}: every plan in its
support is optimal at every KL-minimizing model. In particular, any
unperturbed argmax plan that is eventually followed with positive
probability must be a VE plan.

\item If $\phi$ is a menu-wise strict VE, then for every
$\epsilon\in(0,1)$ there is $\delta>0$ such that every
product-Dirichlet prior of the form \eqref{app:fls:prior} satisfying
\begin{equation}
 \nu\!\left(\max_s\|q_s-q_s^\phi\|_\infty<\delta\right)>1-\delta
 \label{app:fls:concentration}
\end{equation}
induces, with probability greater than $1-\epsilon$, eventual adherence
to $\phi$ in every menu and $v_k\to v^\phi$.

\item Let $(v^\sigma,\sigma)$ be any VE. For every $\tau>0$ there
is a finite iid recommendation signal and a known payoff bonus taking
values in $\{0,\tau\}$ such that the following holds. For every
$\epsilon\in(0,1)$ there is $\delta>0$ for which every prior
\eqref{app:fls:prior} satisfying \eqref{app:fls:concentration} with
$q^\sigma$ in place of $q^\phi$ induces, with probability greater
than $1-\epsilon$, convergence $v_k\to v^\sigma$ and
\begin{equation}
 \frac{\sum_{k<K}\mathbf1\{\omega_k=\omega,\ s_k=s\}}
      {\sum_{k<K}\mathbf1\{\omega_k=\omega\}}
 \longrightarrow\sigma(s\mid\omega)
 \quad\text{for every }\omega\in\Omega^+\text{ and }s\in\omega.
 \label{app:fls:mixed-frequencies}
\end{equation}
Valuations in this assertion continue to update the base reward $y$;
the known bonus is separately included in choice utility.
\end{enumerate}
The concentration neighborhoods in (ii) and (iii) contain positive
product-Dirichlet priors. In (iii), the required neighborhood may depend
on the perturbation size $\tau$.
\end{proposition}

\begin{proof}
\emph{Regularity and exact updating.}
For a realized history consistent with the fixed policy, its likelihood
as a function of $q$ is
\[
 \prod_{t<k}\widetilde p(\omega_t)q_{s_t}(r_t).
\]
Its factors involving $q$ separate by class, giving Dirichlet conjugacy
and \eqref{app:fls:harmonic}. The subjective model family is the image
$\Theta=\{(P_\phi^q)_{\phi\in\Phi}:q\in Q\}$, with the prior
induced by \eqref{app:fls:prior}. Singleton menus make the parameterization
injective, since every plan chooses $s$ in $\{s\}$ and
$\widetilde p(\{s\})>0$.

We verify the weaker support condition stated in Section~B.2 of the
FLS supplemental material: every model in the prior support can be
approximated by models assigning positive probability to every
objectively possible outcome. For $\eta>0$, replace every $q_s$ by
$(1-\eta)q_s+\eta\operatorname{Unif}(Y)$. These models belong to
$\Theta$, converge to the original model, and assign positive
probability to every pair $(\omega,y)$. This is FLS Assumption~1(i$'$).
It permits zero probabilities in the original model and in the objective
reward laws. As established there, the results used here remain valid
under this condition; the limit-point theorem is invoked only for the
myopic case, with economic discount factor zero.

For completeness, the Dirichlet prior also satisfies the required
subexponential thickness condition. Let $m=|Y|\geq2$ and
$0<r\leq1$. For every $q_s^0\in\Delta(Y)$, a relative simplex box
around
$(1-r/2)q_s^0+(r/2)\operatorname{Unif}(Y)$ has volume at least
$C_mr^{m-1}$, lies within sup-norm distance $r$ of $q_s^0$, and has
all coordinates at least $r/(4m)$; the constants can be chosen uniformly
in $q_s^0$. On this box the Dirichlet density is bounded below by a
positive constant times $r^{\sum_y(a_{sy}-1)_+}$. Consequently the
product prior assigns at least $Cr^d$ to every radius-$r$ parameter
ball, where
\[
 d=|\mathcal S|(m-1)+\sum_{s,y}(a_{sy}-1)_+.
\]
The map $q\mapsto(P_\phi^q)_\phi$ is $1$-Lipschitz in the sup
norm, so model-law balls inherit this lower bound. In particular,
$C(K/n)^d e^n\to\infty$ for every $K>0$, as required by FLS
Assumption~1(ii). The case $m=1$ is immediate. Thus the embedding
satisfies the hypotheses of the FLS results being applied, without
truncating the simplex or requiring common objective reward support.

\emph{KL minimization and part (i).}
Let $\alpha\in\Delta(\Phi)$ have marginals $\sigma$. Such a
distribution always exists; for example,
\begin{equation}
 \alpha(\phi)=\prod_{\omega\in\Omega^+}
                         \sigma(\phi(\omega)\mid\omega).
 \label{app:fls:plan-mixture}
\end{equation}
With the usual extended-value convention for KL divergence,
\begin{align}
 \sum_\phi\alpha(\phi)D_{\rm KL}(P_\phi^*\|P_\phi^q)
 &=D_{\rm KL}(p\|\widetilde p)
   +\sum_{\omega,s\in\omega}p(\omega)\sigma(s\mid\omega)
                      D_{\rm KL}(\mathsf R_{s\omega}\|q_s)
 \notag\\
 &=C_\sigma+\sum_sD_s(\sigma)
                         D_{\rm KL}(q_s^\sigma\|q_s),
 \label{app:fls:kl}
\end{align}
where $C_\sigma$ is finite and independent of $q$. Every coefficient
$D_s(\sigma)$ is positive, so the unique minimizer is $q^\sigma$.
This conclusion remains valid when its coordinates lie on simplex
boundaries. All beliefs supported on the minimizer set therefore induce
the same valuation vector $v^\sigma$.

A plan is optimal at these beliefs exactly when it selects a maximizing
valuation in every supported menu, because every subjective menu weight
is positive. Hence every plan in the support of $\alpha$ is optimal if
and only if \eqref{app:fls:ve} holds. Uniqueness of the KL minimizer
also makes the equilibrium strongly uniform in the sense of FLS
Definition~13. For pure plans, their Theorem~1 then implies the final
claim in (i): an eventual limit plan must be a uniform Berk--Nash
equilibrium, hence a VE plan in this embedding. This does not assert
convergence of the argmax process from arbitrary initial beliefs.

\emph{Strict stability, part (ii).}
When $\phi$ is menu-wise strict, it is the unique optimal plan at
$q^\phi$ and therefore a uniformly strict Berk--Nash equilibrium.
FLS Theorem~2 gives, for each $\epsilon>0$, a concentration neighborhood
of its KL-minimizing outcome law such that every admissible initial belief
in that neighborhood generates eventual adherence to $\phi$ with
probability greater than $1-\epsilon$. Since the parameter-to-law map
is continuous, choosing $\delta$ sufficiently small in
\eqref{app:fls:concentration} places the induced prior in that
neighborhood. The result applies to every fixed measurable optimal
tie-breaking rule.

On the adherence event, all observations after a finite date are generated
by the fixed plan $\phi$. The strong law for these iid menu--reward
draws, together with $D_s(\phi)>0$, makes each class's empirical reward
distribution converge to $q_s^\phi$. Finite initial histories and prior
pseudocounts have vanishing weight, so \eqref{app:fls:harmonic} gives
$v_k\to v^\phi$. These strong-law statements hold on an event of
probability one and hence also on the positive-probability adherence
event; no independence conditional on that event is being assumed.

\emph{Purification by finite recommendations, part (iii).}
Fix the distribution $\alpha$ in \eqref{app:fls:plan-mixture}, and
draw a recommendation $z\in Z:=\operatorname{supp}\alpha$ iid
with law $\alpha$, independently of the menu and reward draws. The
recommended class in menu $\omega$ is $z(\omega)$. Give class $s$
the known additional utility $\tau\mathbf1\{s=z(\omega)\}$,
so the decision maker chooses
\begin{equation}
 s_k\in\arg\max_{s\in\omega_k}
       \{v_k(s)+\tau\mathbf1\{s=z_k(\omega_k)\}\}.
 \label{app:fls:perturbed-choice}
\end{equation}
The reward entered in \eqref{app:fls:harmonic} is the base reward
$r_k=y_k$. Equivalently, the decision maker subtracts the known bonus
from total realized utility before updating the unknown class payoff.

Apply the finite-action embedding again, now with actions all maps
$F:Z\times\Omega^+\to\mathcal S$ satisfying $F(z,\omega)\in\omega$.
An observed outcome is $(z,\omega,y)$, with laws
\begin{align*}
 P_F^*(z,\omega,y)&=\alpha(z)p(\omega)
                         \mathsf R_{F(z,\omega),\omega}(y),\\
 P_F^q(z,\omega,y)&=\alpha(z)\widetilde p(\omega)
                         q_{F(z,\omega)}(y).
\end{align*}
The utility of this outcome is
$y+\tau\mathbf1\{F(z,\omega)=z(\omega)\}$. Myopic optimization
over $F$ is separable across recommendation--menu pairs, so this
ex ante construction implements exactly \eqref{app:fls:perturbed-choice}.
The posterior for $q$ is unchanged by the independent known signal
probabilities and still obeys \eqref{app:fls:harmonic}.

Consider $F^*(z,\omega)=z(\omega)$. Its induced menu-wise
frequencies are $\sigma$, so \eqref{app:fls:kl} gives its unique
KL minimizer as $q^\sigma$. At this model, the recommended class
already maximizes base expected payoff by \eqref{app:fls:ve}; its bonus
then makes it uniquely optimal. Every other contingent plan differs at
some pair $(z,\omega)$ and loses expected utility by at least
\[
 \tau\min_{z\in Z}\alpha(z)
       \min_{\omega\in\Omega^+}\widetilde p(\omega)>0.
\]
Thus $F^*$ is uniformly strict. The dense-positivity and prior-thickness
arguments above apply unchanged to this finite enlarged outcome space.
FLS Theorem~2 yields, from sufficiently concentrated product-Dirichlet
priors, eventual adherence to $F^*$ with probability greater than
$1-\epsilon$.

On that event, the strong law for the iid pairs $(z_k,\omega_k)$
gives \eqref{app:fls:mixed-frequencies}, and the class-specific reward
strong law gives $v_k\to v^\sigma$. Singleton support again makes
every class's limiting sampling frequency positive. The conclusion
holds for each fixed $\tau>0$; no exchange of the long-run and
vanishing-perturbation limits is used.

Finally, the concentration neighborhoods in (ii) and (iii) are nonempty
within the stipulated prior family. Given the relevant target $q^*$,
fix $b_y>0$ and take $a_{sy}=Lq_s^*(y)+b_y$. As $L\to\infty$,
the Dirichlet means converge to $q_s^*$ and their variances vanish,
so their product concentrates at $q^*$. All parameters remain strictly
positive, even if the target distribution has zero coordinates.
\end{proof}

\paragraph{Scope of the connection.}
The stability neighborhood controls posterior concentration, not only
initial posterior mean valuations. Part (ii) is the strict-stability
implication needed here; no converse is asserted within the restricted
family of product-Dirichlet priors. Part (iii) is a finite-contingent-plan
application of the same FLS theorem, in the spirit of their purification
result for strongly uniform mixed equilibria. It constructs signals and
bonuses specifically for the target VE, and its admissible neighborhood
may shrink as $\tau\downarrow0$. It establishes neither stability of
that mixed VE under unperturbed argmax learning nor approximation under
the particular logit family. The logit refinement, local ODE stability,
and global convergence results of the main analysis therefore remain
separate conclusions. The finite-outcome assumption here is also
additional to the maintained moment assumptions used elsewhere.

\section{Convergence in balanced binary constant-sum decision trees}
\label{app:balanced}

Suppose every supported menu is binary, write $p_{sr}=p_{rs}$ for the
probability of $\{s,r\}$, and suppose its payoffs sum to $2c$, independently
of the pair. Center rewards and valuations at $c$ and write
\[
 \pi_s(\{s,r\})-c=a_{sr}=-a_{rs},\qquad
 B_{sr}=p_{sr}a_{sr},\qquad P_s=\sum_rp_{sr}>0.
\]
Uniform menu-wise choice gives every class the same conditional mean
precisely when $\sum_rB_{sr}=0$ for every $s$. Indeed the common mean must
be $c$, since the ex ante payoff of uniform choice is $c$ and
$\sum_sP_s/2=1$. In what follows valuations are centered, so the candidate
SVE is $v=0$.

\begin{proposition}[Global convergence under harmonic visits]
\label{app:balanced:global}
At every finite sensitivity the unique SVE is $v=0$, with
$x_s=P_s/2$. It is globally asymptotically stable under the joint dynamics.
Harmonic visit-count learning converges to it almost surely under the
maintained second-moment assumption. The normalized equilibrium is
locally exponentially stable at every finite sensitivity, but normalized
global convergence need not hold.
\end{proposition}
\begin{proof}
For $\beta>0$, let $t_s=e^{\beta v_s}$. The centered unconditional reward
contribution is $R_s(v)=\sum_rB_{sr}t_s/(t_s+t_r)$. Antisymmetry and the
zero row sums give the exact cancellation
\begin{align}
 \sum_s\sinh(\beta v_s)R_s(v)
 &=\sum_{s<r}B_{sr}
       \frac{\sinh(\beta v_s)t_s-\sinh(\beta v_r)t_r}{t_s+t_r}\nonumber\\
 &=\frac12\sum_{s<r}B_{sr}(t_s-t_r)=0.
 \label{app:balanced:cancellation}
\end{align}
Consider the nonnegative function
\[
 L(v,x)=\sum_s x_s\bigl(\cosh(\beta v_s)-1\bigr).
\]
Using $\dot v_s=(R_s-\Lambda_s v_s)/x_s$ and
$\dot x_s=\Lambda_s-x_s$, its derivative is
\begin{align*}
 \dot L
 &=-L+\sum_s\Lambda_s(v)
       \{\cosh(\beta v_s)-1-\beta v_s\sinh(\beta v_s)\}\
 &\leq-L.
\end{align*}
The inequality follows because $\cosh u-1-u\sinh u\leq0$: its derivative
is $-u\cosh u$, and its value at zero is zero. Thus $L(t)\leq e^{-t}L(0)$.
Valuations are bounded and frequencies stay uniformly positive along
each trajectory. Since $\cosh u-1\geq u^2/2$, this proves
$v(t)=O(e^{-t/2})$. The frequency filter then gives
$x(t)-P/2=O(e^{-t/2})$. The same estimates in a fixed neighborhood prove
Lyapunov stability. At any equilibrium $\dot L=0$ forces $v=0$, proving
uniqueness. If $\beta=0$, centered $g$ is zero and the constant-target
argument applies. Almost-sure convergence follows from the mean-field
approximation and the globally asymptotically stable equilibrium.

To check normalized local stability, differentiation at zero gives
\[
 Dg(0)=-\frac\beta2Q^{-1}B,\qquad Q=\operatorname{diag}(P_s).
\]
The diagonal terms vanish because $B\mathbf1=0$. The matrix
$Q^{-1}B$ is similar to the real skew-symmetric matrix
$Q^{-1/2}BQ^{-1/2}$. Hence all eigenvalues of $Dg(0)-I$ have real part
$-1$. The joint spectral correspondence yields the same nontrivial
eigenvalues and the additional stable frequency directions. The following
explicit example proves the last assertion.
\end{proof}

\subsection{A normalized periodic orbit despite a stable unique SVE}
Consider three binary menus with centered payoffs and probabilities
\[
\begin{array}{c|c|cc}
 \text{menu}&\text{probability}&\text{first payoff}&\text{second payoff}\\\hline
 \{1,2\}&1/28&1&-1\\
 \{1,3\}&20/28&-1/20&1/20\\
 \{2,3\}&7/28&1/7&-1/7
\end{array}
\]
Here $B_{12}=B_{23}=-B_{13}=1/28$, so $B\mathbf1=0$.
The unique equilibrium has normalized eigenvalues
$-1,-1+i\beta/18,-1-i\beta/18$.
We next prove analytically that, for every sufficiently large finite
$\beta$, a nonconstant periodic orbit also exists.

Set $X=\beta(v_1-v_3)$, $Y=\beta(v_2-v_3)$, $a=e^X$, and $b=e^Y$.
The three conditional means are
\[
 g_1=\frac{1-b}{21a+20b+1},\qquad
 g_2=\frac{a-1}{7a+8b+1},\qquad
 g_3=\frac{b-a}{7a+20b+27}.
\]
With time $\tau=\beta t$ and $\epsilon=1/\beta$, the gap system is
\[
 (X,Y)'=H(X,Y)-\epsilon(X,Y),\qquad
 H=(g_1-g_3,g_2-g_3).
\]
Its Jacobian at zero, before subtracting $\epsilon I$, is
\[
 DH(0)=\begin{pmatrix}1/54&-8/189\\35/432&-1/54\end{pmatrix}.
\]
The change $u=X$, $w=-X/3+16Y/21$ makes the linear part
$(u',w')=(-w/18,u/18)$ at $\epsilon=0$. Write its nonlinear terms as
$(f,g)$. Direct differentiation at the origin gives
\begin{align*}
 (f_{uu},f_{uw},f_{ww})&=(11/576,11/1728,-53/5184),\\
 (g_{uu},g_{uw},g_{ww})&=(-25/1728,-23/576,-5/1728),\\
 (f_{uuu},f_{uww},g_{uuw},g_{www})
   &=(-7/3072,-757/82944,77/3072,-145/82944).
\end{align*}
With $\omega=1/18$, the cubic radial coefficient is
\begin{align*}
 a_1={}&\frac{f_{uuu}+f_{uww}+g_{uuw}+g_{www}}{16}\\
 &+\frac{f_{uw}(f_{uu}+f_{ww})-g_{uw}(g_{uu}+g_{ww})
               -f_{uu}g_{uu}+f_{ww}g_{ww}}{16\omega}
 =\frac{247}{663552}>0.
\end{align*}
The eigenvalues cross the imaginary axis with real part $-\epsilon$.
The nondegenerate subcritical Hopf theorem therefore produces a small
repelling periodic gap orbit for every sufficiently small
$\epsilon>0$. Equivalently it exists for all sufficiently large finite
$\beta$. The remaining normalized equation is the stable linear filter
$\dot v_3=g_3(X(t),Y(t))-v_3$. A periodic forcing has exactly one periodic
solution, so the gap orbit lifts to a periodic orbit in all three
valuations. This establishes failure of normalized global convergence.


\end{document}